\documentclass[a4paper,UKenglish, thm-restate, cleveref]{lipics-v2021}
\hideLIPIcs  %

\title{Dynamic Planar Graph Isomorphism is in DynFO}

\author{Samir Datta}{Chennai Mathematical Institute \& UMI ReLaX, Chennai, India}{sdatta@cmi.ac.in}{https://orcid.org/0000-0003-2196-2308}{}
\author{Asif Khan}{Chennai Mathematical Institute, Chennai, India \and Indian Institute of Technology Delhi, New Delhi, India}{asifkhan@cmi.ac.in}{https://orcid.org/0009-0001-5950-8891}{}
\author{Felix Tschirbs}{Ruhr University Bochum, Bochum, Germany}{felix.tschirbs@rub.de}{https://orcid.org/0009-0005-7824-6507}{}
\author{Nils Vortmeier}{Ruhr University Bochum, Bochum, Germany}{nils.vortmeier@rub.de}{https://orcid.org/0009-0000-2821-7365}{}
\author{Thomas Zeume}{Ruhr University Bochum, Bochum, Germany}{thomas.zeume@rub.de}{https://orcid.org/0000-0002-5186-7507}{}

\authorrunning{S. Datta and A. Khan and F. Tschirbs and N. Vortmeier and T. Zeume} 

\Copyright{Samir Datta and Asif Khan and Felix Tschirbs and Nils Vortmeier and Thomas Zeume} 


\relatedversion{This article is the full version of a paper that appeared
	at LICS 2026.}
\relatedversiondetails[cite=DattaKTVZ26]{Conference Version}{https://doi.org/10.4230/LIPIcs.LICS.2026.36}

\nolinenumbers %

\usepackage{amssymb,amsmath}
\usepackage{amsthm}
\usepackage{thmtools}

\usepackage{regexpatch}
\makeatletter
\xpatchcmd\thmt@restatable{%
	\csname #2\@xa\endcsname\ifx\@nx#1\@nx\else[{#1}]\fi
}{%
	\ifthmt@thisistheone
	\csname #2\@xa\endcsname\ifx\@nx#1\@nx\else[{#1}]\fi
	\else
	\csname #2\@xa\endcsname[{Restated}]
	\fi}{}{}
\makeatother

\usepackage{xspace}
\usepackage{enumerate}
\usepackage{tikz}
\usetikzlibrary{arrows,arrows.meta, decorations.pathmorphing,backgrounds,calc,positioning,fit,petri,matrix,backgrounds, decorations.pathreplacing, shapes.geometric,intersections,patterns}
\usepackage{soul}
\usepackage{subcaption}
\usepackage[above]{placeins}

\newif\ifcomments
\newif\ifchanges
\commentsfalse\changesfalse
\newcommand{\operator}[1]{\textsf{\textup{#1}}}

\newcommand{\biTree}{\operator{bi-tree}}
\newcommand{\triTree}{\operator{tri-tree}}
\newcommand{\mat}[1]{\mathbf{#1}}
\newcommand{\M}{\mathcal{M}}
\newcommand{\transpose}{\mathsf{T}}

\newcommand{\X}{\ensuremath{\text{X}}}
\newcommand{\rcX}{\ensuremath{\text{rcX}}}

\newcommand{\ST}{\ensuremath{\text{ST}}}
\newcommand{\rcST}{\ensuremath{\text{rcST}}}

\newcommand{\graph}{\ensuremath{\text{graph}}}
\newcommand{\cgraph}{\text{col-graph}}
\newcommand{\Disc}{\operator{Disc}}

\newcommand{\discDist}{\RelName{DiscDistance}}
\newcommand{\combEmb}{\RelName{CombEmbedding}}
\newcommand{\Emb}{\operator{Emb}}
\newcommand{\cycleDist}{\RelName{CycleDistance}}

\newcommand{\RelName}[1]{\mtext{#1}}
\newcommand{\tpl}{\bar}

\newcommand{\mtext}[1]{\textsc{#1}}
\newcommand{\dist}{\mtext{dist}\xspace}

\newcommand{\ins}{\mtext{ins}\xspace}
\newcommand{\del}{\mtext{del}\xspace}

\newcommand{\schema}{\ensuremath{\sigma}\xspace}

\newcommand{\query}{\ensuremath{Q}}
\newcommand{\inp}{\ensuremath{\calI}\xspace}
\newcommand{\aux}{\ensuremath{\calA}\xspace}%

\newcommand{\N}{\ensuremath{\mathbb{N}}}

\newcommand{\Q}{\ensuremath{\mathbb{Q}}}
\newcommand{\Z}{\ensuremath{\mathbb{Z}}}
\newcommand{\bigO}{\ensuremath{\mathcal{O}}}

\newcommand{\df}{\ensuremath{\mathrel{\smash{\stackrel{\scriptscriptstyle{
    \text{def}}}{=}}}} \;}

\DeclareMathOperator{\polylog}{polylog}

\newcommand  {\myclass} [1]  {\ensuremath{\textsf{\upshape #1}}}

\newcommand{\StaClass}[1]{\myclass{#1}\xspace}

\newcommand{\DynClass}[1]{\myclass{Dyn#1}\xspace}

\newcommand     {\LOGSPACE}     {\StaClass{LOGSPACE}}
\newcommand     {\NL}   {\StaClass{NL}}
\renewcommand   {\P}    {\StaClass{P}}
\newcommand     {\NP}   {\StaClass{NP}}
\newcommand     {\NC}   {\StaClass{NC}}

\newcommand     {\AC}   {\StaClass{AC}}

\newcommand{\FO}{\StaClass{FO}}

\newcommand{\FOar}{\StaClass{FO$(\leq,+,\times)$}}%

\newcommand{\CQ}[1][]{\StaClass{CQ}}
\newcommand{\UCQ}[1][]{\StaClass{UCQ}}
\newcommand{\CQneg}[1][]{\StaClass{CQ\ensuremath{^{\mneg}}}}
\newcommand{\UCQneg}[1][]{\StaClass{UCQ\ensuremath{^{\mneg}}}}

\newcommand{\mneg}{\neg} %

\newcommand{\DynFO}{\DynClass{FO}}

\theoremstyle{plain}
\newtheorem*{theorem*}{Theorem}
\theoremstyle{definition}
\newtheorem*{question*}{Question}
\newtheorem*{openquestion*}{Open question}

\newenvironment{proofof}[1]{\begin{proof}[Proof (of #1).]}{\end{proof}}

\providecommand {\calA}      {{\mathcal A}\xspace}

\providecommand {\calI}      {{\mathcal I}\xspace}

\providecommand {\calT}      {{\mathcal T}\xspace}

\newcommand{\prog}{\ensuremath{\Pi}\xspace}

\newcommand{\auxSchema}{\ensuremath{\schema_{\text{aux}}}\xspace}

\ifcomments
\newcommand{\commentbox}[1]{\noindent\framebox{\parbox{0.98\linewidth}{#1}}}

\newcommand{\acomment}[2]{\ \\ \fbox{\parbox{0.98\linewidth}{{\sc #1}: #2}}}
\newcommand{\mcomment}[2]{{\color{blue}(#1)}\footnote{#1: #2}} %
\else
\newcommand{\commentbox}[1]{}
\newcommand{\mcomment}[2]{}
\newcommand{\acomment}[2]{}
\fi

\ifchanges

\setul{}{0.2mm}
\setstcolor{red}

				\else

				\fi

				\newcommand{\tzm}[1]{\mcomment{TZ}{#1}}

				\newcommand{\akm}[1]{\mcomment{AK}{#1}}

\definecolor{iltisBeige1}{HTML}{fef6ee}
\definecolor{iltisBeige2}{HTML}{e3d5c8}
\definecolor{iltisBeige3}{HTML}{cfbeb0}
\definecolor{iltisBeige4}{HTML}{bfac9b}

\definecolor{iltisGrey1}{HTML}{edebe8}
\definecolor{iltisGrey2}{HTML}{d9d6d2}
\definecolor{iltisGrey3}{HTML}{c2bfbc}
\definecolor{iltisGrey4}{HTML}{a6a4a1}

\definecolor{iltisLightGreen1}{HTML}{f4ffe9}
\definecolor{iltisLightGreen2}{HTML}{d4f7b2}
\definecolor{iltisLightGreen3}{HTML}{bde697}
\definecolor{iltisLightGreen4}{HTML}{a3d177}

\definecolor{iltisGreen1}{HTML}{cfffe1}
\definecolor{iltisGreen2}{HTML}{a2e8bd}
\definecolor{iltisGreen3}{HTML}{86d1a2}
\definecolor{iltisGreen4}{HTML}{6fbf8d}

\definecolor{iltisYellow1}{HTML}{fef2d0}
\definecolor{iltisYellow2}{HTML}{ffe3a2}
\definecolor{iltisYellow3}{HTML}{ffd77d}
\definecolor{iltisYellow4}{HTML}{f2c55e}

\definecolor{iltisRed1}{HTML}{ffe9e6}
\definecolor{iltisRed2}{HTML}{eda498}
\definecolor{iltisRed3}{HTML}{e08475}
\definecolor{iltisRed4}{HTML}{c26e60}

\definecolor{iltisOrange1}{HTML}{ffdfb3}
\definecolor{iltisOrange2}{HTML}{ffc97d}
\definecolor{iltisOrange3}{HTML}{ebb467}
\definecolor{iltisOrange4}{HTML}{e0a34c}

\definecolor{iltisCyan1}{HTML}{e0fffe}
\definecolor{iltisCyan2}{HTML}{b4e0df}
\definecolor{iltisCyan3}{HTML}{95c7c5}
\definecolor{iltisCyan4}{HTML}{81b3b0}

\definecolor{iltisBlue1}{HTML}{cce8ff}
\definecolor{iltisBlue2}{HTML}{8db8d9}
\definecolor{iltisBlue3}{HTML}{6f9abd}
\definecolor{iltisBlue4}{HTML}{5c86a8}
\definecolor{iltisBlue5}{HTML}{2e5b80}

\definecolor{iltisViolet1}{HTML}{ede6ff}
\definecolor{iltisViolet2}{HTML}{b1a5cc}
\definecolor{iltisViolet3}{HTML}{9b8eba}
\definecolor{iltisViolet4}{HTML}{8578a6}

\colorlet{midgray}{black!50}

\tikzstyle{virtual} = [blue,line width=1.3pt, densely dotted]

\tikzstyle{Pnode} = [ellipse, draw=black, inner sep=0pt, minimum width=3mm, minimum height=7mm, fill=white]
\tikzstyle{Adot} = [circle, fill=black, inner sep=1.1pt]

\tikzstyle{dEdge}=[
  -latex', %
  thick, 
  shorten >=2pt, 
  shorten <=2pt,
  draw=black!80,
]

\tikzstyle{dhEdge}=[
  -latex', %
  thick, 
  shorten >=3pt, 
  shorten <=3pt,
  draw=black!80,
]
\tikzstyle{uEdge}=[
  thick, 
  shorten >=3pt, 
  shorten <=3pt,
  draw=black!80,
]
\tikzstyle{uhEdge}=[
  thick, 
  shorten >=3pt, 
  shorten <=3pt,
  draw=black!80,
]

\tikzstyle{cEdge}=[
  ultra thick, 
  shorten >=3pt, 
  shorten <=3pt,
  draw=black!80,
]

\tikzstyle{dotsEdge}=[
  very thick, 
  loosely dotted, 
  shorten >=3pt, 
  shorten <=3pt
]

\tikzstyle{snakeEdge}=[
  ->, 
  decorate, 
    shorten >=2pt, 
  shorten <=2pt,
  decoration={snake,amplitude=.4mm,segment length=2.5mm,post length=0.0mm},
]

\tikzstyle{snakeEdgea}=[
  ->, 
  decorate, 
  decoration={snake,amplitude=.4mm,segment length=3mm,post length=0.5mm}
]
\tikzstyle{blackNode}=[
	shape=circle, draw, fill,inner sep=-1.5pt,minimum size=2pt
]

\newcommand{\es}[3][\sigma]{[\smash{{#2}\xrightarrow{#1}{#3}}]}
\newtheorem{fact}[theorem]{Fact}
\newcommand{\ISOt}{\RelName{iso$_3$}}
\newcommand{\spqrDist}{\RelName{dist$_2$}}
\newcommand{\spqrIso}{\RelName{iso$_2$}}
\newcommand{\spqrCIso}{\RelName{x-iso$_2$}}
\newcommand{\spqrSTIso}{\RelName{st-iso$_2$}}
\newcommand{\spqrIsoW}{\RelName{iso-witness$_2$}}
\newcommand{\spqrCIsoW}{\RelName{x-iso-witness$_2$}}
\newcommand{\spqrSTIsoW}{\RelName{st-iso-witness$_2$}}
\newcommand{\spqrSibIsoCount}{\RelName{$\#$iso-siblings$_2$}}
\newcommand{\spqrSibIsoRank}{\RelName{iso-siblings-rank$_2$}}

\newcommand{\sameFace}{\RelName{same-face}}
\newcommand{\IsoW}{\RelName{iso-witness}}
\newcommand{\bcDist}{\RelName{dist$_1$}}
\newcommand{\bcIso}{\RelName{iso$_1$}}
\newcommand{\bcIsoW}{\RelName{iso-witness$_1$}}
\newcommand{\bcCIso}{\RelName{x-iso$_1$}}
\newcommand{\bcSibIsoCount}{\RelName{$\#$iso-siblings$_1$}}
\newcommand{\bcCIsoW}{\RelName{x-iso-witness$_1$}}
\newcommand{\bcSibIsoRank}{\RelName{iso-siblings-rank$_1$}}
\newcommand{\IsoCount}{\RelName{$\#$iso-components}}
\newcommand{\IsoRank}{\RelName{iso-components-rank}}

\begin{document}
	
	\maketitle

\begin{abstract}
	Consider two planar graphs which are subject to edge insertions and deletions. We show that whether the two graphs are isomorphic can be maintained with first-order logic formulas and auxiliary data of polynomial size. This places the dynamic planar graph isomorphism problem into the dynamic descriptive complexity class $\DynFO$. As a consequence, there is a dynamic constant-time parallel algorithm with polynomial-size auxiliary data which maintains whether two dynamic planar graphs are isomorphic.
\end{abstract}

\section{Introduction}\label{section:intro}

The \emph{graph isomorphism problem} asks whether two given graphs are isomorphic, i.e., whether there is a bijection between their vertices that preserves their edges. It is one of the few algorithmic problems known to be in $\NP$, but not known to be either in $\P$ or $\NP$-complete. In fact, it is only known that the problem can be solved in quasipolynomial time \cite{Babai16} and that it is at least as hard as computing determinants \cite{Toran04} (i.e., it is not even known to be hard for $\P$).

While it is open  whether the graph isomorphism problem can be solved in polynomial time for general graphs, several restricted classes of graphs have been identified for which polynomial time or even logarithmic space suffices. For instance, early on it was shown that isomorphism of trees is complete for logarithmic space \cite{Lindell92}, which was generalized in subsequent work to graphs of bounded tree-width \cite{Bodlaender90, GroheV06,DasTW12,ElberfeldS17}, planar graphs \cite{DattaLNTW22}, and graphs of bounded genus \cite{ElberfeldK14}. In particular, the graph isomorphism problem on these graph classes can be solved efficiently in parallel. It can be solved by a CRCW PRAM with $n^{O(1)}$ processors in $O(\log{n})$ time or equivalently by $\AC^1$-circuits, i.e., circuits with $\neg$-gates and unbounded fan-in $\wedge$- and $\vee$-gates of polynomial size and logarithmic depth. From a logical perspective this means that the graph isomorphism problem for these restricted graph classes can be expressed by first-order formulas of quantifier-depth $\bigO(\log n)$ \cite{BarringtonIS90}.

In this work we study the isomorphism problem for graphs that change dynamically over time, i.e., graphs where edges may be inserted or deleted. Our guiding question is whether less than $\bigO(\log n)$ parallel time  suffices to solve the graph isomorphism problem for dynamic graphs for restricted graph classes. Almost three decades ago, Etessami made a first step in this direction by proving that the dynamic graph isomorphism problem for trees can be maintained in constant parallel time by (uniform) $\AC^0$-circuits with auxiliary data of polynomial size \cite{Etessami98}. Later Mehta showed that also $3$-connected planar isomorphism can be maintained in constant parallel time, assuming the graph stays $3$-connected at all times~\cite{Mehta14}.

In this work, we generalize both results and establish that the graph isomorphism problem for planar graphs can be maintained in constant parallel time.

\begin{theorem*}[Main theorem]
  The dynamic graph isomorphism problem for planar graphs can be maintained by $\AC^0$-circuits with auxiliary data of polynomial size.
\end{theorem*}

The changing graphs are promised\footnote{It can be maintained whether an edge change destroys planarity of an input graph \cite{DattaKM23}, so one can equivalently consider the setting where changes are rejected if they would violate the promise of planarity.} to remain planar and one bit of the auxiliary data indicates whether the graphs are isomorphic after each change.
Additionally, if the graphs are isomorphic, a witnessing isomorphism is part of the maintained auxiliary data.

For proving the result, like Etessami and Mehta, we work in the dynamic descriptive complexity framework introduced by Patnaik and Immerman \cite{PatnaikI94, PatnaikI97}. In this framework, the input graphs as well as the auxiliary data are represented by relational structures. After each change, all auxiliary relations are updated by first-order formulas with the changed edge and a tuple of nodes as parameters. The tuple of nodes will be included into an updated auxiliary relation if the formula evaluates to true in the relational auxiliary structure before the change. The problems maintainable in this way by first-order formulas constitute the dynamic descriptive complexity class $\DynFO$ and coincide with the problems maintainable by $\AC^0$-circuits with auxiliary data of polynomial size (due to the correspondence of $\AC^0$ and first-order logic, see \cite{BarringtonIS90}). In their proofs, Etessami and  Mehta\footnote{Mehta showed membership in the class $\DynFO^+$ that allows for polynomial-time initialization of the auxiliary relations.} established that dynamic tree isomorphism and dynamic $3$-connected planar isomorphism are in $\DynFO$.

\begin{theorem*}[Main theorem, rephrased]
  The dynamic graph isomorphism problem for planar graphs is in $\DynFO$.
\end{theorem*}

This result continues the exploration of the power of first-order logic as an update mechanism for dynamic problems, which intensified with the positive resolution of a conjecture by Patnaik and Immerman that directed graph reachability is in $\DynFO$ \cite{DKMSZ18}. This and further results of the last decade --  including, for instance, that MSO-definable algorithmic problems for graphs of bounded treewidth \cite{DMSVZ19} as well as testing planarity and maintaining plane embeddings \cite{DattaKM23} are in $\DynFO$ -- are proof that $\DynFO$, and dynamic constant-parallel time, respectively, is surprisingly powerful. However, while $\DynFO$ contains \NL-complete problems like directed graph reachability and artificial \P-complete problems \cite{PatnaikI97}, there are also simple queries like ``Is the number of nodes that have an edge to a red node odd?'', which are definable in first-order logic extended by parity quantifiers but not known to be in $\DynFO$~\cite{VortmeierZ21}.

\subparagraph*{Outline.} We recall basic notions for graphs and dynamic complexity theory in Section~\ref{section:preliminaries}. After outlining high-level proof ideas in Section \ref{section:proof-overview} and discussing the structural impact of edge changes on planar graphs in Section~\ref{section:changes}, we provide the detailed proof of our main result in four stages. Section~\ref{section:componenttrees} deals with the maintenance of biconnected and triconnected component trees, together with necessary distance information. In Section~\ref{section:triconnected}, Section~\ref{section:biconnected} and Section~\ref{section:connected}, respectively, we show how to maintain whether $3$-connected, $2$-connected and connected components of planar graphs are isomorphic. Section~\ref{section:finally} finalizes the proof of our main result for possibly disconnected planar graphs.

\section{Preliminaries and Notations}\label{section:preliminaries}
We recall elementary notions from graph theory based on \cite{DattaLNTW22}.

\subparagraph*{Basic graph notions}
In this work, all graphs $G = (V, E)$, where $V$ is the set of vertices (also called nodes) and $E$ is the set of edges, are undirected. For a vertex $v \in V$, the neighbourhood $N(v) = \{u \mid (v,u) \in E\}$ contains all vertices adjacent to $v$. The graph induced by a subset  $X\subseteq V(G)$ is denoted by $G[X]$ and $G-X$ denotes the graph $G[V(G)\setminus X]$. Two graphs $G = (V, E)$ and $G^* = (V^*, E^*)$ are \emph{isomorphic}, denoted by $G\simeq G^*$, if there is an isomorphism from $G$ to $G^*$, i.e., a bijection $\pi:V\to V^*$ such that $(u,v) \in E$ if and only if $(\pi(u),\pi(v))\in E^*$ for all $u,v\in V$. %

A graph is \emph{planar} if it can be drawn on the plane such that edges do not cross except at their endpoints. Connected regions in a drawing of a planar graph are called \emph{faces} \cite{Diestel12}.%

A \emph{(undirected) tree} is a graph in which there is exactly one path between any pair of nodes; a forest is a disjoint union of trees. Given a tree $T$, we denote by $\ST(x,r)$ the subtree rooted at $r$, in the tree $T$ rooted at $x$. 
A \emph{(rooted) context} $X = (V, E, r, h)$ is a tree $(V, E)$ with root $r \in V$ and one distinguished leaf $h \in V$, called the \emph{hole}. Two contexts $X = (V, E, r, h)$, $X^* = (V^*, E^*, r^*, h^*)$ are isomorphic if there is a root- and hole-preserving isomorphism between them, i.e., an isomorphism that maps $r$ to $r^*$ and $h$ to $h^*$. For a forest $F$ and nodes $x$, $r$, $h$ occurring in this order on some path, the context  $\X(x, r, h)$ is defined as the context we obtain by taking $\ST(x,r)$, removing all children of $h$, and taking $r$ as root and $h$ as hole.

\subparagraph*{Connectivity and decomposition trees} We now recall basic notions related to connectivity, roughly following \cite{DattaLNTW22}. A graph is \emph{connected} if there is a path between any two of its vertices.
A subset $S \subseteq V$ with $|S| = k$ is a \emph{$k$-separating} set if $G - S$ is not connected. Such a set $S$ is a \emph{separating vertex} or \emph{cut vertex} if $k = 1$ and a \emph{separating pair} if $k = 2$. A graph $G$ is $k$-connected if it has no $(k-1)$-separating set.  If a graph is $2$-connected, we also call it \emph{biconnected}.  

A biconnected subgraph $B= G[U]$ is a \emph{biconnected component} of $G$ if $U$ is an inclusion-maximal set of vertices such that $G[U]$ is biconnected. The decomposition of a connected graph into its biconnected components can be given as a tree:
\begin{definition}[Biconnected component tree]
	The \emph{biconnected component tree} $\biTree(G)$ of a connected graph $G$ has a node for each biconnected component and each cut vertex of~$G$. There is an edge in $\biTree(G)$ between the node for a biconnected component $B$ and the node for a cut vertex $c$ if $c$ occurs in $B$.
\end{definition}

We further aim to decompose biconnected components (that are not just single edges).
A separating pair $\{a,b\}$ of a biconnected graph $G$ is a \emph{$3$-connected separating pair} if there are three vertex-disjoint paths between $a$ and $b$ in $G$. 
Let $U \subseteq V$ be an inclusion-maximal set such that $G[U]$ is not separated by any $3$-connected separating pair, that is, $G[U \setminus  \{a,b\}]$ is connected for any $3$-connected separating pair $\{a,b\}$.
The \emph{triconnected component} associated with $U$ is obtained from $G[U]$ by adding an edge $(a,b)$ for each $3$-connected separating pair $\{a,b\}$ of $G$ with $a, b \in U$, if it is not already present. Such added edges are called \emph{virtual edges}. A triconnected component is either $3$-connected or a (chordless) cycle and we call every graph \emph{triconnected} if it is either $3$-connected or a cycle~\cite[p. 137]{HopcroftT73}.

\begin{definition}[Triconnected component tree]
	The \emph{triconnected component tree} $\triTree(G)$ of a biconnected graph $G$ has a node for each triconnected component and for each $3$-connected separating pair of $G$. There is an edge between the node for a triconnected component $C$ and the node for a $3$-connected separating pair $\{a, b\}$, if $a$ and $b$ both occur in $C$.
\end{definition}

Both the biconnected component tree of a connected graph and the triconnected component tree of a biconnected graph are indeed trees and unique (up to isomorphism). Triconnected component trees are also known under the name SPQR trees \cite{BattistaT96,GutwengerM01}.

If $X$ is a context of a biconnected or triconnected component tree $T$ for a graph $G$, we denote by $\graph(X)$ the subgraph of $G$ induced by the set of vertices represented in $X$, that is, that occur as a cut vertex or in a biconnected component respectively in a separating pair or a triconnected component represented in $X$.

We also use the notion of \emph{coherent paths} defined in~\cite{DattaKM23}.
\begin{definition}[Coherent paths]
	Let $C_1, C_2$ be two triconnected components of a triconnected component tree $\triTree(G)$ and let $a_1, a_2$ be vertices of $C_1, C_2$, respectively, such that no separating pair on the path from $C_1$ to $C_2$ in $\triTree(G)$ contains $a_1$ or $a_2$. The path between $C_1, C_2$ in $\triTree(G)$ is an \emph{$(a_1,a_2)$-coherent path} if $G+\{(a_1,a_2)\}$ is planar.
\end{definition}
	We also denote the $(a_1,a_2)$-coherent path from $C_1$ to $C_2$ by $\rho(C_1, C_2,(a_1, a_2))$. Often, we do not mention the vertices $a_1, a_2$ explicitly and say that the path between $C_1, C_2$ in $\triTree(G)$ is a \emph{coherent path} if there are $a_1, a_2$ such that the path is $(a_1,a_2)$-coherent.
	
	Note that the path between triconnected components $C_1, C_2$ may be non-coherent: for example, if $C_1$ and $C_2$ need to be embedded into different faces of another triconnected component, any edge insertion between $C_1$ and $C_2$ destroys planarity. 

\subparagraph*{Combinatorial embeddings.}

We give some definitions regarding combinatorial embeddings of planar graphs, see also \cite{Diestel12} for more details.

A \emph{topological embedding} $\calT$ of a planar graph $G=(V,E)$ maps each vertex of $G$ to a distinct point in $\mathbb{R}^2$, and each edge $e=(u,v)\in E$ to a simple curve in $\mathbb{R}^2$ with endpoints at the images of $u$ and $v$, such that (i) the interior of every edge-arc contains no vertex image, and (ii) no point of an edge-arc lies on any other edge-arc except possibly at common endpoints. Let $X$ denote the union of all vertex images and edge-arcs in the embedding. The connected components of $\mathbb{R}^2\setminus X$ are called \emph{faces}. The unique unbounded face is the \emph{outer face}, and all other faces are \emph{interior faces}. For a face $F$, a vertex (respectively, edge) is said to be \emph{incident} to $F$ if its image lies on the boundary (equivalently, in the closure) of $F$. Let $V(F)$ denote the set of vertices incident to $F$.

Given a topological embedding $\mathcal{T}$ of $G$, the associated \emph{combinatorial embedding} is specified by the cyclic order of neighbours around each vertex. Formally, for each $v\in V(G)$, let $\sigma_v:N(v)\to N(v)$ be the cyclic permutation that lists the neighbours of $v$ in clockwise order in the embedding. Then the set $\{(v,\sigma_v)\mid v\in V(G)\}$ is the combinatorial embedding of $G$ induced by the topological embedding.
We will also use an equivalent description in terms of faces. For each face $F$, let $\tau_F:V(F)\to V(F)$ be the cyclic permutation giving the clockwise order of the vertices along the boundary of $F$. We also call the cyclic sequence of $\tau_F$ the \emph{combinatorial face} $F$. The tuple $\mathsf{Emb}(G)\df(F_o,\{F\mid F\ \text{is a combinatorial face of}\ \mathcal{T}\})$ is also called the combinatorial embedding of $G$, where $F_o$ indicates the outer (combinatorial) face of $\mathcal{T}$.
From now on, we use the term \emph{planar embedding} to mean this last combinatorial embedding.

Since $3$-connected planar graphs have a unique embedding in the plane (up to the choice of outer face and reflection)~\cite{Whitney33}, just specifying the combinatorial outer face of any embedding suffices to give the whole embedding.
For a $3$-connected planar graph $G$, we denote by $\mathsf{Emb}(G,F)$ the combinatorial embedding of $G$ where $F$ is the combinatorial outer face.

\subparagraph*{Dynamic complexity.}

The goal of dynamic (descriptive) complexity theory is to study the complexity of {maintaining}
the answer to a query where the input structure is subject to changes \cite{PatnaikI97}.
In this paper, the input structure is a planar graph, which is the disjoint union of the graphs $G$ and $G^*$ (see Section~\ref{section:finally} for details), and the changes are insertions $\ins(e)$ and deletions $\del(e)$ of edges $e$ such that the input graph remains planar.

A \emph{dynamic program} $\prog$ stores the current input structure as well as a set~$\calA$ of auxiliary relations over some (relational) schema $\auxSchema$ and over the same domain as the input structure. For every auxiliary relation symbol $R \in \auxSchema$ and every change operation $\delta\in\{\ins,\del\}$, the program $\prog$ has an \emph{update formula} $\varphi_\delta^R(\tpl x;u,v)$, which can access input and auxiliary relations, and has the changed edge $e = (u,v)$ among its parameters. After a change $\delta(u,v) \in \{\ins(u,v), \del(u,v)\}$, for every auxiliary relation symbol $R \in \auxSchema$ the auxiliary relation $R^\calA$ is replaced by the updated auxiliary relation $R^{\calA'}$ that contains all tuples $\tpl x$ such that $\varphi_\delta^R(\tpl x; u,v)$ is satisfied when evaluated over the old input and auxiliary relations. %

A dynamic program $\prog$ \emph{maintains} a dynamic query $\query$, if after applying a sequence $\alpha$ of changes to an initial input structure $\inp_0$ with empty relations and applying the corresponding update formulas to $(\inp_0, \aux_0)$, where $\aux_0$ is an initial auxiliary structure, a dedicated auxiliary relation always equals the result of $\query$ on the input structure.

The class \DynFO contains all dynamic queries that are maintained by a dynamic program with \FOar formulas as update formulas and first-order definable initialization of $\aux_0$.

By allowing update formulas to use a linear order $\leq$, we can identify the vertex set $V$ with an initial segment of the natural numbers. The relations $+$ and $\times$ provide addition and multiplication on these numbers. Assuming the existence of these relations does not increase the expressive power of \DynFO, see \cite[Proposition~7]{DKMSZ18}. 
We use that $\FOar$ corresponds to uniform $\AC^0$ \cite{BarringtonIS90} and note that $\FOar$ formulas can express iterated addition and iterated multiplication of polylogarithmically many numbers that consist of $\polylog(n)$ bits, see \cite[Theorem 5.1]{HesseAB02}.

\section{Proof Overview}\label{section:proof-overview}

The inspiration for our dynamic program for planar isomorphism comes from  algorithms that test planar isomorphism in linear time~\cite{HopcroftW74} and logarithmic space~\cite{DattaLNTW22}. These algorithms decompose given graphs $G$ and $G^*$ into their biconnected and triconnected component trees. They then proceed by identifying pairs of isomorphic triconnected components; isomorphic biconnected components via their triconnected component trees; and finally isomorphic connected components via their biconnected component trees.

Our dynamic program follows this decomposition and maintains
\begin{description}
 \item[(1)] isomorphisms between triconnected components,
 \item[(2)] which biconnected components are isomorphic by maintaining isomorphism information for their triconnected component trees, and
 \item[(3)] which connected components are isomorphic by maintaining isomorphism information for their biconnected component trees.
\end{description}

For connected input graphs, our main theorem follows: the input to our dynamic program is the disjoint union of $G$ and $G^*$; it holds $G \simeq  G^*$ if these connected components of the input graph are isomorphic.
The generalization for disconnected input graphs is presented in Section~\ref{section:finally}.

\subparagraph*{(1) Isomorphisms between triconnected components}
Crucial for maintaining isomorphisms between triconnected components
is that $3$-connected planar graphs have a unique embedding (up to reflection and choice of outer face) in the plane \cite{Whitney33}. This fact has been used for deciding isomorphism between two $3$-connected planar graphs in various settings, see~\cite{Weinberg66,ThieraufW10}. Under the promise that graphs remain 3-connected, Mehta used their unique embeddings to maintain whether two $3$-connected planar graphs are isomorphic by maintaining canonical breadth-first search trees \cite{Mehta14}.

Without the promise that graphs remain 3-connected, there are changes that (i) do not change the triconnected component tree and only modify existing 3-connected components, and changes that (ii) change the triconnected component tree, e.g., by merging triconnected components along a path of the tree into a new 3-connected component. Mehta's approach handles (i), but not (ii).

We use, instead, Tutte's canonical spring embedding~\cite{Tutte63} which is obtained by fixing three arbitrary vertices on a common face of the graph, putting identical springs between adjacent pairs of vertices, and computing an equilibrium. Tutte showed that this equilibrium yields a canonical plane embedding of 3-connected planar graphs, called the \emph{Tutte embedding}. To check whether two $3$-connected planar graphs are isomorphic, their Tutte embeddings can be compared for all choices of three vertices (cf.~\cite{KieferPS19}).

Our dynamic algorithm maintains representations of all Tutte embeddings for all 3-connected components  as well as for all such components that arise from edge insertions that merge multiple triconnected components along a path of the triconnected component tree into a new 3-connected component. The latter helps to deal with changes of type~(ii). For maintaining embeddings, we use that coordinates of graph vertices in the plane --- $\mat x\in\mathbb{Q}^n$ and $\mat y\in\mathbb{Q}^n$ for the $x$ and $y$-coordinates of the $n$ vertices --- can be computed by solving a system of linear equations of the form $\mat T \: \mat x= \mat b$ and $\mat T \: \mat y = \mat d$. The \emph{Tutte matrix} $\mat T$ is almost the Laplacian of the graph, except that the three rows corresponding to the fixed vertices have $1$ at diagonal entries and $0$ elsewhere. For solving these systems, it is sufficient to maintain the inverse of $\mat T$ and then compute $\mat T^{-1} \mat b$ and $\mat T^{-1}\mat d$ (exploiting that $\mat b$ and $\mat d$ only have a constant number of non-zero entries).

We use that edge insertions and deletions in the graph can be translated to the following operations on Tutte matrices:
\begin{itemize}
 \item[(a)] Modifying constantly many entries of a Tutte matrix $\mat T$, or
 \item[(b)] Combining two Tutte matrices $\mat T_1$ and $\mat T_2$ into a new Tutte matrix $\mat T$ in an overlapped block diagonal fashion with small overlap. %
\end{itemize}
We show that the inverses of Tutte matrices can be maintained under both kinds of changes using the Sherman-Morrison-Woodbury identity (see, e.g., \cite{HendersonS81}).

A technical issue in this approach is that entries of $\mat T^{-1}$ may be large, and hence the arithmetic operations required may not be carried out in $\FOar$.  Therefore, similarly as done in \cite{DattaMVZ18}, we maintain inverses modulo polynomially many small ($\bigO(\log n)$ bits) primes and compute the Chinese remainder representation of the positions of the vertices in the Tutte embedding. During updates, primes may become invalid because $\mat T^{-1}$ does not exist for them. For this reason we refresh available primes regularly using a ``muddling technique'' from \cite{DMSVZ19}.

From this information we can infer whether any two $3$-connected components $C$ and $C^*$ of the graph are isomorphic, as well as all isomorphisms between them (if there are any).

\subparagraph*{(2) Isomorphic biconnected components}
For maintaining isomorphic biconnected components, we exploit that two biconnected components are isomorphic if and only if there is an isomorphism between their triconnected component trees that maps triconnected components to triconnected components and separating pairs to separating pairs in such a way that
\begin{itemize}
 \item triconnected components $C$ are mapped to isomorphic triconnected components $C^*$;
 \item separating pairs $\{a_1, b_1\}, \ldots, \{a_k, b_k\}$  of $C$ are mapped to separating pairs $\{a^*_1, b^*_1\} \df \pi(\{a_1, b_1\}), \ldots, \{a^*_k, b^*_k\} \df \pi(\{a_k, b_k\})$ of $C^*$ such that the subtree rooted at $\{a_i, b_i\}$ in the triconnected component tree is isomorphic to the subtree rooted at $\{a^*_i, b^*_i\}$, for all $i \in \{1, \ldots, k\}$; and
 \item the isomorphism from $C$ to $C^*$ is compatible with the mapping $\pi$ of the separating pairs.
\end{itemize}
Thus, we have to maintain whether (a) triconnected component trees are isomorphic, and (b) their isomorphisms can be ``extended'' to isomorphisms of the underlying graphs. We extend the dynamic algorithm for tree isomorphism due to Datta et al.~\cite{DattaK0TVZ24} -- a variant of Etessami's algorithm \cite{Etessami98} -- for (a), but exploit our knowledge of isomorphisms between triconnected components for~(b).

The dynamic algorithm for tree isomorphism by Datta et al.~maintains auxiliary relations
\begin{itemize}
 \item $\textsc{x-iso}$ that stores tuples $(X, X^*) \df (x, r, h, x^*, r^*, h^*)$ such that the contexts $X \df \X(x, r, h)$ and $X^* \df \X(x^*, r^*, h^*)$ are isomorphic,
  \item $\#\textsc{iso-siblings}$ that stores tuples $(x, r, y, m)$ such that $y$ has $m$ isomorphic siblings within the subtree $\ST(x,r)$, and
 \item $\textsc{dist}$ that stores tuples $(x, y, d)$ such that the distance between vertices $x$ and $y$ is $d$.
\end{itemize}

For maintaining isomorphic biconnected components, we maintain similar relations for the triconnected component trees but ensure that both contexts within those trees as well as their induced subgraphs within the biconnected components are isomorphic.

The main challenge is that the triconnected component tree can change significantly. For instance, a triconnected component node can unfurl into a path of triconnected component nodes and separating pair nodes, when an edge is deleted.\footnote{Similarly, such a path of triconnected component nodes and separating pair nodes can merge into a single triconnected component node when an edge is inserted. We refer to Section \ref{section:changes} for more details on effects of edge insertions and deletions on the triconnected component tree.}

For updating whether contexts $X$ and $X^*$ are isomorphic where one of the contexts, say $X$, contains parts of such an unfurled path, we follow Etessami's approach and match subtrees of nodes along this path with their counterparts in $X^*$ using the distance information. Isomorphism of these unchanged subtrees can be verified using the stored auxiliary data.  A difference to tree isomorphism is that we also need to map the separating pair nodes, for which, a priori, there are two possible ways for each pair $\{a, b\}$ and $\{a^*, b^*\}$ which might lead to exponentially many possible mappings if the path contains many such pairs. Fortunately, the unfurled path is a coherent path which implies that if $\{a_1, b_1\}$ and $\{a_2, b_2\}$ are adjacent separating pairs on the path, then the mapping of $\{a_1, b_1\}$ to $\{a_1^*, b_1^*\}$ determines the mapping of $\{a_2, b_2\}$ to $\{a_2^*, b_2^*\}$. Thus, the mapping of the first separating pair on the unfurled path determines the mappings for all subsequent pairs. To exploit this, we use auxiliary information from \cite{DattaKM23} for determining this mapping (see Section~\ref{section:biconnected} for details).

\begin{figure}[t!]
	\centering

\tikzset{every picture/.style={line width=0.75pt}} %
\scalebox{0.4}{
	\begin{tikzpicture}[
		node distance=1.25cm and 2.35cm,
		vertex/.style={
			circle,
			fill=black,
			inner sep=0pt,
			minimum size=8pt
		},
		outer/.style={
			draw,
			fill={rgb, 255:red, 184; green, 233; blue, 134 },
			dash pattern=on 3.75pt off 3pt on 7.5pt off 1.5pt
		},
		hole/.style={
			draw,
			fill=white,
			dash pattern=on 3.75pt off 3pt on 7.5pt off 1.5pt,
			line width=.75pt
		},
		greyedge/.style={
			draw=gray!65,
			line width=1.5pt
		},
		blueedge/.style={
			draw=blue,
			line width=1.5pt
		},
		every label/.style={font=\Huge}
		]

		\node[vertex] (a) {};
		\node[vertex, right=2.55cm of a] (b) {};

		\node[vertex, below left=1.48cm and 2.05cm of a] (c) {};
		\node[vertex, below right=1.48cm and 2.05cm of b] (d) {};
		\node[vertex, below left =2.25cm  and 0.05cm of a] (e) {};
		\node[vertex, below right =2.25cm  and 0.05cm of b] (f) {};

		\node[vertex, below left =1.35cm and 0.5cm of e] (g) {};
		\node[vertex, below right =1.35cm and 0.5cm of f] (h) {};
		\node[vertex, left=2.0cm of g] (i) {};
		\node[vertex, right=2.0cm of h] (j) {};

		\node[vertex, below right =1.35cm and 0.5cm of g] (k) {};
		\node[vertex, below left =1.35cm and 0.5cm of h] (l) {};
		\node[vertex, below=4.45cm of c] (m) {};
		\node[vertex, below=4.45cm of d] (n) {};

		\node[vertex, below right=2.20cm and 0.05cm of k] (o) {};
		\node[vertex, below left=2.20cm and 0.05cm of l] (p) {};
		\node[vertex, below right=.58cm and 1.47cm of k] (q) {};

		\path[outer]
		(j.center) --
		(n.center) --
		(p.center) --
		(o.center) --
		(m.center) --
		(i.center) --
		(c.center) --
		(a.center) --
		(e.center) --
		(g.center) --
		(k.center) --
		(q.center) --
		(l.center) --
		(h.center) --
		(f.center) --
		(b.center) --
		(d.center) --
		cycle;

		\path[hole]
		(c.center) -- ++(0.8,0.3) -- ++ (1.1,-0.2) --
		(e.center) -- ++(-1.1, -0.3) -- ++ (-0.8,0.4) --
		(c.center);
		
		\path[hole]
		(d.center) -- ++(-0.8,0.3) -- ++ (-1.1,-0.2) --
		(f.center) -- ++(1.1, -0.3) -- ++ (0.8,0.4) --
		(d.center) ;

		\path[hole]
		(i.center) -- ++(0.8,0.4) -- ++ (0.8,0) --
		(g.center) -- ++(-0.8, -0.4) -- ++ (-0.8,0) --
		(i.center);

		\path[hole]
		(j.center) -- ++(-0.8,0.4) -- ++ (-0.8,0) --
		(h.center) -- ++(0.8, -0.4) -- ++ (0.8,0) --
		(j.center);
		
		\path[hole]
		(l.center) -- ++(0.5,0.4) -- ++ (0.9,-0.2) --
		(n.center) -- ++(-1.1, 0.2) -- 
		(l.center);
		
		\path[hole]
		(k.center) -- ++(-0.5,0.4) -- ++ (-0.9,-0.2) --
		(m.center) -- ++(1.1, 0.2) -- 
		(k.center);
		
		\path[hole]
		(q.center) -- ++(1.1,-0.2) -- ++ (0.3,-0.6) --
		(p.center) -- ++(-1.1, 0.5) -- 
		(q.center);
		
		\path[hole]
		(q.center) -- ++(-1.1,-0.2) -- ++ (-0.3,-0.6) --
		(o.center) -- ++(1.1, 0.5) -- 
		(q.center);

		\draw[blueedge] (a) -- node[above=0pt, black] {\Huge $e_4$} (b);
		
		\draw[greyedge] (e) -- node[above=0pt, gray] {\Huge $e_3$} (f);
		\draw[greyedge] (g) -- node[above=0pt, gray] {\Huge $e_2$} (h);
		\draw[greyedge] (k) -- node[above=0pt, gray] {\Huge $e_1$} (l);

		\foreach \v in {a,b,c,d,e,f,g,h,i,j,k,l,m,n,o,p,q}
		\node[vertex] at (\v) {};
		
	\end{tikzpicture}
}
 	\caption{A $3$-connected component that unfurls into a path of $3$-connected components after deleting $e_1$, $e_2$, $e_3$, and finally $e_4$.} %
	\label{fig:delayed-unfurl}
\end{figure}

For updating distance information when unfurling a path, it is necessary to maintain additional information for triconnected components. To see this, consider the 3-connected component in Figure~\ref{fig:delayed-unfurl}, which remains 3-connected when deleting the edges $e_1$, $e_2$, and $e_3$; and unfurls into a path only when also deleting $e_4$. When the path unfurls, distance information for this path needs to be available immediately.  A careful examination of the separating pair vertices of the unfurled path just before the deletion of an edge $e$ reveals that they lie on a disc formed by the faces adjacent to $e$ (see Figure~\ref{fig:distSPQRdel} and Section \ref{section:biconnected:dist}).%

Technically, we show that isomorphism can be maintained even for coloured biconnected components, where the isomorphism needs to preserve colours as well. Considering coloured biconnected components is helpful for maintaining isomorphic connected components. Intuitively, two vertices of a biconnected component will have the same colour, if they are cut vertices in the biconnected component tree of a connected component that have isomorphic subtrees.

\subparagraph*{(3) Isomorphic connected components.}
For maintaining isomorphic connected components, similarly to biconnected components, we exploit that two connected components are isomorphic if and only if there is an isomorphism between their biconnected component trees that maps biconnected components to biconnected components and cut vertices to cut vertices such that
\begin{itemize}
 \item biconnected components $B$ are mapped to isomorphic biconnected components $B^*$;
 \item cut vertices $c_1, \ldots, c_k$  of $B$ are mapped to cut vertices $c^*_1 \df \pi(c_1), \ldots, c^*_k \df \pi(c_k)$ of $B^*$ such that the subtree rooted at ${c_i}$ in the biconnected component tree is isomorphic to the subtree rooted at ${c^*_i}$, for all $i \in \{1, \ldots, k\}$; and
 \item some isomorphism from $B$ to $B^*$ is compatible with the mapping $\pi$ of the cut vertices.
\end{itemize}
Thus, we have to maintain whether (a) the biconnected component trees are isomorphic, and (b) the isomorphism can be ``extended'' to an isomorphism of the underlying graphs. We again build on the dynamic algorithm for tree isomorphism and maintain auxiliary information for isomorphic contexts, distances, and number of isomorphic siblings in the biconnected component tree. Again, a challenge is to handle when an edge deletion unfurls a biconnected component into a path of biconnected components and cut vertices. However, updating isomorphic contexts is easier to handle as there is only one mapping from a cut vertex $c$ to its corresponding cut vertex $c^*$ (as opposed to potentially two such mappings from separating pairs $\{a, b\}$ to $\{a^*, b^*\}$ in the case of the triconnected component trees). Also, maintaining distances is easier in this case.

One complication arises from the fact that there may be many ways to isomorphically map a biconnected component $B$ to a biconnected component $B^*$. For triconnected components (see (2) above), there is a unique mapping (up to
reflection and choice of outer face) which can be inferred from the Tutte information. For biconnected components, we cannot deduce such a mapping from the isomorphism information maintained for biconnected components. To ensure that a biconnected component $B$ can be mapped to a biconnected component $B^*$ such that cut vertices $c$ are mapped to cut vertices $c^*$ with isomorphic subtrees, we colour cut vertices of $B$ and $B^*$ with their subtree isomorphism types. If there is an isomorphism between such coloured biconnected components then this isomorphism maps cut vertices $c$ to cut vertices $c^*$ with isomorphic subtrees.

\section{Effects of Changes on Planar Graph Decompositions}\label{section:changes}

When an edge of a planar graph $G$ is modified, the decomposition of $G$ into its biconnected and triconnected component trees can change significantly. %
In this section, building on \cite[Definition 8]{DattaKM23}, we systematically discuss the structural impact of edge changes.

The effect of modifying an edge $(u, v)$ depends on whether (i) the edge $(u,v)$ is inserted or deleted, (ii) the vertices $u$ and $v$ are in the same $k$-connected component before the change, for $k \leq 3$, and whether (iii) $u$ and $v$ are in the same $k'$-connected component after the change, for $k' \leq 3$.
We denote the \emph{type} of an edge insertion (respectively edge deletion) with $\es[+]{k}{k'}$ (respectively $\es[-]{k}{k'}$), where $k, k' \in \{0,\ldots,3\}$ are the maximal numbers such that the two affected vertices are in the same $k$-connected component before and in the same $k'$-connected component after the change.

Depending on the type of a change, the change may affect only triconnected component trees, only biconnected component trees, or both.

\subparagraph*{Changes only relevant for the triconnected component tree}

The following types of changes only affect the triconnected components, any biconnected component tree remains unchanged.

\subparagraph{$\es[+]{3}{3}$.} The vertices $u$ and $v$ belong to a common $3$-connected component before and after the edge change. In this case, the triconnected component tree remains unchanged.

\begin{figure*}[t]
	\centering
	\begin{subfigure}{.52\textwidth}\centering

		\scalebox{0.66}{
      	\begin{tikzpicture}[
		xscale=0.9,
		yscale=0.7,
		vertex/.style={circle,fill=black,inner sep=1.4pt},
		hollow/.style={circle,draw, fill=white,inner sep=1.6pt},
		region/.style={draw=black,fill=gray!35},
		edge/.style={line width=1pt},
		newedge/.style={draw={rgb:red,1;green,2;blue,1},line width=2pt}
		]

		\coordinate (a)  at (-4.25,0);
		\coordinate (b)  at ( 4.25,0);

		\coordinate (a1) at (-3,1.2);
		\coordinate (b1) at (-3,-1.2);

		\coordinate (a2) at (-1,1.2);
		\coordinate (b2) at (-1,-1.2);

		\coordinate (a3) at ( 1,1.2);
		\coordinate (b3) at ( 1,-1.2);

		\coordinate (a4) at ( 3,1.2);
		\coordinate (b4) at ( 3,-1.2);

		\filldraw[region]
		(a1) -- (a4) .. controls(5.3,1.2) and (5.3,-1.2) .. (b4) -- (b1) .. controls(-5.3,-1.2) and (-5.3,1.2) .. cycle;

		\filldraw[fill=white, bend right=80]
		(a1) to (b1) to (a1);

		\filldraw[fill=white, bend right=80]
		(a2) to node[vertex,pos=0.4] (inner1) {} node[vertex, pos=0.6] (inner2) {} (b2) to (a2);

		\filldraw[region]
		(inner1) ..controls ++(-10:0.7) and ++(10:0.7) .. (inner2.center) .. controls ++(30:0.3) and ++(-30:0.3) .. (inner1);

		\filldraw[fill=white, bend right=80]
		(a3) to (b3) to (a3);

		\filldraw[region, bend right=15]
		(a3) to (b3) to (a3);

		\filldraw[region]
		(a3) to[bend right=50] (b3) to[bend left=30] (a3);

		\filldraw[region]
		(a3) to[bend left=50] (b3) to[bend right=30] (a3);

		\filldraw[fill=white, bend right=80]
		(a4) to (b4) to (a4);

		\coordinate[right= 0.5 of b1] (lower1);
		\coordinate[left= 0.5 of b2] (lower2);
		\coordinate[right= 0.5 of a2] (upper1);
		\coordinate[left= 0.5 of a3] (upper2);

		\filldraw[region]
		(lower1) ..controls ++(-70:0.7) and ++(-110:0.7) .. (lower2.center) .. controls ++(-70:1.5) and ++(-110:1.5) ..
		node[vertex, pos=0.6] (lower3) {}
		node[vertex, pos=0.8] (lower4) {}
		(lower1);

		\filldraw[region]
		(lower3) ..controls ++(-135:0.3) and ++(-135:0.3) .. (lower4.center) .. controls ++(-135:0.6) and ++(-135:0.6) ..
		(lower3);

		\filldraw[region]
		(upper1) ..controls ++(70:0.7) and ++(110:0.7) .. (upper2.center) .. controls ++(70:1.5) and ++(110:1.5) ..
		node[vertex, pos=0.6] (upper3) {}
		node[vertex, pos=0.8] (upper4) {}
		(upper1);

		\filldraw[region]
		(upper3) ..controls ++(135:0.3) and ++(135:0.3) .. (upper4.center) .. controls ++(135:0.6) and ++(135:0.6) ..
		(upper3);

		\draw[newedge]
		(a) .. controls (-3,4) and (3,4) .. (b);
		
		\draw[virtual]
		(a1) -- (b1)
		(a2) -- (b2)
		(a4) -- (b4)
		(inner1) to[bend left] (inner2)
		(lower1) to[bend right] (lower2)
		(lower3) to[bend left=40] (lower4)
		(upper1) to[bend left=40] (upper2)
		(upper3) to[bend right=40] (upper4)
		(a3) to[bend right=23] (b3)
		;

		\foreach \p in {a1,b1,a2,b2,a3,b3,a4,b4,inner1,inner2, lower1, lower2, upper1, upper2, lower3, lower4}{
			\node[vertex] at (\p) {};
		}

		\node[hollow, label=left:$u$] at (a) {};
		\node[hollow, label=right:$v$] at (b) {};

		\node[above] at (a1) {$a_1$};
		\node[below] at (b1) {$b_1$};

		\node[above] at (a2) {$a_2$};
		\node[below] at (b2) {$b_2$};

		\node[above] at (a3) {$a_3$};
		\node[below] at (b3) {$b_3$};

		\node[above] at (a4) {$a_4$};
		\node[below] at (b4) {$b_4$};

		\node at (-2,0) {$R_1$};
		\node at (0,0) {$R_2$};
		\node at ( 2,0) {$R_3$};

		\node at (-4,-0.6) {$R_u$};
		\node at ( 4,-0.6) {$R_v$};

	\end{tikzpicture}
	}
	\caption{Graph $G$ before the insertion of $(u,v)$. Grey areas indicate $3$-connected components. The blue, dotted lines represent virtual edges that are included in the corresponding components.}
\end{subfigure}	\hspace{3mm}
	\begin{subfigure}{.44\textwidth}\centering
	\scalebox{0.66}{
		\begin{tikzpicture}[
			xscale=0.9,
			yscale=0.7,
			vertex/.style={circle,fill=black,inner sep=1.4pt},
			hollow/.style={circle,draw, fill=white,inner sep=1.6pt},
			region/.style={draw=black,fill=gray!35},
			regionlight/.style={draw=black,fill=gray!15},
			edge/.style={line width=1pt},
			newedge/.style={draw={rgb:red,1;green,2;blue,1},line width=2pt}
			]
			\clip (-5,3.3) rectangle (5,-2.6);
			
			\coordinate (a)  at (-4.25,0);
			\coordinate (b)  at ( 4.25,0);

		\coordinate (a1) at (-3,1.2);
\coordinate (b1) at (-3,-1.2);			
			
			\coordinate (a2) at (-1,1.2);
			\coordinate (b2) at (-1,-1.2);
			\coordinate (a22) at (-1.2,1.2);
			\coordinate (b22) at (-1.2,-1.2);
			
			\coordinate (a3) at ( 1,1.2);
			\coordinate (b3) at ( 1,-1.2);
			
			\coordinate (a4) at ( 3,1.2);
			\coordinate (b4) at ( 3,-1.2);

\filldraw[region, even odd rule]
(-2.5,3.2) -- (2.5,3.2)
.. controls (4,3.2)  and (6.5,-1.2)  .. (3,-1.2)
-- (-3,-1.2)
.. controls (-6.5,-1.2) and (-4,3.2) .. (-2.5,3.2)
-- cycle
(-2.5,2.7) -- (2.5,2.7)
.. controls (3.1,2.7) and (3.1,1.2) .. (2.5,1.2)
-- (-2.5,1.2)
.. controls (-3.1,1.2) and (-3.1,2.7) .. (-2.5,2.7)
-- cycle;

\filldraw[fill=white, bend right=80]
(a3) to (b3) to (a3);

\filldraw[fill=white]
(-1.2,0) ellipse [x radius=0.5cm, y radius=1.1cm];

			\path[bend right=80]
			(a2) to node[vertex,pos=0.4] (inner1) {} node[vertex, pos=0.6] (inner2) {} (b2) to (a2);

			\filldraw[region]
			(inner1) ..controls ++(-10:0.7) and ++(10:0.7) .. (inner2.center) .. controls ++(30:0.3) and ++(-30:0.3) .. (inner1);
			
			\filldraw[region, bend right=15]
			(a3) to (b3) to (a3);
			
			\filldraw[region]
			(a3) to[bend right=50] (b3) to[bend left=30] (a3);
			
			\filldraw[region]
			(a3) to[bend left=50] (b3) to[bend right=30] (a3);

		\coordinate[right= 0.5 of b1] (lower1);
		\coordinate[left= 0.5 of b2] (lower2);
			\coordinate[right= 0.5 of a2] (upper1);
			\coordinate[left= 0.5 of a3] (upper2);
			
			\filldraw[region]
			(lower1) ..controls ++(-70:0.7) and ++(-110:0.7) .. (lower2.center) .. controls ++(-70:1.5) and ++(-110:1.5) ..
			node[vertex, pos=0.6] (lower3) {}
			node[vertex, pos=0.8] (lower4) {}
			(lower1);
			
			\filldraw[region]
			(lower3) ..controls ++(-135:0.3) and ++(-135:0.3) .. (lower4.center) .. controls ++(-135:0.6) and ++(-135:0.6) ..
			(lower3);
			
			\filldraw[region]
			(upper1) ..controls ++(70:0.7) and ++(110:0.7) .. (upper2.center) .. controls ++(70:1.5) and ++(110:1.5) ..
			node[vertex, pos=0.6] (upper3) {}
			node[vertex, pos=0.8] (upper4) {}
			(upper1);
			
			\filldraw[region]
			(upper3) ..controls ++(135:0.3) and ++(135:0.3) .. (upper4.center) .. controls ++(135:0.6) and ++(135:0.6) ..
			(upper3);

				\draw[virtual]
		(inner1) to[bend left] (inner2)
		(lower1) to[bend right] (lower2)
		(lower3) to[bend left=40] (lower4)
		(upper1) to[bend left=40] (upper2)
		(upper3) to[bend right=40] (upper4)
		(a3) to[bend right=23] (b3)
		;

			\foreach \p in {a3,b3,inner1,inner2, lower1, lower2, upper1, upper2, lower3, lower4}{
				\node[vertex] at (\p) {};
			}
			
			\node[hollow, label=left:$u$] at (a) {};
			\node[hollow, label=right:$v$] at (b) {};

			\node[above] at (a3) {$a_3$};
			\node[below] at (b3) {$b_3$};
			
			\node at (-0.2,0) {$R_{uv}$};
			
		\end{tikzpicture}
	}
	\caption{Graph $G'$ after the insertion.}
\end{subfigure}

\vspace{3mm}
\begin{subfigure}{.52\textwidth}\centering
    \scalebox{0.7}{
      	\begin{tikzpicture}[
		every node/.style={node distance=0.5cm},
		tri/.style={
			draw,
			fill=gray!40,
			regular polygon,
			regular polygon sides=3,
			minimum size=6mm
		},
		rect/.style={
			draw,
			fill=black,
			minimum width=3mm,
			minimum height=5mm
		},
		edge/.style={line width=1.2pt},
		highlight/.style={yellow!45,opacity=0.7}
		]
		\node[tri, label=below:$R_2$] (R2) {};

		\node[Pnode,left=of R2, label={[font=\small, yshift=0.1cm]below:$a_2\,b_2$}] (a2b2){};
		\node[Adot] at ($(a2b2)+(0,0.15)$){};
		\node[Adot] at ($(a2b2)+(0,-0.15)$){};
		\node[tri,left=of a2b2, label={[below right= 0.5cm and -0.1cm]$R_1$}] (R1){};
		\node[Pnode,left=of R1, label={[font=\small, yshift=0.1cm]below:$a_1\,b_1$}] (a1b1){};
		\node[Adot] at ($(a1b1)+(0,0.15)$){};
		\node[Adot] at ($(a1b1)+(0,-0.15)$){};
		\node[tri,left=of a1b1, label=below:$R_u$] (Ra){};
		\node[Pnode, right=of R2, label={[font=\small, yshift=-0.1cm]above:$a_3\,b_3$}] (a3b3) {};
		\node[Adot] at ($(a3b3)+(0,0.15)$){};
		\node[Adot] at ($(a3b3)+(0,-0.15)$){};
		\node[tri, right=of a3b3, label=below:$R_3$] (R3) {};
		\node[Pnode, right=of R3, label={[font=\small, yshift=0.1cm]below:$a_4\,b_4$}] (a4b4) {};
		\node[Adot] at ($(a4b4)+(0,0.15)$){};
		\node[Adot] at ($(a4b4)+(0,-0.15)$){};
		\node[tri, right=of a4b4, label=below:$R_v$] (Rb) {};
	 	\draw[edge] (Ra) -- (a1b1) -- (R1) -- (a2b2) -- (R2) -- (a3b3) -- (R3) -- (a4b4) -- (Rb);

		\node[Pnode, below=of R1] (R1-1) {};
		\node[Adot] at ($(R1-1)+(0,0.15)$){};
		\node[Adot] at ($(R1-1)+(0,-0.15)$){};
		\node[tri, below=of R1-1] (R1-2) {};
		\node[Pnode, left=of R1-2] (R1-3) {};
		\node[Adot] at ($(R1-3)+(0,0.15)$){};
		\node[Adot] at ($(R1-3)+(0,-0.15)$){};
		\node[tri, left=of R1-3] (R1-4) {};
		\draw[edge] (R1) -- (R1-1) -- (R1-2) -- (R1-3) -- (R1-4);

		\node[Pnode, above=of R1] (R1b-1) {};
		\node[Adot] at ($(R1b-1)+(0,0.15)$){};
		\node[Adot] at ($(R1b-1)+(0,-0.15)$){};
		\node[tri, left=of R1b-1] (R1b-2) {};
		\draw[edge] (R1) -- (R1b-1) -- (R1b-2);

		\node[Pnode, above= of R2] (R2-1) {};
		\node[Adot] at ($(R2-1)+(0,0.15)$){};
		\node[Adot] at ($(R2-1)+(0,-0.15)$){};
		\node[tri, above= of R2-1] (R2-2) {};
		\node[Pnode, left= of R2-2] (R2-3) {};
		\node[Adot] at ($(R2-3)+(0,0.15)$){};
		\node[Adot] at ($(R2-3)+(0,-0.15)$){};
		\node[tri, left= of R2-3] (R2-4) {};
		\draw[edge] (R2) -- (R2-1) -- (R2-2) -- (R2-3) -- (R2-4);

		\node[tri, below =1cm of a3b3] (a3b3-1) {};
		\node[tri, left =1cm of a3b3-1] (a3b3-2) {};
		\node[tri, right=1cm of a3b3-1] (a3b3-3) {};
		\draw[edge] (a3b3) -- (a3b3-1);
		\draw[edge] (a3b3) -- (a3b3-2);
		\draw[edge] (a3b3) -- (a3b3-3);

		\begin{pgfonlayer}{background}
			\coordinate[left= 0.4cm of Ra] (left);
			\coordinate[right= 0.4cm of Rb] (right);
			\draw[highlight, line width=14mm]
			(left) --(right);
		\end{pgfonlayer}

	\end{tikzpicture}

    }
		\caption{Triconnected component tree $\triTree(G)$ before the insertion. Grey triangles correspond to triconnected component nodes, separating pair nodes are drawn as ellipses with two vertices inside. The highlighted path from $R_u$ to $R_v$ collapses after the insertion.}
		\label{fig:sfig:SPQRpathB}
	\end{subfigure}\hspace{3mm}
	\begin{subfigure}{.44\textwidth}\centering
    \scalebox{0.7}{
      \begin{tikzpicture}[
		every node/.style={node distance=0.5cm},
		tri/.style={
			draw,
			fill=gray!40,
			regular polygon,
			regular polygon sides=3,
			minimum size=6mm
		},
		rect/.style={
			draw,
			fill=black,
			minimum width=3mm,
			minimum height=5mm
		},
		edge/.style={line width=1.2pt},
		highlight/.style={yellow!45,opacity=0.7}
		]
		\node[tri, label={[above right= -0.3cm and 0.1cm]$R_{uv}$}] (Rab) {};

		\node[Pnode, below left= 0.5cm and 1cm of Rab] (R1-1) {};
		\node[Adot] at ($(R1-1)+(0,0.15)$){};
		\node[Adot] at ($(R1-1)+(0,-0.15)$){};
		\node[tri, below=of R1-1] (R1-2) {};
		\node[Pnode, left=of R1-2] (R1-3) {};
		\node[Adot] at ($(R1-3)+(0,0.15)$){};
		\node[Adot] at ($(R1-3)+(0,-0.15)$){};
		\node[tri, left=of R1-3] (R1-4) {};
		\draw[edge] (Rab) -- (R1-1) -- (R1-2) -- (R1-3) -- (R1-4);

		\node[Pnode, above left=0.5cm and 1cm of Rab] (R1b-1) {};
		\node[Adot] at ($(R1b-1)+(0,0.15)$){};
		\node[Adot] at ($(R1b-1)+(0,-0.15)$){};
		\node[tri, left=of R1b-1] (R1b-2) {};
		\draw[edge] (Rab) -- (R1b-1) -- (R1b-2);

		\node[Pnode, above= of Rab] (R2-1) {};
		\node[Adot] at ($(R2-1)+(0,0.15)$){};
		\node[Adot] at ($(R2-1)+(0,-0.15)$){};
		\node[tri, above= of R2-1] (R2-2) {};
		\node[Pnode, left= of R2-2] (R2-3) {};
		\node[Adot] at ($(R2-3)+(0,0.15)$){};
		\node[Adot] at ($(R2-3)+(0,-0.15)$){};
		\node[tri, left= of R2-3] (R2-4) {};
		\draw[edge] (Rab) -- (R2-1) -- (R2-2) -- (R2-3) -- (R2-4);

		\node[Pnode, below right= 0cm and 1cm of Rab, label={[font=\small]right:$a_3\,b_3$}] (a3b3) {};
		\node[Adot] at ($(a3b3)+(0,0.15)$){};
		\node[Adot] at ($(a3b3)+(0,-0.15)$){};
		\draw[edge]	(Rab) -- (a3b3);
		\node[tri, below =1cm of a3b3] (a3b3-1) {};
		\node[tri, left =1cm of a3b3-1] (a3b3-2) {};
		\node[tri, right=1cm of a3b3-1] (a3b3-3) {};
		\draw[edge] (a3b3) -- (a3b3-1);
		\draw[edge] (a3b3) -- (a3b3-2);
		\draw[edge] (a3b3) -- (a3b3-3);

		\begin{pgfonlayer}{background}
			\coordinate[left= 0.5cm of Rab] (left);
			\coordinate[right= 0.7cm of Rab] (right);
			\draw[highlight, line width=12mm]
			(left) --(right);
		\end{pgfonlayer}

	\end{tikzpicture}
    }
		\caption{Triconnected component tree $\triTree(G')$ after the insertion.}
		\label{fig:sfig:SPQRpathA}
	\end{subfigure}
	\caption{Illustration of the effect of inserting an edge $(u,v)$ of type $\es[+]{2}{3}$ on a graph and its triconnected component tree. In the triconnected component tree, the path from $R_u$ to $R_v$ is merged into a new triconnected component node $R_{uv}$.}	\label{fig:SPQRpath}
\end{figure*}

\subparagraph{$\es[+]{2}{3}$.} The vertices $u$ and $v$ lie in distinct triconnected components $R_u$ and $R_v$ in the same biconnected component of~$G$; in the resulting graph $G'$ they are in a common $3$-connected component $R_{uv}$.  Let the $R_u$ to $R_v$ path\footnote{If $u$ or $v$ are in multiple triconnected components, $R_u$ and $R_v$ are defined such that they are the only triconnected components on the path that contain $u$ and $v$, respectively. This choice is unique.} in $\triTree(G)$ be $R_uP_1R_1P_2\cdots P_kR_kP_{k+1}R_v$, where the $R_i$ are triconnected component nodes and the $P_i$ are $3$-connected separating pair nodes.

For simplicity, let us first consider the case where all nodes $R_i$ represent $3$-connected components (and not cycle components), see Figure~\ref{fig:SPQRpath} for an illustration. After the insertion, all $3$-connected components on the path merge into a $3$-connected component $R_{uv}$.
The new triconnected component tree $\triTree(G')$ is obtained from $\triTree(G)$ by replacing the $R_u$ to $R_v$ path by a new node for $R_{uv}$. To this end, the separating pair nodes of degree $2$ and all triconnected component nodes are replaced by $R_{uv}$; all remaining separating pair nodes of degree more than $2$ are connected by an edge to the new node $R_{uv}$.

\begin{figure*}[t]
	\centering
	\begin{subfigure}{.48\textwidth}
		\centering
		\scalebox{0.7}{
			\begin{tikzpicture}[
				scale=1,
				every node/.style={inner sep=0pt, label distance=1mm},
				vertex/.style={circle,fill=black,inner sep=1.6pt},
				hollow/.style={circle,draw, fill=white,inner sep=1.6pt},
				region/.style={draw=black,fill=gray!35},
				edge/.style={line width=1pt},
				newedge/.style={draw={rgb:red,1;green,2;blue,1},line width=2pt}
				]

				\def\rad{.9}
				\filldraw[region] (0,0) circle (\rad)
				node (C2) {}
				node[coordinate,label=below left:$a_1$] (2-3) at ++(45:\rad) {}
				node[coordinate] (2-1a) at ++(150:\rad) {}
				node[coordinate] (2-1b) at ++(210:\rad) {}
				node[coordinate,label=above left:$b_1$] (2-7) at ++(-45:\rad) {}
				;
				
				\filldraw[region] (2-3) ++ (45:\rad) circle (\rad)
				node {}
				node[coordinate,label=right:$c$] (3-4) at ++(0:\rad) {}
				;

				\filldraw[region] (3-4) ++ (0:\rad) circle (\rad)
				node {}
				node[coordinate,label=below right:$a_2$] (4-5) at ++(-45:\rad) {}
				;

				\node[label=right:$c$] at (3-4){};
				
				\filldraw[region] (4-5) ++ (-45:\rad) circle (\rad)
				node {}
				node[coordinate,label=above right:$b_2$] (5-6) at ++(-135:\rad) {}
				node[coordinate] (5-8a) at ++(-30:\rad) {}
				node[coordinate] (5-8b) at ++(30:\rad) {}
				;
				
				\node[coordinate,label=below right:$a_2$] at (4-5) {};			
				\filldraw[region] (5-6) ++ (-135:\rad) circle (\rad)
				node {}
				node[coordinate,label= right:$d$] (6-7) at ++(180:\rad) {}
				;
				
				\filldraw[region] (6-7) ++ (180:\rad) circle (\rad)
				node {}
				;
				
				\def\innerdist{0.7} %
				\def\outerdist{1.5} %
				\def\innerangle{0} %
				\def\outerangle{-40} %
				\def\dir{180}
				\filldraw[region] (2-1a) .. controls ++(\dir+\outerangle/2:\outerdist) and ++(\dir-\outerangle/2:\outerdist) ..
				node[pos=0.5] (a) {}
				(2-1b)
				..controls ++(\dir+\innerangle/2:\innerdist) and ++(\dir-\innerangle/2:\innerdist) .. (2-1a)
				;
				
				\def\dir{0}
				\filldraw[region] (5-8a) .. controls ++(\dir+\outerangle/2:\outerdist) and ++(\dir-\outerangle/2:\outerdist) ..
				node[pos=0.5] (b) {}
				(5-8b)
				..controls ++(\dir+\innerangle/2:\innerdist) and ++(\dir-\innerangle/2:\innerdist) .. (5-8a)
				;

				\node[label=left:$u$] at (a) {};
				\node[right = 0.05cm of a] {$R_u$};
				\node[label=right:$v$] at (b) {};
				\node[left = 0.05cm of b] {$R_v$};
				\draw[newedge]
				(a) to[bend left=90, looseness=1.2] (b);
				
				\draw[virtual]
				(2-1a) to[bend right= 60] (2-1b)
				(5-8a) to[bend right=60] (5-8b)
				(2-3) to[bend right=70, looseness=1.4] (3-4)
				(3-4) to[bend right=70, looseness=1.4] (4-5)
				(4-5) to[bend right=50, looseness=1.3] (5-6)
				(5-6) to[bend right=70, looseness=1.4] (6-7)
				(6-7) to[bend right=70, looseness=1.4] (2-7)
				(2-7) to[bend right=50, looseness=1.3] (2-3)
				;
				
				\foreach \p in {2-1a,2-1b,5-8a,5-8b,2-3,3-4,4-5,5-6,6-7,2-7,a,b}{
					\node[vertex] at (\p) {};
				}
				
			\end{tikzpicture}
		}
		\caption{A biconnected graph before the insertion of the edge $(u,v)$. Virtual edges between the vertices of $3$-connected separating pairs are again depicted as blue, dotted lines.\label{fig:22-23graph}}
	\end{subfigure}\quad
	\begin{subfigure}{.48\textwidth}
		\centering
		\scalebox{0.7}{
			\begin{tikzpicture}[
				scale=1,
				every node/.style={inner sep=0pt, label distance=1mm},
				vertex/.style={circle,fill=black,inner sep=1.6pt},
				hollow/.style={circle,draw, fill=white,inner sep=1.6pt},
				region/.style={draw=black,fill=gray!35},
				edge/.style={line width=1pt},
				newedge/.style={draw={rgb:red,1;green,2;blue,1},line width=2pt}
				]

				\def\rad{.9}
				\path (0,0) circle (\rad)
				node (C2) {}
				node[coordinate,label=below left:$a_1$] (2-3) at ++(45:\rad) {}
				node[coordinate] (2-1a) at ++(150:\rad) {}
				node[coordinate] (2-1b) at ++(210:\rad) {}
				node[coordinate,label=above left:$b_1$] (2-7) at ++(-45:\rad) {}
				;
				
				\filldraw[region] (2-3) ++ (45:\rad) circle (\rad)
				node {}
				node[coordinate,label=right:$c$] (3-4) at ++(0:\rad) {}
				;

				\filldraw[region] (3-4) ++ (0:\rad) circle (\rad)
				node {}
				node[coordinate,label=below right:$a_2$] (4-5) at ++(-45:\rad) {}
				;

				\node[label=right:$c$] at (3-4){};
				
				\path (4-5) ++ (-45:\rad) circle (\rad)
				node {}
				node[coordinate,label=above right:$b_2$] (5-6) at ++(-135:\rad) {}
				node[coordinate] (5-8a) at ++(-30:\rad) {}
				node[coordinate] (5-8b) at ++(30:\rad) {}
				;
				
				\node[coordinate,label=below right:$a_2$] at (4-5) {};			
				\filldraw[region] (5-6) ++ (-135:\rad) circle (\rad)
				node {}
				node[coordinate,label= right:$d$] (6-7) at ++(180:\rad) {}
				;
				
				\filldraw[region] (6-7) ++ (180:\rad) circle (\rad)
				node {}
				;
				
				\def\innerdist{0.7} %
				\def\outerdist{1.5} %
				\def\innerangle{0} %
				\def\outerangle{-40} %
				\def\dir{180}
				\path (2-1a) .. controls ++(\dir+\outerangle/2:\outerdist) and ++(\dir-\outerangle/2:\outerdist) ..
				node[pos=0.5] (a) {}
				(2-1b)
				..controls ++(\dir+\innerangle/2:\innerdist) and ++(\dir-\innerangle/2:\innerdist) .. (2-1a)
				;
				
				\def\dir{0}
				\path (5-8a) .. controls ++(\dir+\outerangle/2:\outerdist) and ++(\dir-\outerangle/2:\outerdist) ..
				node[pos=0.5] (b) {}
				(5-8b)
				..controls ++(\dir+\innerangle/2:\innerdist) and ++(\dir-\innerangle/2:\innerdist) .. (5-8a)
				;

				\filldraw[region, even odd rule] 	
($(a)+(-0.4,-0.2)$) to[bend left=90, looseness=1.2] ($(b)+(0.4,-0.2)$)
to[bend left, looseness=1.3]  (5-6) to[bend right=70, looseness=0.7] (2-7) to[bend left, looseness=1.3] cycle

($(a)+(0.3,0.2)$) to[bend left=80, looseness=1.15] ($(b)+(-0.3,.2)$)
to[bend right, looseness=1.3] (4-5) to[bend left=70, looseness=0.7] (2-3) to [bend right, looseness=1.3] cycle
;				
				
				\node[label=left:$u$] at (a) {};
				\node[label=right:$v$] at (b) {};

				\draw[virtual]
				(2-3) to[bend right=70, looseness=1.4] (3-4)
				(3-4) to[bend right=70, looseness=1.4] (4-5)
				(2-3) to[bend right=60, looseness=0.65] (4-5)
				(5-6) to[bend right=70, looseness=1.4] (6-7)
				(6-7) to[bend right=70, looseness=1.4] (2-7)
				(5-6) to[bend right=60, looseness=0.65] (2-7)
				;
				
				\node[vertex,label=above left:$b_1$] at (2-7) {};
				\node[coordinate,label=below right:$a_2$] at (4-5) {};	
				\node[coordinate,label=above right:$b_2$] at (5-6) {};
				\node[coordinate,label=below left:$a_1$] at (2-3) {};
				\foreach \p in {2-3,3-4,4-5,5-6,6-7,2-7,a,b}{
					\node[vertex] at (\p) {};
				}
				
				\node[right = 1cm of a] {$R_{uv}$};
				
			\end{tikzpicture}
		}
		\caption{The graph after the insertion of the edge $(u,v)$. The vertices $a_1, a_2, b_1, b_2$ become part of the coalesced $3$-connected component $R_{uv}$.}
	\end{subfigure}
	
	\vspace{3mm}
	\begin{subfigure}{.48\textwidth}
		\centering
		\scalebox{0.7}{
			\begin{tikzpicture}[
				scale=1,
				every node/.style={inner sep=0pt, label distance=1mm},
				vertex/.style={circle,fill=black,inner sep=1.6pt},
				C/.style={fill=white, inner sep=1mm},
				hollow/.style={circle,draw, fill=white,inner sep=1.6pt},
				region/.style={fill=gray!35},
				edge/.style={line width=1pt},
				tri/.style={
					draw,
					line width=0.5pt,
					fill=gray!40,
					regular polygon,
					regular polygon sides=3,
					minimum size=6mm
				},
				rect/.style={
					draw,
					fill=black,
					minimum width=3mm,
					minimum height=5mm
				},
				]
				
				\node[tri, minimum size=5.2cm] at (0,0) {};

				\newcommand{\seppair}[3]{
					\begin{scope}[rotate=#2]
						\node[Pnode, fill=white, rotate=#2, #3] at (#1) (#1p){};
						\node[Adot] at ($(#1p)+(0,0.15)$){};
						\node[Adot] at ($(#1p)+(0,-0.15)$){};
					\end{scope}
				}

				\def\dist{.7}
				
				\def\dir{0}
				\draw[line width=1pt]
				(\dir:1.5) -- ++(\dir:\dist)
				node[pos=1] (8-1) {}
				-- ++(\dir:\dist)
				node[pos=1, tri] (8-2) {}
				-- ++(\dir:\dist)
				node[pos=1] (8-3) {}
				-- ++(\dir:\dist)
				node[pos=1, tri, label=below:$R_v$] (8-4) {}
				;
				
				\seppair{8-1}{\dir}{label={[font=\small, yshift=0.05cm, xshift=0.05cm]below:$a_2\,b_2$}};
				\seppair{8-3}{\dir}{};

				\def\dir{60}
				\draw[line width=1pt]
				(\dir:1.5) -- ++(\dir:\dist)
				node[pos=1] (4-1) {}
				-- ++(\dir:\dist)
				node[pos=1, tri] (4-2) {}
				;
				
				\seppair{4-1}{\dir}{label={[font=\small, yshift=0.1cm]below:$a_2\,c$}};
				
				\def\dir{120}
				\draw[line width=1pt]
				(\dir:1.5) -- ++(\dir:\dist)
				node[pos=1] (3-1) {}
				-- ++(\dir:\dist)
				node[pos=1, tri] (3-2) {}
				;
				
				\seppair{3-1}{\dir}{label={[font=\small, yshift=0.1cm]above:$a_1\,c$}};
				
				\def\dir{180}
				\draw[line width=1pt]
				(\dir:1.5) -- ++(\dir:\dist)
				node[pos=1] (1-1) {}
				-- ++(\dir:\dist)
				node[pos=1, tri] (1-2) {}
				-- ++(\dir:\dist)
				node[pos=1,] (1-3) {}
				-- ++(\dir:\dist)
				node[pos=1, tri, label=below:$R_u$] (1-4) {}
				;
				
				\seppair{1-1}{\dir}{label={[font=\small, yshift=0.05cm, xshift=-0.05cm]above:$a_1\,b_1$}};
				\seppair{1-3}{\dir}{};
				
				\def\dir{-120}
				\draw[line width=1pt]
				(\dir:1.5) -- ++(\dir:\dist)
				node[pos=1] (7-1) {}
				-- ++(\dir:\dist)
				node[pos=1, tri] (7-2) {}
				;
				
				\seppair{7-1}{\dir}{label={[font=\small, yshift=-0.1cm]below:$b_1\,d$}};
				
				\def\dir{-60}
				\draw[line width=1pt]
				(\dir:1.5) -- ++(\dir:\dist)
				node[pos=1] (6-1) {}
				-- ++(\dir:\dist)
				node[pos=1, tri] (6-2) {}
				;
				
				\seppair{6-1}{\dir}{label={[font=\small, yshift=-0.1cm]above:$b_2\,d$}};

				\foreach \p in {1,2,3,4,5,6}{
					\coordinate (\p) at ($(-60*\p-150:0.9)+(0,.1)$);
				}
				\draw[virtual] 
				(1)
				to[bend right=30] (2)
				to[bend right=30] (3)
				to[bend right=30] (4)
				to[bend right=30] (5)
				to[bend right=30] (6)
				to[bend right=30] (1)
				;
				\foreach \p in {1,2,3,4,5,6}{
					\node[vertex,label] at(\p){};
				}
				
				\node[label=above:$a_1$] at (1) {};
				\node[label=above:$c$] at (2) {};
				\node[label=above:$a_2$] at (3) {};
				\node[label=below:$b_2$] at (4) {};
				\node[label=below:$d$] at (5) {};
				\node[label=below:$b_1$] at (6) {};
				
			\end{tikzpicture}
		}
		\caption{The triconnected component tree before the insertion. Its central triconnected component node represents the cycle component formed by the virtual edges. \label{fig:22-23tri-treeB}}
	\end{subfigure}
	\quad
	\begin{subfigure}{.42\textwidth}
		\centering
		\scalebox{0.7}{
			\begin{tikzpicture}[
				scale=1,
				every node/.style={inner sep=0pt, label distance=2mm},
				vertex/.style={circle,fill=black,inner sep=1.2pt},
				C/.style={fill=white, inner sep=1mm},
				hollow/.style={circle,draw, fill=white,inner sep=1.6pt},
				region/.style={fill=gray!35},
				edge/.style={line width=1pt},
				tri/.style={
					draw,
					line width=0.5pt,
					fill=gray!40,
					regular polygon,
					regular polygon sides=3,
					minimum size=6mm
				},
				rect/.style={
					draw,
					fill=black,
					minimum width=3mm,
					minimum height=5mm
				},
				]
				
				\newcommand{\seppair}[3]{
					\begin{scope}[rotate=#2]
						\node[Pnode, fill=white, rotate=#2, #3] at (#1) (#1p){};
						\node[Adot] at ($(#1p)+(0,0.15)$){};
						\node[Adot] at ($(#1p)+(0,-0.15)$){};
					\end{scope}
				}
				
				\def\dist{0.7}
				\def\rad{0.5}

				\begin{scope}[on background layer]
					\node[tri, label=right:$R_{uv}$] at (0,0) (center) {};
				\end{scope}

				\def\dir{90}
				\draw[line width=1pt] (center) ++(0,0) -- ++(\dir:\dist)
				coordinate (up)
				-- ++(\dir:\dist)
				coordinate (uppertri-root)
				;
				
				\seppair{up}{\dir}{label={[font=\small, yshift=-0cm]below:$a_1\,a_2$}};
				
				\path[] (uppertri-root) ++(\dir:\rad) circle (\rad)
				coordinate (upper)
				;
				\node[tri, minimum size=2cm] at (upper) {};
				\path (upper) ++ (-30:0.4)
				coordinate[label={[below right=0.1cm, yshift=-0.2cm]$a_2$}] (upper1);
				\path (upper) ++ (-150:0.4)
				coordinate[label={[below left=0.05cm, yshift=-0.2cm]$a_1$}] (upper2);
				\path (upper) ++ (90:0.4)
				coordinate[label=above:$c$] (upper3);
				\draw[virtual]
				(upper1)
				to[bend right=20] (upper2)
				to[bend right=20] (upper3)
				to[bend right=20] (upper1)
				;
				\node[vertex] at (upper1) {};
				\node[vertex] at (upper2) {};
				\node[vertex] at (upper3) {};

				\def\dir{45}
				\draw[line width=1pt] (upper) ++(\dir:\rad) -- ++(\dir:\dist)
				coordinate (ur-1)
				-- ++(\dir:\dist)
				node[tri] {}
				;
				\seppair{ur-1}{\dir}{label={[font=\small, yshift=0.1cm]below:$a_2\,c$}};
				
				\def\dir{135}
				\draw[line width=1pt] (upper) ++(\dir:\rad) -- ++(\dir:\dist)
				coordinate (ul-1)
				-- ++(\dir:\dist)
				node[tri] {}
				;
				\seppair{ul-1}{\dir}{label={[font=\small, yshift=0.1cm]above:$a_1\,c$}};

				\def\dir{-90}
				\draw[line width=1pt] (center) ++(0,0) -- ++(\dir:\dist)
				coordinate (down)
				-- ++(\dir:\dist)
				coordinate (lowertri-root)
				;
				
				\seppair{down}{\dir}{label={[font=\small, yshift=0cm]above:$b_1\,b_2$}};
				
				\path[] (lowertri-root) ++(\dir:\rad) circle (\rad)
				coordinate (lower)
				;
				\node[tri, minimum size=2cm, rotate=180] at (lower) {};
				
				\path (lower) ++ (30:0.4)
				coordinate[label={[above right=0.1cm, yshift=-0.3cm]$b_2$}] (lower1);
				\path (lower) ++ (-90:0.4)
				coordinate[label={[below=0.3cm]$d$}] (lower2);
				\path (lower) ++ (150:0.4)
				coordinate[label={[above left=0.1cm, yshift=-0.3cm]$b_1$}] (lower3);
				\draw[virtual]
				(lower1)
				to[bend right=20] (lower2)
				to[bend right=20] (lower3)
				to[bend right=20] (lower1)
				;
				\node[vertex] at (lower1) {};
				\node[vertex] at (lower2) {};
				\node[vertex] at (lower3) {};
				
				\def\dir{-45}
				\draw[line width=1pt] (lower) ++(\dir:\rad) -- ++(\dir:\dist)
				coordinate (lr-1)
				-- ++(\dir:\dist)
				node[tri] {}
				;
				\seppair{lr-1}{\dir}{label={[font=\small, yshift=-0.1cm]above:$b_2\,d$}};
				
				\def\dir{-135}
				\draw[line width=1pt] (lower) ++(\dir:\rad) -- ++(\dir:\dist)
				coordinate (ll-1)
				-- ++(\dir:\dist)
				node[tri] {}
				;
				\seppair{ll-1}{\dir}{label={[font=\small, yshift=-0.1cm]below:$b_1\,d$}};
				
				\node[tri] at (0,0) {};

			\end{tikzpicture}
		}
		\caption{The triconnected component tree after the insertion. The former central cycle component has split into two parts.\label{fig:23-tri-treeA}}
	\end{subfigure}
	\caption{\label{fig:2-3treeChanges} Illustration of the effect of inserting an edge $(u, v)$ of type $\es[+]{2}{3}$ such that the path from $R_u$ to $R_v$ in the triconnected component tree includes a cycle component.}
\end{figure*}
 
Let us now consider the case where an $R_i$ on the $R_u$ to $R_v$ path in $\triTree(G)$ is a cycle component, see Figure~\ref{fig:2-3treeChanges} for an illustration. Let $P_i = \{a_1,b_1\}$ and $P_{i+1} = \{a_2,b_2\}$ be the two separating pairs adjacent to $R_i$ on the path and let $a_1, \ldots, a_2,b_2,\ldots, b_1,a_1$ be the cyclic order of the vertices in the cycle.
In $G'$, the vertices $a_1,a_2,b_1,b_2$ get absorbed into the merged $3$-connected component $R_{uv}$; the cycle $R_i$ is split into two cycles $C_a=a_1\cdots a_2 a_1$ and $C_b= b_2 \cdots b_1 b_2$ with corresponding triconnected component nodes in $\triTree(G')$. New separating pair nodes for $\{a_1,a_2\}$ and $\{b_1,b_2\}$ are introduced, which connect $R_{uv}$ to $C_a$ and to $C_b$, respectively.
Any separating pair node that is attached to $R_i$ in $\triTree(G)$ but does not lie on the path from $R_u$ to $R_v$ gets attached to either $C_a$ or $C_b$ in $\triTree(G')$, depending on whether the separating pair nodes lie in $C_a$ or $C_b$.  In case $R_u$ (resp. $R_v$) itself is a cycle component, the cycle $R_u$ (resp. $R_v$) is to be split into possibly two cycles, analogous to a cycle component that lies internally on the $R_u$ to $R_v$ path, by insertion of a chord from $u$ (resp. $v$) to each vertex of $P_1$ (resp. $P_{k+1}$).

\subparagraph{$\es[+]{2}{2}$.} If $u$ and $v$ belong to the same $2$-connected component but not to the same $3$-connected component before and after the insertion of the edge $(u,v)$, they belong to a common cycle component $C$ in $G$ (otherwise this would be a change of type $\es[+]{2}{3}$). %
\begin{figure*}[t]
	\centering
	\begin{subfigure}{.48\textwidth}
		\centering
		\scalebox{0.7}{
			\begin{tikzpicture}[
				scale=1,
				every node/.style={inner sep=0pt, label distance=1mm},
				vertex/.style={circle,fill=black,inner sep=1.6pt},
				hollow/.style={circle,draw, fill=white,inner sep=1.6pt},
				region/.style={draw=black,fill=gray!35},
				edge/.style={line width=1pt},
				newedge/.style={draw={rgb:red,1;green,2;blue,1},line width=2pt}
				]

				\def\rad{.9}
				\filldraw[region] (0,0) circle (\rad)
				node (C2) {}
				node[coordinate,label=below left:$a_1$] (2-3) at ++(45:\rad) {}
				node[coordinate] (2-1a) at ++(150:\rad) {}
				node[coordinate] (2-1b) at ++(210:\rad) {}
				node[coordinate,label=above left:$b_1$] (2-7) at ++(-45:\rad) {}
				;
				
				\filldraw[region] (2-3) ++ (45:\rad) circle (\rad)
				node {}
				node[coordinate,label=right:$u$] (3-4) at ++(0:\rad) {}
				;

				\filldraw[region] (3-4) ++ (0:\rad) circle (\rad)
				node {}
				node[coordinate,label=below right:$a_2$] (4-5) at ++(-45:\rad) {}
				;

				\node[label=right:$u$] at (3-4){};
				
				\filldraw[region] (4-5) ++ (-45:\rad) circle (\rad)
				node {}
				node[coordinate,label=above right:$b_2$] (5-6) at ++(-135:\rad) {}
				node[coordinate] (5-8a) at ++(-30:\rad) {}
				node[coordinate] (5-8b) at ++(30:\rad) {}
				;
				
				\node[coordinate,label=below right:$a_2$] at (4-5) {};			
				\filldraw[region] (5-6) ++ (-135:\rad) circle (\rad)
				node {}
				node[coordinate,label= right:$v$] (6-7) at ++(180:\rad) {}
				;
				
				\filldraw[region] (6-7) ++ (180:\rad) circle (\rad)
				node {}
				;
				
				\def\innerdist{0.7} %
				\def\outerdist{1.5} %
				\def\innerangle{0} %
				\def\outerangle{-40} %
				\def\dir{180}
				\filldraw[region] (2-1a) .. controls ++(\dir+\outerangle/2:\outerdist) and ++(\dir-\outerangle/2:\outerdist) ..
				(2-1b)
				..controls ++(\dir+\innerangle/2:\innerdist) and ++(\dir-\innerangle/2:\innerdist) .. (2-1a)
				;
				
				\def\dir{0}
				\filldraw[region] (5-8a) .. controls ++(\dir+\outerangle/2:\outerdist) and ++(\dir-\outerangle/2:\outerdist) ..
				(5-8b)
				..controls ++(\dir+\innerangle/2:\innerdist) and ++(\dir-\innerangle/2:\innerdist) .. (5-8a)
				;
				\draw[newedge]
				(3-4) to[] (6-7);
				
				\draw[virtual]
				(2-1a) to[bend right= 60] (2-1b)
				(5-8a) to[bend right=60] (5-8b)
				(2-3) to[bend right=70, looseness=1.4] (3-4)
				(3-4) to[bend right=70, looseness=1.4] (4-5)
				(4-5) to[bend right=50, looseness=1.3] (5-6)
				(5-6) to[bend right=70, looseness=1.4] (6-7)
				(6-7) to[bend right=70, looseness=1.4] (2-7)
				(2-7) to[bend right=50, looseness=1.3] (2-3)
				;
				
				\foreach \p in {2-1a,2-1b,5-8a,5-8b,2-3,3-4,4-5,5-6,6-7,2-7}{
					\node[vertex] at (\p) {};
				}
				
			\end{tikzpicture}
		}
		\caption{A biconnected graph before the insertion of the edge $(u,v)$.}
	\end{subfigure}\quad
	\begin{subfigure}{.48\textwidth}
		\centering
		\scalebox{0.7}{
			\begin{tikzpicture}[
				scale=1,
				every node/.style={inner sep=0pt, label distance=1mm},
				vertex/.style={circle,fill=black,inner sep=1.6pt},
				hollow/.style={circle,draw, fill=white,inner sep=1.6pt},
				region/.style={draw=black,fill=gray!35},
				edge/.style={line width=1pt},
				newedge/.style={draw={rgb:red,1;green,2;blue,1},line width=2pt}
				]

				\def\rad{.9}
				\filldraw[region] (0,0) circle (\rad)
				node (C2) {}
				node[coordinate,label=below left:$a_1$] (2-3) at ++(45:\rad) {}
				node[coordinate] (2-1a) at ++(150:\rad) {}
				node[coordinate] (2-1b) at ++(210:\rad) {}
				node[coordinate,label=above left:$b_1$] (2-7) at ++(-45:\rad) {}
				;
				
				\filldraw[region] (2-3) ++ (45:\rad) circle (\rad)
				node {}
				node[coordinate,label=right:$u$] (3-4) at ++(0:\rad) {}
				;

				\filldraw[region] (3-4) ++ (0:\rad) circle (\rad)
				node {}
				node[coordinate,label=below right:$a_2$] (4-5) at ++(-45:\rad) {}
				;

				\node[label=right:$u$] at (3-4){};
				
				\filldraw[region] (4-5) ++ (-45:\rad) circle (\rad)
				node {}
				node[coordinate,label=above right:$b_2$] (5-6) at ++(-135:\rad) {}
				node[coordinate] (5-8a) at ++(-30:\rad) {}
				node[coordinate] (5-8b) at ++(30:\rad) {}
				;
				
				\node[coordinate,label=below right:$a_2$] at (4-5) {};			
				\filldraw[region] (5-6) ++ (-135:\rad) circle (\rad)
				node {}
				node[coordinate,label= right:$v$] (6-7) at ++(180:\rad) {}
				;
				
				\filldraw[region] (6-7) ++ (180:\rad) circle (\rad)
				node {}
				;
				
				\def\innerdist{0.7} %
				\def\outerdist{1.5} %
				\def\innerangle{0} %
				\def\outerangle{-40} %
				\def\dir{180}
				\filldraw[region] (2-1a) .. controls ++(\dir+\outerangle/2:\outerdist) and ++(\dir-\outerangle/2:\outerdist) ..
				(2-1b)
				..controls ++(\dir+\innerangle/2:\innerdist) and ++(\dir-\innerangle/2:\innerdist) .. (2-1a)
				;
				
				\def\dir{0}
				\filldraw[region] (5-8a) .. controls ++(\dir+\outerangle/2:\outerdist) and ++(\dir-\outerangle/2:\outerdist) ..
				(5-8b)
				..controls ++(\dir+\innerangle/2:\innerdist) and ++(\dir-\innerangle/2:\innerdist) .. (5-8a)
				;
				\draw[edge]
				(3-4) to[] (6-7);
				
				\draw[virtual]
				(2-1a) to[bend right= 60] (2-1b)
				(5-8a) to[bend right=60] (5-8b)
				(2-3) to[bend right=70, looseness=1.4] (3-4)
				(3-4) to[bend right=70, looseness=1.4] (4-5)
				(4-5) to[bend right=50, looseness=1.3] (5-6)
				(5-6) to[bend right=70, looseness=1.4] (6-7)
				(6-7) to[bend right=70, looseness=1.4] (2-7)
				(2-7) to[bend right=50, looseness=1.3] (2-3)
				;
				
				\foreach \p in {2-1a,2-1b,5-8a,5-8b,2-3,3-4,4-5,5-6,6-7,2-7}{
					\node[vertex] at (\p) {};
				}
				
			\end{tikzpicture}
		}
		\caption{The graph after the insertion of the edge $(u,v)$. The $3$-connected components did not change.}
	\end{subfigure}
	
	\vspace{3mm}
	\begin{subfigure}{.48\textwidth}
		\centering
		\scalebox{0.7}{
			\begin{tikzpicture}[
				scale=1,
				every node/.style={inner sep=0pt, label distance=1mm},
				vertex/.style={circle,fill=black,inner sep=1.6pt},
				C/.style={fill=white, inner sep=1mm},
				hollow/.style={circle,draw, fill=white,inner sep=1.6pt},
				region/.style={fill=gray!35},
				edge/.style={line width=1pt},
				tri/.style={
					draw,
					line width=0.5pt,
					fill=gray!40,
					regular polygon,
					regular polygon sides=3,
					minimum size=6mm
				},
				rect/.style={
					draw,
					fill=black,
					minimum width=3mm,
					minimum height=5mm
				},
				]
				
				\node[tri, minimum size=5.2cm] at (0,0) {};

				\newcommand{\seppair}[3]{
					\begin{scope}[rotate=#2]
						\node[Pnode, fill=white, rotate=#2, #3] at (#1) (#1p){};
						\node[Adot] at ($(#1p)+(0,0.15)$){};
						\node[Adot] at ($(#1p)+(0,-0.15)$){};
					\end{scope}
				}

				\def\dist{.7}
				
				\def\dir{0}
				\draw[line width=1pt]
				(\dir:1.5) -- ++(\dir:\dist)
				node[pos=1] (8-1) {}
				-- ++(\dir:\dist)
				node[pos=1, tri] (8-2) {}
				-- ++(\dir:\dist)
				node[pos=1] (8-3) {}
				-- ++(\dir:\dist)
				node[pos=1, tri] (8-4) {}
				;
				
				\seppair{8-1}{\dir}{label={[font=\small, yshift=0.05cm, xshift=0.05cm]below:$a_2\,b_2$}};
				\seppair{8-3}{\dir}{};

				\def\dir{60}
				\draw[line width=1pt]
				(\dir:1.5) -- ++(\dir:\dist)
				node[pos=1] (4-1) {}
				-- ++(\dir:\dist)
				node[pos=1, tri] (4-2) {}
				;
				
				\seppair{4-1}{\dir}{label={[font=\small, yshift=0.1cm]below:$a_2\,u$}};
				
				\def\dir{120}
				\draw[line width=1pt]
				(\dir:1.5) -- ++(\dir:\dist)
				node[pos=1] (3-1) {}
				-- ++(\dir:\dist)
				node[pos=1, tri] (3-2) {}
				;
				
				\seppair{3-1}{\dir}{label={[font=\small, yshift=0.1cm]above:$a_1\,u$}};
				
				\def\dir{180}
				\draw[line width=1pt]
				(\dir:1.5) -- ++(\dir:\dist)
				node[pos=1] (1-1) {}
				-- ++(\dir:\dist)
				node[pos=1, tri] (1-2) {}
				-- ++(\dir:\dist)
				node[pos=1,] (1-3) {}
				-- ++(\dir:\dist)
				node[pos=1, tri] (1-4) {}
				;
				
				\seppair{1-1}{\dir}{label={[font=\small, yshift=0.05cm, xshift=-0.05cm]above:$a_1\,b_1$}};
				\seppair{1-3}{\dir}{};
				
				\def\dir{-120}
				\draw[line width=1pt]
				(\dir:1.5) -- ++(\dir:\dist)
				node[pos=1] (7-1) {}
				-- ++(\dir:\dist)
				node[pos=1, tri] (7-2) {}
				;
				
				\seppair{7-1}{\dir}{label={[font=\small, yshift=-0.1cm]below:$b_1\,v$}};
				
				\def\dir{-60}
				\draw[line width=1pt]
				(\dir:1.5) -- ++(\dir:\dist)
				node[pos=1] (6-1) {}
				-- ++(\dir:\dist)
				node[pos=1, tri] (6-2) {}
				;
				
				\seppair{6-1}{\dir}{label={[font=\small, yshift=-0.1cm]above:$b_2\,v$}};

				\foreach \p in {1,2,3,4,5,6}{
					\coordinate (\p) at ($(-60*\p-150:0.9)+(0,.1)$);
				}
				\draw[virtual] 
				(1)
				to[bend right=30] (2)
				to[bend right=30] (3)
				to[bend right=30] (4)
				to[bend right=30] (5)
				to[bend right=30] (6)
				to[bend right=30] (1)
				;
				\foreach \p in {1,2,3,4,5,6}{
					\node[vertex,label] at(\p){};
				}
				
				\node[label=above:$a_1$] at (1) {};
				\node[label=above:$u$] at (2) {};
				\node[label=above:$a_2$] at (3) {};
				\node[label=below:$b_2$] at (4) {};
				\node[label=below:$v$] at (5) {};
				\node[label=below:$b_1$] at (6) {};
				
			\end{tikzpicture}
		}
		\caption{The triconnected component tree before the insertion.}
	\end{subfigure}
	\quad
		\begin{subfigure}{.48\textwidth}
			\centering
			\scalebox{0.7}{
				\begin{tikzpicture}[
						scale=1,
						every node/.style={inner sep=0pt, label distance=1mm},
						vertex/.style={circle,fill=black,inner sep=1.2pt},
						C/.style={fill=white, inner sep=1mm},
						hollow/.style={circle,draw, fill=white,inner sep=1.6pt},
						region/.style={fill=gray!35},
						edge/.style={line width=1pt},
						tri/.style={
								draw,
								line width=0.5pt,
								fill=gray!40,
								regular polygon,
								regular polygon sides=3,
								minimum size=6mm
							},
						rect/.style={
								draw,
								fill=black,
								minimum width=3mm,
								minimum height=5mm
							},
						]
						
				\newcommand{\seppair}[3]{
	\begin{scope}[rotate=#2]
		\node[Pnode, fill=white, rotate=#2, #3] at (#1) (#1p){};
		\node[Adot] at ($(#1p)+(0,0.15)$){};
		\node[Adot] at ($(#1p)+(0,-0.15)$){};
	\end{scope}
}
						
						\def\elllength{1.2}
						\def\ellheight{0.85}

						\begin{scope}[on background layer]
								\path (0,0) ellipse ({\elllength} and \ellheight)
								node (upper) {};
							\end{scope}		
						\def\dir{-90}
						\def\dist{1}
						\draw[line width=1pt] (upper) ++(0,-\ellheight) -- ++(\dir:\dist)
						coordinate (mid)
						-- ++(\dir:\dist)
						coordinate (lower-root)
						;
						
						\seppair{mid}{\dir}{label={[font=\small, yshift=-0cm]above:$u\,v$}};
						\path (mid) ++ (180:\dist)
						node {};
						\path (mid) ++ (0:\dist)
						node {};
						
						\begin{scope}[on background layer]
								\path (lower-root) ++(\dir:0.85) ellipse ({\elllength} and \ellheight)
								node (lower) {};
							\end{scope}

						\def\dist{.7}
						\def\dir{60}
						\draw[line width=1pt] (upper) ++ (\dir:0.75)
						-- ++(\dir:\dist)
						coordinate (u)
						-- ++(\dir:\dist)
						node[tri] (C4) {}
						;
						\seppair{u}{\dir}{label={[font=\small, yshift=0.2cm]below:$a_2\,u$}};
						
						\def\dir{0}
						\draw[line width=1pt] (upper) ++ (\dir:1.3)
						-- ++(\dir:\dist)
						coordinate (ur1)
						-- ++(\dir:\dist)
						node[tri] (C5) {}
						-- ++(\dir:\dist)
						coordinate (ur2)
						-- ++(\dir:\dist)
						node[tri] (C8) {}
						;
						\seppair{ur1}{\dir}{label={[font=\small, yshift=0.05cm, xshift=0.05cm]above:$a_2\,b_2$}};
						\seppair{ur2}{\dir}{};
						
						\def\dir{-40}
						\draw[line width=1pt] (upper) ++ (\dir:0.75)
						-- ++(\dir:\dist)
						coordinate (ul)
						-- ++(\dir:\dist)
						node[tri] (C3) {}
						;
						\seppair{ul}{\dir}{label={[font=\small, yshift=-0.1cm]above:$b_2\,v$}};

						\def\dir{-120}
						\draw[line width=1pt] (lower) ++ (\dir:0.75)
						-- ++(\dir:\dist)
						coordinate (l)
						-- ++(\dir:\dist)
						node[tri] (C7) {}
						;
						\seppair{l}{\dir}{label={[font=\small, yshift=-0.2cm]below:$b_1\,v$}};
						
						\def\dir{180}
						\draw[line width=1pt] (lower) ++ (\dir:1.3)
						-- ++(\dir:\dist)
						coordinate (ul1)
						-- ++(\dir:\dist)
						node[tri] (C2) {}
						-- ++(\dir:\dist)
						coordinate (ul2)
						-- ++(\dir:\dist)
						node[tri] (C1) {}
						;
						\seppair{ul1}{\dir}{label={[font=\small, yshift=0.05cm]above:$a_1\,b_1$}};
						\seppair{ul2}{\dir}{};
						
						\def\dir{130}
						\draw[line width=1pt] (lower) ++ (\dir:0.75)
						-- ++(\dir:\dist)
						coordinate (ur)
						-- ++(\dir:\dist)
						node[tri] (C6) {}
						;
						\seppair{ur}{\dir}{label={[font=\small, yshift=0.1cm]above:$a_1\,u$}};
						
						\node[tri, minimum size=3cm, rotate=-90] at (upper) {};	
						\node[tri, minimum size=3cm, rotate=90] at (lower) {};	
						\begin{scope}[rotate=-90]
							\path (upper) ++(-0.6*\elllength,-0.25*\ellheight)++(0,-0.1)
							coordinate (u-c);
							\path (upper) ++(0.6*\elllength,-0.25*\ellheight)++(0,-0.1)
							coordinate (u-d);
							\path (upper) ++(-0.2*\elllength,0.5*\ellheight)++(0,-0.1)
							coordinate (u-1);
							\path (upper) ++(0.2*\elllength,0.5*\ellheight)++(0,-0.1)
							coordinate (u-2);
						\end{scope}

						\draw[virtual] (u-c)
						to[bend right=30] (u-1)
						to[bend right=30] (u-2)
						to[bend right=30] (u-d)
						;
						\draw[edge](u-d) to[bend left=10] (u-c);
						
						\begin{scope}[rotate=-90]
							\path (lower) ++(-0.6*\elllength,0.25*\ellheight)++(0,0.1)
							coordinate (l-c);
							\path (lower) ++(0.6*\elllength,0.25*\ellheight)++(0,0.1)
							coordinate (l-d);
							\path (lower) ++(-0.2*\elllength,-0.5*\ellheight)++(0,0.1)
							coordinate (l-1);
							\path (lower) ++(0.2*\elllength,-0.5*\ellheight)++(0,0.1)
							coordinate (l-2);
						\end{scope}
						
						\draw[virtual] (l-c)
						to[bend left=30] (l-1)
						to[bend left=30] (l-2)
						to[bend left=30] (l-d)
						;
						\draw[edge](l-d) to[bend right=10] (l-c);

						\foreach \p in {u-c,u-d,u-1,u-2,l-c,l-d,l-1,l-2}{
							\node[vertex] at (\p) {};
						}
						\node[label=left:$u$] at (u-c) {};
						\node[label=right:$u$] at (l-c) {};
						\node[label=left:$v$] at (u-d) {};
						\node[label=right:$v$] at (l-d) {};

						\node[label=right:$a_2$] at (u-1) {};
						\node[label=right:$b_2$] at (u-2) {};
						\node[label=left:$a_1$] at (l-1) {};
						\node[label=left:$b_1$] at (l-2) {};
						
					\end{tikzpicture}
				}
			\caption{The triconnected component tree after the insertion. The former central cycle component has split along the new chord $(u,v)$.\label{fig:22-tri-treeA}}
			\end{subfigure}
	\caption{\label{fig:2-2treeChanges} Illustration of the effect of inserting an edge $(u, v)$ of type $\es[+]{2}{2}$. Compare with Figure~\ref{fig:2-3treeChanges}.}
\end{figure*}
 See Figure~\ref{fig:2-2treeChanges} for an illustration.
The triconnected component tree $\triTree(G')$ is obtained from $\triTree(G)$ by replacing the triconnected component node of $C$ with a path $C_1 P_{uv} C_2$, where $P_{uv}$ corresponds to the new separating pair $\{u,v\}$ and $C_1, C_2$ correspond to the two cycles obtained by splitting $C$ along the chord $(u, v)$.
Each neighbour $P$ of $C$ in $\triTree(G)$ is made a neighbour of either $C_1$ or $C_2$ in $\triTree(G')$, depending on the location of the vertices of the separating pair $P$.

\smallskip

Edge changes of types $\es[-]{3}{3}$ and $\es[-]{3}{2}$, $\es[-]{2}{2}$ reverse edge changes of type $\es[+]{3}{3}$, $\es[+]{2}{3}$, and $\es[+]{2}{2}$. Effects on the decompositions are also reversed.

\subparagraph*{Changes affecting both component trees}

Changes that create or destroy biconnected components affect both the triconnected and biconnected component trees.

\subparagraph{$\es[+]{1}{2}$.} Before the insertion of $(u,v)$, the vertices $u$ and $v$ belong to distinct biconnected components $B_u$ and $B_v$; after the insertion, $u$ and $v$ end up in a common biconnected component $B_{uv}$. See Figure~\ref{fig:BCpath} for an illustration.

\begin{figure*}[tp]
	\centering
	\begin{subfigure}{.5\textwidth}\centering
		\scalebox{0.7}{
			\begin{tikzpicture}[
				scale=1.1,
				label distance= 2mm,
				vertex/.style={circle,fill=black,inner sep=1.6pt},
				hollow/.style={circle,draw, fill=white,inner sep=1.6pt},
				region/.style={draw=black,fill=gray!35},
				edge/.style={line width=1pt},
				newedge/.style={draw={rgb:red,1;green,2;blue,1},line width=2pt},
				redpath/.style={red!60,line width=1pt}
				]
				
				\newcommand{\ellipse}[3]{
					\filldraw[region] #1 ellipse ({#2} and {#3})
				};

				\newcommand{\cone}[3]{
					\def\ratio{0.3} %
					\def\myangle{34.9} %
					\draw[region]
					({#1}) -- ++({#2+\myangle/2}:{#3}) arc ({#2+90+\myangle /2}:{#2-90-\myangle /2}:{\ratio*#3}) -- ({#1})
				};

				\def\length{1}
				\def\height{0.5}
				\ellipse{(0,0)}{\length}{\height}
				node (B1) {$B_1$}
				node[coordinate, label={below:$c_1$}] (c1) at ++(-\length,0) {}
				node[coordinate, label={below:$c_2$}] (c2) at ++(\length,0) {}
				node[coordinate] (B1lower) at ++(0,-\height) {}
				;
				\ellipse{(B1) ++ (-2*\length,0)}{\length}{\height}
				node (Ba) {$B_u$}
				node[coordinate, label={[label distance=0mm]below right:$u$}] (a) at ++(-0.8*\length,0){}
				node[coordinate] (Balower) at ++(0,-\height){}
				;
				\ellipse{(B1) ++ (2*\length,0)}{\length}{\height}
				node (B2) {$B_2$}
				node[coordinate, label={below:$c_3$}] (c3) at ++(\length,0) {}
				node[coordinate] (B2upper) at ++(0,\height) {}
				;
				\ellipse{(B2) ++(2*\length,0)}{\length}{\height}
				node (Bb) {$B_v$}
				node[coordinate, label={[label distance=0mm]below left:$v$}] (b) at ++(0.8*\length,0) {}
				node[coordinate] (Bblower) at ++(0,-\height) {}
				;

				\cone{c1}{90}{1};
				\cone{B1lower}{-90}{0.7};
				\cone{Balower}{-135}{0.7};
				\cone{Balower}{-45}{0.7};
				\cone{Bblower}{-45}{1};
				\cone{Bblower}{-135}{0.5};

				\def\one{B2upper}
				\def\two{45}
				\def\three{0.5}
				\def\ratio{0.3} %
				\def\myangle{34.9} %
				\draw[region]
				(\one) -- ++({\two+\myangle/2}:{\three}) arc 
				({\two+90+\myangle /2}:{\two-90-\myangle /2}:{\ratio*\three})
				node[pos=0] (B2-1) {}
				-- ({\one});
				
				\filldraw[region] (B2-1) ++ (135:0.3) circle (0.3)
				node[coordinate] (B2-2) at ++(180:0.3) {}
				;
				\cone{B2-2}{135}{0.5};

				\draw[newedge]
				(a) .. controls ++(80:2.5) and ++(110:2.5) .. (b)
				;
				
				\foreach \p in {Balower,B1lower,B2upper,Bblower,c1,c2,c3,B2-1,B2-2}{
					\node[vertex] at (\p) {};
				}
				
				\foreach \p in {a,b}{
					\node[hollow] at (\p) {};
				}
			\end{tikzpicture}
		}
		\caption{A graph $G$ before the insertion of the edge $(u,v)$. It includes the cut vertices $c_1, c_2, c_3$. Grey areas indicate $2$-connected components.}
		\end{subfigure} \quad
			\begin{subfigure}{.45\textwidth}\centering
			\scalebox{0.7}{
				\begin{tikzpicture}[
					scale=1.4,
					vertex/.style={circle,fill=black,inner sep=1.6pt},
					hollow/.style={circle,draw, fill=white,inner sep=1.6pt},
					region/.style={draw=black,fill=gray!35},
					edge/.style={line width=1pt},
					bluepath/.style={blue!60,line width=1pt},
					redpath/.style={red!60,line width=1pt}
					]

					\newcommand{\cone}[3]{
						\def\ratio{0.3} %
						\def\myangle{34.9} %
						\draw[region]
						({#1}) -- ++({#2+\myangle/2}:{#3}) arc ({#2+90+\myangle /2}:{#2-90-\myangle /2}:{\ratio*#3}) -- ({#1})
					};
					
					\def\length{1.5}
					\def\height{0.6}
					\filldraw[draw=black,fill=gray!15,name path=ellipse, line width=1pt] (0,0) ellipse ({\length} and {\height})
					node (Bab) at ++(0,-0.5*\height) {$B_{uv}$}
					node[coordinate,label={left:$u$}] (a) at ++(-0.5*\length,0.4*\height) {}
					node[coordinate,label={right:$v$}] (b) at ++(0.5*\length,0.4*\height) {}
					node[coordinate, label=below:$c_2$] (r2) at (0,0.6*\height) {}
					node[coordinate, label=below:$c_3$] (r3) at (0.25*\length,0.6*\height) {}
					;
					
					\path[name path=p-1] (-0.8*\length,0) -- ++(0,-1.1*\height);
					\path[name path=p-2] (-0.5*\length,0) -- ++(0,1.1*\height);
					\path[name path=p-3] (0*\length,0) -- ++(0,-1.1*\height);
					\path[name path=p-4] (0.7*\length,0) -- ++(0,-1.1*\height);
					\path[name path=p-5] (0.8*\length,0) -- ++(0,1.1*\height);
					
					\path[name intersections={of= ellipse and p-1, by=1}];
					\path[name intersections={of= ellipse and p-2, by=2}];
					\path[name intersections={of= ellipse and p-3, by=3}];
					\path[name intersections={of= ellipse and p-4, by=4}];
					\path[name intersections={of= ellipse and p-5, by=5}];
					
					\cone{1}{-180}{0.7};
					\cone{1}{-90}{0.7};
					
					\cone{2}{90}{1};
					
					\cone{3}{-90}{0.7};
					
					\cone{4}{-20}{1};
					\cone{4}{-100}{0.5};

					\def\one{5}
					\def\two{45}
					\def\three{0.5}
					\def\ratio{0.3} %
					\def\myangle{34.9} %
					\draw[region]
					(\one) -- ++({\two+\myangle/2}:{\three}) arc 
					({\two+90+\myangle /2}:{\two-90-\myangle /2}:{\ratio*\three})
					node[pos=0] (B2-1) {}
					-- (\one);
					
					\filldraw[region] (B2-1) ++ (135:0.3) circle (0.3)
					node[coordinate] (B2-2) at ++(180:0.3) {}
					;
					\cone{B2-2}{135}{0.5};
					
					\draw[edge]
					(a) to[bend right=35] (b)
					;
					\draw[virtual]
					(a) -- (2) -- (r2) -- (r3) -- (b)
					;
					
					\foreach \p in {r2,r3}{
						\node[vertex] at (\p) {};
					}

					\foreach \p in {1,2,3,4,5,B2-1,B2-2}{
						\node[vertex] at (\p) {};
					}
					\foreach \p in {a,b}{
						\node[hollow] at (\p) {};
					}
					
					\node[label={[label distance=-1.5mm]above right:$c_1$}] at (2) {};
					
				\end{tikzpicture}
			}
			\caption{The graph $G'$ after the insertion. It includes new virtual edges for the new $3$-connected separating pairs $\{u,c_1\}, \allowdisplaybreaks \{c_1,c_2\}, \allowdisplaybreaks \{c_2,c_3\}, \allowdisplaybreaks \{c_3,v\}$.}
	\end{subfigure}
	
	\vspace{3mm}	
	\begin{subfigure}{.5\textwidth}\centering
		\scalebox{0.6}{
			\begin{tikzpicture}[
				scale=1,
				every node/.style={node distance=.7cm},
				circ/.style={
					draw,
					fill=gray!40,
					circle,
					minimum size=7mm
				},
				dot/.style={
					draw,
					fill=black,
					circle,
					minimum size=2mm
				},
				edge/.style={line width=2pt},
				highlight/.style={yellow!45,opacity=0.7}
				]
				\node[circ, label=below:$B_2$] (B2) {};

				\node[dot,left=of B2, label=below:$c_2$] (c2){};
				\node[circ,left=of c2, label=above:$B_1$] (B1){};
				\node[dot,left=of B1, label=below:$c_1$] (c1){};
				\node[circ,left=of c1, label=left:$B_u$] (Ba){};

				\node[dot, right=of B2, label=below:$c_3$] (c3) {};
				\node[circ, right=of c3, label=right:$B_v$] (Bb) {};
				
				\draw[edge]  (Ba) -- (c1) -- (B1) -- (c2) -- (B2) -- (c3) -- (Bb);

				\node[dot, below=of Ba] (Ba-1) {};
				\node[circ, below left=of Ba-1] (Ba-2) {};
				\node[circ, below right=of Ba-1] (Ba-3) {};
				\draw[edge] (Ba) -- (Ba-1) -- (Ba-2);
				\draw[edge] (Ba-1) -- (Ba-3);

				\node[circ, above=of c1] (c1-1) {};
				\draw[edge] (c1) -- (c1-1);

				\node[dot, below=  of B1] (B1-1) {};
				\node[circ, below= of B1-1] (B1-2) {};
				\draw[edge] (B1) -- (B1-1) -- (B1-2);

				\node[dot, above= 0.5cm of B2] (B2-1) {};
				\node[circ] (B2-2) at ([shift=(45:1cm)]B2-1) {};
				\node[dot] (B2-3) at ([shift=(160:0.7cm)]B2-2) {};
				\node[circ] (B2-4) at ([shift=(160:0.7cm)]B2-3) {};
				\node[dot] (B2-5) at ([shift=(160:0.7cm)]B2-4) {};
				\node[circ] (B2-6) at ([shift=(160:0.7cm)]B2-5) {};
				\draw[edge] (B2) -- (B2-1) -- (B2-2) -- (B2-3) -- (B2-4) -- (B2-5) -- (B2-6);

				\node[dot,below=of Bb] (Bb-1) {};
				\node[circ,below left=0.5 of Bb-1] (Bb-2) {};
				\node[circ,below right=of Bb-1] (Bb-3) {};
				\draw[edge] (Bb) -- (Bb-1) -- (Bb-2);
				\draw[edge] (Bb-1) -- (Bb-3);

				\begin{scope}[on background layer]
					\coordinate[left= 0.6cm of Ba] (left);
					\coordinate[right= 0.6cm of Bb] (right);
					\draw[highlight, line width=17mm]
					(left) --(right);
				\end{scope}
			\end{tikzpicture}
		}
		\caption{The biconnected component tree $\biTree(G)$ before the insertion. Biconnected component nodes are represented by grey circles; cut vertex nodes by black circles. The highlighted path from $B_u$ to $B_v$ collapses after the insertion.}
		\label{fig:BCpath:B}
	\end{subfigure} \quad
	\begin{subfigure}{.42\textwidth}\centering
		\scalebox{0.6}{
			\begin{tikzpicture}[
				scale=1,
				every node/.style={node distance=.7cm},
				circ/.style={
					draw,
					fill=gray!40,
					circle,
					minimum size=7mm
				},
				dot/.style={
					draw,
					fill=black,
					circle,
					minimum size=2mm
				},
				edge/.style={line width=2pt},
				highlight/.style={yellow!45,opacity=0.7}
				]

				\node[circ,line width=1.3pt, label=right:$B_{uv}$] (Bab) {};

				\node[above=of Bab] (upper) {};
				\node[dot,left=2cm of upper, label=left:$c_1$] (c1) {};
				\node[dot,right=2cm of upper] (upper-right) {};
				\node[dot,below=of Bab] (lower) {};
				\node[dot, left=2cm of lower] (lower-left) {};
				\node[dot, right=2cm of lower] (lower-right) {};
				\foreach \p in {c1,upper-right,lower-left,lower,lower-right}{
					\draw[edge] (Bab) -- (\p);
				}

				\node[circ, below left=of lower-left] (Ba-2) {};
				\node[circ, below right=of lower-left] (Ba-3) {};
				\draw[edge] (lower-left) -- (Ba-2);
				\draw[edge] (lower-left) -- (Ba-3);

				\node[circ, above=of c1] (c1-1) {};
				\draw[edge] (c1) -- (c1-1);

				\node[circ, below= of lower] (B1-2) {};
				\draw[edge] (lower) -- (B1-2);

				\node[circ] (B2-2) at ([shift=(45:1cm)]upper-right) {};
				\node[dot] (B2-3) at ([shift=(160:0.7cm)]B2-2) {};
				\node[circ] (B2-4) at ([shift=(160:0.7cm)]B2-3) {};
				\node[dot] (B2-5) at ([shift=(160:0.7cm)]B2-4) {};
				\node[circ] (B2-6) at ([shift=(160:0.7cm)]B2-5) {};
				\draw[edge] (upper-right) -- (B2-2) -- (B2-3) -- (B2-4) -- (B2-5) -- (B2-6);

				\node[circ,below left=0.5cm of lower-right] (Bb-2) {};
				\node[circ,below right=of lower-right] (Bb-3) {};
				\draw[edge] (lower-right) -- (Bb-2);
				\draw[edge] (lower-right) -- (Bb-3);

				\begin{scope}[on background layer]
					\coordinate[left= 0.6cm of Bab] (left);
					\coordinate[right= 0.7cm of Bab] (right);
					\draw[highlight, line width=15mm]
					(left) --(right);
				\end{scope}
			\end{tikzpicture}
		}
		\caption{The biconnected component tree $\biTree(G')$ after the insertion.}
		\label{fig:BCpath:A}
	\end{subfigure}
	
	\vspace{3mm}
	\begin{subfigure}{.90\textwidth}
		\centering
		\scalebox{1}{
			\begin{tikzpicture}[
				scale=1,
				every node/.style={inner sep=0pt, label distance=1mm},
				vertex/.style={circle,fill=black,inner sep=1pt},
				C/.style={fill=white, inner sep=1mm},
				hollow/.style={circle,draw, fill=white,inner sep=1.2pt},
				region/.style={fill=gray!35, draw=black},
				edge/.style={line width=1pt},
				tri/.style={
					draw,
					line width=0.5pt,
					fill=gray!40,
					regular polygon,
					regular polygon sides=3,
					minimum size=6mm
				},
				onhatch/.style={fill=white, inner sep=0.5pt, rounded corners=3pt}
				]

				\draw[edge] (0,0) -- (-2,-1)
				coordinate (1-1)
				-- (-3,-1.5)
				coordinate(1-2)
				;
				\draw[edge] (0,0) -- (-2/3,-1)
				coordinate(2-1)
				-- (-1,-1.5)
				coordinate(2-2)
				; 
				\draw[edge] (0,0) -- (2/3,-1)
				coordinate(3-1)
				-- (1,-1.5)
				coordinate(3-2)
				; 
				
				\draw[edge] (0,0) -- (2,-1)
				coordinate(4-1)
				-- (3,-1.5)
				coordinate (4-2)
				;
				
				\node[Pnode, fill=white, label=right:\footnotesize$u\,c_1$] at (1-1) (R1-1) {};
				\node[Adot] at ($(R1-1)+(0,0.15)$){};
				\node[Adot] at ($(R1-1)+(0,-0.15)$){};
				\node[Pnode, fill=white, label=right:\footnotesize$c_1\,c_2$] at (2-1) (R2-1) {};
				\node[Adot] at ($(R2-1)+(0,0.15)$){};
				\node[Adot] at ($(R2-1)+(0,-0.15)$){};
				\node[Pnode, fill=white, label=right:\footnotesize$c_2\,c_3$] at (3-1) (R3-1) {};
				\node[Adot] at ($(R3-1)+(0,0.15)$){};
				\node[Adot] at ($(R3-1)+(0,-0.15)$){};
				\node[Pnode, fill=white, label=right:\footnotesize$c_3\,v$] at (4-1) (R4-1) {};
				\node[Adot] at ($(R4-1)+(0,0.15)$){};
				\node[Adot] at ($(R4-1)+(0,-0.15)$){};

				\newcommand{\tri}[1]{
					\filldraw[draw, pattern=north east lines, pattern color=gray!60,
					regular polygon, regular polygon sides=3, minimum size=11mm, inner sep=0pt,
					rounded corners=1.5mm] (#1) -- ++(-60:1.7) -- ++(180:1.7) -- (#1);
				}
				
				\tri{1-2};
				\tri{2-2};
				\tri{3-2};
				\tri{4-2};
				
				\node[label={[onhatch, label distance=1.1cm]below:\scriptsize$\triTree(B'_u)$}] at (1-2) {};
				\node[label={[onhatch, label distance=1.1cm]below:\scriptsize$\triTree(B'_1)$}] at (2-2) {};
				\node[label={[onhatch, label distance=1.1cm]below:\scriptsize$\triTree(B'_2)$}] at (3-2) {};
				\node[label={[onhatch, label distance=1.1cm]below:\scriptsize$\triTree(B'_v)$}] at (4-2) {};

				\node[tri, minimum size=1.8cm] at (0,0.1) {};
				\draw[edge] (-0.375,0.1) to[bend left=50] (0.375,0.1);
				\draw[virtual] (0.375,0.1) to[bend right=20] (0.27,-0.1)
				to[bend right=60] (0.0,-0.1)
				to[bend right=60] (-0.27,-0.1)
				to[bend right=20] (-0.375,0.1)
				;
				\node[hollow, label={[font=\tiny, above right=0.05cm]$u$}] at (-0.375,0.1) {};
				\node[hollow, label={[font=\tiny, above left=0.05cm]$v$}] at (0.375,0.1) {};
				\node[vertex, label={[font=\tiny, below=0.23cm]$c_1$}] at (-0.27,-0.1) {};
				\node[vertex, label={[font=\tiny, below=0.23cm]$c_2$}] at (-0,-0.1) {};
				\node[vertex, label={[font=\tiny, below=0.23cm]$c_3$}] at (0.27,-0.1) {};

			\end{tikzpicture}		
		}
		\caption{The triconnected component tree $\triTree(B_{uv})$ after the insertion. It consists of a triconnected component node for the new cycle component with $u, c_1, c_2, c_3, v$, separating pair nodes for the corresponding new $3$-connected separating pairs, and the triconnected component trees of the biconnected components that result from inserting the appropriate new virtual edge into $B_u, B_1, B_2, B_v$, respectively.}
		\label{fig:BCpath:3tree}
	\end{subfigure}
	\caption{
		Illustration of the effect of a $\es[+]{1}{2}$ edge insertion. 
	} 	\label{fig:BCpath}
\end{figure*}

Let $(B_u=B_0) c_1B_1c_2\cdots c_{k-1} B_{k-1} c_{k} (B_k = B_v)$ be the $B_u$ to $B_v$ path in $\biTree(G)$, where the $B_i$ are biconnected component nodes and the $c_i$ are cut vertex nodes.
After the insertion, all the biconnected components on this path merge into one biconnected component $B_{uv}$. The new biconnected component tree $\biTree(G')$ is obtained from $\biTree(G)$ by replacing the $B_u$ to $B_v$ path by a new node for $B_{uv}$. To this end, the cut vertex nodes of degree $2$ and all biconnected component nodes are replaced by $B_{uv}$; all remaining cut vertex nodes of degree more than $2$ are connected by an edge to the new node $B_{uv}$.

The triconnected component tree for the newly formed biconnected component $B_{uv}$ is obtained from the triconnected component trees of the old biconnected components that merge into $B_{uv}$ as follows. Let $c_0 = u, c_1, \ldots, c_k, c_{k+1} = v$ be the list of cut vertex nodes between $B_u$ and $B_v$ in the order they appear, preceded by $u$ and followed by $v$. Every pair $\{c_i, c_{i+1}\}$ of consecutive vertices in the list forms a $3$-connected separating pair in $B_{uv}$.
In every component $B_i$, the edge $(c_i, c_{i+1})$ is inserted as virtual edge if it is not present. A new separating pair node for $\{c_i, c_{i+1}\}$ is added to the resulting triconnected component tree for the component, if it not already exists (and if $B_i$ not only consists of the single edge between these nodes), and it becomes a neighbour of the unique triconnected component node whose component contains $c_i$ and~$c_{i+1}$.

The tree $\triTree(B_{uv})$ contains all resulting triconnected component trees and an additional new triconnected component node for the cycle that consists of the (potentially virtual) edges $(c_i, c_{i+1})$ and the inserted edge $(u,v)$. The latter node becomes a neighbour of the separating pair nodes $\{c_i, c_{i+1}\}$.

\smallskip

Edge changes of type $\es[-]{2}{1}$ reverse edge changes of type $\es[+]{1}{2}$. Effects on the decompositions are also reversed.

\subparagraph*{Changes only relevant for the biconnected component tree}

The following types of changes only affect $\biTree(G)$, the triconnected component tree $\triTree(G)$ is unchanged.

\subparagraph{$\es[+]{0}{2}$.} %
Suppose $u$ and $v$ belong to distinct connected components $A_u$ and $A_v$ before the edge insertion. After the insertion, $A_{u}$ and $A_v$ merge into one connected component $A_{uv}$.
	The biconnected component tree $\biTree(A_{uv})$ of $A_{uv}$ is obtained from $\biTree(A_u)$ and $\biTree(A_v)$ as follows. It contains a node $B_{uv}$ for the trivial biconnected component that only consists of the edge $(u,v)$.
	If $u$ is already a cut vertex node in $\biTree(A_u)$, this cut vertex node is connected to $B_{uv}$. Otherwise, first a new cut vertex node for $u$ is introduced to $\biTree(A_u)$ as a neighbour to the unique biconnected component node whose component includes $u$.
	Analogously, the cut vertex node for $v$ from $\biTree(A_v)$ is connected to $B_{uv}$, potentially after introducing that cut vertex node.

\smallskip

Edge changes of type $\es[-]{2}{0}$ reverse edge changes of type $\es[+]{0}{2}$. Effects on the decompositions are also reversed.
 
\section{Maintaining Biconnected and Triconnected Component Trees}\label{section:componenttrees}
Biconnected and triconnected component trees are the fundamental data structures we use to identify isomorphic components of planar graphs.
In this section, we explain how these trees can be maintained, including additional information such as distances.

Nodes of biconnected and triconnected component trees correspond to biconnected component and cut vertices as well as triconnected components and separating pairs, respectively. 
Our dynamic program represents these objects by tuples of vertices of the underlying graph: a cut vertex is represented by the vertex itself, a separating pair by a $2$-tuple containing both vertices of that pair, and a biconnected respectively triconnected component by a tuple of two respectively three vertices that are included in that component.
Biconnected and triconnected component trees are then maintained as auxiliary relations over this encoding.

\begin{lemma}
\label{theorem:maintainability-decomposition}
	Biconnected respectively triconnected component trees for connected respectively biconnected components of arbitrary graphs can be maintained in $\DynFO$ under insertions and deletions of edges.
\end{lemma}
\begin{proof}
	Reachability and connectivity in undirected graphs can be maintained in \DynFO \cite{PatnaikI97}. 
	By maintaining connectivity information when ignoring a constant number $k$ of vertices and their edges, for all choices of $k$ vertices, also $k$-connectivity can be maintained. 
	From this information we can identify cut vertices, $3$-connected separating pairs and biconnected and triconnected components in \FO. 
\end{proof}

In addition to the component trees themselves, our dynamic programs require information about the distances between nodes of the trees. 
Note that while distances in undirected graphs -- and thus also for trees -- can be maintained in $\DynFO$ under single-edge changes \cite{GradelS12}, this is more difficult for biconnected and triconnected component trees as these are subject to larger changes when the underlying graph changes.

In the remainder of this section, we explain how to maintain distances for triconnected and biconnected component trees.

	\subsection{Distances in triconnected component trees}\label{section:biconnected:dist}

In this section, we explain how to maintain the relation $\spqrDist$ that stores tuples $(\tpl x, \tpl y, d)$ where $\tpl x$ and $\tpl y$ are nodes of a triconnected component tree and $d$ is their distance in the tree.

\begin{lemma}\label{lemma:tritree-distances}
	The relation $\spqrDist$ can be maintained in $\DynFO$ under insertions and deletions of edges in the underlying graph $G$, provided that $G$ stays planar.
\end{lemma}

Distances in triconnected component trees are not affected by changes of type $\es[+/-]{3}{3}$ and can be easily updated after changes of type $\es[+]{2}{3}$, as one only has to deduct the known length of the coalescing path from the affected distances. More precisely, if an edge $(u,v)$ is inserted that collapses a path $\pi_{uv}$ between components $R_u$ and $R_v$ of a triconnected component tree, then any path $\pi_{xy}$ between components $R_x$ and $R_y$ that has an overlap of length $d$ with $\pi_{uv}$ is shortened by this length $d$.

Changes of type $\es[-]{3}{2}$ are more difficult, as a node for a 3-connected component $C$ can unfurl into a path of arbitrary length, see Figure  \ref{fig:distSPQRdel} for an illustration.
We exploit the unique combinatorial embedding of $C$ to determine the length of this path. A key observation is that the separating pairs $\{a_1,b_1\}, \ldots, \{a_k, b_k\}$ in the unfurled path appear in this order in the \emph{disc} formed by the two faces $F_1$ and $F_2$ adjacent to the deleted edge in $C$. Within~$C$, for each such pair $\{a_i, b_i\}$, one of the vertices lies on $F_1$ and the other on $F_2$.

We next define such discs for arbitrary (not necessarily adjacent) faces $F_1$ and $F_2$, and show how to use them to maintain distances, see also Figure~\ref{fig:distSPQR:discgraph} for an illustration. We also refer to the notions of societies and rural neighbourhoods of Robertson and Seymour~\cite{RobertsonS90} for related constructions based on cyclically ordered boundary vertices.

\begin{figure}[t]
	\centering
	\begin{subfigure}{0.48\textwidth}
		\resizebox{.97\textwidth}{!}{%

\tikzset{every picture/.style={line width=0.75pt}} %

\begin{tikzpicture}[
	vertex/.style={
		circle,
		fill=black,
		inner sep=0pt,
		minimum size=8pt
	},
	outer/.style={
		draw,
		fill={rgb, 255:red, 184; green, 233; blue, 134 },
		dash pattern=on 3.75pt off 3pt on 7.5pt off 1.5pt
	},
	hole/.style={
		draw,
		fill=white,
		dash pattern=on 3.75pt off 3pt on 7.5pt off 1.5pt,
		line width=.75pt
	},
	every label/.style={font=\normalsize}
	]

	\node[vertex, label={[above left]$u$}] (a) {};
	\node[vertex, label={[above right]$v$}, right=2.55cm of a] (b) {};	
	\node[vertex, label={[above left]$a_{10}$}, below left=1.48cm and 2.05cm of a] (c) {};
	\node[vertex, label={[above right]$a_{3}$}, below right=1.48cm and 2.05cm of b] (d) {};
	\node[vertex, label={[below right=0.1cm]$b_{7}$}, below left =2.25cm  and 0.05cm of a] (e) {};
	\node[vertex, label={[below left=0.1cm]$b_{1}$}, below right =2.25cm  and 0.05cm of b] (f) {};

	\node[vertex, label=right:$b_{6}$, below left =1.35cm and 0.5cm of e] (g) {};
	\node[vertex, label=left:$b_{2}$, below right =1.35cm and 0.5cm of f] (h) {};
	\node[vertex, label=left:$a_{9}$, left=2.0cm of g] (i) {};
	\node[vertex, label=right:$a_{4}$, right=2.0cm of h] (j) {};

	\node[vertex, label={[above right=-0.1cm]$b_{5}$}, below right =1.35cm and 0.5cm of g] (k) {};
	\node[vertex, label={[above left=-0.1cm]$b_{3}$}, below left =1.35cm and 0.5cm of h] (l) {};
	\node[vertex, label={[below left=0.1cm]$a_{8}$}, below=4.45cm of c] (m) {};
	\node[vertex, label={[below right=0.1cm]$a_{5}$}, below=4.45cm of d] (n) {};

	\node[vertex, label={[below left=0.1cm]$a_{7}$}, below right=2.20cm and 0.05cm of k] (o) {};
	\node[vertex, label={[below right=0.1cm]$a_{6}$}, below left=2.20cm and 0.05cm of l] (p) {};
	\node[vertex, label=above:$b_{4}$, below right=.58cm and 1.47cm of k] (q) {};

	\path[outer]
	(j.center) --
	(n.center) --
	(p.center) --
	(o.center) --
	(m.center) --
	(i.center) --
	(c.center) --
	(a.center) --
	(e.center) --
	(g.center) --
	(k.center) --
	(q.center) --
	(l.center) --
	(h.center) --
	(f.center) --
	(b.center) --
	(d.center) --
	cycle;

	\path[hole]
	(c.center) -- ++(0.8,0.3) -- ++ (1.1,-0.2) --
	(e.center) -- ++(-1.1, -0.3) -- ++ (-0.8,0.4) --
	(c.center);
	
	\node[below right= 0cm and 0.3cm of c] {$F_{a_{10},b_{7}}$};
	
	\path[hole]
	(d.center) -- ++(-0.8,0.3) -- ++ (-1.1,-0.2) --
	(f.center) -- ++(1.1, -0.3) -- ++ (0.8,0.4) --
	(d.center) ;
	
	\node[below left= 0cm and 0.3cm of d] {$F_{a_{3},b_{1}}$};
	
	\path[hole]
	(i.center) -- ++(0.8,0.4) -- ++ (0.8,0) --
	(g.center) -- ++(-0.8, -0.4) -- ++ (-0.8,0) --
	(i.center);
	
	\node[right= 0.4cm of i] {$F_{a_{9},b_{6}}$};
	
	\path[hole]
	(j.center) -- ++(-0.8,0.4) -- ++ (-0.8,0) --
	(h.center) -- ++(0.8, -0.4) -- ++ (0.8,0) --
	(j.center);
	
	\node[left= 0.5cm of j] {$F_{a_{4},b_{2}}$};
	
	\path[hole]
	(l.center) -- ++(0.5,0.4) -- ++ (0.9,-0.2) --
	(n.center) -- ++(-1.1, 0.2) -- 
	(l.center);
	
	\node[above left= 0.4cm and 0.5cm of n] {$F_{a_{5},b_{3}}$};
	
	\path[hole]
	(k.center) -- ++(-0.5,0.4) -- ++ (-0.9,-0.2) --
	(m.center) -- ++(1.1, 0.2) -- 
	(k.center);
	
	\node[above right= 0.4cm and 0.5cm of m] {$F_{a_{8},b_{5}}$};
	
	\path[hole]
	(q.center) -- ++(1.1,-0.2) -- ++ (0.3,-0.6) --
	(p.center) -- ++(-1.1, 0.5) -- 
	(q.center);
	
	\node[below right= 0.3cm and 0.1cm of q] {$F_{a_{6},b_{4}}$};
	
	\path[hole]
	(q.center) -- ++(-1.1,-0.2) -- ++ (-0.3,-0.6) --
	(o.center) -- ++(1.1, 0.5) -- 
	(q.center);
	
	\node[below left= 0.3cm and 0.1cm of q] {$F_{a_{7},b_{4}}$};

	\draw[red] (a) edge[ultra thick]  (b);

	\foreach \v in {a,b,c,d,e,f,g,h,i,j,k,l,m,n,o,p,q}
	\node[vertex] at (\v) {};
	
	\node[right = 1.85cm of g, font=\LARGE]  (f2) {$F_{2}$};
	\node[above right = 0.1cm and 1.2cm of a.center, font=\LARGE]  {$F_{1}$};
	
\end{tikzpicture} %
}
		\caption{A $3$-connected graph $C$ that after deletion of edge $(u,v)$ unfurls into a path of triconnected components with separating pairs $\{a_3,b_1\},\allowbreak \{a_4,b_2\},\allowbreak \{a_5,b_3\},\allowbreak \{a_6,b_4\},\allowbreak \{a_7,b_4\},\allowbreak \{a_8,b_5\},\allowbreak \{a_9,b_6\},\allowbreak \{a_{10},b_7\}$. These pairs are also the nodes of the disc graph $\Disc(C, F_1, F_2)$.}   
		\label{fig:distSPQRdel}
	\end{subfigure}	
	\hfill
	\begin{subfigure}{0.48\textwidth}
		\resizebox{.97\textwidth}{!}{%

\tikzset{every picture/.style={line width=0.75pt}} %

\begin{tikzpicture}[
	vertex/.style={
		circle,
		fill=black,
		inner sep=0pt,
		minimum size=8pt
	},
	outer/.style={
		draw,
		fill={rgb, 255:red, 184; green, 233; blue, 134 },
		dash pattern=on 3.75pt off 3pt on 7.5pt off 1.5pt
	},
	hole/.style={
		draw,
		fill=white,
		dash pattern=on 3.75pt off 3pt on 7.5pt off 1.5pt,
		line width=.75pt
	},
	every label/.style={font=\normalsize}
	]

	\node[vertex, label={[above left]$a_1$}] (a) {};
	\node[vertex, label={[above right]$a_2$}, right=2.55cm of a] (b) {};	
	\node[vertex, label={[above left]$a_{10}$}, below left=1.48cm and 2.05cm of a] (c) {};
	\node[vertex, label={[above right]$a_{3}$}, below right=1.48cm and 2.05cm of b] (d) {};
	\node[vertex, label={[below right=0.1cm]$b_{7}$}, below left =2.25cm  and 0.05cm of a] (e) {};
	\node[vertex, label={[below left=0.1cm]$b_{1}$}, below right =2.25cm  and 0.05cm of b] (f) {};

	\node[vertex, label=right:$b_{6}$, below left =1.35cm and 0.5cm of e] (g) {};
	\node[vertex, label=left:$b_{2}$, below right =1.35cm and 0.5cm of f] (h) {};
	\node[vertex, label=left:$a_{9}$, left=2.0cm of g] (i) {};
	\node[vertex, label=right:$a_{4}$, right=2.0cm of h] (j) {};

	\node[vertex, label={[above right=-0.1cm]$b_{5}$}, below right =1.35cm and 0.5cm of g] (k) {};
	\node[vertex, label={[above left=-0.1cm]$b_{3}$}, below left =1.35cm and 0.5cm of h] (l) {};
	\node[vertex, label={[below left=0.1cm]$a_{8}$}, below=4.45cm of c] (m) {};
	\node[vertex, label={[below right=0.1cm]$a_{5}$}, below=4.45cm of d] (n) {};

	\node[vertex, label={[below left=0.1cm]$a_{7}$}, below right=2.20cm and 0.05cm of k] (o) {};
	\node[vertex, label={[below right=0.1cm]$a_{6}$}, below left=2.20cm and 0.05cm of l] (p) {};
	\node[vertex, label=above:$b_{4}$, below right=.58cm and 1.47cm of k] (q) {};

	\path[outer]
	(j.center) --
	(n.center) --
	(p.center) --
	(o.center) --
	(m.center) --
	(i.center) --
	(c.center) --
	(a.center) --
	(e.center) --
	(g.center) --
	(k.center) --
	(q.center) --
	(l.center) --
	(h.center) --
	(f.center) --
	(b.center) --
	(d.center) --
	cycle;

	\path[hole]
	(c.center) -- ++(0.8,0.3) -- ++ (1.1,-0.2) --
	(e.center) -- ++(-1.1, -0.3) -- ++ (-0.8,0.4) --
	(c.center);
	
	\path[hole]
	(d.center) -- ++(-0.8,0.3) -- ++ (-1.1,-0.2) --
	(f.center) -- ++(1.1, -0.3) -- ++ (0.8,0.4) --
	(d.center) ;
	
	\path[hole]
	(i.center) -- ++(0.8,0.4) -- ++ (0.8,0) --
	(g.center) -- ++(-0.8, -0.4) -- ++ (-0.8,0) --
	(i.center);
	
	\path[hole]
	(j.center) -- ++(-0.8,0.4) -- ++ (-0.8,0) --
	(h.center) -- ++(0.8, -0.4) -- ++ (0.8,0) --
	(j.center);
	
	\path[hole]
	(l.center) -- ++(0.5,0.4) -- ++ (0.9,-0.2) --
	(n.center) -- ++(-1.1, 0.2) -- 
	(l.center);
	
	\path[hole]
	(k.center) -- ++(-0.5,0.4) -- ++ (-0.9,-0.2) --
	(m.center) -- ++(1.1, 0.2) -- 
	(k.center);
	
	\path[hole]
	(q.center) -- ++(1.1,-0.2) -- ++ (0.3,-0.6) --
	(p.center) -- ++(-1.1, 0.5) -- 
	(q.center);
	
	\path[hole]
	(q.center) -- ++(-1.1,-0.2) -- ++ (-0.3,-0.6) --
	(o.center) -- ++(1.1, 0.5) -- 
	(q.center);

	\draw (a) edge[very thick]  (b);
	\draw (e.center) edge[very thick]  (f.center);
	\foreach \v in {a,b,c,d,e,f,g,h,i,j,k,l,m,n,o,p,q}
	\node[vertex] at (\v) {};
	
	\node[right = 1.85cm of g, font=\LARGE]  (f2) {$F_{2}$};
	\node[above right = 0.1cm and 1.2cm of a.center, font=\LARGE]  {$F_{1}$};
	
	\draw[gray!80, thick,arrows = {-Stealth}]
	($(f2)+(100:0.7cm)$)
	arc[start angle=100, end angle=-210, radius=0.7cm];
	
	\draw[gray!80, thick,arrows = {-Stealth}]
	($(f2)+(100:5.5cm and 5cm)$)
	arc[start angle=100, end angle=-210, x radius=5.5cm, y radius = 5cm];

\end{tikzpicture} %
}
		\caption{Two faces $F_1$ and $F_2$ in an embedding of a $3$-connected component $C$. The disc graph $\Disc(C,F_1,F_2)$ consists of the vertices $(a_1,b_7),\allowbreak (a_2,b_1),\allowbreak (a_3,b_1),\allowbreak (a_4,b_2),\allowbreak (a_5,b_3),\allowbreak (a_6,b_4),\allowbreak (a_7,b_4),\allowbreak (a_8,b_5),\allowbreak (a_9,b_6),\allowbreak (a_{10},b_7)$.}
		\label{fig:distSPQR:discgraph}
	\end{subfigure}	
	\caption{Illustration of how to address the unfurling of a 3-connected component into a path of triconnected components using disc graphs (see Definition \ref{def:potdist}).}
	\label{fig:distSPQR}
\end{figure}

\begin{definition}[Disc graph]\label{def:potdist}
	Let $C$ be a $3$-connected graph, $\Emb(C, F_1)$ the embedding of $C$ with the combinatorial outer face $F_1$, and $F_2$ another combinatorial face of $\Emb(C, F_1)$. 
	The directed \emph{disc graph} $\Disc(C, F_1, F_2)$ consists of nodes $(a,b)$, where $a$, $b$ are distinct vertices of $C$, such that
	\begin{enumerate}[(1)]
		\item $a$ appears on $F_1$ and not on $F_2$, and $b$ appears on $F_2$ and not on $F_1$,
		\item there exists a face $F_{a,b}$ distinct from $F_1$ and $F_2$ such that both $a$ and $b$ appear on $F_{a,b}$,
		\item there does not exist another pair $(a',b')$ such that
		\begin{enumerate}[(i)]
			\item $a\neq a'$ and $b\neq b'$,
			\item $a', b'$ satisfy (1) and (2), and
			\item the vertices $a,a',b,b'$ appear on any face of $\Emb(C, F_1)$ in that cyclic order.
		\end{enumerate}
	\end{enumerate}
	The graph $\Disc(C,F_1,F_2)$ has an edge from a node $(a,b)$ to a node $(a',b')$ if for any other node $(a'',b'')$ of $\Disc(C, F_1, F_2)$ the vertices $a,a',a''$ and the vertices $b,b',b''$ appear in that (clockwise) cyclic order on the faces $F_1$ and $F_2$, respectively.
\end{definition}

We observe that $\Disc(C, F_1, F_2)$ is a directed cycle. %
Its distances correspond to distances in an unfurled path of the triconnected component tree, if an edge deletion merges the faces $F_1$ and $F_2$ and destroys a $3$-connected component. %

\begin{lemma}
	\label{lem:spqrDist}
	Let $C$ be a $3$-connected graph and let $e=(u,v)$ be an edge of $C$ such that $C-\{e\}$ is not $3$-connected. Let $F_1$ and $F_2$ be the faces of $C$ adjacent to $e$. %
	Then the following holds:
	\begin{enumerate}[(i)]
		\item The triconnected component tree $\triTree(C-\{e\})$ is a path between triconnected component nodes $R_u$ and $R_v$ that contain $u$ and $v$, respectively.
		\item $\{a,b\}$ is a $3$-connected separating pair of $C-\{e\}$ if and only if $(a,b)$ or $(b,a)$
		is a node of $\Disc(C, F_1, F_2)$.
		\item The path $\triTree(C-\{e\})$ has length $2d$, where $d$ is the number of nodes of $\Disc(C, F_1, F_2)$, and the $3$-connected separating pairs of $C-\{e\}$ appear in the same order in this path as their corresponding nodes in $\Disc(C, F_1, F_2)$.
	\end{enumerate}
	
\end{lemma}

In the proof of the lemma we will repeatedly use the following simple consequence of Jordan's Curve Theorem \cite[Chapter 4]{Diestel12}.
	\begin{fact}\label{fact:curve:disc}
		Let $G$ be a planar graph with two vertices $a,a'$ such that in some planar embedding of $G$ the vertices $a$ and $a'$ are both incident to two faces $F, F'$.
		Suppose there are vertices $b, b'$ such that the vertices $a, b, a', b'$ appear in this cyclic order either on the face $F$ or on the face $F’$.  Then $\{a,a'\}$ is a separating pair of $G$ and $b$ and $b'$ are in different connected components of $G-\{a,a'\}$.
	\end{fact}
	\begin{proofof}{Fact~\ref{fact:curve:disc}}
		We can draw a simple closed curve through the faces $F$ and $F'$ such that the curve touches their boundaries exactly at the vertices $a$ and $a'$. Exactly one of $b, b'$ is in the interior of this curve, the other one is in the exterior, so these vertices are disconnected in $G - \{a,a'\}$.
	\end{proofof}

\begin{proofof}{Lemma \ref{lem:spqrDist}}
	For notational convenience in the proof, we denote $\Disc(C,F_1,F_2)$ by~$\mathcal{D}_e$.

 \subparagraph*{Proof of part (i).} See also the proof of Lemma 10 in \cite{HolmIKLR18,HolmIKLR18arx}.
  Suppose towards a contradiction that $\triTree(C-\{e\})$ is not a path. Then there must be a $3$-connected separating pair $\{a,b\}$ such that $(C-\{e\}) - \{a,b\}$ contains at least three connected components. But then $C - \{a,b\}$ is also disconnected, so $\{a,b\}$ must be a separating pair in $C$, contradicting the assumption that $C$ is $3$-connected.
  We argue similarly that this path is from a triconnected component $R_u$ that contains $u$ to a triconnected component $R_v$ that contains $v$: if there would be a $3$-connected separating pair $\{a,b\}$ that is not on the path from $R_u$ to $R_v$, it would already be a separating pair in $C$.
  
  We note that the path is non-trivial, i.e. it does not consist of a single triconnected component. It cannot consist of a single $3$-connected component, as $C-\{e\}$ is not $3$-connected by assumption. It cannot consist of a single cycle component, as otherwise $\{u,v\}$ would be a separating pair of $C$.
  \subparagraph*{Proof of part (ii).}
 We first show that every node of $\mathcal{D}_e$ corresponds to a $3$-connected separating pair of $C-\{e\}$.
Consider a node $(a,b)$ of $V(\mathcal{D}_e)$. For $\{a,b\}$ to be a $3$-connected separating pair, we need to show that (a) $a$ and $b$ are $3$-connected in $C-\{e\}$, and (b) removing $a$ and $b$ from $C-\{e\}$ disconnects the graph. 

For (a), let $F_{a,b}$ be a face distinct from $F_1$ and $F_2$ such that $a$ and $b$ are both incident to this face, which needs to exist according to condition (2) of Definition \ref{def:potdist}.
Observe that $F_{a,b}$ still exists in $C-\{e\}$, as $e$ cannot be incident to $F_{a,b}$ in $C$: it is already incident to $F_1$ and $F_2$ and cannot be incident to a third face.
The boundary of $F_{a,b}$ corresponds to two internally vertex disjoint paths between $a$ and $b$ in $C-\{e\}$.

For the third path between $a$ and $b$ that is internally vertex disjoint from $F_{a,b}$ in $C-\{e\}$, consider the path $\rho_1$ that is the composition of the path from $a$ to $u$ along the face $F_1$ and the path from $u$ to $b$ along the face $F_2$, both avoiding $v$, and the path $\rho_2$ that is the composition of the path from $a$ to $v$ along the face $F_1$ and the path from $v$ to $b$ along the face $F_2$, both avoiding $u$.
Both paths exist in $C-\{e\}$, as the edge $e = (u,v)$ is not used.

We argue that at least one of $\rho_1$ and $\rho_2$ is internally vertex disjoint from $F_{a,b}$.
Suppose, towards a contradiction, that this is not the case. 
Let $a_1$ and $a_2$ be the successors of $a$ in $\rho_1$ and $\rho_2$, respectively, and $b_1, b_2$ the predecessors of $b$ in the respective paths. It cannot be that both $a_1$ and $a_2$ are incident to $F_{a,b}$, as they are also incident to $F_1$ and as $a_1, a, a_2, b$ appear in that cyclic order on $F_{a,b}$, by Fact~\ref{fact:curve:disc} we have that $C$ is not $3$-connected. Analogously, not both $b_1$ and $b_2$ are incident to $F_{a,b}$.
So, either $a_1$ and $b_2$ are incident to $F_{a,b}$ or $a_2$ and $b_1$ are. Suppose the former case, the other case is analogous.
Notice that $a_1$ cannot be $u$ and $b_2$ cannot be $v$: if $a_1$ equals $u$, then $\{a_1, b\}$ is a separating pair of $C$ according to Fact~\ref{fact:curve:disc}, using that $a_1$ and $b$ are on the faces $F_2$ and  $F_{a,b}$ and the vertices $a_1,v,b,b_1$ appear in this cyclic order on $F_2$ (and $b_1 \neq u$ as $b_1$ is not on $F_{a,b}$ by assumption). Analogously we can argue $b_2 \neq v$.
So, the pair $(a_1, b_2)$ witnesses that Condition (3) of Definition~\ref{def:potdist} is not satisfied for $(a,b)$.

It remains to argue (b), that $\{a,b\}$ is a separating pair of $C-\{e\}$. 
Notice that the faces $F_1$ and $F_2$ that are incident to $e$ are merged in the embedding of $C-\{e\}$. So, $a,u,b,v$ appear in this cyclic order on this merged face and therefore $\{a,b\}$ is a separating pair according to Fact~\ref{fact:curve:disc}, using that $a$ and $b$ are still also incident to $F_{a,b}$. 

We now show that every $3$-connected separating pair of $C-\{e\}$ corresponds to a node in $\mathcal{D}_e$.
Let $\{a,b\}$ be any $3$-connected separating pair of $C-\{e\}$. We argue that either $(a,b)$ or $(b,a)$ satisfies the conditions of Definition~\ref{def:potdist}.
\begin{enumerate}[(1)]
	\item Let us assume, towards a contradiction, that Condition (1) is not satisfied. First observe that $a$ and $b$ cannot be incident to both $F_1$ and $F_2$, as only the endpoints $u,v$ of the edge $e$ have this property and neither $u$ nor $v$ is part of a separating pair of $C-\{e\}$, because this would also be a separating pair of $C$, contradicting the assumption that $C$ is $3$-connected. %
	If neither $a$ nor $b$ is incident to $F_1$ or $F_2$ then $u$ and $v$ are still connected in $(C-\{e\}) - \{a,b\}$ via the (remaining) boundary of that face, which implies that $\{a,b\}$ does not separate $u$ and $v$ in $C-\{e\}$. This is a contradiction to part (i).
	So, Condition~(1) is satisfied and without loss of generality we assume that $a$ appears on $F_1$ and $b$ appears on $F_2$.
	\item Let us consider the case when $(a,b)$ is not an edge in $C$. Since $\{a,b\}$ is a separating pair of $C-\{e\}$, a biconnected planar graph, the vertices $a$ and $b$ must lie on at least two common faces in any embedding of $C-\{e\}$~\cite[Theorem 2]{JajaK88}. In particular, let us consider the embedding of $C-\{e\}$ obtained from the embedding of $C$ ($\textsf{Emb}(C,F_1)$) by removing the edge $e$. One of these faces (where $a$ and $b$ lie together) in the embedding of $C-\{e\}$ is the outer face (that $F_1,F_2$ have merged to), because $a$ lay on $F_1$ and $b$ lay on $F_2$ in $\textsf{Emb}(C,F_1)$ (see Condition (1)). So there is at least one more internal face in the embedding of $C-\{e\}$ where $a$ and $b$ both lie. This face also exists in $\textsf{Emb}(C,F_1)$, and is distinct from both $F_1$ and $F_2$: in the embedding of $C-\{e\}$ we can draw the edge $e$ in the outer face, without affecting any internal face, to obtain the embedding of $C$.  
	
	In the case when $(a,b)$ is an edge in $C$ (and thus in $C-\{e\}$), there are at least three common faces, in any embedding of $C-\{e\}$, on which $a$ and $b$ both lie~\cite[Theorem 2]{JajaK88}. Again, let us consider the embedding of $C-\{e\}$ obtained from $\textsf{Emb}(C,F_1)$. Let us denote the outer face of this embedding by $F_{12}$.  As argued earlier, $a$ and $b$ lie on $F_{12}$. Furthermore, there are two faces, say $F_a$ and $F_b$, that are adjacent to $(a,b)$ in the embedding of $C-\{e\}$, and one of them\footnote{In fact, both of the faces that are adjacent to $(a,b)$ are distinct from $F_{12}$, but for our purpose just one suffices.}, say $F_a$, must be distinct from $F_{12}$, and hence it should be an internal face. Then the face $F_a$ also exists in $\textsf{Emb}(C,F_1)$ as argued earlier.%

	\item Suppose towards a contradiction that there is a pair $(a',b')$ of vertices witnessing that Condition (3) of Definition~\ref{def:potdist} is not satisfied. Then, by Fact~\ref{fact:curve:disc}, $a$ and $b$ are disconnected in $C - \{a',b'\}$, contradicting the assumption that $\{a,b\}$ is a $3$-connected separating pair.
\end{enumerate}
  \subparagraph*{Proof of part (iii).}
  We show that if two $3$-connected separating pairs $\{a,b\}$ and $\{a',b'\}$ are not disconnected by another $3$-connected separating pair $\{a'',b''\}$ in $C-\{e\}$, that is, are part of a common triconnected component, then the nodes $(a,b)$ and $(a',b')$ are connected by an edge in $\mathcal{D}_e$.
  We show the contraposition and assume that the nodes $(a,b)$ and $(a',b')$ are not connected by an edge. Then there must be another node $(a'',b'')$ such that either $a,a'',a'$ or $b,b'',b'$ appear in this cyclic order on $F_1$ or $F_2$, respectively.
  Suppose the former. Then the vertices $a,a'',a',b''$ appear in this cyclic order on the outer face of the embedding of $C-\{e\}$ and by Fact~\ref{fact:curve:disc}, the $3$-connected separating pair $\{a'',b''\}$ disconnects $a$ and $a'$.
  So $a$ and $a'$ are not in the same triconnected component in $\triTree(C-\{e\})$ and therefore $\{a,b\}$ and $\{a',b'\}$ are disconnected by another $3$-connected separating pair in $C-\{e\}$.
  
  We now argue the length of the path $\triTree(C-\{e\})$.
  Between every two consecutive separating pairs in this path there is a triconnected component, while in $\mathcal{D}_e$ there is an edge between the corresponding nodes, so the length of $\triTree(C-\{e\})$, that is, the distance from $R_u$ to $R_v$, is $2|V(\mathcal{D}_e)|$.
\end{proofof}

Thus distances between separating pair nodes of the triconnected component tree $\triTree(C-\{e\})$ can be inferred from the distances between the nodes of $\Disc(C, F_1, F_2)$.

Therefore, apart from maintaining distances on triconnected component trees, we also maintain the distances on the cycles $\Disc(C, F_1, F_2)$, for every $3$-connected component $C$ and all combinatorial faces $F_1, F_2$. 
Formally, we represent a face $F$ of an embedded $3$-connected planar graph by triples $\tpl f = (f_1, f_2, f_3)$ of vertices that appear in this (clockwise) cyclic order on $F$. This is possible as three vertices can only appear together on at most one face of a $3$-connected planar graph (that is not a single triangle).

Our dynamic algorithm maintains the relations

\begin{itemize}
	\item $\spqrDist$ that stores tuples $(\tpl x, \tpl y, d)$ such that the distance between nodes $\tpl x, \tpl y$ of the triconnected component tree is~$d$;
	\item $\combEmb$ that contains tuples $(\tpl f_1,\tpl f_2)$ such that $\tpl f_1,\tpl f_2$ specify faces $F_1, F_2$ of a 3-connected component $C$ and $F_2$ is a face occurring in the embedding of $C$ determined by the outer face $F_1$; and
	\item $\discDist$ that contains tuples $(\tpl f_1,\tpl f_2,\tpl a,\tpl b, k)$ such that $\tpl f_1,\tpl f_2$ specify faces $F_1, F_2$ of a 3-connected component $C$, $\tpl a$ and $\tpl b$ are nodes of $\Disc(C, F_1, F_2)$, and $k$ is the distance from $\tpl a$ to $\tpl b$ in $\Disc(C, F_1, F_2)$.
	
\end{itemize}

The relation $\combEmb$ can be maintained using the dynamic algorithm from
\cite{DattaKM23}.

The following extends Lemma~\ref{lemma:tritree-distances}.
\akm{Lemma 7 talks about a 3-connected graph unfurling into a path. Then lemma 9 considers a 3-connected component unfurling into a path. We should include a short argument (if there is time) bridging the two settings.}
\begin{lemma}
	The relations $\spqrDist$ and $\discDist$ can be maintained in $\DynFO$ under insertions and deletions of edges in the underlying graph $G$, provided that $G$ stays planar.
\end{lemma}
\begin{proof}
	We distinguish the different cases of the changed edge and explain how the updates required to these relations in each case are expressible in $\FO$.
	Note that given $\combEmb$, any disc graph $\Disc(C, F_1, F_2)$ can directly be defined using first-order formulas, so for the maintenance of $\discDist$ we only need to explain how the distances in these graphs are obtained.
	\subparagraph*{Case $\es[+]{3}{3}$.}
	Let the inserted edge be $e = (u,v)$.
	Clearly, $\spqrDist$ does not require updating since the triconnected component tree remains unchanged. So, we focus on updating $\discDist$. 
	
	Let $C$ be the $3$-connected component where the edge $e$ is being inserted and let $C'$ be the component after the edge insertion. Let $F_1$ and $F_2$ be two faces of $C'$.
	We distinguish the subcases of how many of the faces $F_1$ and $F_2$ are not already present in the embedding of $C$. 
	We show that we can compose $\Disc(C',F_1,F_2)$ from constantly many segments of disc graphs for $C$, for which distances are known, to which we add and from which we remove a constant number of nodes.
	It will follow that the updates can be expressed with first-order formulas.
	
	First suppose that both $F_1$ and $F_2$ are new faces, so, these faces result from the splitting of some face $F$ of $C$ by the inserted edge $e = (u,v)$.
	Then $\Disc(C',F_1,F_2)$ has no nodes. Any node $(a,b)$ of $\Disc(C',F_1,F_2)$ cannot include $u$ or $v$ because of condition (1) of Definition~\ref{def:potdist}. Because of condition (2), there must be another face $F_{a,b}$ that includes these vertices, and this face is already a face of $C$. 
	But $a, u, b, v$ appear in this (clockwise or anti-clockwise) cyclic order on the face $F$ of $C$, and then $C$ is not $3$-connected by Fact~\ref{fact:curve:disc}.
	
	Let us now assume that $F_1$ is also a face in the embedding of $C$, but $F_2$ is not. The argument for the symmetric case is analogous.
So, $F_2$ is obtained by splitting a face $F$ of $C$ along the new edge $(u,v)$. Let the other emerging face be $F_3$.
As an example, consider the insertion of the edge $(b_5,b_2)$ in the $3$-connected component $C$ in Figure~\ref{fig:distSPQRins}. 
Let the vertices of $F$ appear in the cyclic order $u, S_1, v, S_2, u$, where $S_1, S_2$ are sequences of vertices. On the two faces $F_2$ and $F_3$ obtained from $F$, the vertices appear in the cyclic order $u, S_1, v, u$ and $u,v,S_2,u$, respectively. Suppose that $F_2$ is the face that consists of the vertices $u, S_1, v, u$.
Then the cycle $\Disc(C',F_1,F_2)$ contains a node $(a,b)$ of $\Disc(C,F_1,F)$ if $b$ is one of $u, S_1, v$, and the cyclic order of these nodes is inherited from $\Disc(C,F_1,F)$.

To verify this claim, observe that for every node $(a,b)$ of $\Disc(C,F_1,F)$, conditions~(2) and (3) of Definition~\ref{def:potdist} are still satisfied after the edge insertion: the common face $F_{a,b}$ from condition (2) still exists, as $F$ is the only affected face of $C$; if condition (3) would be violated, only the two new faces could witness this, but then $a$ and $b$ would be incident to $F$, violating condition (1) and $(a,b)$ would not have been a node of $\Disc(C,F_1,F)$.

The cycle $\Disc(C',F_1,F_2)$ may contain a constant number of nodes that are not present in $\Disc(C,F_1,F)$.
Suppose the nodes $(a_1,b_1), \ldots, (a_m,b_m)$ of $\Disc(C',F_1,F_2)$ are not already in $\Disc(C,F_1,F)$.
We distinguish which conditions of Definition~\ref{def:potdist} are not satisfied for the pairs $(a_i, b_i)$ in $\Disc(C,F_1,F)$.
\begin{itemize}
	\item If condition (1) is not satisfied, then $a_i$ must appear on $F_3$ but not on $F_2$. We further argue that $b_i$ must be either $u$ or $v$: otherwise, as $a_i, u, b_i, v$ appear in this order on $F$ and as $a_i$ and $b_i$ are also on another common face that already is a face of $C$, by Fact~\ref{fact:curve:disc} we have that $C$ is not $3$-connected. Because $a_i$ also appears on $F_1$ and two faces of a $3$-connected planar graph can only share two vertices, this case can apply for no more than four pairs $(a_i, b_i)$.
	\item If condition (1) is satisfied and condition (2) is not, $F_3$ must be the common face of $C'$ for $a_i$ and $b_i$, as this is the only face apart from $F_2$ that is not already present in $C$. This is a contradiction to the assumption that condition (1) is satisfied.
	\item If conditions (1) and (2) are satisfied, but condition (3) is not, let $(a',b')$ be the pair and let $F_w$ be the face of $\Emb(C,F_1)$ witnessing that the condition is not satisfied. It holds $F_w = F$, as all other faces still exist in $\Emb(C',F_1)$. But then $a_i$ and $b_i$ appear on $F$, contradicting the assumption that condition (1) is satisfied.
\end{itemize}

In summary, only a constant number of nodes can occur in $\Disc(C',F_1,F_2)$ in addition to those from $\Disc(C,F_1,F)$. The distances inherited from $\Disc(C,F_1,F)$ can easily be adapted to take these nodes into account.

As the last case, we consider that neither $F_1$ nor $F_2$ is incident to $e$, so $\Disc(C,F_1,F_2)$ already existed.
We show that $\Disc(C,F_1,F_2)$ and $\Disc(C',F_1,F_2)$ differ only by a constant number of nodes, so the distances can be updated easily.
\begin{itemize}
	\item As $F_1$ and $F_2$ do not change, no node can vanish or emerge because of condition (1) of Definition~\ref{def:potdist}.
	\item Nodes can only vanish because of condition (2), which happens if the new edge splits the face $F_{a,b}$ that witnesses condition (2) for a node $(a,b)$ such that these vertices are not on a common face any more. See for example the node $(a_3, b_1)$ in Figure~\ref{fig:distSPQRins} when the edge $(c,d)$ is inserted. Only for at most two nodes $(a_1, b_1)$ and $(a_2,b_2)$ it is possible that the split face was their witness for condition (2), as otherwise condition (3) would be violated.
	\item Nodes can only emerge because of condition (3). For example, if the edge $(a_2,b_7)$ is inserted in Figure~\ref{fig:distSPQRins}, the updated disc graph contains the node $(a_2,b_7)$. 
	Suppose a node $(a,b)$ occurs in $\Disc(C',F_1,F_2)$ but not in $\Disc(C,F_1,F_2)$. It must be that the edge $e$ splits a face $F$ of $C$ such that there are vertices $a', b'$ that appear on $F_1$ and $F_2$, respectively, and on the face $F$ the vertices $a, a', b, b'$ appear in this cyclic order. Note that there cannot be more vertices on $F$ that are incident to either $F_1$ or $F_2$, as otherwise by Fact~\ref{fact:curve:disc} we have that $C$ is not $3$-connected. 
	For the pairs $(a, b')$ and $(a', b)$, the face $F$ does not contradict condition (3), so they must already be nodes of $\Disc(C,F_1,F_2)$ if they are in $\Disc(C',F_1,F_2)$.
	Because $F$ is the only face that $a$ and $b$ as well as $a'$ and $b'$ have in common, as otherwise by Fact~\ref{fact:curve:disc} again $C$ is not $3$-connected, after the insertion $a$ and $b$ must appear together on one of the two emerging faces, while $a'$ and $b'$ are not both on this face. But then $a'$ and $b'$ are not on a common face of $C'$.
	In summary, at most one new node of $\Disc(C',F_1,F_2)$ is not already in $\Disc(C,F_1,F_2)$.
\end{itemize}
	\begin{figure}[t]
		\centering
		\resizebox{0.4\textwidth}{!}{%

\tikzset{every picture/.style={line width=0.75pt}} %

\begin{tikzpicture}[
	vertex/.style={
		circle,
		fill=black,
		inner sep=0pt,
		minimum size=8pt
	},
	outer/.style={
		draw,
		fill={rgb, 255:red, 184; green, 233; blue, 134 },
		dash pattern=on 3.75pt off 3pt on 7.5pt off 1.5pt
	},
	hole/.style={
		draw,
		fill=white,
		dash pattern=on 3.75pt off 3pt on 7.5pt off 1.5pt,
		line width=.75pt
	},
	every label/.style={font=\normalsize}
	]

	\node[vertex, label={[above left]$a_1$}] (a) {};
	\node[vertex, label={[above right]$a_2$}, right=2.55cm of a] (b) {};	
	\node[vertex, label={[above left]$a_{10}$}, below left=1.48cm and 2.05cm of a] (c) {};
	\node[vertex, label={[above right]$a_{3}$}, below right=1.48cm and 2.05cm of b] (d) {};
	\node[vertex, label={[below right=0.1cm]$b_{7}$}, below left =2.25cm  and 0.05cm of a] (e) {};
	\node[vertex, label={[below left=0.1cm]$b_{1}$}, below right =2.25cm  and 0.05cm of b] (f) {};

	\node[vertex, label=right:$b_{6}$, below left =1.35cm and 0.5cm of e] (g) {};
	\node[vertex, label={[left=0.1cm]$b_{2}$}, below right =1.35cm and 0.5cm of f] (h) {};
	\node[vertex, label=left:$a_{9}$, left=2.0cm of g] (i) {};
	\node[vertex, label=right:$a_{4}$, right=2.0cm of h] (j) {};

	\node[vertex, label={[above right=-0.1cm]$b_{5}$}, below right =1.35cm and 0.5cm of g] (k) {};
	\node[vertex, label={[above left=-0.1cm]$b_{3}$}, below left =1.35cm and 0.5cm of h] (l) {};
	\node[vertex, label={[below left=0.1cm]$a_{8}$}, below=4.45cm of c] (m) {};
	\node[vertex, label={[below right=0.1cm]$a_{5}$}, below=4.45cm of d] (n) {};

	\node[vertex, label={[below left=0.1cm]$a_{7}$}, below right=2.20cm and 0.05cm of k] (o) {};
	\node[vertex, label={[below right=0.1cm]$a_{6}$}, below left=2.20cm and 0.05cm of l] (p) {};
	\node[vertex, label=above:$b_{4}$, below right=.58cm and 1.47cm of k] (q) {};

	\path[outer]
	(j.center) --
	(n.center) --
	(p.center) --
	(o.center) --
	(m.center) --
	(i.center) --
	(c.center) --
	(a.center) --
	(e.center) --
	(g.center) --
	(k.center) --
	(q.center) --
	(l.center) --
	(h.center) --
	(f.center) --
	(b.center) --
	(d.center) --
	cycle;

	\path[hole]
	(c.center) -- ++(0.8,0.3) -- ++ (1.1,-0.2) --
	(e.center) -- ++(-1.1, -0.3) -- ++ (-0.8,0.4) --
	(c.center);
	
	\path[hole]
	(d.center) -- ++(-0.8,0.3) -- ++ (-1.1,-0.2) --
	(f.center) -- ++(1.1, -0.3) -- ++ (0.8,0.4) --
	(d.center) ;
	
	\path[hole]
	(i.center) -- ++(0.8,0.4) -- ++ (0.8,0) --
	(g.center) -- ++(-0.8, -0.4) -- ++ (-0.8,0) --
	(i.center);
	
	\path[hole]
	(j.center) -- ++(-0.8,0.4) -- ++ (-0.8,0) --
	(h.center) -- ++(0.8, -0.4) -- ++ (0.8,0) --
	(j.center);
	
	\path[hole]
	(l.center) -- ++(0.5,0.4) -- ++ (0.9,-0.2) --
	(n.center) -- ++(-1.1, 0.2) -- 
	(l.center);
	
	\path[hole]
	(k.center) -- ++(-0.5,0.4) -- ++ (-0.9,-0.2) --
	(m.center) -- ++(1.1, 0.2) -- 
	(k.center);
	
	\path[hole]
	(q.center) -- ++(1.1,-0.2) -- ++ (0.3,-0.6) --
	(p.center) -- ++(-1.1, 0.5) -- 
	(q.center);
	
	\path[hole]
	(q.center) -- ++(-1.1,-0.2) -- ++ (-0.3,-0.6) --
	(o.center) -- ++(1.1, 0.5) -- 
	(q.center);

	\draw (a) edge[very thick]  (b);
	\draw (e.center) edge[very thick]  (f.center);
	
	\draw (e.center) edge[ultra thick, blue]  (b.center);
	\draw (k.center) edge[ultra thick, blue]  (h.center);

	\node[right = 1.85cm of g, font=\LARGE]  (f2) {$F_{2}$};
	\node[above right = 0.1cm and 1.2cm of a.center, font=\LARGE]  {$F_{1}$};	
	
	\node[vertex, label=above:$c$] at ($(d.center)+(-1.9,0.1)$) (edgec) {};
	\node[vertex, label=below:$d$] at ($(f.center)+(1.9,0.1)$) (edged) {};
	\draw (edgec.center) edge[ultra thick, blue]  (edged.center);
	
	\foreach \v in {a,b,c,d,e,f,g,h,i,j,k,l,m,n,o,p,q,edgec,edged}
	\node[vertex] at (\v) {};
	
\end{tikzpicture} %
}
		\caption{$\Disc(C,F_1,F_2)$. Changes to the $\discDist$ relation under insertions $\es[+]{3}{3}$. Compare with Figure~\ref{fig:distSPQR}.}
		\label{fig:distSPQRins}
	\end{figure}
	\subparagraph*{Case $\es[-]{3}{3}$.}
	This case is the reverse of the case $\es[+]{3}{3}$ and the arguments are analogous. If the face $F$ of the changed component $C'$ emerges from two faces $F_2, F_3$ after the edge deletion and the face $F_1$ is unaffected by the change, distances in $\Disc(C',F_1,F)$ can be obtained from the distances in $\Disc(C,F_1,F_2)$ and $\Disc(C,F_1, F_3)$, using the arguments from above in the reverse direction.
	\subparagraph*{Case $\es[+]{2}{3}$.}
	
	\begin{figure}[t]
		\centering
		\begin{subfigure}{.45\textwidth}
			\resizebox{0.97\textwidth}{!}{%

\tikzset{every picture/.style={line width=0.75pt}} %

 %
}
			\caption{Coherent path $\rho(R_u,R_v,(u,v))$. Dashed edges in blue represent virtual edges.} \label{fig:spqrDistMergea}
		\end{subfigure} \\
		\centering	
		\begin{subfigure}{.45\textwidth}
			\resizebox{0.97\textwidth}{!}{%

\tikzset{every picture/.style={line width=0.75pt}} %

%
 %
}
			\caption{$\Disc(R_1,F_s,F_t)$ and $\Disc(R_2,F_p,F_q)$ for the $3$-connected components $R_1$ and $R_2$ on the coherent path, respectively, before insertion. $\Disc(R_1,F_s,F_t)$ is $((s_1,t_1),\allowbreak (s_2,t_2),\allowbreak (s_3,t_3),\allowbreak (s_4,t_4),\allowbreak (s_5,t_5),\allowbreak (s_6,t_6))$. $\Disc(R_2,F_p,F_q)$ for $R_2$ is similarly available. Individual embeddings of $R_1$ and $R_2$ are shown.} \label{fig:spqrDistMergeb}
		\end{subfigure} \quad
		\begin{subfigure}{.45\textwidth}
			\resizebox{0.97\textwidth}{!}{%

\tikzset{every picture/.style={line width=0.75pt}} %

 %
}
			\caption{After the insertion of $(u,v)$, $\discDist$ with respect to a glue and a straddling face pair ($(F_1,F_{sp})$ and $(F_2,F_{sp})$), two straddling faces $F_1$ and $F_2$, and two glue faces $(F_g,F_h)$. } \label{fig:spqrDistMergec}
		\end{subfigure}
		\caption{Changes in $\discDist$. Case $\es[+]{2}{3}$: insertion of an edge along a coherent path, $\rho(R_u,R_v,(u,v))$.}
		\label{fig:spqrDistMerge}
	\end{figure}
	
	In the case that a path of the triconnected component tree coalesces after an edge insertion, the $\spqrDist$ relation obviously needs to be updated. If for two nodes $\tpl x$ and $\tpl y$ of the triconnected component tree the unique path from $\tpl x$ to $\tpl y$ contains the coalesced component, the old distance of the relevant part of the coalescing path can be subtracted to obtain the new distance. Note that the path from $\tpl x$ to $\tpl y$ may contain up to two newly created cycle components, which need to be considered as well.
	
	We need to update the $\discDist$ relation for the coalesced $3$-connected component $C'$.
	The faces of this component are of the following type:
	\begin{itemize}
		\item \emph{existing} faces already existed in one of the $3$-connected components on the coalesced path,
		\item the boundary of a \emph{cycle} face has the four vertices of a cycle component from the coalescing path that are included in $C'$,
		\item a \emph{glue} face is created for two consecutive triconnected components $C_1, C_2$ of the coalesced path if a virtual edge due to a former $3$-connected separating pair between them is removed. The boundary of the face consists of one segment from one face of $C_1$ and from one face of $C_2$ that both do not exist in $C'$. 
		A \emph{superglue} face results when the coalesced path had a subpath of a $3$-connected component $C_1$, a cycle component and another $3$-connected component $C_2$, separated by separating pairs $\{a_1, b_1\}$ and $\{a_2, b_2\}$, such that both $(a_1, b_1)$ and $(a_2, b_2)$ are virtual edges that do not exist in $C'$ any more. In this case, the boundary of the new face has a segment from a face of $C_1$ and of $C_2$ and the (possibly virtual) edges $(a_1, a_2)$ and $(b_1, b_2)$, and
		\item the two faces that are adjacent to the inserted edge. We call these faces \emph{straddling} faces.
	\end{itemize}
	
We now explain how to determine distances in $\Disc(C',F_1, F_2)$, for any type of the faces $F_1$ and $F_2$, see also Figure~\ref{fig:spqrDistMerge} for an illustration of some important cases.
We will omit cycle faces: as such a face only has constantly many vertices, any disc graph that involves a cycle face has only constantly many nodes, and distances are trivial to express using first-order formulas.

If $F_1$ and $F_2$ are the two straddling faces that are incident to the inserted edge $e$, using Lemma~\ref{lem:spqrDist} the distances can be inferred from the old distances on the triconnected component tree as given by $\spqrDist$.

Suppose $F_1$ is a straddling face and $F_2$ is a (super)glue face and let $R_1, R_2$ be the $3$-connected components of the unchanged graph that together with the separating pairs $\{a_1, b_1\}$ and $\{a_2, b_2\}$ have a part in forming the face $F_2$. Let $F^o_1$ and $F^o_2$ be the faces of $R_1$ and $R_2$, respectively, whose boundaries contain the vertices of the boundary of $F_1$ that are in $R_1$ and $R_2$. Similarly, let $F^g_1$ and $F^g_2$ be the faces of $R_1$ and $R_2$ that become part of the glue face $F_2$. 
The cycle $\Disc(C', F_1, F_2)$ consists of one path from $\Disc(R_1, F^o_1, F^g_1)$ and one path from $\Disc(R_2, F^o_2, F^g_2)$, consisting of those nodes $(a,b)$ such that $a$ is on $F_1$ and not on the other straddling face.

For example, consider the straddling face $F_1$ and glue face $F_{sp}$ in Figure~\ref{fig:spqrDistMergec}. The glue face is built up from the faces $F_s$ and $F_p$ of the $3$-connected components $R_1$ and $R_2$ respectively, as illustrated in Figure~\ref{fig:spqrDistMergeb}. The outer faces $F_t$ and $F_q$ contain the vertices that are in the boundary of $F_1$. We obtain $\Disc(C',F_1,F_{sp})$, consisting of the green pairs in Figure~\ref{fig:spqrDistMergec}, by merging the paths of $\Disc(R_1,F_t,F_s)$ and $\Disc(R_2,F_q,F_p)$ that have nodes with vertices lying on $F_1$. The remaining pairs are coloured blue in Figure~\ref{fig:spqrDistMergec}.

If $F_2$ is an existing face, the approach is analogous, but simpler, as one only needs to consider one disc graph of one previously existing $3$-connected component instead of two.

Now suppose that $F_1$ and $F_2$ are (super)glue faces. Let $F^a_1, F^b_1$ and $F^a_2, F^b_2$ be the faces of $3$-connected components of the unchanged graph that are joined to form $F_1$ and $F_2$, respectively. 
Note that only two of these faces may come from the same $3$-connected component, as they are incident to separating pairs of the coalescing path. If, say, $F^b_1$ and $F^a_2$ are from the same $3$-connected component $R$, the nodes of $\Disc(R, F^b_1, F^a_2)$ are included in $\Disc(C', F_1, F_2)$ without changing their order, potentially excluding up to two nodes for which the outer face of $R$ witnessed condition (2) of Definition~\ref{def:potdist}. 
Only a constant number of further nodes can be part of $\Disc(C', F_1, F_2)$: only the two straddling faces and potentially one glue face (in the case that there is exactly one separating pair in the coalescing path between the separating pairs at which the faces $F_1$ and $F_2$ are glued) are candidate witnesses for condition (2) of Definition~\ref{def:potdist}, and each of these faces can be a witness of at most two nodes without violating condition (3).
For example, in Figure~\ref{fig:spqrDistMergec}, for the face pair $F_g,F_h$ the disc graph contains only the nodes $(a_1,a_3)$ and $(b_1,b_3)$.

If $F_2$ is an existing face instead of a glue face, the argument is again similar, but simpler.
	\subparagraph*{Case $\es[-]{3}{2}$.}
	This case is the reverse of the case $\es[+]{2}{3}$ and the required updates can be expressed performing the operations of the case $\es[+]{2}{3}$ in reverse.
	\subparagraph*{Cases $\es[+]{2}{2}$ and $\es[-]{2}{2}$.}
	Insertion or deletion of an edge inside a cycle component does not affect $\discDist$, since no $3$-connected component is affected.
	
	The $\spqrDist$ relation however does change. In case of an insertion, a cycle component splits into two, and a new separating pair is introduced. Thus the distance between any two nodes of the triconnected component tree that contains the cycle component increases by two if the tree path between the two nodes passes through the newly created separating pair. The deletion case is analogous.  
	\subparagraph*{Cases $\es[+]{1}{2}$ and $\es[-]{2}{1}$.}
	We recall that when a path of the biconnected component tree merges into one biconnected component after a $\es[+]{1}{2}$ insertion, a virtual edge may be inserted into each old (non-trivial) biconnected component on this path. These independent $\es[+]{3}{3}$ or $\es[+]{2}{3}$ insertions are handled as explained above.
	After a $\es[-]{2}{1}$ deletion, the corresponding $\es[-]{3}{3}$ or $\es[-]{3}{2}$ deletions of virtual edges are analogously processed.
	\subparagraph*{Cases  $\es[+]{0}{2}$ and $\es[-]{2}{0}$.}
	The relations $\discDist$ and $\spqrDist$ are not affected.
\end{proof}

\subsection{Distances in biconnected component trees}

The goal of this section is to show, analogously to the result of the previous section, that distances for biconnected component trees can be maintained. We show how to maintain the relation $\bcDist$ that contains tuples $(\tpl x, \tpl y, d)$ such that the distance between nodes $\tpl x, \tpl y$ of some biconnected component tree is $d$.

\begin{lemma}\label{lemma:bc-distances}
	The relation $\bcDist$ can be maintained in $\DynFO$ under insertions and deletions of edges in the underlying graph $G$, provided that $G$ stays planar.
\end{lemma}

The challenging case for maintaining $\bcDist$ is $\es[-]{2}{1}$ edge deletions, that is, deletions that cause a biconnected component to decompose. In the biconnected component tree, the node for this biconnected component is replaced by a path whose length cannot be bounded by a constant.

We will observe that a biconnected component $G$ becomes non-biconnected after the deletion of an edge $e$ if and only if $e$ belongs to a unique cycle component in the triconnected component tree of $G$. When $e$ is deleted, the biconnected component unfurls into a path whose cut vertex nodes correspond to the vertices of the former cycle.
Thus, we also maintain the distances between the vertices of each cycle component of the graph, in addition to the distances for the biconnected component trees themselves. We will also use these distances within cycles in later sections to determine isomorphisms of cycle components.
More formally, we maintain the relation $\cycleDist$ that contains a tuple $(\tpl c = (c_1, c_2,c_3),a,b,d)$ if
\begin{enumerate}[(i)]
	\item the vertices $c_1, c_2, c_3$ appear together in a cycle component $C$ of some triconnected component tree that also includes $a$ and $b$, and
	\item the distance between $a$ and $b$ in $C$ is $d$ when the cycle $C$ is oriented as $c_1\cdots c_2\cdots c_3\cdots c_1$. 
\end{enumerate}

\begin{lemma}\label{lemma:cycle-distances}
	The relation $\cycleDist$ can be maintained in $\DynFO$ under insertions and deletions of edges in the underlying graph $G$, provided that $G$ stays planar.
\end{lemma}

Towards a proof of these lemmas, we first formalize the observation mentioned above.

\begin{lemma}\label{lem:distBC}
	Let $G$ be a biconnected graph. Let $e=(u,v)$ be an edge of $G$ such that $G-\{e\}$ is not biconnected. Then,
	\begin{enumerate}[(a)]
		\item there exists a unique cycle component $C$ in the triconnected component tree of $G$ that contains $e$,
		\item the cut vertices of $G-\{e\}$ are exactly the vertices of $C$ except for $u$ and $v$, and
		\item $\biTree(G-\{e\})$ is a path of length $2d-4$, where $d$ is the number of vertices in $C$, and the cut vertices of $G-\{e\}$ appear in the same order in this path as the corresponding vertices in $C$.
	\end{enumerate} 
\end{lemma}  
\begin{proof}
Suppose $G$ is a biconnected graph and $e= (u,v)$ is one of its edges such that $G-\{e\}$ is not biconnected.
We first show that $\biTree(G-\{e\})$ is a path from some biconnected component $B_u$ containing $u$ to some other biconnected component $B_v$ containing $v$.
Note that $u$ and $v$ cannot be in the same biconnected component in $G-\{e\}$, as then any cut vertex of $G-\{e\}$ would also be a cut vertex of $G$, contradicting the assumption that $G$ is biconnected.
Now suppose, towards a contradiction, that there is a biconnected component $B_x$ that is not on the path from $B_u$ to $B_v$ in $\biTree(G-\{e\})$ and let $x$ be the cut vertex of $B_x$ towards this path.
Again, $x$ must also be a cut vertex in $G$, contradicting the assumption that $G$ is biconnected.

Towards part (a), observe that $u,v$ must be contained in some common triconnected component $C$ in $G$: as they are connected by an edge, they cannot be separated by any separating pair. 
Such a component cannot be a $3$-connected component, as that would imply the existence of three internally vertex-disjoint paths from $u$ to $v$ in $G$; as at least two of them would still exist in $G-\{e\}$, the vertices $u$ and $v$ would be in the same biconnected component of $G-\{e\}$, which we ruled out above. 
So $C$ must be a cycle component. It must be unique, because if $u$ and $v$ are included in two cycle components of $\triTree(G)$, they also need to form a $3$-connected separating pair in $\triTree(G)$, but we ruled out above the existence of three internally vertex-disjoint paths between $u$ and $v$.

Towards parts (b) and (c), we inspect $\biTree(G-\{e\})$ in more detail. We already determined that $\biTree(G-\{e\})$ is a path $B_u c_1 B_1 c_2 \cdots c_k B_v$, where the $B_i$ are biconnected components, $B_u$ and $B_v$ are the biconnected components that contain $u$ and $v$, respectively, and the $c_i$ are cut vertices.
We claim that the cycle component $C$ is the cycle $u c_1 c_2 \cdots c_k v u$. 
To show this, we observe first that each pair of consecutive vertices in that sequence must be in a common triconnected component of $G$ and then that no $3$-connected separating pair separates the vertices from that sequence, so in fact they are all in the same triconnected component, which accordingly is a cycle.

For the first part, we show that every pair of consecutive vertices in the sequence $u c_1 c_2 \cdots c_k v u$ that is not connected by an edge forms a $3$-connected separating pair in $G$, and are therefore in a common triconnected component and are connected by a virtual edge. 
The vertices $u$ and $v$ are connected by the edge $e$ in $G$. If for some other pair of consecutive vertices this is not the case, then there are three vertex-disjoint paths between them: two paths are guaranteed to exist because these vertices are in a common biconnected component in $G-\{e\}$, the third one is traversing the remaining vertices of the sequence $u c_1 c_2 \cdots c_k v u$ and the edge $e$. This pair of vertices is also a separating pair, as it separates the remaining vertices of the sequence $u c_1 c_2 \cdots c_k v u$ from the remaining vertices that are in the mentioned common biconnected component of $G-\{e\}$.

For the second part, observe that if any pair of non-consecutive vertices from the sequence $u c_1 c_2 \cdots c_k v u$ is a $3$-connected separating pair in $G$, then they are in the same biconnected component in $G-\{e\}$, contradicting the choice of these vertices.

This reasoning proves part (b); for part (c), it only remains to argue the length of the path, which can be easily seen to be $2k = 2(d-2) = 2d- 4$, where $d = k + 2$ is the number of vertices of the cycle $u c_1 c_2 \cdots c_k v u$.
\end{proof}

Thanks to this observation, maintenance of $\bcDist$ is easy, as we can use $\cycleDist$ for the only non-trivial case of $\es[-]{2}{1}$ deletions. So, in the following combined proof, we can focus on maintaining $\cycleDist$.

\begin{proofof}{\cref{lemma:bc-distances,lemma:cycle-distances}}
	We distinguish the different types of the changed edge $e = (u,v)$.
	\subparagraph*{Cases $\es[+]{3}{3}$ and $\es[-]{3}{3}.$}
	No updates are required, since no biconnected component tree and no cycle component is affected by the change.
	\subparagraph*{Case $\es[+]{2}{3}$.}
	Again, no updates are required for $\bcDist$, as the biconnected component trees are not affected.
	
	As discussed in Section~\ref{section:changes}, cycle components may change as a result of the insertion. Let $B$ be the biconnected component that contains both $u$ and $v$ and let $R_u$ and $R_v$ be the triconnected components of $\triTree(B)$ that contain $u$ and $v$, respectively.
	All cycle components $C$ on the path from $R_u$ to $R_v$ in $\triTree(B)$ do not exist any more in the triconnected component tree $\triTree(B')$ of the biconnected component $B'$ after the change.
	Instead, new cycle components are introduced as follows. 
	
	\akm{The following could be a separate lemma/claim that explains formally what happens to cycle components that lie on the coalescing path (the change is described declaratively in Section~\ref{section:changes}).}Suppose $C$ lies internally on the path from $R_u$ to $R_v$ and $\{a_1,b_1\}$ and $P_{i+1} = \{a_2,b_2\}$ are the two separating pairs that are adjacent to $C$ on the path. Let the cycle $C$ consist of the sequence $a_1, x_1, \ldots, x_k, a_2,b_2, y_1, \ldots, y_\ell, b_1, a_1$ of vertices.
	We claim that if $k \geq 1$ or $\ell \geq 1$, then $\{a_1, a_2\}$ respectively $\{b_1, b_2\}$ are $3$-connected separating pairs in $B'$.
	They are $3$-connected, as two internally vertex-disjoint paths are given by $C$ and a third internally vertex-disjoint path goes via the inserted edge $e$. They are separating pairs, as they separate the vertices $x_1, \ldots, x_k$ respectively $y_1, \ldots, y_\ell$ from the remaining graph.
	If $k = 1$, then the triangle $a_1, x_1, a_2$ is a $3$-connected component and if $k \geq 2$, then $a_1, x_1, \ldots, x_k, a_2, a_1$ is a new cycle component for which distances can easily be inferred from the distances in $C$. Analogously, $b_2, y_1, \ldots, y_\ell, b_1, b_2$ is a new cycle component if $\ell \geq 2$.
	
	If $R_u$ or $R_v$ are themselves cycles, we replace $a_1$ and $b_1$ by $u$ and $a_2$ and $b_2$ by $v$ in the reasoning above.
	
	\subparagraph*{Case $\es[-]{3}{2}$.}
	This case is the reverse of the $\es[+]{2}{3}$ case. No updates are required for the relation $\bcDist$. New cycle components consist of two formerly existing cycles that are merged at common vertices, so for the update of $\cycleDist$, their distances can easily be expressed using the existing auxiliary data.
	
	\subparagraph*{Cases $\es[+]{2}{2}$ and $\es[-]{2}{2}$.}
	We consider the case of a $\es[+]{2}{2}$ insertion.
	Again, no updates are required for $\bcDist$. For the update of $\cycleDist$, we observe that $u$ and $v$ are in a common cycle component of their biconnected component $B$: they are not in a common $3$-connected component by assumption and they must be in a common triconnected component, as otherwise they would be in a common $3$-connected component after the change.
	
	If $u$ and $v$ are in more than one common cycle component, they form a $3$-connected separating pair of $B$ and a virtual edge $(u,v)$ is already present in the cycle components. So, no update of $\cycleDist$ is necessary.
	Otherwise, if $u$ and $v$ are in a unique cycle component $C$ of $B$, $\{u,v\}$ becomes a $3$-connected separating pair in the changed biconnected component $B'$: they form a separating pair already in $B$, as they are a pair of non-adjacent vertices in a cycle, and the new edge is the (trivial) third internally vertex-disjoint path in $B'$ from $u$ to $v$, in addition to the two paths along the cycle.
	In $\triTree(B')$, the cycle $C$ is split along the new chord $(u,v)$, and the distances can easily be obtained from the distances in $C$.
	
	The case of a $\es[-]{2}{2}$ edge deletion is symmetric.
	
	\subparagraph*{Case $\es[+]{1}{2}$.}
	Let $A$ be the connected component that contains the vertices $u$ and $v$, and let $B_u$ and $B_v$ be the biconnected components containing these vertices such that on the path from $B_u$ to $B_v$ in $\biTree(A)$, no other biconnected component contains one of $u$ and $v$.
	Let this path be $B_u c_1 B_1 c_2\cdots c_k B_v$, where $B_1,\ldots,B_{k-1}$ are biconnected components and $c_1,\ldots,c_k$ are cut vertices.
	
	In the connected component $A'$ that results from $A$ after the insertion of $(u,v)$, the vertices from $B_u, B_1, \ldots, B_{k-1}, B_v$ obviously make up a new biconnected component $B_{uv}$.
	We obtain the distance between any pair of nodes from $\biTree(A')$ from the distances in $\biTree(A)$ by subtracting the length of the intersection of the path between the given nodes with the path $B_u c_1 B_1 c_2\cdots c_k B_v$. The necessary information can be obtained from $\bcDist$.
	
	For updating $\cycleDist$, we claim that $\triTree(B_{uv})$ contains a new cycle component $uc_1c_2\cdots c_k v u$. The argument is analogous to the arguments in the proof of Lemma~\ref{lem:distBC}: 
	if consecutive vertices $c_i, c_{i+1}$ in this sequence are not already connected by an edge, then $\{c_i, c_{i+1}\}$ forms a $3$-connected separating pair in $B_{uv}$, thus a virtual edge $(c_i, c_{i+1})$ is present. So, $uc_1c_2\cdots c_k v u$ is a cycle.
	It is a single cycle component, because non-consecutive vertices from the sequence $uc_1c_2\cdots c_k v u$ cannot form a $3$-connected separating pair of $B_{uv}$: otherwise there are already two internally vertex-disjoint paths between these vertices in $A$, meaning that they would be in the same biconnected component, which contradicts our choice of these vertices.
	
	The distances in this cycle $uc_1c_2\cdots c_k v u$ can be obtained from $\bcDist$.
	
	Other cycles may be affected by insertions of virtual edges $(c_i, c_{i+1})$. These insertions are of type $\es[+]{2}{2}$, $\es[+]{2}{3}$ or $\es[+]{3}{3}$ and can be handled as discussed above.
	 
	\subparagraph*{Case $\es[-]{2}{1}$.}
	This case is the reverse of the case $\es[+]{1}{2}$. Using Lemma~\ref{lem:distBC}, we can obtain the distances on the unfurled path from $\cycleDist$ and then update $\bcDist$ by adding the length of the relevant part of the unfurled path if necessary.
	
	Any deletions of virtual edges are of type $\es[-]{3}{3}$, $\es[-]{3}{2}$ or $\es[-]{2}{2}$ and are handled as explained above.
	\subparagraph*{Cases $\es[+]{0}{2}$ and $\es[-]{2}{0}$.}
	In these cases, $\cycleDist$ does not need to be updated. 
	After a $\es[+]{0}{2}$ insertion, two distinct biconnected component trees are merged along the new trivial biconnected component that corresponds to the inserted edge. The case of an $\es[-]{2}{0}$ deletion is the reverse.
	In either case, $\bcDist$ can be updated in a straightforward manner.
\end{proofof}
  
\section{Maintaining Isomorphisms between Triconnected Components}\label{section:triconnected}
As a first step towards our main result, we now show how to maintain isomorphisms between triconnected components of planar graphs.
An isomorphism $\pi$ between two triconnected components $C$ and $C^*$ is uniquely determined by giving the images $\pi(a), \pi(b), \pi(c)$ in $C^*$ of any three vertices $a,b,c$ of $C$ that lie on a common face (see~\cite[Lemma 23]{KieferPS19}).
So, the following proposition tells us that we can maintain and access all possible isomorphisms between triconnected components.

\begin{proposition}\label{prop:3iso}
	Let $G$ be a planar graph. The relation $\ISOt$ that contains all tuples $(a,b,c,d, a^*, b^*, c^*, d^*)$ such that
	\begin{itemize}
		\item $a,b,c,d$ and $a^*, b^*, c^*, d^*$ are vertices of triconnected components $C$ and $C^*$ of $G$, respectively,
		\item $a,b,c$ lie on a common face of $C$,
		\item $C$ and $C^*$ are isomorphic via an isomorphism $\pi_{abc}$ fixed by $\pi_{abc}(a) = a^*, \pi_{abc}(b) = b^*$, $\pi_{abc}(c) = c^*$, and
		\item $\pi_{abc}(d) = d^*$ (i.e., how $\pi_{abc}$ maps vertices can also be accessed via $\ISOt$),
	\end{itemize}
	can be maintained in \DynFO under insertions and deletions of edges, provided that $G$ stays planar.
\end{proposition}

For triconnected components $C, C^*$ that are cycles, this follows as distances are maintained for them (see Lemma~\ref{lemma:cycle-distances}).

For the proof for $3$-connected components, we follow the outline presented in Section~\ref{section:proof-overview} and maintain Tutte embeddings for all tuples $(a,b,c)$ and $(a^*, b^*, c^*)$ of vertices that lie on a common face. Finding isomorphic components then boils down to checking that two such embeddings are equal. The induced isomorphism maps a vertex to the vertex of the other component with the same position in the embedding.

To check whether three vertices lie on a common face of a $3$-connected component, we maintain an auxiliary relation that contains all such triples of vertices, which is possible in \DynFO, see \cite[Lemma~14 in conjunction with Chapter~5]{DattaKM23arx}.

We formalize Tutte embeddings and the auxiliary data required to maintain them in Section~\ref{sec:tutte:overview}.
In Section~\ref{sec:tutte:algebra} and Section~\ref{sec:tutte:coherent} we introduce linear algebra tools for maintaining the data under several operations and use them to prove Proposition~\ref{prop:3iso} in Section~\ref{sec:tutte:final}.

\subsection{The Tutte embedding and maintained auxiliary data}\label{sec:tutte:overview}

We make the definition of the Tutte embedding more formal.
Assuming three vertices $v_{i_1}, v_{i_2}, v_{i_3}$ of the input graph $G$ are pinned and the edges are replaced by springs, we consider the forces that act on the non-pinned vertices through the springs. At equilibrium, every vertex is stationary and hence the forces acting on it must all sum to zero.
We obtain the following system of equations for the first coordinate $x_i$ of each vertex $v_i$ of $G$ (and a second, analogous system for the second coordinate $y_i$ of each $v_i$), where the first coordinates of the pinned vertices $v_{i_1}, v_{i_2}, v_{i_3}$ are $p_{i_1}, p_{i_2}, p_{i_3}$, respectively:
\begin{align*}
	\deg(v_i) x_i& = \sum_{(v_i,v_j)\in E}^{} x_j & \forall v_i\not\in \{v_{i_1}, v_{i_2}, v_{i_3}\}\\
	x_i& =  p_i & \forall v_i\in \{v_{i_1}, v_{i_2}, v_{i_3}\}
\end{align*}
Writing the equations in linear algebraic terms, we have $\mat T \: \mat x = \mat b$
where $\mat T = \mat T(G, v_{i_1}, v_{i_2}, v_{i_3})$ is a square matrix with entries
\begin{align*}
	T_{ij}=
	\begin{cases}
		~-1 & \text{if } i\neq j,  (v_i,v_j)\in E \text{ and } v_i\notin \{v_{i_1}, v_{i_2}, v_{i_3}\}\\
		~\deg(v_i) & \text{if } i=j  \text{ and }  v_i\notin \{v_{i_1}, v_{i_2}, v_{i_3}\}\\
		~1 & \text{if } i=j  \text{ and } v_i\in \{v_{i_1}, v_{i_2}, v_{i_3}\}\\
		~0 & \text{otherwise}\\
	\end{cases}
\end{align*}
and $\mat{b}$ is such that
\begin{align*}
b_i=
\begin{cases}
	~0 & \text{if } v_i \notin \{v_{i_1}, v_{i_2}, v_{i_3}\} \\
	~p_{i} & \text{if } v_{i} \in \{v_{i_1}, v_{i_2}, v_{i_3}\}.
\end{cases}
\end{align*}

The matrix $\mat T$ is very similar to the \emph{Laplacian} $\mat L$ of the graph $G$, which has the entries $L_{ii}= \deg(v_i)$ for all $i$ and $L_{ij}= -1$ if $(v_i,v_j)\in E$ and $0$ otherwise for all $i \neq j$. So, $\mat T$ is obtained from $\mat L$ by replacing the rows for the pinned vertices $v_{i_1}, v_{i_2}, v_{i_3}$ with the row that has $1$ at the diagonal entry and $0$ otherwise. Note that $\mat L$ is symmetric, but $\mat T$ is not.

Our goal is to maintain the inverse $\mat T^{-1}$ of $\mat T$. From this, we obtain $\mat x$ as $\mat x = \mat T^{-1} \: \mat b$, and analogously we obtain the second coordinates $\mat y $ as $\mat y =  \mat T^{-1} \: \mat d$, where $\mat d$ has as only non-zero entries the second coordinates of the fixed vertices.
As $\mat b$ and $\mat d$ each have only three non-zero values, this computation can be performed in \FOar, assuming that $\mat T^{-1}$ is available in a suitable form and all values are polynomially bounded.

Note that $\mat T^{-1}$ is a matrix over $\Q$. It follows that $\mat x$ and $\mat y$ are also over the rationals if the coordinates of the fixed vertices are rational numbers.
As $\mat T^{-1}$ may have entries $\frac{a}{b}$ that involve exponentially large numbers, we maintain $\mat T^{-1}$ modulo many polynomially large primes.

We now present the auxiliary information that is maintained.
Let $G$ be a graph of size $n$ (or a subgraph of a graph of size $n$) and let $v_{i_1}, v_{i_2}, v_{i_3}$ be three of its vertices. Let $p$ be a prime that is polynomially bounded in $n$. Let $\mat L$ be the Laplacian of $G$ and let $\mat T = \mat T(G, v_{i_1}, v_{i_2}, v_{i_3})$ be the matrix as explained above with respect to $G$ and  $v_{i_1}, v_{i_2}, v_{i_3}$. We call $(G, v_{i_1}, v_{i_2}, v_{i_3})$ \emph{embedding-inducing} if $\mat T^{-1}$ exists and $(G, v_{i_1}, v_{i_2}, v_{i_3}, p)$ \emph{modulo embedding-inducing} if $\mat T \bmod p = \mat T(G, v_{i_1}, v_{i_2}, v_{i_3},p)$ has an inverse over $\Z_p$.
\begin{lemma} \label{lem:lina-invertable}
If $G$ is a connected graph and $v_{i_1}, v_{i_2}, v_{i_3}$ are arbitrary vertices of $G$ then $(G, v_{i_1}, v_{i_2}, v_{i_3})$ is embedding-inducing. %
\end{lemma}
\begin{proof}
	A matrix $\mat A$ is \emph{weakly chained diagonally dominant} if it is weakly diagonally dominant (that is, the absolute value of each diagonal entry is at least as large as the sum of the absolute values of the non-diagonal entries of its row), has at least one strictly diagonally dominant row and for every non-strictly diagonally dominant row~$i$ there is a strictly diagonally dominant row $j$ and a sequence $A_{i i_1}, A_{i_1 i_2}, \ldots, A_{i_r j}$ of non-zero entries. Such a matrix is non-singular \cite{ShivakumarC74}.
	
	For any connected graph $G$ with vertices $v_{i_1}, v_{i_2}, v_{i_3}$, the matrix $\mat T = \mat T(G, v_{i_1}, v_{i_2}, v_{i_3})$ is weakly chained diagonally dominant: the rows of $v_{i_1}, v_{i_2}, v_{i_3}$ are strictly diagonally dominant and from every vertex of $G$, at least one of the vertices $v_{i_1}, v_{i_2}, v_{i_3}$ is reachable without using an outgoing edge of $v_{i_1}, v_{i_2}$, or $v_{i_3}$.
\end{proof}

For a modulo embedding-inducing $(G, v_{i_1}, v_{i_2}, v_{i_3}, p)$, where $G$ contains $n$ vertices, let $\M(G,v_{i_1}, v_{i_2}, v_{i_3}, p)$ contain the following information:
\begin{enumerate}
 \item the inverse $\mat T^{-1} \in \Z_p^{n \times n}$,
 \item the products $\mat T^{-1} \, \mat c \in \Z_p^{n}$ and $\mat c^\transpose \, \mat T^{-1} \in \Z_p^{1 \times n}$ for each column $\mat c$ of the Laplacian $\mat L$,
 \item the products $\mat c_1^\transpose \, \mat T^{-1} \, \mat c_2 \in \Z_p$ for each pair $\mat c_1, \mat c_2$ of columns of $\mat L$.
\end{enumerate}

In the remainder of this section, we will show how to maintain $\M(G,v_{i_1}, v_{i_2}, v_{i_3}, p)$ under various operations and that this suffices to maintain isomorphism of $3$-connected components.

\subsection{Tools for Maintaining Tutte auxiliary information}\label{sec:tutte:algebra}

We discuss now how $\M(G,v_{i_1}, v_{i_2}, v_{i_3}, p)$ can be maintained using \FOar formulas under operations like edge changes, exchanging of pinned vertices and merging of graphs (to be defined later), provided the inverses continue to exist.
Intuitively, these operations will allow us to reflect the modification of existing triconnected components (e.g., for changes of type $\es[-]{3}{3}$), combine existing triconnected components into new ones (e.g., for changes of type $\es[+]{2}{3}$), or split an existing triconnected component into multiple new ones (e.g., for changes of type $\es[-]{3}{2}$).
Section~\ref{sec:tutte:coherent} uses these insights for maintaining Tutte auxiliary information for $3$-connected components, which in turn is used in Section~\ref{sec:tutte:final} for maintaining isomorphisms for $3$-connected components in \DynFO.

Large parts of our reasoning rely on the Sherman-Morrison-Woodbury (SMW) formula, see \cite{HendersonS81} and \cite{DattaMVZ18} for a  prior use in dynamic complexity theory. The formula explains how to obtain the inverse of a matrix that is the result of changing an $n \times n$ matrix $\mat A$ by adding the product $\mat U \mat V^\transpose$ of two arbitrary $n \times k$ matrices $\mat U$ and $\mat V$, as long as the original and the new matrix are invertible:
\begin{align}
	\label{eq:smw}
	(\mat A+ \mat U \mat V^\transpose)^{-1} = \mat A^{-1} - \mat A^{-1}\mat U{(\mat I+\mat V^\transpose \mat A^{-1}\mat U)}^{-1}\mat V^\transpose \mat A^{-1}
\end{align}

In this formula, $\mat I$ denotes the $k \times k$ identity matrix.

Given the matrix $\mat T$ for an embedding-inducing $(G, v_{i_1}, v_{i_2}, v_{i_3})$, inserting an edge $(v_i, v_j)$ leads to the change of four entries of $\mat T$, assuming that $v_i, v_j$ are neither of $v_{i_1}, v_{i_2}, v_{i_3}$: the diagonal entries $T_{ii}$ and $T_{jj}$ are incremented, the entries $T_{ij}$ and $T_{ji}$ are decremented. The new matrix $\mat T'$ can be expressed as $\mat T + \mat a \mat a^\transpose$, where $\mat a$ is the (column) vector with entries $a_i =1$, $a_j = -1$ and $0$ everywhere else.

\begin{lemma}\label{lem:lina-edges}
Fix a domain of size $n$. Let $G$ be a graph with at most $n$ vertices, $v_{i_1}, v_{i_2}, v_{i_3}$ three of its vertices, and $p$ a prime of magnitude $\bigO(n^c)$, for some constant $c$, such that $(G, v_{i_1}, v_{i_2}, v_{i_3}, p)$ is modulo embedding-inducing.
If $G'$ results from $G$ by inserting or deleting some edge, then given $\M(G,v_{i_1}, v_{i_2}, v_{i_3}, p)$ one can determine whether $(G', v_{i_1}, v_{i_2}, v_{i_3}, p)$ is also modulo embedding-inducing in \FOar and if so, $\M(G',v_{i_1}, v_{i_2}, v_{i_3}, p)$ can be defined.
\end{lemma}
\begin{proof}
Suppose that $(G, v_{i_1}, v_{i_2}, v_{i_3}, p)$ is modulo embedding-inducing and that $G'$ results from $G$ by inserting the edge $(v_i, v_j)$. Let $\mat T \df \mat T(G,v_{i_1}, v_{i_2}, v_{i_3}, p)$ and let $\mat a$ be the vector whose only non-zero entries are the entries $a_i = 1$ and $a_j = -1$.
We assume that $v_i, v_j$ are neither of $v_{i_1}, v_{i_2}, v_{i_3}$, as for these vertices the rows of $\mat T$ do not change. Then $\mat T + \mat a \mat a^\transpose = \mat T' = \mat T(G',v_{i_1}, v_{i_2}, v_{i_3}, p)$.
According to the SMW formula it holds
\[(\mat T+ \mat a \mat a^\transpose)^{-1} = \mat T^{-1} - \mat T^{-1}\mat a{(\mat I+\mat a^\transpose \mat T^{-1}\mat a)}^{-1}\mat a^\transpose \mat T^{-1} \]
and $(G', v_{i_1}, v_{i_2}, v_{i_3}, p)$ is modulo embedding-inducing, that is, $\mat T'$ is invertible over $\Z_p$, exactly if $\mat I+\mat a^\transpose \mat T^{-1}\mat a$ is invertible, that is, has a non-zero determinant. As $\mat a$ has only two non-zero entries, computing the scalar $\mat a^\transpose \mat T^{-1}\mat a \in \Z_p$ involves only a constant number of additions and multiplications of constantly many numbers, and therefore is possible in \FOar. Then, it can be determined whether $1 + \mat a^\transpose \mat T^{-1}\mat a \equiv 0 \pmod p$ holds. If not, the inverse exists and can be computed, the column vector $\mat T^{-1}\mat a$ and the row vector $\mat a^\transpose \mat T^{-1}$ can be computed, and finally the full right-hand side of the SMW formula can be evaluated.

It remains to explain how the products $\mat T'^{-1} \, \mat c, \mat c^\transpose \, \mat T'^{-1}$ and $\mat c_1^\transpose \, \mat T'^{-1} \, \mat c_2$ can be computed, for all columns $\mat c, \mat c_1, \mat c_2$ of the Laplacian $\mat L'$ of $G'$.
For the products of the form $\mat c_1^\transpose \, \mat T'^{-1} \, \mat c_2$, we observe
\[\mat c_1^\transpose \,(\mat T+ \mat a \mat a^\transpose)^{-1} \, \mat c_2= \mat c_1^\transpose \, \mat T^{-1} \, \mat c_2 -  \mat c_1^\transpose \, \mat T^{-1}\mat a{(\mat I+\mat a^\transpose \mat T^{-1}\mat a)}^{-1}\mat a^\transpose \mat T^{-1} \, \mat c_2 \]
If $\mat c_1$ and $\mat c_2$ are columns of $\mat L'$ that are equal in $\mat L$, the additionally required partial products are contained in $\M(G,v_{i_1}, v_{i_2}, v_{i_3}, p)$. If the columns are different from the respective columns $\mat d_1, \mat d_2$ of $\mat L$ then these columns differ in at most two entries, so $\mat c_k = \mat d_k + \mat \Delta_k$, where $\mat \Delta_k$ has at most two non-zero entries, for each $k = 1, 2$.
So, we can compute the products as
\begin{align*}
 \mat c_1^\transpose (\mat T+ \mat a \mat a^\transpose)^{-1}  \mat c_2 &=  (\mat d_1^\transpose + \mat \Delta_1^\transpose) \, \mat T^{-1} \, (\mat d_2 + \mat \Delta_2) \\
    & \qquad - (\mat d_1^\transpose + \mat \Delta_1^\transpose) \, \mat T^{-1}\mat a{(\mat I+\mat a^\transpose \mat T^{-1}\mat a)}^{-1}\mat a^\transpose \mat T^{-1} \, (\mat d_2 + \mat \Delta_2) \\
    &=  \mat d_1^\transpose \mat T^{-1} \mat d_2 + \mat \Delta_1^\transpose \mat T^{-1}  \mat d_2 + \mat d_1^\transpose \mat T^{-1} \mat  \Delta_2  + \mat \Delta_1^\transpose \mat T^{-1} \mat  \Delta_2 \\
    & \qquad - \mat d_1^\transpose \, \mat T^{-1}\mat a{(\mat I+\mat a^\transpose \mat T^{-1}\mat a)}^{-1}\mat a^\transpose \mat T^{-1} \, \mat d_2 \\
    &  \qquad - \mat \Delta_1^\transpose \, \mat T^{-1}\mat a{(\mat I+\mat a^\transpose \mat T^{-1}\mat a)}^{-1}\mat a^\transpose \mat T^{-1} \, \mat d_2 \\
    & \qquad - \mat d_1^\transpose \, \mat T^{-1}\mat a{(\mat I+\mat a^\transpose \mat T^{-1}\mat a)}^{-1}\mat a^\transpose \mat T^{-1} \, \mat  \Delta_2 \\
    & \qquad -  \mat \Delta_1^\transpose \, \mat T^{-1}\mat a{(\mat I+\mat a^\transpose \mat T^{-1}\mat a)}^{-1}\mat a^\transpose \mat T^{-1} \, \mat  \Delta_2
\end{align*}
where the products containing  $\mat d_k$  are contained in $\M(G,v_{i_1}, v_{i_2}, v_{i_3}, p)$ and the products involving $\mat \Delta_k$ can be computed in \FOar, as these vectors have a constant number of non-zero entries and thus the matrix product requires a constant number of additions.

The remaining products can be computed analogously.

In case of an edge deletion, the new matrix $\mat T'$ can be expressed as $\mat T - \mat a \mat a^\transpose$ and the approach is analogous.
\end{proof}

Other operations on $\M(G,v_{i_1}, v_{i_2}, v_{i_3}, p)$ can be implemented similarly, see Lemmas~\ref{lem:lina-pins} to~\ref{lem:lina-removepair} in the Appendix, for instance,
\begin{itemize}
 \item exchanging pinned vertices, i.e., given $\M(G,v_{i_1}, v_{i_2}, v_{i_3}, p)$ for some pinned vertices $v_{i_1}, v_{i_2}, v_{i_3}$, one can determine $\M(G,v'_{i_1}, v'_{i_2}, v'_{i_3}, p)$ (if it exists) for arbitrary other pinned vertices $v'_{i_1}, v'_{i_2}, v'_{i_3}$ (Lemma~\ref{lem:lina-pins}),
 \item merging of components, i.e., given $\M(G_1,v_{i_1}, v_{i_2}, v_{i_3}, p)$  and $\M(G_2,v_{i_1}, v_{i_2}, v_{i_4}, p)$ for graphs $G_1$ and $G_2$ sharing a separating pair $\{v_{i_1}, v_{i_2}\}$, for the graph $G$ obtained from the union of $G_1$ and $G_2$ by adding the edge $(v_{i_3}, v_{i_4})$ one can determine $\M(G,v_{i_1}, v_{i_3}, v_{i_4}, p)$ (if it exists)  (Lemma~\ref{lem:lina-merging}), and
 \item splitting of components, i.e., the reverse of the previous operation (Lemma~\ref{lem:lina-splitting}).
\end{itemize}

\subsection{Maintaining Tutte information for coherent paths}\label{sec:tutte:coherent}

We now show that the information $\M(H,v_{i_1}, v_{i_2}, v_{i_3}, p)$ can be maintained for $3$-connected components $H$ of planar graphs, as long as the involved inverses exist modulo the prime $p$.
To do so, we also need this information for 3-connected components induced by coherent paths in triconnected component trees, since edge insertions may merge paths of such trees into one $3$-connected component.

Let $\rho = \rho(C_1,C_2,(a_1,a_2))$ be a ($(a_1,a_2)$-) coherent path. By $G[\rho]$ we denote the $3$-connected component that is created by the $\es[+]{2}{3}$ insertion of the edge $(a_1, a_2)$.

\begin{lemma}
\label{lem:tutte:coherent}
	Let $G$ be a planar graph with $n$ vertices and let $p$ be a prime of magnitude $\bigO(n^c)$, for some constant $c$. Assume that $\M(H,v_{i_1}, v_{i_2}, v_{i_3}, p)$ is available for every $H, v_{i_1}, v_{i_2}, v_{i_3}$ where $H$ is a
	\begin{itemize}
		\item $3$-connected component of $G$, or
		\item a graph $ G[\rho]$ induced by a coherent path $\rho$ in $\triTree(G)$
	\end{itemize}
	and $v_{i_1}, v_{i_2}, v_{i_3}$ are vertices on the same face of $H$.

	Let $G'$ be a planar graph that results from $G$ by inserting or deleting some edge. Suppose $H'$ is any $3$-connected component of $G'$ or equal to $G'[\rho']$ for any coherent path $\rho'$ in $\triTree(G')$, and $v'_{i_1}, v'_{i_2}, v'_{i_3}$ are any of its vertices that are on the same face.

	There is a constant number of embedding-inducing $(H_j, v^j_1, v^j_2, v^j_3)_{j \leq d}$, where the $H_j$ are $3$-connected planar graphs and $v^j_1, v^j_2, v^j_3$ are vertices that are on the same face of $H_j$, such that one can determine in \FOar whether all of $(H', v'_{i_1}, v'_{i_2}, v'_{i_3}, p)$ and $(H_j, v^j_1, v^j_2, v^j_3, p)_{j \leq d}$ are modulo embedding-inducing and if so, $\M(H',v'_{i_1}, v'_{i_2}, v'_{i_3}, p)$ can be defined (using $\M(H_j, v^j_1, v^j_2, v^j_3, p)$).
\end{lemma}
\begin{proof}
	The proof distinguishes all edge type changes  from Section~\ref{section:changes} where the triconnected component tree is affected. 
	\subparagraph*{Cases $\es[+]{3}{3}$ and $\es[-]{3}{3}$.} The statement directly follows from Lemma~\ref{lem:lina-edges}.
	\subparagraph*{Case $\es[+]{2}{3}$.} Consider an edge change of type $\es[+]{2}{3}$ and suppose $G'$ results from $G$ by insertion of the edge $(u,v)$, where $u$ and $v$ are in different triconnected components $D_1$, $D_2$ on a coherent path $\rho(D_1,D_2,(u,v))$ in the triconnected component tree $\triTree(G)$ of $G$. Accordingly, the components of that path coalesce into a single $3$-connected component in the triconnected component tree of $G'$.

We show how to maintain the auxiliary information for a coherent path $\rho(C_1, C_2, (a_1,a_2))$ in $\triTree(G')$ between some components $C_1$ and $C_2$ that include the vertices $a_1, a_2$, respectively. We assume that the coalesced component lies on $\rho(C_1, C_2, (a_1,a_2))$ in $\triTree(G')$, as otherwise no update is necessary.

We only consider the general case that $C_1, C_2$ are not the newly coalesced component, otherwise the following explanations can easily be adapted. It follows that $\rho(C_1, C_2, (a_1,a_2))$ is also a coherent path in $\triTree(G)$. See Figure~\ref{lm:tutte2conn2and3} for a sketch.

\begin{figure}[t]\centering

	\scalebox{0.7}{
		\begin{tikzpicture}[
			xscale=0.7,
			yscale=0.7,
			every node/.style={inner sep=0pt},
			vertex/.style={circle,fill=black,inner sep=1.6pt},
			hollow/.style={circle,draw, fill=white,inner sep=1.6pt},
			region/.style={draw=black,fill=gray!35},
			edge/.style={line width=1pt}
			]
			
			\def\sepdist{0.2} %
			\def\sepdistoncircle{0.1} %
			\def\innerdist{1} %
			\def\outerdist{3} %
			\def\innerangle{0} %
			\def\outerangle{70} %

			\filldraw[region] (0,0) arc (360:0:1)
			node[pos=7/12-\sepdistoncircle/2] (ul1a) {}
			node[pos=7/12+\sepdistoncircle/2] (ul1b) {}
			node[pos=0-\sepdistoncircle/2] (c1a) {}
			node[pos=0+\sepdistoncircle/2] (c1b) {}
			node[pos=5/12-\sepdistoncircle/2] (ll1a) {}
			node[pos=5/12+\sepdistoncircle/2] (ll1b) {}
			;

			\def\dir{150} %
			\filldraw[region] (ul1a) .. controls ++(\dir+\outerangle/2:\outerdist) and ++(\dir-\outerangle/2:\outerdist) ..
			node[pos=0.5 - \sepdist/2] (ul2a) {}
			node[pos=0.5 + \sepdist/2] (ul2b) {}
			(ul1b)
			..controls ++(\dir+\innerangle/2:\innerdist) and ++(\dir-\innerangle/2:\innerdist) .. (ul1a)
			;
			
			\filldraw[region] (ul2a) .. controls ++(\dir+\outerangle/2:\outerdist) and ++(\dir-\outerangle/2:\outerdist) ..
			node[pos=0.5 - \sepdist/2] (ul3a) {}
			node[ pos=0.5 + \sepdist/2] (ul3b) {}
			(ul2b)
			..controls ++(\dir+\innerangle/2:\innerdist) and ++(\dir-\innerangle/2:\innerdist) .. (ul2a)
			;
			
			\filldraw[region] (ul3a) .. controls ++(\dir+\outerangle/2:\outerdist) and ++(\dir-\outerangle/2:\outerdist) ..
			node[pos=0.67] (u) {}
			(ul3b)
			..controls ++(\dir+\innerangle/2:\innerdist) and ++(\dir-\innerangle/2:\innerdist) .. (ul3a)
			;

			\def\dir{0}
			\filldraw[region] (c1a) .. controls ++(\dir+\outerangle/2:\outerdist) and ++(\dir-\outerangle/2:\outerdist) ..
			node[pos=0.5 - \sepdist/2] (c2a) {}
			node[pos=0.5 + \sepdist/2] (c2b) {}
			(c1b)
			..controls ++(\dir+\innerangle/2:\innerdist) and ++(\dir-\innerangle/2:\innerdist) .. (c1a)
			;
			
			\filldraw[region] (c2a) .. controls ++(\dir+\outerangle/2:\outerdist) and ++(\dir-\outerangle/2:\outerdist) ..
			node[pos=0.5 - \sepdist/2] (c3a) {}
			node[pos=0.5 + \sepdist/2] (c3b) {}
			(c2b)
			..controls ++(\dir+\innerangle/2:\innerdist) and ++(\dir-\innerangle/2:\innerdist) .. (c2a)
			;
			
			\filldraw[region] (c3a) .. controls ++(\dir+\outerangle/2:\outerdist) and ++(\dir-\outerangle/2:\outerdist) ..
			node[pos=0.33 - \sepdist/2] (ur1a) {}
			node[pos=0.33 + \sepdist/2] (ur1b) {}
			node[pos=0.67 - \sepdist/2] (lr1a) {}
			node[pos=0.67 + \sepdist/2] (lr1b) {}
			(c3b)
			..controls ++(\dir+\innerangle/2:\innerdist) and ++(\dir-\innerangle/2:\innerdist) .. (c3a)
			;

			\def\dir{210} %
			\filldraw[region] (ll1a) .. controls ++(\dir+\outerangle/2:\outerdist) and ++(\dir-\outerangle/2:\outerdist) ..
			node[pos=0.5 - \sepdist/2] (ll2a) {}
			node[pos=0.5 + \sepdist/2] (ll2b) {}
			(ll1b)
			..controls ++(\dir+\innerangle/2:\innerdist) and ++(\dir-\innerangle/2:\innerdist) .. (ll1a)
			;
			
			\filldraw[region] (ll2a) .. controls ++(\dir+\outerangle/2:\outerdist) and ++(\dir-\outerangle/2:\outerdist) ..
			node[pos=0.5 - \sepdist/2] (ll3a) {}
			node[ pos=0.5 + \sepdist/2] (ll3b) {}
			(ll2b)
			..controls ++(\dir+\innerangle/2:\innerdist) and ++(\dir-\innerangle/2:\innerdist) .. (ll2a)
			;
			
			\filldraw[region] (ll3a) .. controls ++(\dir+\outerangle/2:\outerdist) and ++(\dir-\outerangle/2:\outerdist) ..
			node[pos=0.33] (a1) {}
			(ll3b)
			..controls ++(\dir+\innerangle/2:\innerdist) and ++(\dir-\innerangle/2:\innerdist) .. (ll3a)
			;

			\def\dir{30}
			\filldraw[region] (ur1a) .. controls ++(\dir+\outerangle/2:\outerdist) and ++(\dir-\outerangle/2:\outerdist) ..
			node[pos=0.5 - \sepdist/2] (ur2a) {}
			node[pos=0.5 + \sepdist/2] (ur2b) {}
			(ur1b)
			..controls ++(\dir+\innerangle/2:\innerdist) and ++(\dir-\innerangle/2:\innerdist) .. (ur1a)
			;
			
			\filldraw[region] (ur2a) .. controls ++(\dir+\outerangle/2:\outerdist) and ++(\dir-\outerangle/2:\outerdist) ..
			node[pos=0.5 - \sepdist/2] (ur3a) {}
			node[ pos=0.5 + \sepdist/2] (ur3b) {}
			(ur2b)
			..controls ++(\dir+\innerangle/2:\innerdist) and ++(\dir-\innerangle/2:\innerdist) .. (ur2a)
			;
			
			\filldraw[region] (ur3a) .. controls ++(\dir+\outerangle/2:\outerdist) and ++(\dir-\outerangle/2:\outerdist) ..
			node[pos=0.33] (v) {}
			(ur3b)
			..controls ++(\dir+\innerangle/2:\innerdist) and ++(\dir-\innerangle/2:\innerdist) .. (ur3a)
			;

			\def\dir{330}
			\filldraw[region] (lr1a) .. controls ++(\dir+\outerangle/2:\outerdist) and ++(\dir-\outerangle/2:\outerdist) ..
			node[pos=0.5 - \sepdist/2] (lr2a) {}
			node[pos=0.5 + \sepdist/2] (lr2b) {}
			(lr1b)
			..controls ++(\dir+\innerangle/2:\innerdist) and ++(\dir-\innerangle/2:\innerdist) .. (lr1a)
			;
			
			\filldraw[region] (lr2a) .. controls ++(\dir+\outerangle/2:\outerdist) and ++(\dir-\outerangle/2:\outerdist) ..
			node[pos=0.5 - \sepdist/2] (lr3a) {}
			node[ pos=0.5 + \sepdist/2] (lr3b) {}
			(lr2b)
			..controls ++(\dir+\innerangle/2:\innerdist) and ++(\dir-\innerangle/2:\innerdist) .. (lr2a)
			;
			
			\filldraw[region] (lr3a) .. controls ++(\dir+\outerangle/2:\outerdist) and ++(\dir-\outerangle/2:\outerdist) ..
			node[pos=0.67] (a2) {}
			(lr3b)
			..controls ++(\dir+\innerangle/2:\innerdist) and ++(\dir-\innerangle/2:\innerdist) .. (lr3a)
			;

			\draw [dotted,color=red, line width=1.5pt]
			(u) to[out=0,in=110] (ur1a)
			;
			
			\draw [dotted,color=red, line width=1.5pt]
			(v) ..controls ++(10:3) and ++(-10:3) .. (ur1b)
			;
			
			\draw [dotted,color=red, line width=1.5pt]
			(u) ..controls ++(170:3) and ++(-170:3) .. (ul1a)
			;
			
			\draw [dashed,color=blue, line width=1.5pt]
			(u) to[out=10,in=170] (v)
			;
			
			\draw [dashed,color=purple, line width=1.5pt]
			(a1) to[out=-10,in=-170] (a2)
			;

			\foreach \p in {u,v,a1,a2,ul1a,ul1b,ur1a,ur1b}{
				\node[vertex] at (\p) {};
			}
			
			\def\labeldist{0.3}
			\node[below= \labeldist of u] {$D_1$};
			\node[below= \labeldist of v] {$D_2$};
			\node[below= \labeldist of ul2b] {$B_1$};
			\node[below= \labeldist of ur2a] {$B_2$};
			\node at (-1,0) {$I_1$};
			\node[below= \labeldist of ur1a] {$I_2$};
			\node[above= \labeldist of a1] {$C_1$};
			\node[above= \labeldist of a2] {$C_2$};
			
			\node[above=\labeldist] at (u) {$u$};
			\node[above=\labeldist] at (v) {$v$};
			\node[below left=\labeldist/2] at (ul1a) {$s_1$};
			\node[above right=\labeldist] at (ul1b) {$t_1$};
			\node[above left=\labeldist] at (ur1a) {$s_2$};
			\node[below right=\labeldist-0.2 and \labeldist] at (ur1b) {$t_2$};
			\node[below=\labeldist] at (a1) {$a_1$};
			\node[below=\labeldist] at (a2) {$a_2$};
		\end{tikzpicture}
	}
	\caption{Illustration of Case $\es[+]{2}{3}$ in the proof of Lemma \ref{lem:tutte:coherent}.%
		\label{lm:tutte2conn2and3}}
\end{figure}

The coherent path $\rho(D_1,D_2,(u,v))$ in $G$ can be divided into three subpaths: the intersection with the coherent path $\rho(C_1,C_2,(a_1,a_2))$ and the parts before and after the intersection. In Figure~\ref{lm:tutte2conn2and3}, these are (1) the path between the components $D_1$ and~$B_1$, (2) the path between $I_1$ and $I_2$ and (3) the path between $B_2$ and $D_2$.

We assume here that the paths from $D_1$ to $B_1$ and from $D_2$ to $B_2$ are both not a single cycle component but either a single $3$-connected component or a non-trivial path consisting of triconnected components. %
Let $\{s_1, t_1\}$ be the separating pair between the components $B_1$ and $I_1$ and let $\{s_2, t_2\}$ be the separating pair between the components $B_2$ and $I_2$.
Let $G_1 = G[\rho(D_1,B_1,(u,s_1))]$ and $G_3= G[\rho(D_2,B_2,(v,t_2))]$ be the graphs induced by the coherent paths before and after the intersection. Let $G_2 = G[\rho(C_1,C_2,(a_1,a_2))]$ be the $3$-connected component associated with the coherent path under consideration in the unchanged graph.
According to the assumption of the lemma, $\M(G_1,s_1, t_1, u, p)$, $\M(G_2,s_1, t_1, s_2, p)$ and $\M(G_3,s_2, t_2, v, p)$ are available.

We now explain how to compute $\M(H',v'_{i_1}, v'_{i_2}, v'_{i_3}, p)$, where $H' = G'[\rho(C_1, C_2, (a_1,a_2))]$ and $v'_{i_1},v'_{i_2}, v'_{i_3}$ are three of its vertices that are on the same face. Notice that $H'$ differs from the union of the three graphs $G_1, G_2, G_3$ only by a constant number of edges:
Let $H_1$ be the graph that is obtained from the union of $G_1$ and $G_2$ by adding the edge $(u,s_2)$ and let $H_2$ be the graph that is obtained from the union of $H_1$ and $G_3$ by adding the edge $(u,v)$. Then $H'$ is obtained from $H_2$ by removing the edges $(u,s_1), (u,s_2)$ and $(v, t_2)$ (and possibly virtual edges $(s_1,t_1)$ and $(s_2,t_2)$).
According to Lemmas~\ref{lem:lina-edges}, \ref{lem:lina-pins} and \ref{lem:lina-merging}, we can determine in \FOar for each graph of that sequence whether (together with suitably chosen vertices) it is modulo embedding-inducing and, if they all are, $\M(H',v'_{i_1}, v'_{i_2}, v'_{i_3}, p)$ can be obtained.

We give more details.
Assuming that every intermediate structure is modulo embedding-inducing, from $\M(G_1,s_1, t_1, u, p)$ and $\M(G_2,s_1, t_1, s_2, p)$ one can obtain $\M(H_1,t_1, u, s_2, p)$ via Lemma~\ref{lem:lina-merging} and $\M(H_1,s_2, t_2, u, p)$ via Lemma~\ref{lem:lina-pins}; then using $\M(G_3,s_2, t_2, v, p)$ one can obtain $\M(H_2,s_2, u, v, p)$ again via Lemma~\ref{lem:lina-merging}. Applying Lemma~\ref{lem:lina-edges} three to five times, corresponding to the deletions of the edges $(u,s_1), (u,s_2)$ and $(v, t_2)$ and possibly of virtual edges $(s_1,t_1)$ or $(s_2,t_2)$ if the corresponding vertex set is not a $3$-connected separating pair any more in $H'$, one can obtain $\M(H',s_2, u, v, p)$.
Finally, via Lemma~\ref{lem:lina-pins}, one can obtain $\M(H',v'_{i_1}, v'_{i_2}, v'_{i_3}, p)$.

Note that all intermediate graphs are $3$-connected planar graphs and their three distinguished vertices are on the same face, so all intermediate structures are embedding-inducing by Lemma~\ref{lem:lina-invertable}.

It remains to explain the approach if at least one of the paths from $D_1$ to $B_1$ or from $D_2$ to $B_2$ consists only of a single cycle component. We detail the case that the path from $D_1$ to $B_1$ is a single cycle component. Then only the vertices $u, s_1, t_1$ of this cycle are part of the coalesced $3$-connected component in $G'$. In the explanation above we can therefore replace $H_1$ by the graph that results from $G_2$ by adding the vertex $u$ and the three edges $(u,s_1), (u,t_1), (u,s_2)$. Using Lemma~\ref{lem:lina-addvertex}, we can obtain the required information $\M(H_1,t_1, u, s_2, p)$ also for this graph.
	
	\subparagraph*{Case $\es[-]{3}{2}$.}
	We now consider the case of the deletion of an edge $(u,v)$ such that $u$ and $v$ are in a common $3$-connected component in $G$ but in different $3$-connected components $D_1, D_2$ in the resulting graph $G'$. So, the original $3$-connected components unfurls into a path of triconnected components after the edge deletion. Those components and their separating pairs can be defined using Lemma \ref{theorem:maintainability-decomposition}. This case is the reverse of Case $\es[+]{2}{3}$.
	
	We again show how to maintain the auxiliary information for a coherent path $\rho' = \rho(C_1, C_2, (a_1,a_2))$ in $\triTree(G')$ between some components $C_1$ and $C_2$, assuming that the unfurling $3$-connected component is part of the corresponding coherent path in $G$. After that component unfurls, only a subpath of it is on the coherent path $\rho'$ in $G'$, see again Figure~\ref{lm:tutte2conn2and3}: after the edge $(u,v)$ is deleted only the subpath from component $I_1$ to component $I_2$ is relevant for $\rho'$.
	
	As in the previous case $\es[+]{2}{3}$ we first assume that the subpaths from component $D_1$ to the last triconnected component $B_1$ before $I_1$ and from $D_2$ to the last component $B_2$ before $I_2$ are both not single cycle components but either a $3$-connected component or a non-trivial path of more than one triconnected component.
	
	Let $H = G[\rho]$ be the graph induced by $\rho = \rho(C_1, C_2, (a_1,a_2))$ in the original graph. According to the assumption of the lemma, $\M(H,s_2, u, v, p)$ is available.
	Let $H_3$ be the graph that results from $H$ by adding the edges $(u,s_1), (u,s_2)$ and $(v, t_2)$, where $\{s_1, t_1\}$ is the separating pair between the components $B_1$ and $I_1$ and $\{s_2, t_2\}$ is the separating pair between $B_2$ and $I_2$ (see Figure~\ref{lm:tutte2conn2and3}).
	Assuming that every intermediate structure is modulo embedding-inducing, applying Lemma~\ref{lem:lina-edges} three times, one can obtain $\M(H_3,s_2, u, v, p)$.
	Using Lemma~\ref{lem:lina-splitting} we now aim to compute $\M(H',v'_{i_1}, v'_{i_2}, v'_{i_3}, p)$ for $H' = G'[\rho']$. Intuitively, from $H$ we will ``split off'' the subpaths from $D_1$ to $B_1$ and from $D_2$ to $B_2$ to obtain $H'$.
	
	Let $H'_1$ be the graph $G[\rho(D_2,B_2,(v,t_2))]$ and $H_4$ be the graph such that $H_3$ is the union of $H_4$ and $H'_1$ and the edge $(u,v)$. According to Lemma~\ref{lem:lina-splitting} we can obtain $\M(H_4,s_2, t_2, u, p)$ and then $\M(H_4,s_1, u, s_2, p)$ via Lemma~\ref{lem:lina-pins}, assuming both structures are modulo embedding-inducing.
	The graph $H_4$ is the union of the graph $H'$ and the graph $G[\rho(D_1,B_1,(u,s_1))]$ and the edge $(u,s_1)$. Again using Lemma~\ref{lem:lina-splitting}, we can obtain $\M(H',s_1, t_1, s_2, p)$ and finally $\M(H',v'_{i_1}, v'_{i_2}, v'_{i_3}, p)$ applying Lemma~\ref{lem:lina-pins}.
	
	Again, all intermediate graphs are $3$-connected planar graphs and the respective three distinguished vertices are on the same face, so all intermediate structures are embedding-inducing according to Lemma~\ref{lem:lina-invertable}.
	
	If the path from $D_1$ to $B_1$ or the path from $D_2$ to $B_2$ only consists of a single cycle component, the approach can be altered analogously as in the case $\es[+]{2}{3}$: if the path from $D_1$ to $B_1$ is only one cycle component, $H_4$ is the union of $H'$ and the single vertex $u$, together with the edges $(u,s_1),(u,t_1), (u,s_2)$. We can obtain $\M(H',s_1, t_1, s_2, p)$ using Lemma~\ref{lem:lina-removevertex}.
	
	\subparagraph*{Cases $\es[+]{2}{2}$ and $\es[-]{2}{2}$.}
	The insertion or deletion of an edge $(u,v)$ such that $u$ and $v$ are part of a common cycle component may be relevant for maintaining the auxiliary information for a coherent path, assuming that the cycle component lies on that path. In that case, $u$ and $v$ become part of the $3$-connected component associated with the coherent path, if they were not already included before. See Figure~\ref{fig:tuttecase2to2} for an illustration.
	If the vertex $v$ is added to the $3$-connected component, it subdivides an already existing edge: the (virtual) edge between the vertices $s_1, s_2$ of the two separating pairs $\{s_1, t_1\}$ and $\{s_2, t_2\}$ that are adjacent to the cycle component in the coherent path, such that $s_1, s_2$ are closer to $v$ in the cycle than $t_1$ and $t_2$, respectively.
	\begin{figure}[t]
		\centering
		\begin{subfigure}{.45\textwidth}
			\centering
			\scalebox{1}{
				\begin{tikzpicture}[
					scale=1,
					every node/.style={inner sep=0pt, label distance=0.5mm},
					vertex/.style={circle,fill=black,inner sep=1pt},
					hollow/.style={circle,draw, fill=white,inner sep=1.6pt},
					region/.style={draw=black,fill=gray!35},
					edge/.style={line width=1pt}
					]
					
					\coordinate (v) at (0,1) {};
					\coordinate (u) at (0,-1) {};
					\coordinate (s1) at (-1,0.5) {};
					\coordinate (s1p) at (1,0.5) {};
					\coordinate (s2) at (-1,-0.5) {};
					\coordinate (s2p) at (1,-0.5) {};
					\coordinate (node) at (0.5,-0.75) {};

					\draw[dashed] 
					(v)
					to[bend left=30] (s1)
					to[bend left=30] (s2)
					to[bend left=30] (u)
					to[bend left=30] (node)
					to[bend left=30] (s2p)
					to[bend left=30] (s1p)
					to[bend left=30] (v)
					;

					\def\sepdist{0.3} %
					\def\sepdistoncircle{0.1} %
					\def\innerdist{0.3} %
					\def\outerdist{1} %
					\def\innerangle{80} %
					\def\outerangle{30} %

					\def\dir{180} 
					\filldraw[region] (s2) .. controls ++(\dir+\outerangle/2:\outerdist) and ++(\dir-\outerangle/2:\outerdist) ..
					node[pos=0.5 - \sepdist/2] (s2-1) {}
					node[pos=0.5 + \sepdist/2] (s1-1) {}
					(s1)
					..controls ++(\dir+\innerangle/2:\innerdist) and ++(\dir-\innerangle/2:\innerdist) .. (s2)
					;
					
					\def\outerdist{0.7} %
					
					\filldraw[region] (s2-1) .. controls ++(\dir+\outerangle/2:\outerdist) and ++(\dir-\outerangle/2:\outerdist) ..
					node[pos=0.5] (a) {}
					(s1-1)
					..controls ++(\dir+\innerangle/2:\innerdist) and ++(\dir-\innerangle/2:\innerdist) .. (s2-1)
					;
					
					\def\dir{0} 
					\def\outerdist{1} %
					
					\filldraw[region] (s1p) .. controls ++(\dir+\outerangle/2:\outerdist) and ++(\dir-\outerangle/2:\outerdist) ..
					node[pos=0.5 - \sepdist/2] (s1p-1) {}
					node[pos=0.5 + \sepdist/2] (s2p-1) {}
					(s2p)
					..controls ++(\dir+\innerangle/2:\innerdist) and ++(\dir-\innerangle/2:\innerdist) .. (s1p)
					;
					
					\def\outerdist{0.7} %
					
					\filldraw[region] (s1p-1) .. controls ++(\dir+\outerangle/2:\outerdist) and ++(\dir-\outerangle/2:\outerdist) ..
					node[pos=0.5] (b) {}
					(s2p-1)
					..controls ++(\dir+\innerangle/2:\innerdist) and ++(\dir-\innerangle/2:\innerdist) .. (s1p-1)
					;
					
					\node[vertex,fill=black,label=above:$v$] at (v) {};
					\node[vertex,fill=black,label=below:$u$] at  (u)  {};
					
					\node[vertex,label=above:$s_1$]  at (s1)  {};
					\node[vertex,label=above:$s_2$]  at (s1p)  {};
					
					\node[vertex,label=below:$t_1$] at  (s2) {};
					\node[vertex,label=below:$t_2$]  at (s2p)  {};
					
					\node[vertex] at (node) {};
					\node[vertex] at (s1-1) {};
					\node[vertex] at (s2-1) {};
					\node[vertex] at (s1p-1) {};
					\node[vertex] at (s2p-1) {};

					\node[vertex, label=below left:$a_1$] at (a) {};
					\node[vertex, label=below right:$a_2$] at (b) {};
					
				\end{tikzpicture}
			}
			\caption{A $(a_1,a_2)$-coherent path $\rho$ that contains a cycle component.\label{fig:tutte2to2}}
		\end{subfigure}\quad
		\begin{subfigure}{.45\textwidth}
			\centering
			\scalebox{1}{
				\begin{tikzpicture}[
					scale=1,
					every node/.style={inner sep=0pt, label distance=0.5mm},
					vertex/.style={circle,fill=black,inner sep=1pt},
					hollow/.style={circle,draw, fill=white,inner sep=1.6pt},
					region/.style={draw=black,fill=gray!35},
					edge/.style={line width=1pt}
					]
					
					\coordinate (s1) at (0,0.5) {};
					\coordinate (s1p) at (1,0.5) {};
					\coordinate (s2) at (0,-0.5) {};
					\coordinate (s2p) at (1,-0.5) {};
					\coordinate (node) at (0.5,-0.75) {};

					\draw[dashed] 
					(s1)
					to (s2)
					to (s2p)
					to (s1p)
					to (s1)
					;

					\def\sepdist{0.3} %
					\def\sepdistoncircle{0.1} %
					\def\innerdist{0.3} %
					\def\outerdist{1} %
					\def\innerangle{80} %
					\def\outerangle{30} %

					\def\dir{180} 
					\filldraw[region] (s2) .. controls ++(\dir+\outerangle/2:\outerdist) and ++(\dir-\outerangle/2:\outerdist) ..
					node[pos=0.5 - \sepdist/2] (s2-1) {}
					node[pos=0.5 + \sepdist/2] (s1-1) {}
					(s1)
					..controls ++(\dir+\innerangle/2:\innerdist) and ++(\dir-\innerangle/2:\innerdist) .. (s2)
					;
					
					\def\outerdist{0.7} %
					
					\filldraw[region] (s2-1) .. controls ++(\dir+\outerangle/2:\outerdist) and ++(\dir-\outerangle/2:\outerdist) ..
					node[pos=0.5] (a) {}
					(s1-1)
					..controls ++(\dir+\innerangle/2:\innerdist) and ++(\dir-\innerangle/2:\innerdist) .. (s2-1)
					;
					
					\def\dir{0} 
					\def\outerdist{1} %
					
					\filldraw[region] (s1p) .. controls ++(\dir+\outerangle/2:\outerdist) and ++(\dir-\outerangle/2:\outerdist) ..
					node[pos=0.5 - \sepdist/2] (s1p-1) {}
					node[pos=0.5 + \sepdist/2] (s2p-1) {}
					(s2p)
					..controls ++(\dir+\innerangle/2:\innerdist) and ++(\dir-\innerangle/2:\innerdist) .. (s1p)
					;
					
					\def\outerdist{0.7} %
					
					\filldraw[region] (s1p-1) .. controls ++(\dir+\outerangle/2:\outerdist) and ++(\dir-\outerangle/2:\outerdist) ..
					node[pos=0.5] (b) {}
					(s2p-1)
					..controls ++(\dir+\innerangle/2:\innerdist) and ++(\dir-\innerangle/2:\innerdist) .. (s1p-1)
					;
					
					\node[vertex,label=above:$s_1$]  at (s1)  {};
					\node[vertex,label=above:$s_2$]  at (s1p)  {};
					
					\node[vertex,label=below:$t_1$] at  (s2) {};
					\node[vertex,label=below:$t_2$]  at (s2p)  {};
					
					\node[vertex] at (s1-1) {};
					\node[vertex] at (s2-1) {};
					\node[vertex] at (s1p-1) {};
					\node[vertex] at (s2p-1) {};
					
					\draw[black] (a) .. controls ++(100:2) and ++(80:2) .. (b);
					
					\node[vertex, label=below left:$a_1$] at (a) {};
					\node[vertex, label=below right:$a_2$] at (b) {};
					
				\end{tikzpicture}
			}
			\caption{The $3$-connected component induced by $\rho$. \label{fig:tutte2to2easy}}
		\end{subfigure}\quad
		\begin{subfigure}{.45\textwidth}
			\centering
			\scalebox{1}{
				\begin{tikzpicture}[
					scale=1,
					every node/.style={inner sep=0pt, label distance=0.5mm},
					vertex/.style={circle,fill=black,inner sep=1pt},
					hollow/.style={circle,draw, fill=white,inner sep=1.6pt},
					region/.style={draw=black,fill=gray!35},
					edge/.style={line width=1pt}
					]
					
					\coordinate (v) at (0,1) {};
					\coordinate (u) at (0,-1) {};
					\coordinate (s1) at (-1,0.5) {};
					\coordinate (s1p) at (1,0.5) {};
					\coordinate (s2) at (-1,-0.5) {};
					\coordinate (s2p) at (1,-0.5) {};

					\draw[dashed] 
					(v)
					to(s1)
					to (s2)
					to (u)
					to (s2p)
					to (s1p)
					to (v)
					;
					
					\draw[black, line width=1pt] (u) -- (v);

					\def\sepdist{0.3} %
					\def\sepdistoncircle{0.1} %
					\def\innerdist{0.3} %
					\def\outerdist{1} %
					\def\innerangle{80} %
					\def\outerangle{30} %

					\def\dir{180} 
					\filldraw[region] (s2) .. controls ++(\dir+\outerangle/2:\outerdist) and ++(\dir-\outerangle/2:\outerdist) ..
					node[pos=0.5 - \sepdist/2] (s2-1) {}
					node[pos=0.5 + \sepdist/2] (s1-1) {}
					(s1)
					..controls ++(\dir+\innerangle/2:\innerdist) and ++(\dir-\innerangle/2:\innerdist) .. (s2)
					;
					
					\def\outerdist{0.7} %
					
					\filldraw[region] (s2-1) .. controls ++(\dir+\outerangle/2:\outerdist) and ++(\dir-\outerangle/2:\outerdist) ..
					node[pos=0.5] (a) {}
					(s1-1)
					..controls ++(\dir+\innerangle/2:\innerdist) and ++(\dir-\innerangle/2:\innerdist) .. (s2-1)
					;
					
					\def\dir{0} 
					\def\outerdist{1} %
					
					\filldraw[region] (s1p) .. controls ++(\dir+\outerangle/2:\outerdist) and ++(\dir-\outerangle/2:\outerdist) ..
					node[pos=0.5 - \sepdist/2] (s1p-1) {}
					node[pos=0.5 + \sepdist/2] (s2p-1) {}
					(s2p)
					..controls ++(\dir+\innerangle/2:\innerdist) and ++(\dir-\innerangle/2:\innerdist) .. (s1p)
					;
					
					\def\outerdist{0.7} %
					
					\filldraw[region] (s1p-1) .. controls ++(\dir+\outerangle/2:\outerdist) and ++(\dir-\outerangle/2:\outerdist) ..
					node[pos=0.5] (b) {}
					(s2p-1)
					..controls ++(\dir+\innerangle/2:\innerdist) and ++(\dir-\innerangle/2:\innerdist) .. (s1p-1)
					;
					
					\node[vertex,label=above:$v$] at (v) {};
					\node[vertex,label=below:$u$] at  (u)  {};
					
					\node[vertex,label=above:$s_1$]  at (s1)  {};
					\node[vertex,label=above:$s_2$]  at (s1p)  {};
					
					\node[vertex,label=below:$t_1$] at  (s2) {};
					\node[vertex,label=below:$t_2$]  at (s2p)  {};
					
					\node[vertex] at (s1-1) {};
					\node[vertex] at (s2-1) {};
					\node[vertex] at (s1p-1) {};
					\node[vertex] at (s2p-1) {};
					
					\draw[black] (a) .. controls ++(100:2) and ++(80:2) .. (b);
					
					\node[vertex, label=below left:$a_1$] at (a) {};
					\node[vertex, label=below right:$a_2$] at (b) {};
					
				\end{tikzpicture}
			}
			\caption{The $3$-connected component induced by $\rho'$ (if $(u,v)$ has been inserted).\label{fig:tutte2to2insert}}
		\end{subfigure}
		\caption{Illustration for the case $\es[+]{2}{2}$ in the proof of Lemma~\ref{lem:tutte:coherent}. \label{fig:tuttecase2to2}}
	\end{figure}
	If one or both of $u, v$ are added to the $3$-connected component, this addition together with the corresponding edges can be processed using Lemma~\ref{lem:lina-addvertex} respectively Lemma~\ref{lem:lina-addpair}. The deletion of subdivided edges can be done by applying Lemma~\ref{lem:lina-edges} constantly often.
	
	If vertices are removed from the $3$-connected component following a change of the type $\es[-]{2}{2}$, Lemma~\ref{lem:lina-removevertex} respectively Lemma~\ref{lem:lina-removepair} are used instead.%
	
	\subparagraph*{Case $\es[+]{1}{2}$.} In this case, the triconnected component tree $\triTree(G')$ of $G'$ consists of multiple triconnected component trees $\triTree(B_1), \ldots, \triTree(B_k)$ for biconnected components $B_1, \ldots, B_k$ of $G$ that are joined along a new cycle component. See Figure~\ref{fig:BCpath:3tree} for an illustration.
	Accordingly, all $3$-connected components of $G'$ already exist in one of the components $B_i$ and the required information is already available, bar possible virtual edge insertions inside these components that can be handled as discussed in the cases above.
	
	If $H' = G'[\rho']$ for a coherent path $\rho'$ that did not already exist in one of the components $B_i$ then it can only intersect two of these components and the new cycle component. Let $\rho = \rho(C_1, C_2, (a_1,a_2))$, where $C_1, C_2$ are triconnected components in $B_1, B_2$, respectively. Let $\{s_1, s_2\}, \{t_1,t_2\}$ be the respective separating pairs between these components and the new cycle component and let $D_1, D_2$ be the respective triconnected components that contain $\{s_1, s_2\}$ and $\{t_1,t_2\}$.
	It follows that $\rho_1 = \rho(C_1,D_1,(a_1,s_1))$ and $\rho_2 = \rho(C_2,D_2,(a_2,t_1))$ are coherent paths in $G$ and therefore $\M(G[\rho_1],a_1, s_1, s_2, p)$ and $\M(G[\rho_2],a_2, t_1, t_2, p)$ are available.
	
	The graph $H'$ is obtained from the union of $G[\rho_1]$ and $G[\rho_2]$ by adding the edges $(a_1, a_2)$ and $(s_1, t_1), (s_2,t_2)$ (assuming the cyclic order in the new cycle component is $(s_1, t_1, t_2, s_2)$) and deleting the edges $(a_1, s_1)$ and $(a_2, t_1)$.
	The lemma statement follows from the applications of Lemma~\ref{lem:lina-union}, Lemma~\ref{lem:lina-edges} and Lemma~\ref{lem:lina-pins} in that order.
	
	\subparagraph*{Case $\es[-]{2}{1}$.} The triconnected component tree of $G'$ is obtained from the triconnected component tree of $G$ by removing a cycle component, which disconnects certain subtrees. This means that every coherent path in $\triTree(G')$ is already a coherent path in $\triTree(G)$ and every $3$-connected component of $G'$ already exists in $G$. Starting from the already available information, only the potential deletion of virtual edges has to be processed, which is discussed in the cases above.
\end{proof}

\subsection{Maintaining Isomorphisms between Triconnected Components}\label{sec:tutte:final}
We have seen that information for the Tutte embedding can be maintained modulo primes, as long as the involved inverses exist modulo these primes. To prove Proposition~\ref{prop:3iso}, we show how to extract an isomorphism between isomorphic $3$-connected components.

\begin{proofof}{Proposition~\ref{prop:3iso}}
	If $C, C^*$ are cycle components, then they are isomorphic if and only if their lengths are equal. There is an isomorphism that maps the given vertices as given in the proposition statement if the distances between the vertices are pairwise equal. This information is available according to Lemma~\ref{lemma:cycle-distances}.
	We also need to check that the edges of the cycle components (and the same way in the following for $3$-connected components) are of the same type, that is, whether an edge is virtual or an actual graph edge in both components. This can be looked up in the input edge relation.
	
	Suppose $C, C^*$ are $3$-connected components.
	As explained in Section~\ref{sec:tutte:overview}, for a $3$-connected planar graph $H$, we can obtain its Tutte embedding with respect to three pinned vertices $v_{i_1}, v_{i_2}, v_{i_3}$ that are on a common face from the inverse of the matrix $\mat T = \mat T(H, v_{i_1}, v_{i_2}, v_{i_3})$.
	To determine whether two $3$-connected planar graphs are isomorphic, given three such vertices in one graph and the assumed isomorphic copies in the other graph, we compare the obtained embeddings to find matching pairs of vertices in the two graphs.
	
	We face two problems: (1) we do not have $\mat T^{-1}$ available directly but only $\mat T^{-1} \bmod p$ for many primes $p$; and (2) we permanently ``lose'' the information $\mat T^{-1} \bmod p$ for a prime $p$ if $(H,v_{i_1}, v_{i_2}, v_{i_3}, p)$ is not modulo embedding-inducing at some point, that is, if the encoded matrix inverses do not exist modulo the prime $p$.
	
	To solve (1), observe that the entries of $\mat T$ are bounded by the number $n$ of vertices. As the numbers in any entry $\frac{a}{b}$ in $\mat T^{-1}$ are bounded by the determinant of $\mat T$, the numbers in $\mat T^{-1}$ are at most  $n! n^n$, which for large enough~$n$ is less than $2^{n^2}$. This means that if $\mat T^{-1} \bmod p$ is available for $n^2$ many different primes $p$, we can determine for two embeddings whether they are equal, by checking whether they are equal modulo all $n^2$ different primes.
	
	For (2), assume that the inverse $\mat T(H, v_{i_1}, v_{i_2}, v_{i_3})^{-1} \bmod p$ is available for all primes $p \in P$ for some set $P$, i.e., $\M(H,v_{i_1}, v_{i_2}, v_{i_3}, p)$ is stored as auxiliary information for each $p \in P$.
	According to Lemma~\ref{lem:tutte:coherent}, this information can be updated using first-order formulas, provided that (a) for the changed graph $H'$ the structure $(H',v_{i_1}, v_{i_2}, v_{i_3}, p)$ is modulo embedding-inducing, that is, the inverse $\mat T(H', v_{i_1}, v_{i_2}, v_{i_3})^{-1} \bmod p$ exists, and (b) also a constant number of further matrix inverses exist modulo $p$.
	As the respective structures are embedding-inducing, that is, the inverses exist when evaluated over the rationals instead of modulo some prime, the only reason for non-existence of an inverse is that $p$ divides the determinant of one of the involved matrices. %
	As again the determinants are bounded by $2^{n^2}$, less than $n^2$ many primes can divide each determinant, so only for $\bigO(n^2)$ many primes $p \in P$ the auxiliary information cannot be obtained for the new graph $G'$.
	In total we maintain the auxiliary information for $\bigO(n^6)$ many graphs $H$ (one for each coherent path), in every step the auxiliary information might not be available any more for $\bigO(n^8)$ many primes.

To address that at some point we ``run out'' of primes, we use that whenever a graph problem %
can be maintained for $\log^i(n)$ many steps using first-order formulas, starting from auxiliary relations computed by $\AC^i$ circuits, then the problem is in \DynFO \cite[Theorem~4.2]{DMSVZ19}.
Here, $\AC^i$ circuits are uniform circuit families of polynomial size and depth $\log^i(n)$, and $i \in \N$ is arbitrary.

The $3$-connected components can be identified in $\LOGSPACE \subseteq \AC^1$ \cite[Lemma~4.3]{DattaLNTW22} and matrix inverses can be computed in $\NC^2 \subseteq \AC^2$ (cf.~\cite{Cook85}). Therefore, one can compute in $\AC^2$ the auxiliary information $\M(\cdot)$ for $n^{11}$ many primes, which by the prime number theorem can be found among the first $n^{12}$ many numbers.

According to Lemma~\ref{lem:tutte:coherent}, we can maintain this information for $\log^2(n)$ many steps, losing $\bigO(n^8)$ many primes in every step, but still having the information available for $n^2$ primes in the end.

As reasoned above, this is enough to express the relation $\ISOt$ for vertices $a,b,c$ and $a^*, b^*, c^*$ that are on the same face in their respective $3$-connected components. Whether this is the case can be maintained in \DynFO, cf.~\cite[Lemma~14 in conjunction with Chapter~5]{DattaKM23arx}.
\end{proofof}

\section{Maintaining Isomorphism Information for Biconnected Components}\label{section:biconnected}
\tikzset{
	tritreesketch/.style={
		>={Stealth[length=2mm]},
		tri/.style={draw, fill=gray!40, regular polygon, regular polygon sides=3,
			minimum size=6mm, inner sep=0pt},
		subtree/.style={draw, pattern=north east lines, pattern color=gray!60,
			regular polygon, regular polygon sides=3, minimum size=11mm, inner sep=0pt,
			rounded corners=1.5mm},
		bigtree/.style={draw, line width=0.5pt, pattern=north east lines,
			pattern color=gray!60, rounded corners=1.5mm},
		marked/.style={fill=blue!60!cyan},
		onhatch/.style={fill=white, inner sep=1.5pt, rounded corners=1pt},
		edge/.style={line width=1.2pt, ->},
		longedge/.style={edge, decorate,
			decoration={snake, amplitude=0.5mm, segment length=2.5mm, post length=2mm}},
		ann/.style={font=\scriptsize, align=center},
		highlight/.style={yellow!45, opacity=0.7, line width=9mm, line cap=round, line join=round}
	}
}
\newcommand{\Pn}[3][]{%
	\node[Pnode, #1] (#2) at (#3) {};
	\node[Adot] at ($(#2)+(0,0.15)$) {};
	\node[Adot] at ($(#2)+(0,-0.15)$) {};
}
\newcommand{\siblings}[2]{%
	\node[subtree] (#1a) at ($(#2)+(-0.45,0)$) {};
	\node[subtree] (#1b) at ($(#2)+(0.45,0)$) {};
}
\newcommand{\sibling}[2]{\node[subtree] (#1a) at (#2) {};}

In this section, we show how to maintain which biconnected components of a planar graph are isomorphic. Actually, slightly more generally,  we show this result for \emph{coloured graphs}. This will later be helpful for maintaining isomorphic connected components. Vertices in coloured graphs can be coloured with one of polynomially many colours. Formally, we use an additional relation $F$ that for a vertex~$v$ may contain a tuple $(v, \tpl f)$, for some tuple $\tpl f$ of vertices, with the meaning that $v$ has colour $\tpl f$.

Our goal is to provide a relation $\spqrIso$ that for a coloured planar graph $G$ contains all tuples $(a, b, a^*, b^*)$ such that
\begin{itemize}
	\item $a, b$ and $a^*, b^*$ are vertices of biconnected components $B$ and $B^*$ of $G$, respectively, and
	\item $B$ and $B^*$ are isomorphic via an isomorphism $\pi$ with $\pi(a) = a^*$ and $\pi(b) = b^*$.
\end{itemize}
So, using this relation one can check that two biconnected components are isomorphic. 
Additionally, for every tuple $(a, b, a^*, b^*) \in \spqrIso$, we maintain one witnessing isomorphism (from potentially many candidates), using a relation $\spqrIsoW$.

Let $(a, b, a^*, b^*)$ be any tuple from $\spqrIso$, and let $B$ and $B^*$ be the biconnected components that contain $a,b$ and $a^*,b^*$, respectively. 
For every vertex $c \in B$ there is exactly one vertex $c^* \in B^*$ such that $\spqrIsoW$ contains the tuple $(a, b, c, a^*, b^*, c^*)$ and the function $\pi_{a,b,a^*,b^*}$ that maps each $c \in B$ to the vertex $c^* \in B^*$ with  $(a, b, c, a^*, b^*, c^*) \in \spqrIsoW$ is an isomorphism from $B$ to $B^*$ with $\pi_{a,b,a^*,b^*}(a) = a^*$ and $\pi_{a,b,a^*,b^*}(b) = b^*$.

\begin{proposition}
\label{prop:2iso}
	For any coloured planar graph $G$, the relations $\spqrIso$ and $\spqrIsoW$
	can be maintained in \DynFO, provided that $G$ stays planar, under edge insertions and deletions and changes of vertex colours, with the following restrictions.
	Each connected component of $G$ receives at most one edge insertion or deletion. Additionally, each biconnected component receives at most a constant number of vertex colour changes.
	Beyond that, 
	\begin{enumerate}[(i)]
		\item for changes of type $\es[+]{1}{2}$ and $\es[-]{2}{1}$, all vertices of an emerging or vanishing cycle component may change their colour, and
		\item all vertices of a colour may be recoloured with an unused colour.
	\end{enumerate}
\end{proposition}

As described in the proof overview in Section \ref{section:proof-overview}, we adapt the context-based dynamic tree isomorphism algorithm by Datta et al.~\cite{DattaK0TVZ24}.
In triconnected component trees, we specify contexts by tuples $\tpl x, \tpl r, \tpl h$ of two or three graph vertices. Each tuple uniquely represents a triconnected component (if its three vertices are contained in the component) or a separating pair (if its two vertices are  the vertices of the pair). For a context $\X(\tpl x, \tpl r, \tpl h)$ and a tuple $\tpl f = (\tpl f_{\tpl r}, \tpl f_{\tpl h})$ of colours, the graph of the \emph{recoloured context} $\rcX(\tpl x, \tpl r, \tpl h, \tpl f)$ is obtained from the (coloured) graph of $\X(\tpl x, \tpl r, \tpl h)$ by recolouring the vertices in $\tpl r$ and $\tpl h$ as given by $\tpl f$. 
Two recoloured contexts $X = \rcX(\tpl x, \tpl r, \tpl h, \tpl f)$ and $X^* = \rcX(\tpl x^*, \tpl r^*, \tpl h^*, \tpl f^*)$ of the triconnected component trees of $G$ are \emph{fully isomorphic} if the graphs $\graph(X)$ and $\graph(X^*)$ are isomorphic via an isomorphism~$\pi$ such that $\pi(\tpl r) = \tpl r^*$ and $\pi(\tpl h) = \tpl h^*$. Note that fully isomorphic (recoloured) contexts $X$ and $X^*$ are in particular isomorphic, so, there is a root and hole preserving isomorphism between the (recoloured) contexts. \emph{Recoloured subtrees} $\rcST(\tpl x, \tpl r, \tpl f)$ are  defined similarly by recolouring the root $\tpl r$ of a subtree. 
Two recoloured subtrees $S = \rcST(\tpl x,\tpl r,\tpl f)$ and $S^* = \rcST(\tpl x^*,\tpl r^*,\tpl f^*)$ are fully isomorphic if $\graph(S)$ and $\graph(S^*)$ are isomorphic via an isomorphism~$\pi$ such that $\pi(\tpl r) = \tpl r^*$.
We often refer to recoloured contexts and subtrees simply as contexts and subtrees.

Our dynamic algorithm for $\spqrIso$ and $\spqrIsoW$ maintains the relations

\begin{itemize}
 \item $\spqrCIso$ that stores tuples $(X, X^*) \df (\tpl x,  \tpl r, \tpl h, \tpl f, \tpl x^*, \tpl r^*, \tpl h^*, \tpl f^*)$ such that the recoloured contexts $X \df \rcX(\tpl x, \tpl r, \tpl h, \tpl f)$ and $X^* \df \rcX(\tpl x^*, \tpl r^*, \tpl h^*, \tpl f^*)$ are fully isomorphic,
  \item $\spqrCIsoW$ that stores tuples $(X, c, X^*, c^*) \df (\tpl x,  \tpl r, \tpl f, \tpl h, c, \tpl x^*, \tpl r^*, \tpl h^*, \tpl f^*, c^*)$ that encode a witness isomorphism for every $(X, X^*) \in \spqrCIso$,
  \item $\spqrSibIsoCount$ that stores tuples $(\tpl x, \tpl r, \tpl y, \tpl f, m)$ such that (i) the triconnected component node $\tpl y$ is a child of the separating pair node $\tpl r$, and (ii) in $\rcST(\tpl x, \tpl r, \tpl f)$, the subtree rooted at $\tpl y$ is one of exactly $m$ fully isomorphic siblings,
  \item $\spqrSibIsoRank$ that for each $(\tpl x, \tpl r, \tpl y, \tpl f, m) \in \spqrSibIsoCount$ contains exactly one tuple $(\tpl x, \tpl r, \tpl y, \tpl f, k)$, for some $k \leq m$, that assigns a rank $k$ to the subtree rooted at $\tpl y$ that is unique among the $m$ fully isomorphic siblings, thus providing an enumeration of this set,
 \item $\sameFace$ that stores tuples $(\tpl c_1, \tpl c_2, a_1, a_2, s_1, s_2)$ such that $\rho(\tpl c_1, \tpl c_2, (a_1, a_2))$ is a coherent path and the vertices $s_1$ and $s_2$ that appear in two distinct separating pairs on this path are on the same face after inserting the edge $(a_1, a_2)$. The maintenance of this relation is described in \cite{DattaKM23}.
\end{itemize}

We also assume that a relation $\spqrSTIso$ is available that contains all tuples $(\tpl x, \tpl r, \tpl f, \tpl x^*, \tpl r^*, \tpl f^*)$ representing fully isomorphic recoloured subtrees $S = \rcST(\tpl x,\tpl r,\tpl f)$ and $S^* = \rcST(\tpl x^*,\tpl r^*,\tpl f^*)$, as well as a relation $\spqrSTIsoW$ with tuples of the form $(S, c, S^*, c^*)$ that encode a witness isomorphism for each such pair $S, S^*$. These can be easily defined from $\spqrCIso$ and $\spqrCIsoW$.

Proposition~\ref{prop:2iso} follows from the availability of these auxiliary relations.
\begin{proofof}{Proposition~\ref{prop:2iso}}
	We explain first how to determine whether a tuple $(a,b,a^*,b^*)$ is in $\spqrIso$.
	
	Let $B, B^*$ be the biconnected components that include $a, b$ and $a^*,b^*$, respectively. The tuple $(a,b,a^*,b^*)$ is in $\spqrIso$ if there are
	\begin{itemize}
		\item triconnected components $C_1, C_2, C^*_1$ and $C^*_2$ that include $a, b, a^*$ and $b^*$, respectively, and
		\item further vertices $c, d$ in $C_1$ and $c^*, d^*$ in $C^*_1$,
	\end{itemize}
	such that for the $3$-connected separating pairs $\{s_1, s_2\}$ and $\{s^*_1, s^*_2\}$ that are the parents of $C_2$ and $C^*_2$ in $\triTree(B)$ and $\triTree(B^*)$ when the roots are $C_1$ and $C^*_1$, respectively, it holds that
	\begin{itemize}
		\item the recoloured contexts $\rcX((a,c,d),(a,c,d),(s_1,s_2,b),\tpl f_1)$ and $\rcX((a^*,c^*,d^*),(a^*,c^*,d^*),\allowbreak(s^*_1,s^*_2,b^*),\tpl f^*_1)$ are fully isomorphic, and
		\item the recoloured subtrees $\rcST((a,c,d),(s_1,s_2,b),\tpl f_2)$ and $\rcST((a^*,c^*,d^*),(s^*_1,s^*_2,b^*),\tpl f^*_2)$ are fully isomorphic,
	\end{itemize}
	where $\tpl f_1, \tpl f^*_1, \tpl f_2, \tpl f^*_2$ describe the actual colours of the corresponding vertices.

The conditions can clearly be expressed in \FO using the relations $\spqrCIso$ and $\spqrSTIso$.
Similarly, the corresponding witness to be stored in the relation $\spqrIsoW$ can be read from the relations $\spqrCIsoW$ and $\spqrSTIsoW$.
\end{proofof}

We now show that the mentioned auxiliary relations can be maintained.

\begin{lemma}\label{lem:2isocontexts}
	For any coloured planar graph $G$, the relations $\spqrCIso, \spqrCIsoW, \allowbreak \spqrSibIsoCount$ and $\spqrSibIsoRank$
can be maintained in \DynFO, provided that $G$ stays planar, under edge insertions and deletions and changes of vertex colours, with the following restrictions.
Each connected component of $G$ receives at most one edge insertion or deletion. Additionally, each biconnected component receives at most a constant number of vertex colour changes.
Beyond that, 
\begin{enumerate}[(i)]
	\item for changes of type $\es[+]{1}{2}$ and $\es[-]{2}{1}$, all vertices of an emerging or vanishing cycle component may change their colour, and
	\item all vertices of a colour may be recoloured with an unused colour.
\end{enumerate}
\end{lemma}

An important building block for updating isomorphic contexts, their corresponding witnesses and the other auxiliary relations is that first-order formulas can express whether two contexts rooted at some components $\tpl r$ and $\tpl r^*$  are fully isomorphic when they have access to (1) isomorphisms between the components $\tpl r$ and $\tpl r^*$, if these are triconnected components, and (2) the auxiliary isomorphism information for contexts rooted at children of $\tpl r$ and $\tpl r^*$.
The first part of the information is available thanks to Proposition~\ref{prop:3iso};  the latter information either because these contexts were not affected by the change and therefore the corresponding parts of the auxiliary relations are still valid, or because it was constructed by a previous application of these first-order formulas.

We prove this observation, first for the case that the roots are triconnected components, afterwards for the case that they are separating pairs.

\begin{claim}\label{clm:subtree-iso-tri-fo}
	Let $X = \rcX(\tpl x, \tpl r, \tpl h, \tpl f)$ and $X^* = \rcX(\tpl x^*, \tpl r^*, \tpl h^*, \tpl f^*)$ be recoloured contexts rooted at the triconnected components $\tpl r$ and $\tpl r^*$, respectively. 
	There are first-order formulas that express whether $X$ and $X^*$ are fully isomorphic and, if they are, identify a witnessing isomorphism, if the following information is available:
	\begin{itemize}
		\item the relation $\ISOt$,
		\item for each child $\tpl c$ of $\tpl r$ and each child $\tpl c^*$ of $\tpl r^*$, whether the subcontext of $X$ rooted at $\tpl c$ is fully isomorphic to the subcontext of $X^*$ rooted at $\tpl c^*$, together with a witness in the positive case.
	\end{itemize}
\end{claim}
\begin{proof}
	The first-order formula checks using $\ISOt$ that $\tpl r$ and $\tpl r^*$ are isomorphic via an isomorphism $\pi$ that maps $\tpl r$ to $\tpl r^*$. In particular, it checks whether $\pi$ respects the colouring of the vertices.
	Then, for every child separating pair $\{s_1, s_2\}$ of $\tpl r$ it checks that $\{\pi(s_1), \pi(s_2)\}$ is a child separating pair of $\tpl r^*$ and that the two subcontexts rooted at these separating pairs are fully isomorphic. Finally, it checks that $\tpl r^*$ has no further child separating pairs.
	The witness consists of $\pi$ together with the given witnesses for the pairs of fully isomorphic subcontexts.
\end{proof}

\begin{claim}\label{clm:subtree-iso-sp-fo}
	Let $X = \rcX(\tpl x, \tpl r, \tpl h, \tpl f)$ and $X^* = \rcX(\tpl x^*, \tpl r^*, \tpl h^*, \tpl f^*)$ be recoloured contexts rooted at the separating pairs $\tpl r$ and $\tpl r^*$, respectively.
	There are first-order formulas that express whether $X$ and $X^*$ are fully isomorphic and, if they are, identify a witnessing isomorphism, if the following information is available:
	\begin{itemize}
		\item for each child $\tpl c$ of $\tpl r$ and each child $\tpl c^*$ of $\tpl r^*$, whether the subcontext of $X$ rooted at $\tpl c$ is fully isomorphic to the subcontext of $X^*$ rooted at $\tpl c^*$, together with a witness in the positive case,
		\item the number of fully isomorphic sibling subtrees of each child of $\tpl r$ and $\tpl r^*$, and
		\item an enumeration for each set of children of $\tpl r$ and $\tpl r^*$ with fully isomorphic subtrees.
	\end{itemize}
\end{claim}
\begin{proof}
	The first-order formula first checks that $\tpl r$ and $\tpl r^*$ are coloured equally. Then it checks for each child triconnected component of $\tpl r$ that there is a child triconnected component of $\tpl r^*$ such that the respective subcontexts rooted at these components are fully isomorphic, that the number of isomorphic siblings in the respective triconnected component trees is the same (where the subtree that contains the hole $\tpl h$ of the context is ignored), and that $\tpl r^*$ has no child triconnected component whose subcontext is not fully isomorphic to the subcontext of a child of~$\tpl r$.
	The witness maps $\tpl r$ to $\tpl r^*$ and the vertices from the subcontext of $X$ that contains the hole $\tpl h$ to the vertices from the subcontext of $X^*$ that contains $\tpl h^*$ using the given witness. For any other child $\tpl c$ of $\tpl r$, the witness maps the vertices corresponding to the subtree of $\tpl c$ to the vertices of the subtree of the child $\tpl c^*$ of $\tpl r^*$ such that these subtrees are fully isomorphic and $\tpl c$ and $\tpl c^*$ are at the same position in the enumeration of their sets of isomorphic sibling subtrees.
\end{proof}

We now prove Lemma~\ref{lem:2isocontexts}. For the maintenance of $\spqrCIso$ and $\spqrCIsoW$, the two claims above will be crucial. For maintaining $\spqrSibIsoCount$ and $\spqrSibIsoRank$, we use that for each root $\tpl r$ of a context, each change only affects at most two of the subtrees rooted at the children of $\tpl r$. So, the numbers of isomorphic siblings and their ranks only change slightly.

\begin{proofof}{Lemma~\ref{lem:2isocontexts}}
	We use the two claims above to express whether two contexts are fully isomorphic after a change. To this end, all types of changes have to be explored. We first assume that only one of the involved contexts is affected by the change and later discuss how to extend the reasoning to the case that both contexts are affected.

\paragraph*{Changes that affect one context.}

Let $X = \rcX(\tpl x, \tpl r, \tpl h, \tpl f)$ and $X^* = \rcX(\tpl x^*, \tpl r^*, \tpl h^*, \tpl f^*)$ be recoloured contexts of the triconnected component forest. Assume that the change only affects $\graph(X)$. %

\subparagraph*{Cases $\es[+]{3}{3}$ and $\es[-]{3}{3}$.}
\begin{figure}[t]
	\centering
	\begin{subfigure}{.45\textwidth}
		\resizebox{1\textwidth}{!}{%
\begin{tikzpicture}[tritreesketch]
	\coordinate (r-pos)  at (-1.6,3.6);
	\coordinate (cc-pos) at (-1.1,2.3);
	\coordinate (h-pos)  at (-0.75,-2.0);

	\draw[bigtree] (-6.1,0.6) -- (r-pos) -- (2.4,0.6)
	-- (0.6,0.6) -- (cc-pos) -- (-3.6,0.6) -- cycle;

	\node[tri, label=right:$C$] (C) at (0,0) {};
	\draw[bigtree] (C.center) -- (1.1,-3) to[rounded corners=0pt] (-0.55,-3)
	to[rounded corners=0pt] (h-pos) to[rounded corners=0pt] (-1.7,-3) -- (-2.9,-3) -- cycle;
	\node[tri] at (C) {};   %
	
	\Pn[marked, label=above:$\overline{r}$]{r}{r-pos}
	\Pn[label={[onhatch]left:$c_1\,c_2$}]{cc}{cc-pos}
	\Pn[marked, label={[onhatch]right:$\overline{h}$}]{h}{h-pos}
	
	\draw[longedge] (r)  -- (cc);
	\draw[edge]     (cc) -- (C);
	\draw[longedge] (C)  -- (h);

	\siblings{Cs}{-3,-0.2}
	\draw[edge] (cc) -- (Csa.north);
	\draw[edge] (cc) -- (Csb.north);
\end{tikzpicture} %
}
		\caption{$C$ lies on the root to hole path.}
		\label{fig:iso2:3t3:isLCA}
	\end{subfigure}\hfill
	\begin{subfigure}{.45\textwidth}
		\resizebox{1\textwidth}{!}{%

\begin{tikzpicture}[tritreesketch]
	\coordinate (r-pos) at (-0.6,1.5);
	\coordinate (a-pos) at (0,0);
	\coordinate (d-pos) at (-4,-3.4);
	\coordinate (b-pos) at (2.5,-3.3);
	\coordinate (h-pos) at (-3.9,-4.6);

	\draw[bigtree] (-6,-2) -- (r-pos) -- (4.2,-2)
	-- (2.1,-2) -- (a-pos) -- (-3,-2) -- cycle;

	\draw[bigtree] (d-pos) -- (-2.5,-5.3) to[rounded corners=0pt] (-3.3,-5.3)
	to[rounded corners=0pt] (h-pos) to[rounded corners=0pt] (-4.5,-5.3) -- (-5.4,-5.3) -- cycle;

	\draw[bigtree] (b-pos) -- (3.9,-5.5) -- (1.1,-5.5) -- cycle;

	\Pn[marked, label=above:$\overline{r}$]{r}{r-pos}
	\Pn[label={[onhatch]left:$\overline{a}$}]{a}{a-pos}
	\node[tri, label=left:$\overline{d}$]  (d) at (d-pos) {};
	\node[tri, label=right:$\overline{b}$] (b) at (b-pos) {};
	\Pn[marked, label={[onhatch]right:$\overline{h}$}]{h}{h-pos}
	\node[tri, label={[onhatch, right = 0.2cm]$C$}] (C) at (2.5,-4.8) {};

	\draw[longedge] (r) -- (a);
	\draw[edge]     (a) -- (d);
	\draw[edge]     (a) -- (b);
	\draw[longedge] (d) -- (h);
	\draw[longedge] (b) -- (C);

	\node[subtree] (Sa) at (-1.1,-3.5) {};
	\node[subtree] (Sb) at (0.6,-3.5) {};
	\draw[edge] (a) -- (Sa.north);
	\draw[edge] (a) -- (Sb.north);
\end{tikzpicture} %
}
		\caption{$C$ does not lie on the root to hole path.}
		\label{fig:iso2:3t3:notLCA-SP}
	\end{subfigure}
	\caption{Update of $\spqrCIso$ for edge changes of type $\es[+]{3}{3}$ or $\es[-]{3}{3}$ inside a triconnected component $C$.}	\label{fig:iso2:3t3}
\end{figure}	
Let $C$ be the changed $3$-connected component in the context~$X$. Let $\{c_1, c_2\}$ be the parent separating pair of $C$ in $X$ and let $c_3$ be another graph vertex in~$C$.
We distinguish two cases, depending on the position of $C$ relative to $\tpl h$, see also Figure~\ref{fig:iso2:3t3} for an illustration.

First suppose that $C$ is on the path from $\tpl h$ to $\tpl r$ in $X$.
The  information in $\spqrSibIsoCount$ and $\spqrSibIsoRank$ for $\{c_1, c_2\}$ and its siblings can be updated by a first-order formula as follows. If $C$ was isomorphic to any other child triconnected component of $\{c_1, c_2\}$ before the change, the counts for these children and the affected ranks can be decremented by one. Because the isomorphism information for children of $C$ is not affected by the change, according to Claim~\ref{clm:subtree-iso-tri-fo}, a first-order formula can check whether any other child of $\{c_1, c_2\}$ is isomorphic to $C$ after the change, increment the counts for the corresponding isomorphic siblings by one, copy the resulting number for $C$, and give $C$ the largest rank. 
Any other update of the relations $\spqrSibIsoCount$ and $\spqrSibIsoRank$ can be done along the same lines as for the children of $\{c_1, c_2\}$, we will omit these descriptions in the following.

To update $\spqrCIso$ for $X$ and $X^*$,  a first-order formula can test whether $X$ and $X^*$ are fully isomorphic after the change by first existentially quantifying vertices $c^*_1, c^*_2, c^*_3 = \tpl c^*$, where $\{c^*_1, c^*_2\}$ is the parent separating pair of the component $C^*$ of $\tpl c^*$ that is the isomorphic copy of $C$, and verifying that the contexts $\rcX(\tpl x, \tpl c, \tpl h, \tpl f_1)$ and $\rcX(\tpl x^*, \tpl c^*, \tpl h^*, \tpl f^*_1)$ are fully isomorphic, which is again possible according to Claim~\ref{clm:subtree-iso-tri-fo}.
Also, the formula can check whether the contexts $\rcX(\tpl x, (c_1,c_2), \tpl h, \tpl f_2)$ and $\rcX(\tpl x^*, (c^*_1,c^*_2), \tpl h^*, \tpl f^*_2)$ are fully isomorphic, following Claim~\ref{clm:subtree-iso-sp-fo}.
It remains to check that also the contexts $\rcX(\tpl x, \tpl r, (c_1,c_2), \tpl f_3)$ and $\rcX(\tpl x^*, \tpl r^*, (c^*_1, c^*_2), \tpl f^*_3)$ with holes $(c_1,c_2)$ and $(c^*_1,c^*_2)$, respectively, are fully isomorphic, which can be read from $\spqrCIso$, as these contexts are unaffected by the change.
Here, $\tpl f_i, \tpl f^*_i$ encode the colours of the corresponding vertices in the recoloured graphs $\graph(X)$ and $\graph(X^*)$.

Now suppose that $C$ is not on the path from $\tpl h$ to $\tpl r$ in $X$. Then let $A$ be the lowest common ancestor of $\tpl h$ and $C$ in the context $X$, which can be identified using the distance information on the tree, and let $\tpl a$ be a tuple of vertices representing it. Analogously to the explanation above for the contexts $\rcX(\tpl x, \tpl r, \tpl h, \tpl f)$ and $\rcX(\tpl x^*, \tpl r^*, \tpl h^*, \tpl f^*)$, a first-order formula can determine whether the subtrees $\rcST(\tpl x, \tpl b, \tpl f_4)$ and $\rcST(\tpl x^*, \tpl b^*, \tpl f^*_4)$ are fully isomorphic, where $\tpl b$ represents the child of $A$ on the path to $C$ and $\tpl b^*$ is an existentially quantified tuple of vertices representing its isomorphic copy.
The context rooted at the child $\tpl d$ of $A$ that includes the hole is not affected by the change, so using either Claim~\ref{clm:subtree-iso-tri-fo} or Claim~\ref{clm:subtree-iso-sp-fo}, one can determine whether $\rcX(\tpl x, \tpl a, \tpl h, \tpl f_5)$ and $\rcX(\tpl x^*, \tpl a^*, \tpl h^*, \tpl f^*_5)$ are fully isomorphic, where $\tpl a^*$ is the quantified tuple that represents the isomorphic copy of $A$.

If also the contexts $\rcX(\tpl x, \tpl r, \tpl a, \tpl f_6)$ and $X(\tpl x^*, \tpl r^*, \tpl a^*, \tpl f^*_6)$ are fully isomorphic, so are $X$ and $X^*$.
Again, $\tpl f_i, \tpl f^*_i$ encode the colours of the corresponding vertices in the respective recoloured graphs.

In both cases, witnesses can be obtained according to Claim~\ref{clm:subtree-iso-tri-fo} and Claim~\ref{clm:subtree-iso-sp-fo}.

\subparagraph*{Case $\es[+]{2}{3}$.}

The procedure is very similar to the previous case $\es[+]{3}{3}$. Let $C$ be the newly-coalesced $3$-connected component. To update the isomorphism information for subcontexts with root $C$ via Claim~\ref{clm:subtree-iso-tri-fo}, we need to have the isomorphism information for all its child separating pairs available.

If a child separating pair of $C$ existed before the change, most of its old isomorphism information is still valid. Its subtree is only affected by the change if one of its former child triconnected components was part of the coalesced path and therefore got absorbed into $C$. All other child triconnected components and their subtrees are unaffected. The isomorphism information for the separating pair can be easily updated with the help of Claim~\ref{clm:subtree-iso-sp-fo}.

New separating pairs that are introduced by the change have a simple structure: they are children of $C$, have a new cycle component as their only child, which in turn has only pre-existing separating pairs with unchanged subtrees as children.
So, via Claim~\ref{clm:subtree-iso-tri-fo}, the isomorphism information for the new cycle components can be computed, followed by the isomorphism information for the new separating pairs via Claim~\ref{clm:subtree-iso-sp-fo}.
The remaining procedure is exactly as described for the case $\es[+]{3}{3}$.

\subparagraph*{Case $\es[-]{3}{2}$.}
\begin{figure}[tp]
	\centering
\begin{tikzpicture}[tritreesketch
	]

	\node[tri, label=above:$D_k$] (Dk) at (-6.4,0) {};
	\Pn[label={[above=.3cm]$P_{\ell+1}$}, label={[font=\small, above = -0.1cm]$a_{\ell+1}\,b_{\ell+1}$}]{Pl1}{-5,0}
	\node[tri, label=above:$D_\ell$] (Dl) at (-3.6,0) {};
	\Pn[label={[above=.3cm]$P_\ell$}, label={[font=\small, above = -0.1cm]$a_\ell\,b_\ell$}]{Pl}{-2.4,0}
	\Pn[label={[above=.3cm]$P_1=K_1$}, label={[font=\small, above = -0.1cm]$a_1\,b_1$}]{P1}{-1.1,0}
	\node[tri, label=right:$C$] (C) at (0,0.7) {};
	\Pn[label=above:$K_2$]{K2}{1.5,0}
	\node[tri] (D) at (4.4,0) {};
	
	\draw[edge]     (C)   -- (P1);
	\draw[longedge] (P1)  --  (Pl);
	\draw[edge]     (Pl)  -- (Dl);
	\draw[edge]     (Dl)  -- (Pl1);
	\draw[longedge] (Pl1) --  (Dk);
	\draw[edge]     (C)   -- (K2);
	\draw[longedge] (K2)  --  (D);
	\node[ann] at (3,0.3) {Unfurled path $\rho$};

	\begin{scope}[yshift=-.4cm]
	\coordinate (r-pos) at (-1.2,5.3);
	\coordinate (cc-pos) at (-0.6,3.8);
	\draw[bigtree] (-6.4,1.6) -- (r-pos) -- (4.4,1.6) -- (2.1,1.6) -- (cc-pos) -- (-4.6,1.6) -- cycle;
		\Pn[marked, label=above:$\overline{r}$]{r}{r-pos}
	\Pn[label={[onhatch]left:$c_1\,c_2$}]{cc}{cc-pos}
	\draw[longedge] (r) -- (cc);
	\draw[edge]     (cc) -- (C);
		\sibling{Cs}{-1.2,2.6}
	\draw[edge] (cc) -- (Csa.north);
	\end{scope}

	\coordinate (hp-pos) at (5.1,-1);
	\coordinate (h-pos) at (4.9,-2.7);
	\draw[bigtree] (hp-pos) -- (6.3,-3.4) to[rounded corners=0pt] (5.4,-3.4) -- (h-pos) to[rounded corners=0pt] (4.4,-3.4) -- (3.9,-3.4) -- cycle;
	\Pn{hp}{hp-pos}
	\Pn[marked, label={[onhatch]right:$\overline{h}$}]{h}{h-pos}
	\draw[edge]     (D)  -- (hp);
	\draw[longedge] (hp) -- (h);

	\siblings{Dks}{-6.4,-2.4}
	\draw[edge] (Dk) -- (Dksa.north);  \draw[edge] (Dk) -- (Dksb.north);
	
	\siblings{Pl1s}{-5,-1.4}
	\draw[edge] (Pl1) -- (Pl1sa.north);  \draw[edge] (Pl1) -- (Pl1sb.north);
	
	\sibling{Dls}{-3.6,-2.4}
	\draw[edge] (Dl) -- (Dlsa.north);  %

	\siblings{Pls}{-2.4,-1.4}
	\draw[edge] (Pl) -- (Plsa.north);  \draw[edge] (Pl) -- (Plsb.north);
	
	\sibling{P1s}{-1.0,-1.4}
	\draw[edge] (P1) -- (P1sa.north);
	
	\siblings{Cks}{0,-2.6}
	\draw[edge] (C) -- (Cksa.north);  \draw[edge] (C) -- (Cksb.north);

	\siblings{K2s}{1.5,-1.4}
	\draw[edge] (K2) -- (K2sa.north);  \draw[edge] (K2) -- (K2sb.north);
	
	\siblings{Ds}{3.2,-2.6}
	\draw[edge] (D) -- (Dsa.north);  \draw[edge] (D) -- (Dsb.north);

	\begin{pgfonlayer}{background}
		\draw[highlight] (Dk.center) -- (P1.center) -- (C.center) -- (K2.center) -- (D.center);
	\end{pgfonlayer}
\end{tikzpicture} 	\caption{Illustration for the update of $\spqrCIso$ for edge changes of type $\es[-]{3}{2}$. Because of the deletion, a $3$-connected component node unfurls into a path $\rho$. Among the nodes of this path, the triconnected component node $C$ is closest to the root. The hatched regions represent subtrees that are not affected by the change.}	
	\label{fig:iso2:3t2}
\end{figure}
See Figure~\ref{fig:iso2:3t2} for an illustration of the case that a $3$-connected component unfurls into a path $\rho$ of alternating separating pairs and triconnected components after an edge deletion. 
Let $C$ be the component of this path that is closest to the root in the context $X$.
At most two children $K_1$ and $K_2$ of $C$ are on the path $\rho$; for the other children, the old isomorphism information is still valid. Our goal is to compute the isomorphism information for $K_1$ and $K_2$, then the information for $C$ can be computed according to Claim~\ref{clm:subtree-iso-tri-fo} or Claim~\ref{clm:subtree-iso-sp-fo}, depending on the type of $C$. Isomorphisms of the full context $X$ can then be computed as detailed in the previous cases.

We explain how the isomorphism information for $K_1$ is computed, the computation for $K_2$ is analogous. We assume that $K_1$ is a separating pair $\tpl k_1 = (a_1, b_1)$; in case $K_1$ is a triconnected component, we apply the following to its child separating pair on the path $\rho$ and apply Claim~\ref{clm:subtree-iso-tri-fo} to get the isomorphism information of $K_1$.
We further assume that the hole $\tpl h$ is not in the subtree of $K_1$, the adaptations for the other case are as discussed for Case $\es[+]{3}{3}$.
So, our goal is to decide whether the subtree $\rcST(\tpl x, \tpl k_1, \tpl f_{K_1})$ is fully isomorphic to some subtree $\rcST(\tpl x^*, \tpl k^*_1, \tpl f_{K_1})$, where $\tpl f_{K_1}$ gives the colours of $(a_1, b_1)$ in the context $X$ and $\tpl k^*_1 = (a^*_1, b^*_1)$ describes some separating pair $K^*_1$ in the context $X^*$.

Let $(P_1 = K_1) D_1 P_2 \cdots P_k D_k$ be the subpath of $\rho$ that is below $K_1$, where the $P_i$ are separating pair nodes and the $D_i$ are triconnected component nodes. Note that the old isomorphism information for the subtrees below these nodes that do not intersect $\rho$ is still valid. However, we cannot iteratively use Claim~\ref{clm:subtree-iso-tri-fo} and Claim~\ref{clm:subtree-iso-sp-fo} to lift this information to the subtree of $K_1$, as the length of the path is not bounded by a constant.
Instead, a first-order formula will test in parallel for each $\ell$ from $1$ to $k-1$ whether the context with root $P_{\ell}$ and hole $P_{\ell+1}$ is fully isomorphic to its counterpart.

Putting this information together requires some care, as for each separating pair $P_\ell = \{a_\ell, b_\ell\}$ and its candidate isomorphic copy $P^*_\ell = \{a^*_\ell, b^*_\ell\}$ there is the possibility to map $(a_\ell, b_\ell)$ to $(a^*_\ell, b^*_\ell)$ or to $(b^*_\ell, a^*_\ell)$.
Note however that an ordering of $P_1$ and $P_k$ as well as of the candidate isomorphic copies $P^*_1$ and $P^*_k$ fixes an ordering of the other separating pairs: if the ordering of $P_1$ is $(a_1, b_1)$ and the ordering of $P_k$ is $(a_k, b_k)$, after inserting an edge $(a_1, a_k)$, every separating pair $P_2, \ldots, P_{k-1}$ has exactly one vertex that is on the same face as $a_1$ and $a_k$, see also Lemma~\ref{lem:spqrDist}, and this vertex can be identified using the auxiliary relation $\sameFace$.
Any isomorphism has to respect this property, which allows a first-order formula to globally fix a mapping between the vertices of the separating pairs and then to check whether this mapping can be extended to an isomorphism of the whole context.

We give more details. The first-order formula that expresses whether $\rcST(\tpl x, \tpl k_1, \tpl f_{K_1})$ and $\rcST(\tpl x^*, \tpl k^*_1, \tpl f_{K_1})$ are fully isomorphic has to existentially quantify the separating pair $P_k= (a_k, b_k)$ and its isomorphic copy $P^*_k = (a^*_k, b^*_k)$. It then has to check that $P_k$ is indeed the separating pair from $\rho$ that is closest to the leaves of the subtree of $K_1$. 
As discussed for Case $\es[+]{3}{3}$ and using Claim~\ref{clm:subtree-iso-sp-fo}, it can check that the subtrees rooted at these separating pairs are fully isomorphic. Also, a witness can be identified.
Then, for all pairs $a_\ell, b_\ell$ of vertices such that $\{a_\ell, b_\ell\}$ is some separating pair $P_\ell$ with $\ell < k$ and $a_\ell$ would be on the same face as $a_1$ if the edge $(a_1,a_k)$ would be inserted, it existentially quantifies vertices $a_{\ell+1}, b_{\ell+1}$ and $a^*_\ell, b^*_\ell, a^*_{\ell+1}, b^*_{\ell+1}$ and checks that 
 \begin{enumerate}[(1)]
 	\item  $\{a_{\ell+1}, b_{\ell+1}\}$ is the next separating pair after $P_\ell$ on the path to $P_k$,
 	\item  $\{a^*_\ell, b^*_\ell\}$ and $\{a^*_{\ell+1}, b^*_{\ell+1}\}$ are the separating pairs on the path from $P^*_1 = K^*_1$ to $P^*_k$ with the same distances as $\{a_{\ell}, b_{\ell}\}$ and $\{a_{\ell+1}, b_{\ell+1}\}$ on the path from $P_1$ to $P_k$,
 	\item $a_{\ell+1}$ would be on the same face as $a_1$ after the insertion of $(a_1,a_k)$ and $a^*_\ell, a^*_{\ell+1}$ would be on the same face as $a^*_1$ after the insertion of the edge $(a^*_1, a^*_k)$,
 	\item the contexts $\rcX(\tpl x, (a_\ell, b_\ell), (a_{\ell+1}, b_{\ell+1}),\tpl f_\ell)$ and $\rcX(\tpl x^*, (a^*_\ell, b^*_\ell), (a^*_{\ell+1}, b^*_{\ell+1}),\tpl f^*_\ell)$ are fully isomorphic after the change is applied, where $\tpl f_\ell$ and $\tpl f^*_\ell$ describe the colours of the corresponding vertices in $X$ respectively $X^*$.
 \end{enumerate}
The points (1)--(3) can be checked using the relations $\sameFace$ and $\spqrDist$, point (4) can be checked according to Claim~\ref{clm:subtree-iso-sp-fo}.
The witness isomorphism maps $(a_\ell, b_\ell)$ to $(a^*_\ell, b^*_\ell)$, for the remaining vertices it is obtained through Claim~\ref{clm:subtree-iso-sp-fo}.

\subparagraph*{Cases $\es[+]{2}{2}$ and $\es[-]{2}{2}$.}
These cases are related to Case $\es[-]{3}{2}$ and Case $\es[+]{2}{3}$, respectively, although they are structurally much simpler. As the affected path in the triconnected component tree has constant length, applying Claim~\ref{clm:subtree-iso-tri-fo} and Claim~\ref{clm:subtree-iso-sp-fo} constantly often is sufficient here.

\subparagraph*{Case $\es[+]{1}{2}$.}
In this case, a new cycle component is formed. It may contain new virtual edges that also need to be added to adjacent biconnected components.
Also, we allow that all colours of the vertices in the cycle component change as well.

Let $C$ be the new cycle component and let $B_1, \ldots, B_k$ be the biconnected components that get attached to this cycle. For each component $B_i$, let $c_{i-1}, c_i$ be the two former cut vertices that it shares with $C$.
If $(c_{i-1}, c_i)$ is not an edge in $B_i$, it is added as a virtual edge, as described in the case from $\es[+]{2}{2}$, $\es[+]{2}{3}$ and $\es[+]{3}{3}$ that is applicable.
If not already present, a separating pair node $\{c_{i-1},c_i\}$ is added, whose only neighbour is the unique triconnected component that includes these two vertices. Using Claim~\ref{clm:subtree-iso-sp-fo}, the isomorphism information can be computed for all contexts that include this separating pair.
Then, the potential colour changes of $c_{i-1}$ and $c_i$ can be processed as discussed in the corresponding case.

All these steps can be done in parallel for the affected biconnected components.

The length of the cycle $C$ can be inferred from the distance between $c_0$ and $c_k$ in the biconnected component tree, so the isomorphism information for $C$ can be computed as well, cf.~Lemma~\ref{lemma:cycle-distances}.

The isomorphism information for any context that is rooted at $C$ can then be expressed, again thanks to Claim~\ref{clm:subtree-iso-tri-fo}. %
The isomorphism information for arbitrary contexts with the new cycle component as inner vertex can then be obtained along the lines of Case~$\es[+]{3}{3}$.

\subparagraph*{Case $\es[-]{2}{1}$.}
As a cycle component dissolves, changes of this type lead to the decomposition of a biconnected component into multiple biconnected components. In each one of the resulting biconnected components, additionally one (virtual) edge and one separating pair node might be removed; the two vertices shared with the dissolved cycle may get a new colour.

The context $X$ of the changed graph was already a valid context before the change. To update its isomorphism information, we only have to process the potential deletion of the virtual edge as discussed in the applicable case from $\es[-]{2}{2}$, $\es[-]{3}{2}$ and $\es[-]{3}{3}$ and the colour changes as discussed in the corresponding case.

\subparagraph*{Changing the colour of a vertex.}

Suppose some vertex $v$ that is contained in $\graph(X)$ changes its colour. Let $\tpl c$ describe the node of the context $X$ that is closest to the root such that $\tpl c$ is either a separating pair that contains $v$ or a triconnected component that includes $v$.

We first suppose that $\tpl c$ is a separating pair. To check whether the recoloured contexts $X = \rcX(\tpl x, \tpl r, \tpl h, \tpl f)$ and $X^* = \rcX(\tpl x^*, \tpl r^*, \tpl h^*, \tpl f^*)$ are fully isomorphic, the first-order formula existentially quantifies a separating pair $\tpl c^*$ and checks using the existing auxiliary relations whether the contexts  $\rcX(\tpl x, \tpl c, \tpl h, \tpl f_1)$ and $ \rcX(\tpl x^*, \tpl c^*, \tpl h^*, \tpl f^*_1)$ as well as $\rcX(\tpl x, \tpl r, \tpl c, \tpl f_2)$ and $ \rcX(\tpl x^*, \tpl r^*, \tpl c^*, \tpl f^*_2)$ are fully isomorphic, where $\tpl f^*_1, \tpl f^*_2$ describe the actual colours of the respective vertices in the recoloured graph $\graph(X^*)$ and $\tpl f_1, \tpl f_2$ use the new colour for the vertex $v$.

If $\tpl c$ is a triconnected component, the update works along the lines of the explanation for the case $\es[+]{3}{3}$, but for the isomorphism check of the component $\tpl c$ the new colour of $v$ is taken into account and if $v$ appears in any child separating pairs, also for the checks of their subcontexts the colouring parameter of the recoloured contexts is adapted.

\subparagraph*{Recolouring all vertices of some colour with an unused colour.}

For ``renaming'' a colour, that is, recolouring all vertices with some colour $\tpl f_1$ using some colour $\tpl f_2$ that no vertex was coloured with before the change, observe that colours are only compared for equality. So, we only have to swap the colours $\tpl f_1$ and $\tpl f_2$ in any tuple of any auxiliary relation.

\paragraph*{Changes that affect both contexts.}

So far, we assumed that only one of the two considered contexts is affected by the change. Of course, this is a simplification, as both contexts might include a changed component. This could be the same component if the contexts are not disjoint; or it could be different components, for example if the two contexts include different parts of an unfurled path.

The cases where a change affects both contexts do not introduce new technical challenges, but expand the case distinction considerably.
We restrict ourselves to the discussion of the case $\es[+]{3}{3}$, the other cases can be generalized using the same approach.

Suppose the recoloured contexts $X = \rcX(\tpl x, \tpl r, \tpl h, \tpl f)$ and $X^* = \rcX(\tpl x^*, \tpl r^*, \tpl h^*, \tpl f^*)$ are fully isomorphic after an edge insertion to a $3$-connected component $C$ that is part of both contexts.
Let $\{c_1, c_2\}$ be the parent separating pair of $C$ in $X$ and let $\{d^*_1, d^*_2\}$ be the parent separating pair of $C$ in $X^*$.
Let $\{c^*_1, c^*_2\}$ respectively $\{d_1, d_2\}$ be their isomorphic copies in the other context.
We differentiate the cases of $\{c_1, c_2\}$ and $\{d_1, d_2\}$ being equal, one separating pair being a predecessor of the other, and the two separating pairs being in different subcontexts of $X$.

If $\{c_1, c_2\} = \{d_1, d_2\}$, an isomorphism can map $C$ to itself, although it does not need to map every vertex of $C$ to itself.
Intuitively, the same part of the two contexts is affected.
The existence of a corresponding isomorphism can be verified as discussed for the case $\es[+]{3}{3}$ above.

Suppose now that one separating pair is a predecessor of the other, without loss of generality we assume that $\{c_1, c_2\}$ is a predecessor of $\{d_1, d_2\}$.
As discussed for the case $\es[+]{3}{3}$ above, an update formula can verify that the subcontexts of $X$ and $X^*$ with roots $\{d_1, d_2\}$ and $\{d^*_1, d^*_2\}$ are fully isomorphic, as only the second subcontext is affected by the change.
Then it can analogously be checked in two steps that the subcontexts with roots $\{c_1, c_2\}$ and $\{c^*_1, c^*_2\}$ as well as the whole contexts $X$ and $X^*$ are fully isomorphic.

For the last remaining case, suppose that $\{c_1, c_2\}$ and $\{d_1, d_2\}$ are in different subcontexts of $X$.
Only one of the subcontexts with roots $\{c_1, c_2\}$ and $\{c^*_1, c^*_2\}$ is affected by the change, this holds also for the subcontexts with roots $\{d_1, d_2\}$ and $\{d^*_1, d^*_2\}$, so it can be checked as explained for case $\es[+]{3}{3}$ that the respective pairs of contexts are fully isomorphic.
Then, using the approach discussed for the case $\es[+]{3}{3}$ above, an update formula checks that also the subcontexts that have the lowest common ancestor of $\{c_1, c_2\}$ and $\{d_1, d_2\}$ and respectively of $\{c^*_1, c^*_2\}$ and $\{d^*_1, d^*_2\}$ as roots are fully isomorphic. Finally, it can then be checked that $X$ and $X^*$ are fully isomorphic.
\end{proofof}

\section{Maintaining Isomorphism Information for Connected Components}\label{section:connected}
\tikzset{
	bctreesketch/.style={
		tritreesketch,
		circ/.style={draw, fill=gray!40, circle, minimum size=7mm, inner sep=0pt},
		dot/.style={draw, fill=black, circle, minimum size=2mm, inner sep=0pt}
	}
}

We now lift the result of the previous section and show that one can maintain in $\DynFO$ whether two connected planar graphs are isomorphic.
So, we show that a relation $\bcIso$ can be maintained that for a planar graph $G$ contains all tuples $(a, a^*)$ such that
\begin{itemize}
	\item $a$ and $a^*$ are vertices of connected components $A$ and $A^*$ of $G$, respectively, and
	\item $A$ and $A^*$ are isomorphic via an isomorphism $\pi$ with $\pi(a) = a^*$.
\end{itemize}
We again also maintain a witness using a relation $\bcIsoW$.
For any $(a, a^*) \in \bcIso$ and the connected components $A$ and $A^*$ that contain the vertices $a$ and $a^*$, and for every vertex $b \in A$, this relation $\bcIsoW$ contains a tuple $(a, b, a^*, b^*)$ for exactly one vertex $b^* \in A^*$.
The function $\pi_{a,a^*}$ that maps each $b \in A$ to the vertex $b^* \in A^*$ with  $(a, b, a^*, b^*) \in \bcIsoW$ then has to be an isomorphism from $A$ to $A^*$ with $\pi_{a,a^*}(a) = a^*$.

\begin{proposition}
	\label{prop:1iso}
	For any planar graph $G$, the relations $\bcIso$ and $\bcIsoW$
	can be maintained in \DynFO under insertions and deletions of edges, provided that $G$ stays planar.
\end{proposition}
	
Our approach to prove this result is similar to the approach from Section~\ref{section:biconnected}: we maintain whether contexts of the biconnected component tree $\biTree(G)$ are isomorphic.
We however face an additional challenge in comparison to the approach for maintaining isomorphism of triconnected component trees from Section~\ref{section:biconnected}.
For maintaining isomorphisms of triconnected component trees, we used the unique isomorphisms of triconnected components (see Proposition~\ref{prop:3iso}) to identify isomorphic copies of separating pairs of those components and checked whether the isomorphism extended to the subtrees of the triconnected component tree below these separating pairs.
For isomorphic biconnected components, however, we only have one witness isomorphism available, out of potentially many alternatives, and just because this witness cannot be extended to an isomorphism that includes the subtrees of the cut vertices of these components, this does not mean that no such isomorphism exists.

We solve this problem by colouring cut vertices in such a way that two cut vertices have the same colour if and only if their subtrees are isomorphic. If an isomorphism between biconnected components exists whose cut vertices are coloured this way (and by Proposition~\ref{prop:2iso}, isomorphism of coloured biconnected components can be maintained), we know that such an isomorphism extends to the subtree below these components in the biconnected component tree. 

When we deal with contexts, we use such a colouring only for those cut vertices that do not have the hole of the context in their subtrees. We then mark the vertices of the hole and the cut vertices on the path from the root to the hole with a constant number of reserved colours.

So, the colour of a cut vertex depends on the choice of the root of the biconnected component tree and on whether the cut vertex is, as part of a context, an ancestor of the hole. Accordingly, we maintain for every possible context and every subtree of the biconnected component tree a coloured copy of the graph that is described by this context respectively subtree.

We make this more precise now. Let $G$ be a planar graph. We again identify tuples of graph vertices with nodes of the biconnected component tree $\biTree(G)$ of $G$. If a tuple $\bar x$ consists of two graph vertices of the same biconnected component, then we identify it with this component.
If the tuple consists of a single vertex which is a cut vertex, then we identify the tuple with this cut vertex.

In the following, we often use the notation $\X(\tpl x, \tpl r, ())$ for a subtree $\ST(\tpl x, \tpl r)$ of a biconnected component tree and treat subtrees as a special case of contexts, which have no hole.

For every context $X$ of $\biTree(G)$ (which includes every subtree), we store a coloured copy $\cgraph(X)$ of $\graph(X)$, including its biconnected and triconnected component trees, and maintain $\spqrIso$ for the resulting forest according to Proposition~\ref{prop:2iso}.

Let $X = \X(\bar x, \bar r, \bar h)$ and $X^* = \X(\bar x^*, \bar r^*, \bar h^*)$ be contexts.
We say that $X$ and $X^*$ are \emph{fully isomorphic} if $\cgraph(X)$ and
$\cgraph(X^*)$ are isomorphic via an isomorphism~$\pi$ with
$\pi(\bar r) = \bar r^*$. 
We already note that if $\bar h \neq ()$, then our colouring enforces that also $\pi$ maps $\bar h$ to
$\bar h^*$. 

Our dynamic algorithm maintains an auxiliary relation $\bcCIso$ that contains all tuples $(X, X^*) \df (\bar x, \bar r, \bar h, \bar x^*, \bar r^*, \bar h^*)$ such that $X$ and $X^*$ are fully isomorphic.
To maintain a witnessing isomorphism, it maintains a relation $\bcCIsoW$ that contains tuples $(X, b, X^*, b^*) \df (\bar x, \bar r, \bar h, b, \bar x^*, \bar r^*,
\bar h^*, b^*)$ that encode a witness isomorphism for every $(X, X^*) \in \bcCIso$.

The goal of this section is to prove that coloured copies of the graph according to the intuition given above and the corresponding relations $\bcCIso$ and $\bcCIsoW$ can be maintained in \DynFO.

\begin{lemma}
	\label{lem:1isocontexts}
	The following can be maintained in \DynFO under insertions and deletions of edges to a planar graph $G$, provided that $G$ stays planar:
	\begin{enumerate}[(1)]
		\item a coloured graph $\cgraph(X)$ for every context $X = \X(\tpl x, \tpl r, \tpl h)$ of $\biTree(G)$ such that
		\begin{enumerate}[(a)]
			\item $\cgraph(X)$ has the same vertices and edges as $\graph(X)$,
			\item exactly the following vertices are coloured: 
			  \begin{itemize}
			  	\item the vertices of $\tpl h$, and
			  	\item the cut vertices that are nodes of $X$, except for the root $\tpl r$ of $X$,
			  \end{itemize}
			\item the vertices of $\tpl h$ get reserved colours $H_1$ and $H_2$ if $\tpl h$ consists of two vertices, and colour $H_1$ if $\tpl h$ consists of one vertex,
			\item cut vertices $v$ different from the root such that $\tpl h$ is in the subtree of $v$ in $X$ get the reserved colour $P$, if they are not already coloured $H_1$ or $H_2$, and
			\item for two contexts $X_1 = \X(\tpl x_1, \tpl r_1, \tpl h_1)$ and $X_2 = \X(\tpl x_2, \tpl r_2, \tpl h_2)$, two cut vertices $v_1$ of $X_1$ and $v_2$ of $X_2$ that are not the roots of $X_1$ and $X_2$, respectively, and do not have $\tpl h_1$ respectively $\tpl h_2$ in their subtrees, have the same colour if and only if the subtrees  $\X(\tpl r_1, v_1, ())$ and  $\X(\tpl r_2, v_2, ())$ of the biconnected component trees are fully isomorphic.
		\end{enumerate}
		\item the relations $\bcCIso$ and $\bcCIsoW$ with respect to these coloured graphs.
	\end{enumerate}
\end{lemma}

Lemma~\ref{lem:1isocontexts} implies Proposition~\ref{prop:1iso}.

\begin{proofof}{Proposition~\ref{prop:1iso}}
	An isomorphism $\pi$ between the connected components $A$ and $A^*$ such that $\pi(a) = a^*$ exists if and only if
	\begin{itemize}
		\item there exist vertices $b, b^*$ that are in the same biconnected component as $a$ and $a^*$, respectively, such that
		\item the tuple $(((a,b), (a,b), ()),((a^*,b^*), (a^*,b^*), ()))$ is contained in the relation $\bcCIso$.
	\end{itemize}
	This condition can clearly be expressed by a first-order formula.
	The witnessing isomorphism can then be read from $\bcCIsoW$.
\end{proofof}

We now prove Lemma~\ref{lem:1isocontexts}.

\begin{proofof}{Lemma~\ref{lem:1isocontexts}}
	Similarly to the proof of Lemma~\ref{lem:2isocontexts}, we maintain additional auxiliary relations $\bcSibIsoCount$ and $\bcSibIsoRank$ that help to recognize whether contexts rooted at cut vertices are isomorphic and to identify a witnessing isomorphism. 
The relation $\bcSibIsoCount$ contains a tuple $(\tpl x, \tpl r, \tpl h, v, a, m)$ if in the biconnected component tree of $\cgraph(\X(\tpl x, \tpl r, \tpl h))$ with root $\tpl r$, the cut vertex $v$ has a child biconnected component that contains the vertex $a$ and there are vertices $b_1, \ldots, b_{m-1}$ from exactly $m -1$ other child biconnected components such that the subtrees $\ST(\tpl r, (v,a))$ and $\ST(\tpl r, (v,b_i))$ of the biconnected component tree of $\cgraph(\X(\tpl x, \tpl r, \tpl h))$ are fully isomorphic for each $i \leq m-1$, that is, their graphs are isomorphic via an isomorphism that maps $(v,a)$ to $(v,b_i)$.
The relation $\bcSibIsoRank$ gives an enumeration of these isomorphic siblings: for each $(\tpl x, \tpl r, \tpl h, v, a, m) \in \bcSibIsoCount$ it contains exactly one tuple $(\tpl x, \tpl r, \tpl h, v, a, k)$, for some $k \leq m$, which assigns the rank $k$ to the subtree rooted at $(v,a)$ that is unique among the $m$ fully isomorphic siblings.

	We will use the following claim, which is analogous to Claim~\ref{clm:subtree-iso-sp-fo}.

	\begin{claim}
		\label{clm:subtree-iso-cut-fo}
		Let $X=\X(\tpl x,c,\tpl h)$ and $X^* = \X(\tpl x^*,c^*, \tpl h^*)$ be contexts of biconnected component trees that are rooted at cut vertices $c$ and $c^*$, respectively. 
		There are first-order formulas that express whether $X$ and $X^*$ are fully isomorphic and, if they are, identify a witnessing isomorphism, if the following information is available:
		\begin{itemize}
			\item whether $\X(\tpl x, (c,d), ())$ and $\X(\tpl x^*, (c^*,d^*), ())$ are fully isomorphic, for every $d, d^*$ contained in a child biconnected component of $c,c^*$, respectively, that do not have $\tpl h, \tpl h^*$ in their subtree, together with a witness in the positive case,
			\item in case $\tpl h$ and $\tpl h^*$ are not $()$: whether $\X(\tpl x, (c,d), \tpl h)$ and $\X(\tpl x^*, (c^*,d^*), \tpl h^*)$ are fully isomorphic, where $(c,d)$ and $(c^*,d^*)$ are the children of $c$ and $c^*$ on the path to $\tpl h$ and $\tpl h^*$, respectively, together with a witness in the positive case,
			\item the number of fully isomorphic sibling subtrees of each child of $c$ and $c^*$ in the biconnected component trees of $\cgraph(X)$ and $\cgraph(X^*)$, respectively, and
			\item an enumeration for each set of children of $c$ and $c^*$ with fully isomorphic subtrees.
		\end{itemize}
	\end{claim}
	\begin{proof}
		The required first-order formula checks for each child biconnected component of $c$ whether there is a child biconnected component of $c^*$ such that the respective subtrees rooted at these components are fully isomorphic.
		It also checks that the number of isomorphic siblings in the respective biconnected component trees is the same.
		Then the formula checks that $c^*$ has no child biconnected components whose subtrees are not fully isomorphic to a subtree that is rooted at some child of $c$.
		
		For the witnessing isomorphism, a vertex from a subtree of a child $(c,d)$ is mapped to a vertex from that subtree of $(c^*, d^*)$ such that these subtrees are on the same position in the enumeration of the set of siblings with fully isomorphic subtrees. The exact mapping is according to the witness given for this pair of contexts.
	\end{proof}
	
	We now discuss the maintenance of the auxiliary relations and distinguish the different types of changes as presented in Section~\ref{section:changes}.
	
	For the graph colouring, we maintain the invariant that colours which are different from the reserved colours $H_1, H_2, P$ are represented by tuples of the form $(\tpl p, v)$, where $v \neq \tpl p$ is a cut vertex in the subtree $\ST(\tpl p, \tpl p) =  \X(\tpl p, \tpl p, ())$, and if some cut vertex has the colour $(\tpl p, v)$, then so does $v$ in $\cgraph(\X(\tpl p, \tpl p, ()))$.
	That means that colours are descriptors of a cut vertex with the same colour, so they identify representatives among the colour classes.
	
	Let $X = \X(\tpl x, \tpl r, \tpl h)$ and $X^* = \X(\tpl x^*, \tpl r^*, \tpl h^*)$ be contexts of biconnected component trees.
	We assume that the change only affects $\graph(X)$. The generalization to the case that also $\graph(X^*)$ is affected works analogously as discussed for the proof of Lemma~\ref{lem:2isocontexts}.
	
	\subparagraph*{Cases $\es[+]{3}{3},\es[-]{3}{3}, \es[+]{2}{3},\es[-]{3}{2},\es[+]{2}{2}$ and $\es[-]{2}{2}$.}
	
	In these cases, the change only affects one biconnected component, in particular, no biconnected component tree of any connected component changes. Consequently, there are no changes of the vertex set of any graph $\cgraph(X)$.
	We explain the colour changes of these graphs and the maintenance of $\bcCIso$, $\bcCIsoW$, $\bcSibIsoCount$ and $\bcSibIsoRank$.
	
	Let $B$ be the changed biconnected component in the graph $\cgraph(X)$ and let $b_1$ be its parent cut vertex.
	Let us first assume that $\tpl h$ lies in the subtree of $b_1$ in the context $X$ and let $b_2$ be the cut vertex from $B$ that is on the path to $\tpl h$.
	See Figure~\ref{fig:bc-iso-op} for an illustration.
	
	\begin{figure}[t]
		\centering

\begin{tikzpicture}[bctreesketch]
	\coordinate (r-pos)  at (-1.6,3.6);
	\coordinate (b1-pos) at (-1.1,2.3);
	\coordinate (B-pos)  at (0,0);
	\coordinate (b2-pos) at (0,-1.0);
	\coordinate (h-pos)  at (-0.75,-2.6);

	\draw[bigtree] (-6.1,0.6) -- (r-pos) -- (2.4,0.6)
	-- (0.6,0.6) -- (b1-pos) -- (-3.6,0.6) -- cycle;

	\draw[bigtree] (b2-pos) -- (1.1,-3.5) to[rounded corners=0pt] (-0.55,-3.5)
	to[rounded corners=0pt] (h-pos) to[rounded corners=0pt] (-1.7,-3.5) -- (-2.9,-3.5) -- cycle;

	\node[dot, marked, minimum size=2.5mm, label=above:$\overline{r}$] (r) at (r-pos) {};
	\node[dot, label={[onhatch, left=0.2cm]$b_1$}] (b1) at (b1-pos) {};
	\node[circ, label=right:$B$] (B) at (B-pos) {};
	\node[dot, label={[onhatch]right:$b_2$}] (b2) at (b2-pos) {};
	\node[dot, marked, minimum size=2.5mm, label={[onhatch]right:$\overline{h}$}] (h) at (h-pos) {};

	\draw[longedge] (r)  -- (b1);
	\draw[edge]     (b1) -- (B);
	\draw[edge]     (B)  -- (b2);
	\draw[longedge] (b2) -- (h);

	\node[subtree] (Sa) at (-3.1,-0.2) {};
	\node[subtree] (Sb) at (-1.9,-0.2) {};
	\draw[edge] (b1) -- (Sa.north);
	\draw[edge] (b1) -- (Sb.north);
\end{tikzpicture} 		\caption{Illustration for edge changes of type $\es[+/-]{3}{3},\es[+]{2}{3},\es[-]{3}{2}$ or $\es[+/-]{2}{2}$ in a connected component such that $\tpl h$ is in the subtree of the affected biconnected component $B$.}
		\label{fig:bc-iso-op}
	\end{figure}
	
	We first consider the subcontext $Y = \X(\tpl r, (b_1, b_2), \tpl h)$ and identify contexts $Y^* = \X(\tpl r^*, (b^*_1, b^*_2), \tpl h^*)$ that include a biconnected component $B^*$ with a parent cut vertex $b^*_1$ and a child cut vertex $b^*_2$ such that $Y$ and $Y^*$ are fully isomorphic.
	Note that $b_1$ is not a cut vertex in $Y$, and that for any other cut vertex of $B$, the subtree rooted at this cut vertex did not change, so, the colours of these cut vertices in $\cgraph(Y)$ (apart from $b_2$, which is coloured with the reserved colour $P$) still represent the isomorphism class of the subtrees. The same is true for the cut vertices in $B^*$.
	
	It follows that after the edge change, $Y$ and $Y^*$ are fully isomorphic if and only if 
	\begin{enumerate}[(1)]
		\item the contexts $\X(\tpl r, b_2, \tpl h)$ and $\X(\tpl r^*, b^*_2, \tpl h^*)$ below $b_2$ and $b^*_2$ are fully isomorphic and
		\item the biconnected components $B$ and $B^*$ are isomorphic via an isomorphism that maps $(b_1, b_2)$ to $(b^*_1, b^*_2)$.
	\end{enumerate}
	The first part can be read from $\bcCIso$, the second part holds if $((b_1,b_2),(b^*_1,b^*_2))\in \spqrIso$ after the change.
	
	As $\spqrIso$ can be maintained in \DynFO under edge changes, see Proposition~\ref{prop:2iso}, this condition can be expressed in \FO. From the maintained relation $\spqrIsoW$ we obtain a witnessing isomorphism from $B$ to $B^*$, which can be extended to a witness for $\cgraph(Y)$ and $\cgraph(Y^*)$ using the witnesses for the subtrees of cut vertices of $B$ and their isomorphic copies in $B^*$ that are already stored in $\bcCIsoW$.
	Also the $\bcSibIsoCount$ and the $\bcSibIsoRank$ information for $b_1$ can be updated, as only one child of $b_1$ is affected by the change. 
	Then one can express using a first-order formula whether $\X(\tpl r, b_1, \tpl h)$ and $\X(\tpl r^*, b^*_1, \tpl h^*)$ are fully isomorphic, using Claim~\ref{clm:subtree-iso-cut-fo}, and obtain a witness.
	
	Note that all cut vertices on the path from $\tpl r$ to $b_1$ in $X$, apart from possibly $\tpl r$ itself, are coloured $P$, as they are ancestors of $\tpl h$, so no colour changes are necessary.
	Also, as for every one of these cut vertices the subtree affected by the change is the only subtree that includes vertices with the colour~$P$, no updates of the $\bcSibIsoCount$ and the $\bcSibIsoRank$ information are required.

	Now, we can check whether the contexts $X$ and $X^*$ are fully isomorphic. This is the case if and only if for some descendant cut vertex $b^*_1$ of $\tpl r^*$ it holds that
	\begin{enumerate}[(a)]
		\item $\X(\tpl r, b_1, \tpl h)$ and $\X(\tpl r^*, b^*_1, \tpl h^*)$ are fully isomorphic, and
		\item $\X(\tpl x, \tpl r, b_1)$ and $\X(\tpl x^*, \tpl r^*, b^*_1)$ are fully isomorphic.
	\end{enumerate}
	As for the latter check the corresponding coloured graphs are not affected by the change, the information can be read from $\bcCIso$; a witness is already stored in $\bcCIsoW$ and is combined with the witness for the subtrees of $b_1$ and $b^*_1$. 

	\smallskip
	We now turn to the case that $\tpl h$ is not in the subtree of $b_1$.
	In this case, we proceed similarly as explained in the corresponding case in the proof of Lemma~\ref{lem:2isocontexts} and identify the lowest common ancestor of $\tpl h$ and $b_1$, compute the isomorphism information first for the children of this ancestor and then for the ancestor itself, then we can finally decide whether the contexts in question are fully isomorphic.
	
	Suppose the lowest common ancestor of $\tpl h$ and $b_1$ in $X$ is a cut vertex $a_1$. Let the first cut vertex on the path from $a_1$ to $b_1$ be $a_2$.
	All cut vertices on the path from $a_2$ to $b_1$ need to be recoloured in $\cgraph(X)$: they are not on the path from $\tpl r$ to $\tpl h$ and thus are not coloured with the reserved colour $P$, and their subtrees are affected by the change. But note that these are the only cut vertices in $\cgraph(X)$ for which the colour needs to be adapted.
	Let $c$ be any of these vertices.
	
	Remember that colours have the form $(\tpl p, v)$ and a vertex with this colour has a subtree that is fully isomorphic to the subtree of $v$ in the graph $\cgraph(\X(\tpl p, \tpl p, ()))$. We say that $v$ is the representative of the isomorphism class of its subtree.
	
	In a first step, we check whether, before the change, $c$ represented the isomorphism class of its subtree. In that case, a new representative has to be found.
	Afterwards, the representative for the isomorphism class after the change is identified.
	
	We make this precise now.
	Suppose that $c$ represents the isomorphism class of its subtree before the change, that is, the colour of $c$ in $\cgraph(X)$ is $F^\star = (\tpl p^\star, c)$ for some tuple $\tpl p^\star$.
	If there are cut vertices with the same colour but whose subtree is not affected by the change, we have to identify a new representative for these vertices. We select the lexicographically smallest tuple $F^\bullet = (\tpl p^\bullet, v^\bullet)$ such that $v^\bullet$ has colour $F^\star$ in $\cgraph(\X(\tpl p^\bullet, \tpl p^\bullet, ()))$ and $B$ is not in the subtree of $v^\bullet$.
	All cut vertices with colour $F^\star$ are recoloured with the colour $F^\bullet$; remember that we can maintain $\spqrIso$ under this kind of changes.
	
	To find the new colour of $c$, we have to check whether its subtree after the change is isomorphic to the subtree of another cut vertex that is not affected by the change. The colours of the cut vertices on the path from $c$ to $B$ are updated in the same step, so we cannot use them here.
	Observe that in a context with root $c$ and whose hole is a leaf below $B$, all these cut vertices lie on the path from the root to the hole and are coloured with the reserved colour $P$.
	Observe further that the (uncoloured) graph of the subtree $\X(\tpl x, c, ())$ below $c$ is isomorphic to the graph of the subtree $\X(\tpl x^\circ, c^\circ, ())$ of some other cut vertex $c^\circ$ via an isomorphism that maps $c$ to $c^\circ$ if and only if there are leaves $\tpl h_c$ and $\tpl h^\circ_c$ in the respective subtrees such that $\X(\tpl x, c, \tpl h_c)$ and $\X(\tpl x^\circ, c^\circ, \tpl h^\circ_c)$ are fully isomorphic.
	
	We combine these ideas and check whether for some leaf $\tpl h_c$ below $B$ there is a context $X_c^\circ = \X(\tpl x^\circ, c^\circ, ())$ that does not include $B$ and that has a leaf $\tpl h^\circ_c$ such that $\X(\tpl x, c, \tpl h_c)$ and $\X(\tpl x^\circ, c^\circ, \tpl h_c^\circ)$ are fully isomorphic. This check can be done analogously as explained above for the case that $\tpl h$ is in the subtree of $B$.	
	
	If such a context $X_c^\circ$ is found, the colour of $c^\circ$ becomes the new colour of $c$. Otherwise, it gets the colour $(\tpl p^\diamond, v^\diamond)$, where $(\tpl p^\diamond, v^\diamond)$ is the lexicographically smallest tuple such that $\X(\tpl x, c, \tpl h_c)$ and $\X(\tpl p^\diamond, v^\diamond, \tpl h^\diamond)$ are fully isomorphic after the change, for some leaves $\tpl h_c$ and $\tpl h^\diamond$ in their respective contexts, which again can be checked as explained above.
	
	We also have to update the $\bcSibIsoCount$ and the $\bcSibIsoRank$ information for any cut vertex $c$ on the path from $a_1$ to $b_1$.
	To check whether the subtree that is rooted at a child of such a vertex $c$ and includes the changed component $B$ is fully isomorphic to one of its siblings, one only has to check that after adjusting the colours, the biconnected components at the root of these subtrees are isomorphic via an isomorphism that maps $c$ to itself, which can be read from the updated relation $\spqrIso$.
	As this is the only child of $c$ that is affected by the change, the update of $\bcSibIsoCount$ and $\bcSibIsoRank$ is easy.
	
	To determine now whether the contexts $X = \X(\tpl x, \tpl r, \tpl h)$ and $X^*= \X(\tpl x^*, \tpl r^*, \tpl h^*)$ are fully isomorphic, an update formula has to existentially quantify a cut vertex $a^*_1$ of $X^*$ and check that
	\begin{enumerate}[(1)]
		\item the contexts $\X(\tpl x, \tpl r, a_1)$ and $\X(\tpl x^*, \tpl r^*, a^*_1)$ are fully isomorphic, which can be read from $\bcCIso$, as these contexts are not affected by the change, and
		\item $\X(\tpl x, a_1, \tpl h)$ and $\X(\tpl x^*,  a^*_1, \tpl h^*)$ are fully isomorphic, which can be done thanks to Claim~\ref{clm:subtree-iso-cut-fo} and using the auxiliary information that is still valid respectively updated as explained above.
	\end{enumerate}
	A witnessing isomorphism is obtained by combining the witnesses that are available for the parts (1) and (2) of that check.
	
	In the case that the lowest common ancestor of $\tpl h$ and $b_1$ is a biconnected component $(a_1, a_2)$, where $a_1$ is the parent cut vertex and $a_2$ is the child cut vertex that is the ancestor of $b_1$, we do as explained above for the subtree of $a_2$ and in particular find the new colour of this cut vertex.
	Then we can check whether there are cut vertices $a^*_1$ and $a^*_2$ of $X^*$ that occur in a common biconnected component such that the following holds:
	\begin{enumerate}[(1)]
		\item the contexts $\X(\tpl x, \tpl r, (a_1,a_2))$ and $\X(\tpl x^*, \tpl r^*, (a^*_1,a^*_2))$ are fully isomorphic, which again can be read from $\bcCIso$, and
		\item $\X(\tpl x, (a_1,a_2), \tpl h)$ and $\X(\tpl x^*,  (a^*_1, a^*_2), \tpl h^*)$ are fully isomorphic, which is equivalent to checking whether $((a_1,a_2),(a^*_1, a^*_2)) \in \spqrIso$ after the change.
	\end{enumerate}

	\subparagraph*{Case $\es[+]{1}{2}$.}
	
	Suppose that after the insertion of an edge, a path of the biconnected component tree of $G$ coalesces into a single biconnected component $B$ that is included in the context $X$.
	Let $b_1$ be the parent cut vertex of $B$.
	Remember that all cut vertices that lay on the coalesced path become part of $B$. They are still cut vertices if in the biconnected component tree before the change they had a neighbour that is not on the coalesced path.
	
	We first consider the subtree of $B$ in $X$ and aim to identify biconnected components of $X^*$ such that the subtrees are fully isomorphic.
	We assume that $\tpl h$ is in the subtree of $B$ in $X$. Otherwise, the following approach can be adapted as discussed in the previous case.
	
	For any child cut vertex of $B$, its subtree may be affected by the change: one of its former child biconnected components may be part of the coalesced path and then becomes part of $B$; all other child biconnected components are unaffected and still present.
	That means that the isomorphism information for these child cut vertices can easily be updated, and using Claim~\ref{clm:subtree-iso-cut-fo} and the explanations given in the case above, their new colour as well as witnesses for isomorphic subtrees can be found.
	All former cut vertices that are no longer cut vertices in $B$ lose their colour.
	
	As $\spqrIso$ for coloured graphs can not only be maintained under $\es[+]{1}{2}$ changes but additionally all vertices on the newly-formed cycle may change their colour, see Proposition~\ref{prop:2iso}, a first-order formula can express whether the component $B$ with adjusted colours is isomorphic to some other component $B^*$ inside $X^*$, which is equivalent to $\X(\tpl x, B, \tpl h)$ and $\X(\tpl x^*, B^*, \tpl h^*)$ being fully isomorphic. Also, a witness for this isomorphism can be obtained.
	As a next step, again using Claim~\ref{clm:subtree-iso-cut-fo}, a first-order formula can express whether there is a cut vertex $b^*_1$ such that $\X(\tpl x, b_1, \tpl h)$ and $\X(\tpl x^*, b_1^*, \tpl h^*)$ are fully isomorphic, as well as identifying a witness.
	
	It then remains to check whether $\X(\tpl x, \tpl r, b_1)$ and $\X(\tpl x^*, \tpl r^*, b^*_1)$ are fully isomorphic, which can be looked up from $\bcCIso$, as the corresponding graphs are not affected by the change. Also, a witness can be read from $\bcCIsoW$, which then needs to be combined with the already computed witness for the isomorphism of the subtrees of $b_1$ and $b^*_1$.
	
	\subparagraph*{Case $\es[-]{2}{1}$.}
	This is the reverse of the $\es[+]{1}{2}$ case. It is similar to the case $\es[-]{3}{2}$ from the proof of Lemma~\ref{lem:2isocontexts}, but simpler. See Figure~\ref{fig:21-x-iso-1} for an illustration.
	
	Suppose that after the deletion of an edge, a biconnected component unfurls into a path $\rho$ of alternating cut vertices and biconnected components. Let $N$ be the node of this path that is closest to the root of the context $X$.
	
	We first assume that $N$ is a biconnected component node. At most two of its child cut vertices are on the path $\rho$. The subtree of all other child cut vertices did not change, so their colour still represents the isomorphism class of their subtrees.
	We explain how to decide whether the subtree of a child cut vertex $k$ of $N$ that lies on $\rho$ is fully isomorphic to some other subtree that is not affected by the change. We only consider the case that $\tpl h$ is in the subtree of $k$, which means that the colour of $k$ is the reserved colour $P$. The other case is an easy adaptation and the new colour can be determined as discussed above.
	So, we have to identify contexts $Y^* = X(\tpl x^*, k^*, \tpl h^*)$ that are fully isomorphic to $Y = \X(\tpl x, k, \tpl h)$.
	
	In the biconnected component tree of $Y$, a new child of $k$ is added because of the change, whose subtree consists of a part of the path $\rho$ and the corresponding, pre-existing subtrees. Let $\tpl b$ denote the bottom-most node of $\rho$ in the subtree of $k$.
	We will only consider the case that $\tpl h$ is in the subtree of $\tpl b$, the other case can be handled along the lines of the discussion for the case $\es[+]{3}{3}$.
	\begin{figure}[t]
		\begin{subfigure}{0.45\textwidth}
			\resizebox{\textwidth}{!}{%

\begin{tikzpicture}[bctreesketch,
	comp/.style={circle, draw, fill=gray!40, fill opacity=0.5, minimum size=18mm, inner sep=0pt},
	vdot/.style={dot, minimum size=1.8mm},
	eedge/.style={red, line width=1.5pt},
	lbl/.style={}
	]
	\draw[fill=gray!12, rounded corners=4mm] (-2.1,-1.9) rectangle (10.8,3.1);
	\node[font=\large, anchor=north west] at (-2,3) {$B$};

	\node[comp] (N)  at (0,0)   {};
	\node[comp] (Bd) at (4.3,0) {};
	\node[comp] (Bl) at (8.7,0) {};

	\draw[virtual] (N.west)  .. controls +(-0.1,1.8) and +(0.1,1.8) .. (N.east);
	\draw[virtual] (Bd.west) .. controls +(-0.1,1.8) and +(0.1,1.8) .. (Bd.east);
	\draw[virtual] (Bl.west) .. controls +(-0.1,1.8) and +(0.1,1.8) .. (Bl.east);

	\draw[eedge] (N.west) .. controls +(-1.5,3.4) and +(1.5,3.4) .. node[above] {$e$} (Bl.east);

\node[vdot, label=left:$u$, fill=white]              (u)   at (N.west)  {};
	\node[vdot, label={[inner sep=1pt]above right:$c_p$}] (cp) at (N.south west) {};
\node[vdot, label=below right:$k$]       (k)   at (N.east)  {};
\node[vdot, label=below left:$c_{d-1}$]  (cd-) at (Bd.west) {};
\node[vdot, label=below right:$c_{d+1}$] (cd+) at (Bd.east) {};
\node[vdot, label=below left:$c_\ell$]   (cl)  at (Bl.west) {};
\node[vdot, fill=white, label=right:$v$]             (v)   at (Bl.east) {};

\node[dot, marked, minimum size=2.5mm, label=above:$\overline{r}$] (r) at (0,4) {};
\draw[longedge] (r) to[out=235, in=90] (-1.5,1.2) to[out=270, in=200] (cp);

\draw[longedge] (k)   -- (cd-);
\draw[longedge] (cd+) -- (cl);

	\node[lbl] at ($(N.south)+(0,-0.3)$)  {$N+(u,k)$};
	\node[lbl] at ($(Bd.south)+(0,-0.3)$) {$B_d+(c_{d-1},c_{d+1})$};
	\node[lbl] at ($(Bl.south)+(0,-0.3)$) {$B_\ell+(c_\ell,v)$};
\end{tikzpicture} %
}
			\caption{The biconnected component $B$ before the deletion. \label{fig:21-x-iso-1B}}
		\end{subfigure}
		\hfill
		\begin{subfigure}{0.50\textwidth}
			\resizebox{\textwidth}{!}{%

\begin{tikzpicture}[bctreesketch]
	\coordinate (r-pos)  at (0,6);
	\coordinate (cp-pos) at (0.9,5);
	\coordinate (Bl-pos) at (5.4,0.8);
	\coordinate (h-pos)  at (6.0,-1.0);

	\draw[bigtree] (-5.6,3.3) -- (r-pos) -- (5.6,3.3)
	-- (3.4,3.3) to[rounded corners=0pt] (cp-pos) to[rounded corners=0pt] (-2.6,3.3) -- cycle;
	\node[onhatch] at (-2.4,4.1) {$\X(\bar{x},\bar{r},c_p)$};

	\draw[bigtree] (Bl-pos) -- (7.5,-2.0) -- (6.3,-2.0)
	to[rounded corners=0pt] (h-pos) to[rounded corners=0pt] (5.7,-2.0) -- (4.9,-2.0) -- cycle;
	\node at (6.2,-2.4) {$\X(\bar{x},\bar{b},\bar{h})$};

	\node[dot, marked, minimum size=2.5mm, label=above:$\bar{r}$] (r) at (r-pos) {};
	\node[dot, label={[onhatch]left:$c_p$}] (cp) at (cp-pos) {};
	
	\node[circ, label=left:$N$] (N) at (-3.2,1.9) {};
	\node[dot,  label=above:$k$]         (k)   at (-2.2,0.8) {};
	\node[dot,  label=above:$c_{d-1}$]   (cd-) at (-0.2,0.8) {};
	\node[circ, label=above:$B_d$]       (Bd)  at (1.4,0.8)  {};
	\node[dot,  label=above:$c_{d+1}$]   (cd+) at (3.0,0.8)  {};
	\node[circ, label=right:$B_\ell$, label=above:$\tpl b$]    (Bl)  at (Bl-pos)   {};
	\node[dot, marked, minimum size=2.5mm, label={[onhatch]right:$\bar{h}$}] (h) at (h-pos) {};

	\draw[longedge] (r)   -- (cp);
	\draw[edge]     (cp)  -- (N);
	\draw[edge]     (N)   -- (k);
	\draw[longedge] (k)   -- (cd-);
	\draw[edge]     (cd-) -- (Bd);
	\draw[edge]     (Bd)  -- (cd+);
	\draw[longedge] (cd+) -- (Bl);
	\draw[longedge] (Bl)  -- (h);
	\node[ann] at (-0.7,1.85) {Unfurled path $\rho$ for $B-\{e\}$};

	\siblings{Ns}{1.3,3.5}          %
	\draw[edge] (cp) -- (Nsa.north);  \draw[edge] (cp) -- (Nsb.north);

	\siblings{ks}{-4.6,-1.6}        %
	\draw[edge] (N) -- (ksa.north);  \draw[edge] (N) -- (ksb.north);

	\siblings{kc}{-2.7,-0.6}        %
	\draw[edge] (k) -- (kca.north);  \draw[edge] (k) -- (kcb.north);
	
	\siblings{cdc}{-0.5,-0.6}       %
	\draw[edge] (cd-) -- (cdca.north);  \draw[edge] (cd-) -- (cdcb.north);
	
	\siblings{cds}{1.4,-2.3}        %
	\draw[edge] (Bd) -- (cdsa.north);  \draw[edge] (Bd) -- (cdsb.north);

	\siblings{Bds}{3.3,-0.6}        %
	\draw[edge] (cd+) -- (Bdsa.north);  \draw[edge] (cd+) -- (Bdsb.north);
	\begin{pgfonlayer}{background}
		\draw[highlight] (N.center) -- (k.center) -- (cd-.center) -- (Bd.center)
		-- (cd+.center) -- (Bl.center);
	\end{pgfonlayer}
\end{tikzpicture} %
}
			\caption{The biconnected component tree after the deletion. \label{fig:21-x-iso-1A}}
		\end{subfigure}	
		\caption{Illustration for the update of $\bcCIso$ after a $\es[-]{2}{1}$ edge change.\label{fig:21-x-iso-1}}
	\end{figure}

	To decide whether $Y^*$ is fully isomorphic to $Y$, a first-order formula now has to
	\begin{enumerate}[(1)]
		\item quantify the isomorphic copy $\tpl b^*$ of $\tpl b$ in the subtree of $k^*$,
		\item check that $\X(\tpl x, \tpl b, \tpl h)$ and $\X(\tpl x^*, \tpl b^*, \tpl h^*)$ are fully isomorphic,
		\item check that for every number $d$ smaller than the distance between $k$ and $\tpl b$ in $Y$,
		\begin{enumerate}
			\item if the node with distance $d$ from $k$ on the path to $\tpl b$ is a biconnected component $B_d$ and its cut vertices on $\rho$ are $c_{d-1}$ and $c_{d+1}$, then for the biconnected component $B^*_d$ with cut vertices $c^*_{d-1}$ and $c^*_{d+1}$ that have distance $d-1$ and $d+1$, respectively, from $k^*$ to $\tpl b^*$ it holds that $B_d$ is isomorphic to $B^*_d$ via an isomorphism that maps $(c_{d-1}, c_{d+1})$ to $(c^*_{d-1}, c^*_{d+1})$ when the colours of these cut vertices are ignored,
			\item if the node with distance $d$ from $k$ on the path to $\tpl b$ is a cut vertex $c_d$, then for the cut vertex $c^*_d$ with distance $d$ from $k^*$ to $\tpl b^*$ it holds that the subtrees of $c_d$ and $c^*_d$ are isomorphic when the child biconnected component $B_{d+1}$ of $c_d$ on the path to $\tpl b$ and the child biconnected component $B^*_{d+1}$ of $c^*_d$ on the path to $\tpl b^*$ are being ignored.
		\end{enumerate}
	\end{enumerate}
	
	A first-order update formula can express these conditions as follows.
	For $(2)$, observe that the only effect of the change on the subtree of $\tpl b$ is the potential deletion of a virtual edge in the biconnected component $\tpl b$ itself. This edge deletion is of type $\es[-]{3}{3}$, $\es[-]{3}{2}$ or $\es[-]{2}{2}$. So, the information can be obtained as discussed above.
	For $(3)$, the necessary distance information is maintained according to Lemma~\ref{lemma:bc-distances}. For part (a), isomorphism of biconnected components of coloured graphs can be maintained also under additional colour changes of a constant number of vertices, see Proposition~\ref{prop:2iso}. 
	For part (b), we slightly adapt the proof of Claim~\ref{clm:subtree-iso-cut-fo} to disregard the mentioned children. For this, the child biconnected component $B^*_{d+1}$ of $c^*_d$ is not counted for the number and the rank of isomorphic siblings; for $c_d$ we can use the $\bcSibIsoCount$ and $\bcSibIsoRank$ information from before the change, as adding the child biconnected component $B_{d+1}$ is precisely the effect of the change on $c_d$.
	
	If all checks succeed, a witnessing isomorphism can be stitched together from the existing partial witnesses, starting from mapping $c_d$ to $c^*_d$, for all cut vertices $c_d$ on the path from $k$ to $\tpl b$.
	Also, the $\bcSibIsoCount$ and $\bcSibIsoRank$ information for $k$ can be updated. %
	
	To finally determine whether $X$ and $X^*$ are fully isomorphic, the update formula now has to
	\begin{enumerate}[(1)]
		\item existentially quantify the isomorphic copy $N^*$ of $N$,
		\item check that after adjusting the colour of the at most two cut vertices of $N$ that are on $\rho$, these biconnected components are isomorphic via an isomorphism that maps the parent cut vertex $c_p$ of $N$ to the parent cut vertex $c^*_p$ of $N^*$ and some child cut vertex $k$ of $N$ to some child cut vertex $k^*$ of $N^*$, and
		\item check that $\X(\tpl x, \tpl r, (c_p, k))$ and $\X(\tpl x^*, \tpl r^*, (c^*_p, k^*))$ are fully isomorphic.
	\end{enumerate}
Check (2) is done via $\spqrIso$, for (3) we can use $\bcCIso$, as the involved contexts are not affected by the change. As before, we can build a witnessing isomorphism from the maintained partial witnesses.
	
	In the case that $N$ is a cut vertex, we do as above for its one or two child biconnected components on the path $\rho$ and then employ Claim~\ref{clm:subtree-iso-cut-fo}.

	\subparagraph*{Case $\es[+]{0}{2}$.}
	In this case, a bridge edge $(u,v)$ is inserted between distinct connected components $A_u, A_v$.
	The affected vertices become cut vertices, if they are not already, and they form a new biconnected component $B$ that consists just of these two vertices.
	The isomorphism information and the new cut vertex colours can be determined along the lines of the case $\es[+]{3}{3}$.
	Note that if $u$ or $v$ are not cut vertices before the change, they have at most one child after the change and therefore their $\bcSibIsoCount$ and $\bcSibIsoRank$ information is easy to express.
	
	\subparagraph*{Case $\es[-]{2}{0}$.}
	Suppose a bridge edge $(u,v)$ is deleted and a connected component decomposes into two connected components.
	In the context $X$, one node loses a child because of the change: if the cut vertex node for $u$ is in $X$ after the change and $u$ is still a cut vertex, then this node loses the child connected component $\{u, v\}$ and its subtree; otherwise one biconnected component $B$ loses a child cut vertex.
	In the former case, the $\bcSibIsoCount$ and $\bcSibIsoRank$ information of the cut vertex can be updated easily and its new colour can then be determined as explained above. In the latter case, one vertex of $B$ becomes uncoloured.
	In both cases, the remaining updates can be performed as explained for the case $\es[+]{3}{3}$ above.
\end{proofof}
 
\section{Maintaining Isomorphism of Planar Graphs}\label{section:finally}
In this section, we finalize the proof of our main result and show that one can maintain in \DynFO whether two (possibly disconnected) planar graphs $G$ and $G^*$ are isomorphic, and that a witnessing isomorphism can be maintained whenever the input graphs are indeed isomorphic.
We use a binary relation $\IsoW$ to encode such a witness: if $G$ and $G^*$ are isomorphic, for every vertex $a$ of $G$ there is exactly one vertex $a^*$ of $G^*$ such that $(a,a^*) \in \IsoW$, and the function $\pi$ that maps a vertex $a$ to the vertex $a^*$ with $(a,a^*) \in \IsoW$ is an isomorphism from $G$ to $G^*$.

\begin{theorem}[Main theorem, rephrased]\label{thm:final}
	For any planar graphs $G, G^*$ it can be maintained in \DynFO whether $G$ and $G^*$ are isomorphic, under insertions and deletions of edges to both graphs, provided that they stay planar. 
	Additionally, the relation $\IsoW$ can be maintained.
\end{theorem}

Technically, a dynamic \DynFO program has a single input structure. We use two unary relations $V$ and $V^*$ to be able to represent two possibly disconnected graphs: these relations represent the vertex set of the two graphs $G$ and $G^*$, respectively.
We demand that these relations are disjoint at all times and allow inserting and deleting single disconnected vertices. We assume that all edge changes either affect two vertices from $V$ or two vertices from $V^*$.

\begin{proofof}{Theorem~\ref{thm:final}}
	We use an approach that is similar to the previous sections and maintain two further relations.
	The relation $\IsoCount$ contains tuples of the form $(a,m)$ to denote that the connected component of $G$ or $G^*$ that contains the vertex $a$ is isomorphic to $m-1$ other connected components of the same graph.
	The relation $\IsoRank$ provides an enumeration of the isomorphic connected components and contains for each $(a,m) \in \IsoCount$ a tuple $(a,k)$ with $k \leq m$, such that for all vertices $b$ from the connected component of $a$ the tuple $(b,k)$ is in $\IsoRank$ and the rank $k$ assigned to this connected component is unique among the $m$ isomorphic components.
	
	These relations are easy to maintain, as each change affects at most two connected components and using the relation $\bcIso$ one can check whether two connected components are isomorphic, by checking that there are two vertices $a, a^*$ from the respective connected components such that $(a,a^*) \in \bcIso$ holds.
	
	The two graphs $G$ and $G^*$ are isomorphic if for every connected component of one graph there is an isomorphic connected component in the other graph such that they have the same number of isomorphic components in their respective graphs, as given by $\IsoCount$.
	The witness is obtained from the relation $\bcIsoW$ in the following way. For each pair of isomorphic connected components $A$ of $G$ and $A^*$ of $G^*$ such that $A$ and $A^*$ have the same rank in the enumerations of isomorphic components in their respective graphs as given by $\IsoRank$, let $(a,a^*)$ be the lexicographically smallest tuple in $\bcIso$ with $a\in A$ and $a^* \in A^*$. 
	The witness $\pi$ maps a vertex $b$ from the component $A$ to the vertex $b^* \in A^*$ such that $(a,b,a^*,b^*) \in \bcIsoW$.
\end{proofof} 
\section{Conclusion}\label{section:conclusion}
We have shown that \DynFO is powerful enough to maintain whether two planar graphs are isomorphic.
For isomorphic $3$-connected components of a planar graph, one can maintain all witnessing isomorphisms, which is an important ingredient of our proof. 
For isomorphic biconnected components, connected components, and general planar graphs, we show how to maintain one out of the potentially exponentially many witnesses. 
Whether one can also maintain a canon of a planar graph, that is, a unique representative of the isomorphism type of the input graph, stays an open problem.

As for planar graphs, the isomorphism problem for bounded-treewidth graphs can be solved with logarithmic space \cite{ElberfeldS17}.
At the same time, many properties of graphs with bounded treewidth can be maintained in \DynFO \cite{DMSVZ19}.
It is interesting to study whether isomorphism of bounded-treewidth graphs can be maintained in \DynFO as well. A major difficulty towards such a result is the problem of maintaining isomorphism-invariant tree decompositions of graphs: if two graphs with bounded treewidth are isomorphic, we would like to have tree decompositions of these graphs that are also isomorphic. 
The (approximations of) tree decompositions utilized in \cite{DMSVZ19} do not have this property, as they depend on the order of edge changes to the input graph.
 
\bibliography{biblio}

\begin{thebibliography}{10}

\bibitem{Babai16}
L\'{a}szl\'{o} Babai.
\newblock Graph isomorphism in quasipolynomial time [extended abstract].
\newblock In {\em Proceedings of the Forty-Eighth Annual ACM Symposium on
  Theory of Computing}, STOC '16, page 684–697, New York, NY, USA, 2016.
  Association for Computing Machinery.
\newblock \href {https://doi.org/10.1145/2897518.2897542}
  {\path{doi:10.1145/2897518.2897542}}.

\bibitem{BarringtonIS90}
David A.~Mix Barrington, Neil Immerman, and Howard Straubing.
\newblock On uniformity within {NC}{\({^1}\)}.
\newblock {\em J. Comput. Syst. Sci.}, 41(3):274--306, 1990.
\newblock \href {https://doi.org/10.1016/0022-0000(90)90022-D}
  {\path{doi:10.1016/0022-0000(90)90022-D}}.

\bibitem{BattistaT96}
G.~Battista and R.~Tamassia.
\newblock On-line maintenance of triconnected components with spqr-trees.
\newblock {\em Algorithmica}, 15(4):302–318, April 1996.
\newblock \href {https://doi.org/10.1007/BF01961541}
  {\path{doi:10.1007/BF01961541}}.

\bibitem{Bodlaender90}
Hans~L Bodlaender.
\newblock Polynomial algorithms for graph isomorphism and chromatic index on
  partial k-trees.
\newblock {\em Journal of Algorithms}, 11(4):631--643, 1990.
\newblock \href {https://doi.org/10.1016/0196-6774(90)90013-5}
  {\path{doi:10.1016/0196-6774(90)90013-5}}.

\bibitem{Cook85}
Stephen~A. Cook.
\newblock A taxonomy of problems with fast parallel algorithms.
\newblock {\em Information and Control}, 64(1-3):2--21, 1985.
\newblock \href {https://doi.org/10.1016/S0019-9958(85)80041-3}
  {\path{doi:10.1016/S0019-9958(85)80041-3}}.

\bibitem{DasTW12}
Bireswar Das, Jacobo Torán, and Fabian Wagner.
\newblock Restricted space algorithms for isomorphism on bounded treewidth
  graphs.
\newblock {\em Information and Computation}, 217:71--83, 2012.
\newblock \href {https://doi.org/10.1016/j.ic.2012.05.003}
  {\path{doi:10.1016/j.ic.2012.05.003}}.

\bibitem{DattaKM23}
Samir Datta, Asif Khan, and Anish Mukherjee.
\newblock Dynamic planar embedding is in {DynFO}.
\newblock In J\'{e}r\^{o}me Leroux, Sylvain Lombardy, and David Peleg, editors,
  {\em 48th International Symposium on Mathematical Foundations of Computer
  Science (MFCS 2023)}, volume 272 of {\em Leibniz International Proceedings in
  Informatics (LIPIcs)}, pages 39:1--39:15, Dagstuhl, Germany, 2023. Schloss
  Dagstuhl -- Leibniz-Zentrum f{\"u}r Informatik.
\newblock \href {https://doi.org/10.4230/LIPIcs.MFCS.2023.39}
  {\path{doi:10.4230/LIPIcs.MFCS.2023.39}}.

\bibitem{DattaKM23arx}
Samir Datta, Asif Khan, and Anish Mukherjee.
\newblock Dynamic planar embedding is in {DynFO}.
\newblock {\em CoRR}, abs/2307.09473, 2023.
\newblock \href {https://arxiv.org/abs/2307.09473} {\path{arXiv:2307.09473}},
  \href {https://doi.org/10.48550/ARXIV.2307.09473}
  {\path{doi:10.48550/ARXIV.2307.09473}}.

\bibitem{DattaK0TVZ24}
Samir Datta, Asif Khan, Anish Mukherjee, Felix Tschirbs, Nils Vortmeier, and
  Thomas Zeume.
\newblock Query maintenance under batch changes with small-depth circuits.
\newblock In Rastislav Kr{\'{a}}lovic and Anton{\'{\i}}n Kucera, editors, {\em
  49th International Symposium on Mathematical Foundations of Computer Science,
  {MFCS} 2024, August 26-30, 2024, Bratislava, Slovakia}, volume 306 of {\em
  LIPIcs}, pages 46:1--46:16. Schloss Dagstuhl - Leibniz-Zentrum f{\"{u}}r
  Informatik, 2024.
\newblock \href {https://doi.org/10.4230/LIPICS.MFCS.2024.46}
  {\path{doi:10.4230/LIPICS.MFCS.2024.46}}.

\bibitem{DattaKTVZ26}
Samir Datta, Asif Khan, Felix Tschirbs, Nils Vortmeier, and Thomas Zeume.
\newblock Dynamic planar graph isomorphism is in dynfo.
\newblock In Claudia Faggian and Joost{-}Pieter Katoen, editors, {\em 41st
  Annual Symposium on Logic in Computer Science, {LICS} 2026, Lisbon, Portugal,
  July 20-23, 2026}, volume 380 of {\em LIPIcs}, pages 36:1--36:25. Schloss
  Dagstuhl - Leibniz-Zentrum f{\"{u}}r Informatik, 2026.
\newblock URL: \url{https://doi.org/10.4230/LIPIcs.LICS.2026.36}, \href
  {https://doi.org/10.4230/LIPICS.LICS.2026.36}
  {\path{doi:10.4230/LIPICS.LICS.2026.36}}.

\bibitem{DKMSZ18}
Samir Datta, Raghav Kulkarni, Anish Mukherjee, Thomas Schwentick, and Thomas
  Zeume.
\newblock Reachability is in {DynFO}.
\newblock {\em J. {ACM}}, 65(5):33:1--33:24, 2018.
\newblock \href {https://doi.org/10.1145/3212685} {\path{doi:10.1145/3212685}}.

\bibitem{DattaLNTW22}
Samir Datta, Nutan Limaye, Prajakta Nimbhorkar, Thomas Thierauf, and Fabian
  Wagner.
\newblock Planar graph isomorphism is in log-space.
\newblock {\em ACM Trans. Comput. Theory}, 14(2), sep 2022.
\newblock \href {https://doi.org/10.1145/3543686} {\path{doi:10.1145/3543686}}.

\bibitem{DMSVZ19}
Samir Datta, Anish Mukherjee, Thomas Schwentick, Nils Vortmeier, and Thomas
  Zeume.
\newblock A strategy for dynamic programs: Start over and muddle through.
\newblock {\em Log. Methods Comput. Sci.}, 15(2), 2019.
\newblock \href {https://doi.org/10.23638/LMCS-15(2:12)2019}
  {\path{doi:10.23638/LMCS-15(2:12)2019}}.

\bibitem{DattaMVZ18}
Samir Datta, Anish Mukherjee, Nils Vortmeier, and Thomas Zeume.
\newblock Reachability and distances under multiple changes.
\newblock In Ioannis Chatzigiannakis, Christos Kaklamanis, D{\'{a}}niel Marx,
  and Donald Sannella, editors, {\em 45th International Colloquium on Automata,
  Languages, and Programming, {ICALP} 2018, July 9-13, 2018, Prague, Czech
  Republic}, volume 107 of {\em LIPIcs}, pages 120:1--120:14. Schloss Dagstuhl
  - Leibniz-Zentrum f{\"{u}}r Informatik, 2018.
\newblock URL: \url{https://doi.org/10.4230/LIPIcs.ICALP.2018.120}, \href
  {https://doi.org/10.4230/LIPICS.ICALP.2018.120}
  {\path{doi:10.4230/LIPICS.ICALP.2018.120}}.

\bibitem{Diestel12}
Reinhard Diestel.
\newblock {\em Graph Theory, 4th Edition}, volume 173 of {\em Graduate texts in
  mathematics}.
\newblock Springer, 2012.

\bibitem{ElberfeldK14}
Michael Elberfeld and Ken-ichi Kawarabayashi.
\newblock Embedding and canonizing graphs of bounded genus in logspace.
\newblock In {\em Proceedings of the Forty-Sixth Annual ACM Symposium on Theory
  of Computing}, STOC '14, page 383–392, New York, NY, USA, 2014. Association
  for Computing Machinery.
\newblock \href {https://doi.org/10.1145/2591796.2591865}
  {\path{doi:10.1145/2591796.2591865}}.

\bibitem{ElberfeldS17}
Michael Elberfeld and Pascal Schweitzer.
\newblock Canonizing graphs of bounded tree width in logspace.
\newblock {\em {ACM} Trans. Comput. Theory}, 9(3):12:1--12:29, 2017.
\newblock \href {https://doi.org/10.1145/3132720} {\path{doi:10.1145/3132720}}.

\bibitem{Etessami98}
Kousha Etessami.
\newblock Dynamic tree isomorphism via first-order updates.
\newblock In Alberto~O. Mendelzon and Jan Paredaens, editors, {\em Proceedings
  of the Seventeenth {ACM} {SIGACT-SIGMOD-SIGART} Symposium on Principles of
  Database Systems, June 1-3, 1998, Seattle, Washington, {USA}}, pages
  235--243. {ACM} Press, 1998.
\newblock \href {https://doi.org/10.1145/275487.275514}
  {\path{doi:10.1145/275487.275514}}.

\bibitem{GradelS12}
Erich Gr{\"{a}}del and Sebastian Siebertz.
\newblock Dynamic definability.
\newblock In Alin Deutsch, editor, {\em 15th International Conference on
  Database Theory, {ICDT} '12, Berlin, Germany, March 26-29, 2012}, pages
  236--248. {ACM}, 2012.
\newblock \href {https://doi.org/10.1145/2274576.2274601}
  {\path{doi:10.1145/2274576.2274601}}.

\bibitem{GroheV06}
Martin Grohe and Oleg Verbitsky.
\newblock Testing graph isomorphism in parallel by playing a game.
\newblock In Michele Bugliesi, Bart Preneel, Vladimiro Sassone, and Ingo
  Wegener, editors, {\em Automata, Languages and Programming, 33rd
  International Colloquium, {ICALP} 2006, Venice, Italy, July 10-14, 2006,
  Proceedings, Part {I}}, volume 4051 of {\em Lecture Notes in Computer
  Science}, pages 3--14. Springer, 2006.
\newblock \href {https://doi.org/10.1007/11786986\_2}
  {\path{doi:10.1007/11786986\_2}}.

\bibitem{GutwengerM01}
Carsten Gutwenger and Petra Mutzel.
\newblock A linear time implementation of spqr-trees.
\newblock In Joe Marks, editor, {\em Graph Drawing, 8th International
  Symposium, {GD} 2000, Colonial Williamsburg, VA, USA, September 20-23, 2000,
  Proceedings}, Lecture Notes in Computer Science, pages 77--90. Springer,
  2000.
\newblock \href {https://doi.org/10.1007/3-540-44541-2\_8}
  {\path{doi:10.1007/3-540-44541-2\_8}}.

\bibitem{HendersonS81}
H.~V. Henderson and S.~R. Searle.
\newblock On deriving the inverse of a sum of matrices.
\newblock {\em SIAM Review}, 23(1):53--60, 1981.
\newblock \href {https://doi.org/10.1137/1023004} {\path{doi:10.1137/1023004}}.

\bibitem{HesseAB02}
William Hesse, Eric Allender, and David A.~Mix Barrington.
\newblock Uniform constant-depth threshold circuits for division and iterated
  multiplication.
\newblock {\em J. Comput. Syst. Sci.}, 65(4):695--716, 2002.
\newblock \href {https://doi.org/10.1016/S0022-0000(02)00025-9}
  {\path{doi:10.1016/S0022-0000(02)00025-9}}.

\bibitem{HolmIKLR18}
Jacob Holm, Giuseppe~F. Italiano, Adam Karczmarz, Jakub Lacki, and Eva
  Rotenberg.
\newblock Decremental spqr-trees for planar graphs.
\newblock In Yossi Azar, Hannah Bast, and Grzegorz Herman, editors, {\em 26th
  Annual European Symposium on Algorithms (ESA 2018)}, volume 112 of {\em
  Leibniz International Proceedings in Informatics (LIPIcs)}, pages
  46:1--46:16, Dagstuhl, Germany, 2018. Schloss Dagstuhl -- Leibniz-Zentrum
  f{\"u}r Informatik.
\newblock \href {https://doi.org/10.4230/LIPIcs.ESA.2018.46}
  {\path{doi:10.4230/LIPIcs.ESA.2018.46}}.

\bibitem{HolmIKLR18arx}
Jacob Holm, Giuseppe~F. Italiano, Adam Karczmarz, Jakub Lacki, and Eva
  Rotenberg.
\newblock Decremental spqr-trees for planar graphs.
\newblock {\em CoRR}, abs/1806.10772, 2018.
\newblock \href {https://arxiv.org/abs/1806.10772} {\path{arXiv:1806.10772}}.

\bibitem{HopcroftW74}
J.~E. Hopcroft and J.~K. Wong.
\newblock Linear time algorithm for isomorphism of planar graphs (preliminary
  report).
\newblock In {\em Proceedings of the Sixth Annual ACM Symposium on Theory of
  Computing}, STOC '74, page 172–184, New York, NY, USA, 1974. Association
  for Computing Machinery.
\newblock \href {https://doi.org/10.1145/800119.803896}
  {\path{doi:10.1145/800119.803896}}.

\bibitem{HopcroftT73}
John~E. Hopcroft and Robert~Endre Tarjan.
\newblock Dividing a graph into triconnected components.
\newblock {\em {SIAM} J. Comput.}, 2(3):135--158, 1973.
\newblock \href {https://doi.org/10.1137/0202012} {\path{doi:10.1137/0202012}}.

\bibitem{JajaK88}
J.~Jaja and S.R. Kosaraju.
\newblock Parallel algorithms for planar graph isomorphism and related
  problems.
\newblock {\em IEEE Transactions on Circuits and Systems}, 35(3):304--311,
  1988.
\newblock \href {https://doi.org/10.1109/31.1743} {\path{doi:10.1109/31.1743}}.

\bibitem{KieferPS19}
Sandra Kiefer, Ilia Ponomarenko, and Pascal Schweitzer.
\newblock The {W}eisfeiler-{L}eman dimension of planar graphs is at most 3.
\newblock {\em J. {ACM}}, 66(6):44:1--44:31, 2019.
\newblock \href {https://doi.org/10.1145/3333003} {\path{doi:10.1145/3333003}}.

\bibitem{Lindell92}
Steven Lindell.
\newblock A logspace algorithm for tree canonization (extended abstract).
\newblock In {\em Proceedings of the 24th Annual {ACM} Symposium on Theory of
  Computing, May 4-6, 1992, Victoria, British Columbia, Canada}, pages
  400--404, 1992.
\newblock \href {https://doi.org/10.1145/129712.129750}
  {\path{doi:10.1145/129712.129750}}.

\bibitem{Mehta14}
Jenish~C. Mehta.
\newblock Dynamic complexity of planar 3-connected graph isomorphism.
\newblock In Edward~A. Hirsch, Sergei~O. Kuznetsov, Jean{-}{\'{E}}ric Pin, and
  Nikolay~K. Vereshchagin, editors, {\em Computer Science - Theory and
  Applications - 9th International Computer Science Symposium in Russia, {CSR}
  2014, Moscow, Russia, June 7-11, 2014. Proceedings}, Lecture Notes in
  Computer Science, pages 273--286. Springer, 2014.
\newblock \href {https://doi.org/10.1007/978-3-319-06686-8\_21}
  {\path{doi:10.1007/978-3-319-06686-8\_21}}.

\bibitem{PatnaikI94}
Sushant Patnaik and Neil Immerman.
\newblock Dyn-{FO}: {A} parallel, dynamic complexity class.
\newblock In Victor Vianu, editor, {\em Proceedings of the Thirteenth {ACM}
  {SIGACT-SIGMOD-SIGART} Symposium on Principles of Database Systems, May
  24-26, 1994, Minneapolis, Minnesota, {USA}}, pages 210--221. {ACM} Press,
  1994.
\newblock \href {https://doi.org/10.1145/182591.182614}
  {\path{doi:10.1145/182591.182614}}.

\bibitem{PatnaikI97}
Sushant Patnaik and Neil Immerman.
\newblock Dyn-{FO}: {A} parallel, dynamic complexity class.
\newblock {\em J. Comput. Syst. Sci.}, 55(2):199--209, 1997.
\newblock \href {https://doi.org/10.1006/JCSS.1997.1520}
  {\path{doi:10.1006/JCSS.1997.1520}}.

\bibitem{RobertsonS90}
Neil Robertson and P.D Seymour.
\newblock Graph minors. ix. disjoint crossed paths.
\newblock {\em Journal of Combinatorial Theory, Series B}, 49(1):40--77, 1990.
\newblock \href {https://doi.org/10.1016/0095-8956(90)90063-6}
  {\path{doi:10.1016/0095-8956(90)90063-6}}.

\bibitem{ShivakumarC74}
P.~N. Shivakumar and Kim~Ho Chew.
\newblock A sufficient condition for nonvanishing of determinants.
\newblock {\em Proc. Amer. Math. Soc.}, 43:63--66, 1974.
\newblock \href {https://doi.org/10.2307/2039326} {\path{doi:10.2307/2039326}}.

\bibitem{ThieraufW10}
Thomas Thierauf and Fabian Wagner.
\newblock The isomorphism problem for planar 3-connected graphs is in
  unambiguous logspace.
\newblock {\em Theory Comput. Syst.}, 47(3):655--673, 2010.
\newblock \href {https://doi.org/10.1007/S00224-009-9188-4}
  {\path{doi:10.1007/S00224-009-9188-4}}.

\bibitem{Toran04}
Jacobo Tor\'{a}n.
\newblock On the hardness of graph isomorphism.
\newblock {\em SIAM Journal on Computing}, 33(5):1093--1108, 2004.
\newblock \href {https://doi.org/10.1137/S009753970241096X}
  {\path{doi:10.1137/S009753970241096X}}.

\bibitem{Tutte63}
W.~T. Tutte.
\newblock How to draw a graph.
\newblock {\em Proceedings of the London Mathematical Society},
  s3-13(1):743--767, 1963.
\newblock \href {https://doi.org/10.1112/plms/s3-13.1.743}
  {\path{doi:10.1112/plms/s3-13.1.743}}.

\bibitem{VortmeierZ21}
Nils Vortmeier and Thomas Zeume.
\newblock Dynamic complexity of parity exists queries.
\newblock {\em Log. Methods Comput. Sci.}, 17(4), 2021.
\newblock \href {https://doi.org/10.46298/LMCS-17(4:9)2021}
  {\path{doi:10.46298/LMCS-17(4:9)2021}}.

\bibitem{Weinberg66}
L.~Weinberg.
\newblock A simple and efficient algorithm for determining isomorphism of
  planar triply connected graphs.
\newblock {\em IEEE Transactions on Circuit Theory}, 13(2):142--148, 1966.
\newblock \href {https://doi.org/10.1109/TCT.1966.1082573}
  {\path{doi:10.1109/TCT.1966.1082573}}.

\bibitem{Whitney33}
Hassler Whitney.
\newblock 2-isomorphic graphs.
\newblock {\em American Journal of Mathematics}, 55(1):245--254, 1933.
\newblock URL: \url{http://www.jstor.org/stable/2371127}.

\end{thebibliography}

\appendix

\section{Additional Details for Section~\ref{sec:tutte:algebra}}

We explain here the implementations of further linear algebra operations that are used for maintaining Tutte embeddings.

\subsubsection*{Exchange of pinned vertices}

We show that given an embedding based on three pinned vertices $v_1, v_2, v_3$ we can easily obtain an embedding based on three different pinned vertices $v'_1, v'_2, v'_3$.

\begin{lemma}\label{lem:lina-pins}
	Fix a domain of size $n$. Let $G$ be a graph with at most $n$ vertices, $v_{i_1}, v_{i_2}, v_{i_3}$ three of its vertices, and $p$ a prime of magnitude $\bigO(n^c)$, for some constant $c$, such that $(G, v_{i_1}, v_{i_2}, v_{i_3}, p)$ is modulo embedding-inducing.
	Then given $\M(G,v_{i_1}, v_{i_2}, v_{i_3}, p)$ one can determine in \FOar for any three vertices $v'_{i_1}, v'_{i_2}, v'_{i_3}$ whether $(G, v'_{i_1}, v'_{i_2}, v'_{i_3}, p)$ is also modulo embedding-inducing and if so, $\M(G,v'_{i_1}, v'_{i_2}, v'_{i_3}, p)$ can be defined.
\end{lemma}
\begin{proof}
	We assume, w.l.o.g., that the vertices $v_{i_1}, v_{i_2}, v_{i_3}$ and $v'_{i_1}, v'_{i_2}, v'_{i_3}$ are distinct and for notational simplicity that $v_{i_1}, v_{i_2}, v_{i_3} = v_1, v_2, v_3$ and $v'_{i_1}, v'_{i_2}, v'_{i_3} = v_4, v_5, v_6$.
	The matrices $\mat T = \mat T(G,v_1,v_2,v_3,p)$ and $\mat T' = \mat T(G,v_4,v_5,v_6,p)$ differ only in the first six rows. More precisely, $\mat T' = \mat T + \mat U \mat V^\transpose$, where $\mat U$ is the $n \times 6$ matrix which has $1$ as diagonal entries and $0$ otherwise; $\mat V$ is the $n \times 6$ matrix which is obtained by taking the first six columns of the Laplacian $\mat L$ of $G$, decreasing each value in the diagonal by $1$, and then inverting the sign of the entries in the columns $4$ to $6$.
	It holds
	\[(\mat T+ \mat U \mat V^\transpose)^{-1} = \mat T^{-1} - \mat T^{-1}\mat U{(\mat I+\mat V^\transpose \mat T^{-1}\mat U)}^{-1}\mat V^\transpose \mat T^{-1} \]
	The inverse exists in $\Z_p$ if $\mat I+\mat V^\transpose \mat T^{-1}\mat U$ is invertible over $\Z_p$. This can be checked in \FOar, as the product $\mat V^\transpose \mat T^{-1}$ can be computed using $\M(G,v_1, v_2, v_3, p)$: the $i$-th row of $\mat V^\transpose \mat T^{-1}$ can be written as $\mat v_i^\transpose \mat T^{-1}$, where the $i$-th column $\mat v_i$ of $\mat V$ is either equal to $\mat c_i - \mat e_i$ or to $- \mat c_i + \mat e_i$, where $\mat c_i$ is the $i$-th column of the Laplacian $\mat L$ and $\mat e_i$ is the $i$-th unit vector.
	The product $\mat c_i^\transpose \mat T^{-1}$ is available in $\M(G,v_1, v_2, v_3, p)$, the product $\mat e_i^\transpose \mat T^{-1}$ is trivial to compute. So, $\mat V^\transpose \mat T^{-1}$ can be computed in \FOar and $\mat V^\transpose \mat T^{-1}\mat U$ is just the first six columns of the matrix.
	The matrix $\mat I+\mat V^\transpose \mat T^{-1}\mat U$ has constant size $6 \times 6$, therefore one can compute its determinant in \FOar as well and check whether it is $0$. If not, the inverse ${(\mat I+\mat V^\transpose \mat T^{-1}\mat U)}^{-1}$ can be computed.
	Computing the remaining matrix products also only involves summing at most six numbers that are the product of two numbers, which is possible in $\FOar$.
	
	We explain how to compute the products of the form $\mat c_1^\transpose \, \mat T'^{-1} \, \mat c_2$, where $\mat c_1$ and $\mat c_2$ are columns of the Laplacian of $G$. The other necessary products are computed analogously. We use again the SMW formula:
	\[\mat c_1^\transpose \,(\mat T+ \mat U \mat V^\transpose)^{-1} \, \mat c_2= \mat c_1^\transpose \, \mat T^{-1} \, \mat c_2 -  \mat c_1^\transpose \,  \mat T^{-1}\mat U{(\mat I+\mat V^\transpose \mat T^{-1}\mat U)}^{-1}\mat V^\transpose \mat T^{-1} \, \mat c_2 \]
	The scalar  $\mat c_1^\transpose  \mat T^{-1}  \mat c_2$ and the vector $\mat c_1^\transpose \mat T^{-1}$ are given in $\M(G,v_1, v_2, v_3, p)$,  and $\mat V^\transpose \mat T^{-1}  \mat c_2$ can be computed analogously as explained above for the product $\mat V^\transpose \mat T^{-1}$, using the available results for $\mat c_i^\transpose \mat T^{-1} \mat c_2$ from $\M(G,v_1, v_2, v_3, p)$.
\end{proof}

\subsubsection*{Merging and splitting of components}

Now we explain how to compute $\M(G,v_{i_1}, v_{i_2}, v_{i_3}, p)$ for a graph $G$ that is the union of two subgraphs $G_1, G_2$ that have exactly two vertices in common plus an additional edge between the subgraphs. We will apply this operation in the case that an edge insertion turns two $3$-connected components with a common separating pair into one three-connected component. See Figure~\ref{fig:lina-merge-tutte} for an illustrative example.

\newcommand{\mrow}[3]{%
	\foreach \e [count=\c from #2] in {#3} \node at (\c-0.5,#1-0.5) {$\e$};}
\newcommand{\rlabel}[2]{\node[anchor=east, inner sep=1pt] (rl) at (-0.12,#1-0.5) {$#2$};}
\newcommand{\plabel}[2]{\rlabel{#1}{#2}%
	\draw ([xshift=-1.5pt]rl.west) -- ++(-5pt,0) ++(-1.3pt,0) circle[radius=1.3pt];}
\newcommand{\bigzero}[2]{\node[font=\Large] at (#1,#2) {$0$};}
\newcommand{\mbrace}[4]{%
	\draw[gray!60, decorate, decoration={brace, mirror, amplitude=5pt}]
	(#1+0.05,#3+0.25) -- node[below=6pt, text=black] {$#4$} (#2-0.05,#3+0.25);}

\begin{figure}[t]
	\begin{subfigure}{\textwidth}
		\centering
		\begin{tikzpicture}[x=1pt, y=1pt,
			sv/.style={circle, fill, inner sep=0pt, minimum size=1.8pt},
			bv/.style={circle, fill, inner sep=0pt, minimum size=5.2pt}]
			\coordinate (v3) at (0,0);      \coordinate (b)  at (17.3,-23.8);
			\coordinate (v2) at (45.1,-14.8); \coordinate (v1) at (45.1,14.7);
			\coordinate (t)  at (17.3,23.7);  \coordinate (c)  at (24.9,0);
			\draw (v3) -- (b) -- (v2) -- (v1) -- (t) -- cycle;
			\foreach \p in {v3,b,v2,v1,t} \draw (c) -- (\p);
			\foreach \p in {b,t,c} \node[sv] at (\p) {};
			\node[bv, label=left:$v_3$]        at (v3) {};
			\node[bv, label=right:$v_1$]       at (v1) {};
			\node[bv, label=below right:$v_2$] at (v2) {};
			\node at (-12.8,-29.7) {$G_1$};
		\end{tikzpicture}\quad
		\begin{tikzpicture}[x=1pt, y=1pt,
			sv/.style={circle, fill, inner sep=0pt, minimum size=2pt},
			bv/.style={circle, fill, inner sep=0pt, minimum size=5.2pt}]
			\coordinate (v1) at (0,16);     \coordinate (v2) at (0,-15.9);
			\coordinate (v4) at (47.8,0);   \coordinate (b)  at (29.5,-25.8);
			\coordinate (t)  at (29.5,25.9);
			\coordinate (x)  at (20.9,6.6); \coordinate (y)  at (20.9,-6.5);
			\draw (v4) -- (b) -- (v2) -- (v1) -- (t) -- cycle;
			\draw (v1) -- (x) -- (y) -- (v2);
			\draw (t) -- (x) -- (v4) -- (y) -- (b);
			\foreach \p in {b,t,x,y} \node[sv] at (\p) {};
			\node[bv, label=above left:$v_1$] at (v1) {};
			\node[bv, label=below left:$v_2$] at (v2) {};
			\node[bv, label=right:$v_4$]      at (v4) {};
			\node at (-25.5,-26.5) {$G_2$};
		\end{tikzpicture}\quad
		\begin{tikzpicture}[x=1pt, y=1pt, scale=0.8,
			sv/.style={circle, fill, inner sep=0pt, minimum size=2.2pt},
			bv/.style={circle, fill, inner sep=0pt, minimum size=5.2pt}]
			\coordinate (v3) at (0,0);        \coordinate (v4) at (113.1,0);
			\coordinate (v1) at (56.5,18.3);  \coordinate (v2) at (56.5,-18.3);
			\coordinate (b1) at (21.5,-29.6); \coordinate (t1) at (21.5,29.6);
			\coordinate (c)  at (31.1,0);
			\coordinate (b2) at (91.6,-29.6); \coordinate (t2) at (91.6,29.6);
			\coordinate (x)  at (81.4,7.5);   \coordinate (y)  at (81.4,-7.5);
			\draw (v3) -- (b1) -- (v2) (v1) -- (t1) -- (v3);
			\foreach \p in {v3,b1,v2,v1,t1} \draw (c) -- (\p);
			\draw (v4) -- (b2) -- (v2) (v1) -- (t2) -- (v4);
			\draw (v1) -- (x) -- (y) -- (v2);
			\draw (t2) -- (x) -- (v4) -- (y) -- (b2);
			\draw[] (v1) -- (v2);
			\draw[blue, line width=1.5pt] (v3) .. controls (-60.4,61.1) and (169.9,64.9) .. (v4);
			\foreach \p in {b1,t1,c,b2,t2,x,y} \node[sv] at (\p) {};
			\node[sv, label=below:$v_2$]      at (v2) {};
			\node[bv, label=left:$v_3$]       at (v3) {};
			\node[bv, label=above:$v_1$]      at (v1) {};
			\node[bv, label=right:$v_4$]      at (v4) {};
			\node at (-14.8,-20.5) {$G$};
		\end{tikzpicture}
		\caption{Two $3$-connected planar graphs $G_1$ and $G_2$ that share a separating pair $\{v_1,v_2\}$. The graph $G$ is obtained by inserting the edge $(v_3,v_4)$ into the union of $G_1$ and $G_2$.}
		\label{fig:sfig:lina-tutte-merge}
	\end{subfigure}
	
	\vspace*{3mm}
	\begin{subfigure}{\textwidth}
		\centering
		\resizebox{\textwidth}{!}{
		\begin{tikzpicture}[baseline=(current bounding box.north), x=18.13pt, y=-14.8pt, scale=0.9]
			\fill[red!5]  (0,2) rectangle (2,6);
			\fill[red!20] (2,2) rectangle (6,6);
			\draw (0,0) rectangle (6,6);
			\draw (0,2) -- (6,2)  (2,0) -- (2,6);
			\node at (3,-1.5) {$\mathbf{T}_1$};
			\bigzero{4}{1}
			\mrow{1}{1}{1,0}
			\mrow{2}{1}{0,1}
			\mrow{3}{1}{0,0,1,0,0,0}
			\mrow{4}{1}{-1,0,-1,3,-1,0}
			\mrow{5}{1}{-1,-1,-1,-1,5,-1}
			\mrow{6}{1}{0,-1,-1,0,-1,3}
			\plabel{1}{v_1} \plabel{2}{v_2} \plabel{3}{v_3}
			\mbrace{0}{2}{6}{\mathbf{A}_1}
			\mbrace{2}{6}{6}{\mathbf{S}_1}
		\end{tikzpicture}\quad
		\begin{tikzpicture}[baseline=(current bounding box.north), x=18.13pt, y=-14.8pt, scale=0.9]
			\fill[blue!5]  (0,2) rectangle (2,7);
			\fill[blue!20] (2,2) rectangle (7,7);
			\draw (0,0) rectangle (7,7);
			\draw (0,2) -- (7,2)  (2,0) -- (2,7);
			\node at (3.5,-1.5) {$\mathbf{T}_2$};
			\bigzero{4.5}{1}
			\mrow{1}{1}{1,0}
			\mrow{2}{1}{0,1}
			\mrow{3}{1}{0,0,1,0,0,0,0}
			\mrow{4}{1}{-1,0,-1,3,-1,0,0}
			\mrow{5}{1}{-1,0,-1,-1,4,-1,0}
			\mrow{6}{1}{0,-1,-1,0,-1,4,-1}
			\mrow{7}{1}{0,-1,-1,0,0,-1,3}
			\plabel{1}{v_1} \plabel{2}{v_2} \plabel{3}{v_4}
			\mbrace{0}{2}{7}{\mathbf{A}_2}
			\mbrace{2}{7}{7}{\mathbf{S}_2}
		\end{tikzpicture}\quad
		\begin{tikzpicture}[baseline=(current bounding box.north), x=18.13pt, y=-14.8pt, scale=0.9]
			\fill[red!5]   (0,2) rectangle (2,6);
			\fill[red!20]  (2,2) rectangle (6,6);
			\fill[blue!5]  (0,6) rectangle (2,11);
			\fill[blue!20] (6,6) rectangle (11,11);
			\draw (0,0) rectangle (11,11);
			\draw (0,2) -- (11,2)  (0,6) -- (11,6)  (2,0) -- (2,11)  (6,0) -- (6,11);
			\node at (5.5,-1.5) {$\mathbf{T}$};
			\bigzero{8.5}{4}
			\bigzero{4}{8.5}
			\mrow{1}{1}{1,0,0,0,0,0,0,0,0,0,0}
			\mrow{2}{1}{-1,5,0,0,-1,-1,0,0,0,-1,-1}
			\mrow{3}{1}{0,0,1,0,0,0}
			\mrow{4}{1}{-1,0,-1,3,-1,0}
			\mrow{5}{1}{-1,-1,-1,-1,5,-1}
			\mrow{6}{1}{0,-1,-1,0,-1,3}
			\mrow{7}{1}{0,0}   \mrow{7}{7}{1,0,0,0,0}
			\mrow{8}{1}{-1,0}  \mrow{8}{7}{-1,3,-1,0,0}
			\mrow{9}{1}{-1,0}  \mrow{9}{7}{-1,-1,4,-1,0}
			\mrow{10}{1}{0,-1} \mrow{10}{7}{-1,0,-1,4,-1}
			\mrow{11}{1}{0,-1} \mrow{11}{7}{-1,0,0,-1,3}
			\plabel{1}{v_1} \rlabel{2}{v_2} \plabel{3}{v_3} \plabel{7}{v_4}
		\end{tikzpicture}
	}
		\caption{The Tutte matrices of $G_1$ and $G_2$ before the insertion, and the Tutte matrix of $G$ where the matrices of $G_1$ and $G_2$ are diagonally overlapped as blocks about $v_1$ and $v_2$.}
		\label{fig:sfig:lina-tutte-merged}
	\end{subfigure}
	
	\caption{Illustration for the proof of Lemma~\ref{lem:lina-merging}.}
	\label{fig:lina-merge-tutte}
\end{figure}
\tzm{Figure: explainatory caption, pictures}
\begin{lemma}\label{lem:lina-merging}
	Fix a domain of size $n$. Let $G_1, G_2$ be subgraphs of a graph with at most $n$ vertices such that they have exactly two vertices $v_{i_1}, v_{i_2}$ in common. Let $v_{i_3}$ be another vertex of $G_1$ and $v_{i_4}$ another vertex of $G_2$ and $p$ a prime of magnitude $\bigO(n^c)$, for some constant $c$, such that $(G_1, v_{i_1}, v_{i_2}, v_{i_3}, p)$ and $(G_2, v_{i_1}, v_{i_2}, v_{i_4}, p)$ are modulo embedding-inducing.
	Let $G$ be the graph that results from the union of $G_1$ and $G_2$ by adding the edge $(v_{i_3}, v_{i_4})$. Then given $\M(G_1,v_{i_1}, v_{i_2}, v_{i_3}, p)$  and $\M(G_2,v_{i_1}, v_{i_2}, v_{i_4}, p)$ one can determine in \FOar whether $(G, v_{i_1}, v_{i_3}, v_{i_4}, p)$ and $(G, v_{i_2}, v_{i_3}, v_{i_4}, p)$ are also modulo embedding-inducing and if so, $\M(G,v_{i_1}, v_{i_3}, v_{i_4}, p)$ and $\M(G,v_{i_2}, v_{i_3}, v_{i_4}, p)$ can be defined.
\end{lemma}
\begin{proof}
	Suppose $(G_1, v_{i_1}, v_{i_2}, v_{i_3}, p)$ and $(G_2, v_{i_1}, v_{i_2}, v_{i_4}, p)$ are modulo embedding-inducing, where $G_1$ is over the vertex set $V_1$, $G_2$ is over the vertex set $V_2$, and $V_1 \cap V_2 = \{v_{i_1}, v_{i_2}\}$. We explain how to compute $\M(G,v_{i_1}, v_{i_3}, v_{i_4}, p)$ if it is defined, computing $\M(G, v_{i_2}, v_{i_3}, v_{i_4}, p)$ is analogous.
	
	For ease of notation we assume $v_{i_1}, v_{i_2}, v_{i_3}, v_{i_4} = v_1, v_2, v_3, v_4$.
	Let $\mat T_1 = \mat T(G_1, v_1, v_2, v_3, p)$ and let $\mat T_2 = \mat T(G_2, v_1, v_2, v_4)$. Let $\mat S_1, \mat S_2$ be the submatrices of $\mat T_1, \mat T_2$ that do not include the rows and columns for $v_1, v_2$.
	The matrix $\mat T = \mat T(G, v_1, v_3, v_4)$ has the following non-zero entries: 
	\begin{enumerate}
		\item the non-zero entries of $\mat S_1, \mat S_2$,
		\item the entry $1$ in the diagonal for the vertex $v_1$ and as diagonal entry for $v_2$ its degree in $G$,
		\item in the column for $v_1$ (apart from the diagonal): the remaining non-zero entries from the matrices $\mat T_1$ and $\mat T_2$,
		\item in the row and column for $v_2$ (apart from the diagonal): the remaining non-zero entries from the Laplacians of $G_1$ and $G_2$.
	\end{enumerate}
	Let $\mat R$ be the matrix that only contains the entries from item (1) and diagonal entries $1$ for $v_1$ and $v_2$, that is $\mat R$ is equal\footnote{In general $\mat R$ is a permutation of the given matrix, but we ignore this technicality from now on.}\tzm{Is this footnote necessary? Does this refer to the above naming $v_1, v_2,...$ etc?} to $\left[
	\begin{array}{c |c| c }
		\mat I & \mat 0 & \mat 0\\
		\hline
		\mat 0 & \mat S_1 & \mat 0\\
		\hline
		\mat 0 & \mat 0 & \mat S_2
	\end{array}\right]$.
	
	We next show how to compute the inverse $\mat R^{-1}$ and further necessary products with columns of the Laplacian of $G$.  After that we show how to obtain $\mat T$ as $\mat R + \mat U \mat V^\transpose$ for low-rank matrices $\mat U, \mat V$ and then use the SMW formula to compute $\M(G,v_1, v_3, v_4, p)$.
	
	\smallskip
	\subparagraph*{The  inverse of the matrix $\mat R$:}
	
	The matrix $\mat R$ as defined above  is a (permutation of a) block diagonal matrix, where each block is invertible modulo $p$. This is clear for the block for $v_1, v_2$, as it is an identity matrix; the matrices $\mat S_1, \mat S_2$ are invertible because they are square submatrices of $\mat T_1, \mat T_2$ that are invertible modulo $p$ by assumption.  So, the inverse of $\mat R$ is also a (permutation of a) block diagonal matrix that consists of the inverses of $\mat S_1, \mat S_2$ and the $2 \times 2$ identity matrix, that is,  $\mat R^{-1}$ is equal to $\left[
	\begin{array}{c c c }
		\mat I & \mat 0 & \mat 0\\
		\mat 0 & \mat S_1^{-1} & \mat 0\\
		\mat 0 & \mat 0 & \mat S_2^{-1}
	\end{array}
	\right]$.
	
	We show how to obtain the inverses $\mat S_1^{-1}$ and $\mat S_2^{-1}$. The matrix $\mat T_i$ has the form $\mat T_i = \left[
	\begin{array}{c c}
		\mat I & \mat 0\\
		\mat A_i & \mat S_i
	\end{array}
	\right]$, where $\mat I$ is the $2 \times 2$ identity matrix and $\mat A_i$ coincides with the columns of $v_1$ and $v_2$ in the Laplacian of $G_i$ with the first two rows removed.
	For $\mat U_i = \left[
	\begin{array}{c}
		\mat 0 \\
		\mat A_i
	\end{array}
	\right]$ and $\mat V_i = \left[
	\begin{array}{c}
		\mat I \\
		\mat 0
	\end{array}
	\right]$, it holds $\mat T_i - \mat U_i \mat V_i^\transpose =  \left[
	\begin{array}{c c}
		\mat I & \mat 0\\
		\mat 0 & \mat S_i
	\end{array}
	\right]$. The inverse of this matrix exists and is equal to
	\[ (\mat T_i - \mat U_i \mat V_i^\transpose)^{-1} = \mat T_i^{-1} + \mat T_i^{-1}\mat U_i{(\mat I-\mat V_i^\transpose \mat T_i^{-1}\mat U_i)}^{-1}\mat V_i^\transpose \mat T_i^{-1},\]
	which can be computed in \FOar: the only non-obvious step is the computation of  $\mat T_i^{-1}\mat U_i$, but each column $\mat u$ of $\mat U_i$ differs from some column $\mat c$ of the Laplacian of $G_i$ only by a constant number of entries, so $\mat T_i^{-1} \mat u$ can be computed from the available $\mat T_i^{-1} \mat c$ with only a constant number of parallel multiplications and additions.
	
	The inverse $(\mat T_i - \mat U_i \mat V_i^\transpose)^{-1}$ has the form $\left[
	\begin{array}{c c}
		\mat I & \mat 0\\
		\mat 0 & \mat S_i^{-1}
	\end{array}
	\right]$, so the inverse $\mat S_i^{-1}$ of $\mat S_i$ can be read off from it.
	
	We also need later on the products of the form $\mat S_i^{-1} \mat c, \mat c^\transpose \mat S_i^{-1}$ and $\mat c_1^\transpose \mat S_i^{-1} \mat c_2$, where $\mat c, \mat c_1, \mat c_2$ are obtained from the columns of the Laplacian of $G_i$ by removing the first two entries, corresponding to the rows of $v_1$ and $v_2$.
	For example, the scalar $\mat c_1^\transpose \mat S_1^{-1} \mat c_2$ can be computed from
	\[ \mat d_1^\transpose (\mat T_1 - \mat U_1 \mat V_1^\transpose)^{-1} \mat d_2 = \mat d_1^\transpose \mat T_1^{-1} \mat d_2 + \mat d_1^\transpose \mat T_1^{-1}\mat U_1{(\mat I-\mat V_1^\transpose \mat T_1^{-1}\mat U_1)}^{-1}\mat V_1^\transpose \mat T_1^{-1} \mat d_2,\]
	where $\mat d_1, \mat d_2$ are the columns from the Laplacian that $\mat c_1, \mat c_2$ are obtained from by removing the entries $d_{11}$ and $d_{12}$ as well as $d_{21}$ and $d_{22}$, as $\mat c_1^\transpose \mat S_1^{-1} \mat c_2 = \mat d_1^\transpose (\mat T_1 - \mat U_1 \mat V_1^\transpose)^{-1} \mat d_2 - (d_{11}d_{21} + d_{12}d_{22})$.
	The value $\mat d_1^\transpose (\mat T_1 - \mat U_1 \mat V_1^\transpose)^{-1} \mat d_2$ can be computed analogously to $(\mat T_1 - \mat U_1 \mat V_1^\transpose)^{-1}$ as explained above, using the information available from $\M(G_1,v_1, v_2, v_3, p)$.

	\smallskip
	\subparagraph*{Obtaining the inverse of $\mat T$:}
	We define the matrix $\mat U$ as the matrix with three columns that has
	\begin{itemize}
		\item as first column the column of $v_1$ from $\mat T$, with the first entry set to $0$,
		\item as second column the column of $v_2$ from the Laplacian of $G$, with the second entry set to the degree of $v_2$ in $G$ minus $1$,
		\item as only non-zero entry in the third column a $1$ in the second entry.
	\end{itemize}
	The matrix $\mat V$ has the form  $\left[
	\begin{array}{c c c }
		1 &  0 &  0\\
		0 & 1 &  0\\
		\mat 0 & \mat 0 & \mat c
	\end{array}
	\right]$ , where $\mat c$ is the column of $v_2$ in the Laplacian of $G$ without the first two entries.
	It holds $\mat T = \mat R + \mat U \mat V^\transpose$ and we obtain the inverse $\mat T^{-1}$ via the SMW formula:
	\[ (\mat R + \mat U \mat V^\transpose)^{-1} = \mat R^{-1} - \mat R^{-1}\mat U{(\mat I+\mat V^\transpose \mat R^{-1}\mat U)}^{-1}\mat V^\transpose \mat R^{-1}\]
	
	The right-hand side expression can be evaluated in \FOar using available intermediate results: because $\mat R^{-1}$ is block diagonal, products involving this matrix can be computed based on its three blocks 	$\left[ \begin{array}{c c }
		1 &  0\\
		0 & 1
	\end{array}\right]$, $\mat S_1^{-1}$ and $\mat S_2^{-1}$.
	The computation requires as only non-trivial steps computing the products of $\mat S_1^{-1}$ and $\mat S_2^{-1}$ with (slightly adjusted) partial rows of the Laplacians of the subgraphs $G_1$ and $G_2$, which can be obtained as explained above.
	
	Analogously, one can compute the other parts of $\M(G,v_{i_1}, v_{i_3}, v_{i_4}, p)$.
\end{proof}

When the merging of two subgraphs is undone, the information $\M(\cdot)$ for the subgraphs can be recovered by reversing the operations performed in the proof of Lemma~\ref{lem:lina-merging}.

\begin{lemma}\label{lem:lina-splitting}
	Fix a domain of size $n$. Let $G$ be a graph with at most $n$ vertices and $G_1, G_2$ be subgraphs of $G$ that have exactly two vertices $v_{i_1}, v_{i_2}$ in common such that $G$ equals the union of $G_1$ and $G_2$ and some edge $(v_{i_3}, v_{i_4})$. Let $p$ be a prime of magnitude $\bigO(n^c)$, for some constant $c$, and let $(G, v_{i_1}, v_{i_3}, v_{i_4}, p)$ (or: $(G, v_{i_2}, v_{i_3}, v_{i_4}, p)$) be modulo embedding-inducing.
	
	Given  $\M(G,v_{i_1}, v_{i_3}, v_{i_4}, p)$ (or: $\M(G,v_{i_2}, v_{i_3}, v_{i_4}, p)$) one can determine in \FOar whether $(G_1, v_{i_1}, v_{i_2}, v_{i_3}, p)$ and $(G_2, v_{i_1}, v_{i_2}, v_{i_4}, p)$ are also modulo embedding-inducing and if so, $\M(G_1,v_{i_1}, v_{i_2}, v_{i_3}, p)$ and $\M(G_2,v_{i_1}, v_{i_2}, v_{i_4}, p)$ can be defined.
\end{lemma}

\subsubsection*{Union of two graphs}

A similar operation as merging two graphs along two vertices is the union of two graphs with no vertices in common.

\begin{lemma}\label{lem:lina-union}
	Fix a domain of size $n$. Let $G_1, G_2$ be disjoint $3$-connected planar graphs with at most $n$ vertices combined. Let $v_{i_1},v_{i_2},v_{i_3}$ be vertices of $G_1$ and $v_{i_4}, v_{i_5}, v_{i_6}$ vertices of $G_2$ that each are on the same face of the respective graphs. 
	Let $p$ be a prime of magnitude $\bigO(n^c)$, for some constant $c$, such that $(G_1, v_{i_1}, v_{i_2}, v_{i_3}, p)$ and $(G_2, v_{i_4}, v_{i_5}, v_{i_6}, p)$ are modulo embedding-inducing.
	Let $G$ be the graph that results from the disjoint union of $G_1$ and $G_2$ by adding the edges $(v_{i_1}, v_{i_4}), (v_{i_2}, v_{i_5}),(v_{i_3}, v_{i_6})$. 
	Then the following holds.
	\begin{enumerate}[(a)]
		\item $G$ is a $3$-connected planar graph and $v_{i_1}, v_{i_4}, v_{i_6}$ are on a common face, and
		\item there is a constant number of embedding-inducing $(H_j, v^j_1, v^j_2, v^j_3)_{j \leq d}$, where the $H_j$ are $3$-connected planar graphs and $v^j_1, v^j_2, v^j_3$ are vertices that are on the same face of $H_j$, such that given $\M(G_1,v_{i_1}, v_{i_2}, v_{i_3}, p)$  and $\M(G_2,v_{i_4}, v_{i_5}, v_{i_6}, p)$ one can determine in \FOar whether all of $(G, v_{i_1}, v_{i_4}, v_{i_6}, p)$ and $(H_j, v^j_1, v^j_2, v^j_3, p)_{j \leq d}$ are modulo embedding-inducing and if so, $\M(G,v_{i_1}, v_{i_4}, v_{i_6}, p)$ can be defined.
	\end{enumerate}
\end{lemma}
\begin{proof}
	For Part (a), observe that choosing the face of $G_1$ that contains $v_{i_1},v_{i_2},v_{i_3}$ as an internal face and the face of $G_2$ that contains $v_{i_4}, v_{i_5}, v_{i_6}$ as the outer face such that the vertices occur in this order on their respective faces yields a planar embedding for $G$ when drawing $G_2$ inside the face defined by $v_{i_1},v_{i_2},v_{i_3}$. Also observe that $v_{i_1}, v_{i_4}, v_{i_6}$ are on a common face in this embedding.
	
	We show Part (b).
	Let $H_1$ be the triangle that consists of the vertices $v_{i_2}, v_{i_3}$ and $v_{i_6}$ and $H_2$ the triangle that consists of $v_{i_2}, v_{i_5}$ and $v_{i_6}$. These graphs are clearly $3$-connected, planar and embedding-inducing, with their three vertices chosen as the pinned vertices. 
	As the graphs have constant size, $\M(H_1,v_{i_2}, v_{i_3}, v_{i_6}, p)$ and $\M(H_2,v_{i_2}, v_{i_5}, v_{i_6}, p)$ can be defined in $\FOar$, if they exist.
	
	Let $H_3$ be the graph that results from merging $G_1$ with $H_1$, adding the edge $(v_{i_1}, v_{i_6})$; and let $H_4$ be the graph that results from merging $G_2$ with $H_2$, adding the edge $(v_{i_2}, v_{i_4})$.
	Further, let $H_5$ be the graph obtained by merging $H_3$ and $H_4$ while adding the edge $(v_{i_1}, v_{i_4})$.
	Note that $G$ results from $H_5$ by deleting the edges $(v_{i_1}, v_{i_6}), (v_{i_2}, v_{i_4})$ and $(v_{i_2}, v_{i_6})$, as well as possibly $(v_{i_2}, v_{i_3})$ and $(v_{i_5}, v_{i_6})$ if they were not part of $G_1$ and $G_2$, respectively.
	
	So, assuming that every intermediate structure is modulo embedding-inducing, from $\M(H_1,v_{i_2}, v_{i_3}, v_{i_6}, p)$ and $\M(G_1,v_{i_1}, v_{i_2}, v_{i_3}, p)$ one can obtain $\M(H_3,v_{i_1}, v_{i_2}, v_{i_6}, p)$ using Lemma~\ref{lem:lina-merging}. Also, $\M(H_4,v_{i_2}, v_{i_4}, v_{i_6}, p)$ can be obtained from $\M(H_2,v_{i_2}, v_{i_5}, v_{i_6}, p)$ and $\M(G_2,v_{i_4}, v_{i_5}, v_{i_6}, p)$ via Lemma~\ref{lem:lina-merging}. 
	From these, one can obtain $\M(H_5,v_{i_1}, v_{i_4}, v_{i_6}, p)$ again via Lemma~\ref{lem:lina-merging} and then $\M(G,v_{i_1}, v_{i_4}, v_{i_6}, p)$ by applying Lemma~\ref{lem:lina-edges} three to five times.
	
	Note that all intermediate graphs are $3$-connected and planar and that the three distinguished vertices are on the same face in the respective graphs, so all intermediate structures are embedding-inducing, see Lemma~\ref{lem:lina-invertable}.
\end{proof}

\subsubsection*{Adding and removing a vertex}

The following lemma is basically a special case of Lemma~\ref{lem:lina-union}, where one graph is a single vertex.

\begin{lemma}\label{lem:lina-addvertex}
	Fix a domain of size $n$. Let $G$ be a $3$-connected planar graph with at most $n-1$ vertices and let $v_{i_1},v_{i_2},v_{i_3}$ be vertices that are on the same face. 
	Let $p$ be a prime of magnitude $\bigO(n^c)$, for some constant $c$, such that $(G, v_{i_1}, v_{i_2}, v_{i_3}, p)$ is modulo embedding-inducing.
	Let $G'$ be the graph that results from $G$ by adding a new vertex $v$ and the edges $(v, v_{i_1}), (v, v_{i_2}), (v, v_{i_3})$.
	Then the following holds.
	\begin{enumerate}[(a)]
		\item $G'$ is a $3$-connected planar graph and $v_{i_1}, v_{i_3}, v$ are on a common face, and
		\item there is a constant number of embedding-inducing $(H_j, v^j_1, v^j_2, v^j_3)_{j \leq d}$, where the $H_j$ are $3$-connected planar graphs and $v^j_1, v^j_2, v^j_3$ are vertices that are on the same face of $H_j$, such that given $\M(G,v_{i_1}, v_{i_2}, v_{i_3}, p)$ one can determine in \FOar whether all of $(G', v_{i_1}, v_{i_3}, v, p)$ and $(H_j, v^j_1, v^j_2, v^j_3, p)_{j \leq d}$ are modulo embedding-inducing and if so, $\M(G',v_{i_1}, v_{i_3}, v, p)$ can be defined.
	\end{enumerate}
\end{lemma}
\begin{proof}
	Part (a) is obvious. For Part (b), we argue similarly as in the proof of Lemma~\ref{lem:lina-union}.
	
	Let $H_1$ be the ($3$-connected and planar) triangle that consists of the vertices $v_{i_2}, v_{i_3}$ and $v$.
	As $H_1$ has constant size, $\M(H_1,v_{i_2}, v_{i_3}, v, p)$ can be defined in $\FOar$ if it exists.
	$G'$ results from merging $G$ with $H_1$ while adding the edge $(v_{i_1}, v)$.
	Using Lemma~\ref{lem:lina-merging}, $\M(G',v_{i_1}, v_{i_3}, v, p)$ can be obtained, if it (and $\M(H_1,v_{i_2}, v_{i_3}, v, p)$) exists.
\end{proof}

By reversing the operations in the proof of Lemma~\ref{lem:lina-addvertex} and using Lemma~\ref{lem:lina-splitting} instead of Lemma~\ref{lem:lina-merging}, one can prove the following lemma.

\begin{lemma}\label{lem:lina-removevertex}
	Fix a domain of size $n$. Let $G$ be a $3$-connected planar graph with at most $n$ vertices, among them the vertices $v_{i_1}, v_{i_2}, v_{i_3}, v$ such that $v$ has only the edges $(v, v_{i_1}), (v, v_{i_2}), (v, v_{i_3})$ and such that the graph $G'$ that is obtained from $G$ by removing $v$ and its edges is still $3$-connected and $v_{i_1}, v_{i_2}, v_{i_3}$ are on the same face.
	Let $p$ be a prime of magnitude $\bigO(n^c)$, for some constant $c$, such that $(G, v_{i_1}, v_{i_3}, v, p)$ is modulo embedding-inducing.
	
	Then there is a constant number of embedding-inducing $(H_j, v^j_1, v^j_2, v^j_3)_{j \leq d}$, where the $H_j$ are $3$-connected planar graphs and $v^j_1, v^j_2, v^j_3$ are vertices that are on the same face of $H_j$, such that given $\M(G,v_{i_1}, v_{i_3}, v, p)$ one can determine in \FOar whether all of $(G', v_{i_1}, v_{i_2}, v_{i_3}, p)$ and $(H_j, v^j_1, v^j_2, v^j_3, p)_{j \leq d}$ are modulo embedding-inducing and if so, $\M(G',v_{i_1}, v_{i_2}, v_{i_3}, p)$ can be defined.
\end{lemma}

\subsubsection*{Adding and removing two connected vertices}
At last we explain how to maintain the auxiliary information if the graph changes by adding or removing two vertices that are connected by an edge, each with two further edges to the existing graph.

\begin{lemma}\label{lem:lina-addpair}
Fix a domain of size $n$. Let $G$ be a $3$-connected planar graph with at most $n-2$ vertices and let $v_{i_1},v_{i_2},v_{i_3},v_{i_4}$ be vertices that in this order form an induced cycle in $G$. 
Let $p$ be a prime of magnitude $\bigO(n^c)$, for some constant $c$, such that $(G, v_{i_1}, v_{i_2}, v_{i_3}, p)$ is modulo embedding-inducing.
Let $G'$ be the graph that results from $G$ by adding new vertices $v, v'$ and the edges $(v,v'),(v, v_{i_1}), (v, v_{i_2}), (v', v_{i_3}), \allowbreak (v', v_{i_4})$.
Then the following holds, for each $1 \leq m \leq 3$:
\begin{enumerate}[(a)]
	\item $G'$ is a $3$-connected planar graph and $v_{i_m}, v, v'$ are on a common face, and
	\item there is a constant number of embedding-inducing $(H_j, v^j_1, v^j_2, v^j_3)_{j \leq d}$, where the $H_j$ are $3$-connected planar graphs and $v^j_1, v^j_2, v^j_3$ are vertices that are on the same face of $H_j$, such that given $\M(G,v_{i_1}, v_{i_2}, v_{i_3}, p)$ one can determine in \FOar whether all of $(G', v_{i_m}, v, v', p)$ and $(H_j, v^j_1, v^j_2, v^j_3, p)_{j \leq d}$ are modulo embedding-inducing and if so, $\M(G',v_{i_m}, v, v', p)$ can be defined.
\end{enumerate}	
\end{lemma}
\begin{proof}
Part (a) is again immediate.

Let $H_1$ and $H_2$ be the ($3$-connected and planar) triangles that consist of the vertices $v_{i_1}, v_{i_2}, v$ and $v_{i_3}, v_{i_4}, v'$, respectively.
Observe that $G'$ is obtained from $G$ through the following steps:
\begin{enumerate}
	\item obtain $H_3$ by merging $G$ with $H_1$ while adding the edge $(v,v_{i_3})$,
	\item obtain $H_4$ by merging $H_3$ with $H_2$ while adding the edge $(v,v')$, and
	\item removing the edge $(v,v_{i_3})$.
\end{enumerate}

So, assuming that every intermediate structure is modulo embedding-inducing, from $\M(H_1,v_{i_1}, v_{i_2}, v, p)$ (which can be expressed in $\FO$ since the graph has constant size) and $\M(G,v_{i_1}, v_{i_2}, v_{i_3}, p)$ one can obtain $\M(H_3,v_{i_1}, v_{i_3}, v, p)$ using Lemma~\ref{lem:lina-merging}, as well as $\M(H_3,v_{i_3}, v_{i_4}, v, p)$ via Lemma~\ref{lem:lina-pins}. 
Then, also from the available  $\M(H_2,v_{i_3}, v_{i_4}, v', p)$, one obtains $\M(H_4,v_{i_3}, v, v', p)$ via Lemma~\ref{lem:lina-merging} and finally $\M(G',v_{i_3}, v, v', p)$ via Lemma~\ref{lem:lina-edges}.
If necessary, Lemma~\ref{lem:lina-pins} can be applied again to choose a different pinned first vertex.

All intermediate graphs are $3$-connected and planar and the three distinguished vertices are on the same face in the respective graphs, so all intermediate structures are embedding-inducing following Lemma~\ref{lem:lina-invertable}.
\end{proof}

By reversing the operations in the proof of Lemma~\ref{lem:lina-addpair} one can prove the following lemma.

\begin{lemma}\label{lem:lina-removepair}
Fix a domain of size $n$. Let $G$ be a $3$-connected planar graph with at most $n$ vertices, among them the vertices $v_{i_1}, v_{i_2}, v_{i_3} , v_{i_4}, v, v'$ such that $v, v'$ only have the edges $(v,v'),(v, v_{i_1}), (v, v_{i_2}), (v', v_{i_3}), \allowbreak (v', v_{i_4})$ and such that the graph $G'$ that results from $G$ by removing $v, v'$ and their edges is still $3$-connected and $v_{i_1}, v_{i_2}, v_{i_3} , v_{i_4}$ in this order form an induced cycle.
Let $p$ be a prime of magnitude $\bigO(n^c)$, for some constant $c$, such that $(G, v_{i_m}, v, v', p)$ is modulo embedding-inducing, for some $1 \leq m \leq 3$.

Then there is a constant number of embedding-inducing $(H_j, v^j_1, v^j_2, v^j_3)_{j \leq d}$, where the $H_j$ are $3$-connected planar graphs and $v^j_1, v^j_2, v^j_3$ are vertices that are on the same face of $H_j$, such that given $\M(G,v_{i_m}, v, v', p)$ one can determine in \FOar whether all of $(G', v_{i_1}, v_{i_2}, v_{i_3}, p)$ and $(H_j, v^j_1, v^j_2, v^j_3, p)_{j \leq d}$ are modulo embedding-inducing and if so, $\M(G',v_{i_1}, v_{i_2}, v_{i_3}, p)$ can be defined.
\end{lemma}
 
\end{document}